%% file: main.tex
\documentclass[11pt]{article}
\usepackage{amsmath}
\usepackage{amssymb}

\usepackage{amsfonts}
\usepackage{color}
\usepackage[pdftex]{graphicx}
\usepackage[utf8]{inputenc}
\usepackage{subcaption}
\usepackage{array}
\usepackage{xcolor}
\usepackage[rightbars, color]{changebar}
\cbcolor{brown}
\usepackage[normalem]{ulem}
\newcommand{\rev}[1]{\textcolor{black}{#1}}
\DeclareMathOperator*{\argmin}{arg\,min}
\usepackage{setspace}
\usepackage[letterpaper, margin=1in]{geometry}
\usepackage{accents}

\usepackage{enumitem}
\usepackage{amsthm}
\usepackage{amsmath}
\usepackage{mathtools}
\usepackage{bbm}
\usepackage[longnamesfirst, authoryear]{natbib}
\usepackage{hyperref}
\hypersetup{hidelinks,colorlinks=true,linkcolor=blue,citecolor=blue,linktocpage=true}

\usepackage[title, titletoc]{appendix}

\usepackage{changepage}
\newtheorem{theorem}{Theorem}
\newtheorem{lemma}{Lemma}
\newtheorem{corollary}[theorem]{Corollary}

\usepackage{cancel}
\newcommand{\bounded}{O_p}
\newcommand{\fasterthan}{o_p}
\newcommand{\fasterthandet}{o}
\newcommand{\boundeddet}{O}
\DeclareMathOperator*{\inprob}{\stackrel{P}{\longrightarrow}}

\DeclareMathOperator*{\indist}{\stackrel{d}{\longrightarrow}}
\renewcommand{\b}[1]{\mathbf{#1}}

\newcommand{\s}[1]{\mathcal{#1}}
\renewcommand{\d}[1]{\mathbb{#1}}

\newcommand{\n}[1]{\mathrm{#1}}
\usepackage{authblk}

\title{Debiased inference for a covariate-adjusted regression function}
\author[1]{Kenta Takatsu}
\author[2]{Ted Westling}
\affil[1]{Department of Statistics and Data Science, Carnegie Mellon University}
\affil[2]{Department of Mathematics and Statistics, University of Massachusetts Amherst}
\date{}
\begin{document}
\maketitle
\begin{abstract}
In this article, we study nonparametric inference for a covariate-adjusted regression function. This parameter captures the average association between a continuous exposure and an outcome after adjusting for other covariates. Under certain causal conditions, it also corresponds to the average outcome had all units been assigned to a specific exposure level, known as the causal dose-response curve. We propose a debiased local linear estimator of the covariate-adjusted regression function and demonstrate that our estimator converges pointwise to a mean-zero normal limit distribution. We use this result to construct asymptotically valid confidence intervals for function values and differences thereof. In addition, we use approximation results for the distribution of the supremum of an empirical process to construct asymptotically valid uniform confidence bands. Our methods do not require undersmoothing, permit the use of data-adaptive estimators of nuisance functions, and our estimator attains the optimal rate of convergence for a twice differentiable regression function. We illustrate the practical performance of our estimator using numerical studies and an analysis of the effect of air pollution exposure on cardiovascular mortality.
\end{abstract}

\tableofcontents

\clearpage
\doublespacing
\input{paper/section1_intro}
\input{paper/section2_method}
\input{paper/section3_theory}
\input{paper/section5_sim}

\input{paper/section6_data}
\input{paper/section7_conclusion}
%\clearpage

\singlespacing
\section*{Acknowledgements}
The authors gratefully acknowledge support from the University of Massachusetts Amherst Department of Mathematics and Statistics startup fund and NSF Award 2113171.
\bibliographystyle{apalike}
\bibliography{ref.bib,supref.bib}

\begin{adjustwidth}{-.25in}{-.25in}
\begin{appendices}
\begin{small}
\singlespacing
\input{supp/notation}
\input{supp/conditions}

\input{supp/eif}

\input{supp/decomp}

\input{supp/theorem1}
\input{supp/CLT}

\input{supp/R1}
\input{supp/R2R3}

\input{supp/R4}
\input{supp/R5}
\input{supp/R6}

\input{supp/covar}
\input{supp/unif}
\input{supp/additional_sim}

\end{small}
\end{appendices}

\end{adjustwidth}
\end{document}

%% file: paper/section1_intro.tex
\section{Introduction}
\subsection{Motivation and literature review}
In this article, we study nonparametric inference for a covariate-adjusted regression function, which is also known as a G-computed regression function. This statistical problem arises in the context of observational studies where interest focuses on the causal effect of a continuous exposure, such as the dose of a drug, the amount of air pollution exposure, or the amount of a biochemical in the bloodstream. However, a covariate-adjusted regression function may also be of interest outside of causal contexts as a one-dimensional marginal summary of a multivariate regression function. Despite the simple formulation and many applications of this parameter, a method achieving valid nonparametric inference without undersmoothing is not yet available. In this paper, we close this gap by introducing a novel nonparametric, doubly-robust estimator, pointwise confidence intervals, and uniform confidence bands for the covariate-adjusted regression function that do not require undersmoothing.

% Causal motivation
One setting where a covariate-adjusted regression function arises is causal inference with a continuous exposure. %In many scientific settings, researchers would like to . 
The gold standard for assessing the causal effect of a treatment or exposure on an outcome is a randomized experiment, where units in the population are assigned values of the exposure by a random process known to the researchers. Frequently, however, such an experiment is infeasible, unethical, or cost-prohibitive. For example, it is unethical to purposefully expose people to a chemical or pollutant known to have negative health effects. In such cases, researchers may instead wish to assess the causal effect of the exposure using observational data in which the exposure varies according to an unknown mechanism. Recovering a causal effect with observational data is more challenging due to potential common causes of the exposure and outcome.  However, if all common causes are recorded in the data, then causal effects can be recovered  by appropriately adjusting for them. Specifically, the average outcome had all units been assigned a specific exposure value, \rev{which is known as the \emph{causal dose-response curve} or \emph{average dose-response function}}, coincides with the covariate-adjusted regression function, which is known as the G-formula or G-computation in causal inference \citep{robins1986new, gill2001}.   \rev{Adjusting for covariates can also improve estimator efficiency in the context of randomized experiments  \citep{imbens2015causal}.}

The covariate-adjusted regression function is also of interest outside of causal contexts. The regression (i.e., the conditional expectation) of an outcome on a vector of covariates can be difficult to visualize and summarize when there are more than two regressors and when nonparametric methods are used. The covariate-adjusted regression function summarizes the adjusted association between a single continuous covariate and the outcome by averaging the regression function over all other covariates for each value of the covariate of interest. This is related to the use of \emph{marginal effects} to summarize the results of nonlinear regression models \citep{mize2019marginal}. \rev{Furthermore, since the covariate-adjusted regression is a univariate function, it can serve as a useful visualization tool  \citep{friedman2001greedy, apley2020visualizing, cattaneo2019binscatter}}. 

\rev{The covariate-adjusted regression function has been used in several recent observational studies to describe the association between a continuous exposure and an outcome after adjusting for potential confounders. \cite{josey2023air} and \cite{schwartz2023air} assessed the association between air pollution exposure and health outcomes after adjusting for socioeconomic and demographic factors. \citet{knaus2021double} assessed the association between time spent playing a musical instrument and cognitive improvement in youth after adjusting for socioeconomic factors. \cite{oulhote2019joint} assessed the association between exposure to chemicals and pollutants and neurodevelopment  in children after adjusting for sociodemographic and lifestyle factors and medical history. \cite{colangelo2020double} assessed the association between hours of job training and subsequent employment after adjusting for socioeconomic and health factors. \citet{shroff2022pretrial} assessed the association between the timing of arraignments and judicial decisions after adjusting for defendant, charge, and courtroom characteristics. \cite{Weng2022ContTrtTest} estimated the association between average nurse staffing on hospital readmission rates after adjusting for hospital characteristics. As these examples illustrate, the analysis of continuous exposures is of significant statistical interest across a wide range of disciplines.}

Several nonparametric methods exist for estimation and inference for the covariate-adjusted regression function. \cite{Kennedy2016ptwisecte} proposed an estimator based on local linear smoothing. Theirs was the first \emph{doubly robust} estimator for this parameter, meaning that their estimator is consistent if either of two nuisance estimators is consistent. \cite{vanderLaan2018} and \cite{colangelo2020double} also proposed doubly-robust estimators based on smoothing, and used cross-fitting to remove empirical process conditions for nuisance estimators. \cite{semenova2021debiased} proposed an estimator based on a series expansion. Finally, \cite{westling2020unified} and \cite{westling2018causal} explored inference under a monotonicity assumption. Additional literature related to the covariated-adjusted regression function includes \cite{robins1986new, robins2000msm} and \cite{zhang2016ci}, who proposed plug-in estimators based on a parametric outcome regression model; \cite{hirano2004propensity} and \cite{imai2004causal}, who proposed estimators based on a parametric propensity model;  \cite{NEUGEBAUER2007419}, who studied inference for the best projection onto a working parametric model; \citet{rubin2006extending} and \cite{diaz2013tmledoseresponse}, who proposed data-adaptive methods; \cite{bonvini2022fast}, who proposed a higher-order estimator; and \citet{westling2021causalnull} and \citet{Weng2022ContTrtTest}, who proposed tests of the null hypothesis that the function is flat.

As is the case for an ordinary regression function, nonparametric inference for the covariate-adjusted regression function is a challenging task. The bias of smoothing-based methods is not asymptotically negligible when the bandwidth is chosen to optimize the rate of convergence of the estimator. This bias complicates the task of obtaining valid inference. Some authors have argued for interpreting confidence sets constructed based on the resulting limit theorems as valid for a \textit{smoothed} parameter \citep{wasserman2006all, Kennedy2016ptwisecte}. An alternative approach is to choose the bandwidth to go to zero faster than the optimal rate to guarantee that the bias goes to zero faster than the standard deviation, which is called \textit{undersmoothing} \citep{vanderLaan2018, colangelo2020double, semenova2021debiased}. While this approach theoretically allows valid inference, it yields a suboptimal rate of convergence for the estimator. Furthermore, there is little guidance about how to select a bandwidth for undersmoothing in practice beyond ad-hoc methods. \rev{For instance, a common approach to undersmoothing is to divide the bandwidth selected by cross-validation or another method by a sequence going to infinity slowly with $n$, such as $n^{1/10}$, $\log(n)$, or $\log(\log(n))$. However, the specific choice of this sequence impacts the finite-sample performance of the estimator and confidence intervals, and there is no consensus on which sequence to use in any given situation. It is therefore both theoretically and practically valuable to develop asymptotically valid inference procedures that do not require undersmoothing.}  

Recently, \citet{Calonico2018} proposed a method of  bias correction for density and regression functions estimated using kernel smoothing. Bias correction in the context of kernel smoothing is challenging because the bias depends on the second or higher derivative of the true function, which is more difficult to estimate than the function itself \citep{wasserman2006all, hall1992effect}. However, \citet{Calonico2018} demonstrated that, via a careful choice of the bandwidth parameter of the bias estimator, it is possible to effectively debias kernel smoothing-based estimators. These estimators permit asymptotically valid inference without undersmoothing, and unlike undersmoothing, retain the optimal rate of convergence relative to the assumed smoothness of the true function. 

\subsection{Contribution and organization of the article}
In this article, we contribute to the existing literature in the following ways: (1) we propose a novel debiased estimator of the covariate-adjusted regression function motivated by the approach proposed by \citet{Calonico2018}; (2) we propose methods of pointwise and uniform inference and provide conditions under which our methods are asymptotically valid; and (3) we illustrate the practical performance of our proposed methods using numerical studies and an analysis of the effect of air quality on health. To the best of our knowledge, ours are the first asymptotically valid methods of pointwise and uniform inference for the covariate-adjusted regression function without undersmoothing. We note that our problem is substantively different from that of \citet{Calonico2018} due to the presence of nuisance parameters, which introduces remainder terms and technical considerations not present when estimating density and regression functions. We elucidate these differences more below. \rev{Finally, we have implemented all methods proposed in this article in an \texttt{R} package available at https://github.com/Kenta426/DebiasedDoseResponse.}

The remainder of this article is organized as follows. In Section~\ref{section:method}, we define our statistical setting and proposed estimator. In Section~\ref{section:inference}, we present our approach to inference and our theoretical results. In Section~\ref{section:numerical}, we demonstrate the empirical performance of our proposed methods using numerical studies, and in Section~\ref{section:application}, we use our methods to assess the effect of air pollution on cardiovascular mortality.  Section~\ref{section:conclusion} presents brief concluding remarks.  Proofs of all theorems are provided in Supplementary Material. 

%% file: paper/section2_method.tex
\section{Proposed estimator}\label{section:method}

\subsection{Notation and statistical setting}

We now define the statistical setting we will work in and notation we will use. We consider a univariate outcome $Y \in \s{Y} \subseteq \d{R}$, a univariate exposure $A \in \s{A} \subseteq \d{R}$, and a vector of covariates $W \in \s{W} \subseteq \d{R}^d$. We assume that the distribution of $A$ possesses a Lebesgue density. We define the observed data unit $O \coloneqq (Y,A,W)$, which takes values in $\s{O} \coloneqq \s{Y} \times \s{A} \times \s{W}$. We then observe $n$ independent and identically distributed observations $O_1, \dots, O_n$ from an unknown probability distribution $P_0$. We denote by $\d{P}_n$ the empirical distribution function of the observed data. We index objects by $P$ when they depend on a generic distribution $P$ over the observed data unit $O$, and we use subscript $0$ as short-hand for the true distribution $P_0$. For instance, we denote the expectation under $P_0$ by $E_0$. For a distribution $P$ on $\s{O}$, we denote by $F_P$ the marginal distribution and $f_P$ the Lebesgue density of $A$  under $P$. We denote by $Q_P$ the marginal distribution of $W$ under $P$. We denote by $F_n$ and $Q_n$ the empirical distributions of $A_1, \dotsc, A_n$ and $W_1, \dotsc, W_n$, respectively. We let $\mu_P(a,w) \coloneqq E_P(Y \mid A =a, W=w)$ denote the \emph{outcome regression}  function and $g_P(a,w) \coloneqq \left[ \frac{\partial}{\partial a} P(A \leq a \mid W = w) \right] / f_P(a)$ denote the \emph{standardized propensity} function.  For a probability measure $P$ and $P$-integrable function $h$, we define $Ph \coloneqq \int h \, dP$. For $q \geq 1$, we denote by $\|h\|_{P, q} \coloneqq (P|h|^q)^{1/q}$ the $L_q(P)$ norm of $h$. \rev{For a real-valued function $h$ defined on $\mathcal{X}$, its supremum norm is denoted by $\|h\|_\infty := \sup_{x \in \mathcal{X}}|h(x)|$. We also define $e_1 := (1,0)$ and $e_3 := (0,0,1)$}.

\subsection{Parameter of interest and its interpretation}

Our parameter of interest is the covariate-adjusted regression function $a\mapsto \theta_0(a)$ defined as
\[\theta_0(a) \coloneqq E_0\left\{E_0(Y \mid A=a, W)\right\} = E_0 \left\{ \mu_0(a, W) \right\} = \int \mu_0(a,w) \, dQ_0(w).\]
Under certain conditions, $\theta_0$ has a causal interpretation. \rev{For each $a \in \s{A}$, we define $Y(a)$ as the potential outcome under an intervention that sets the exposure $A$ to $a$. If (1) the potential outcomes of each unit are unaffected by the exposures of all other units, (2) the observed outcome $Y$ equals $Y(A)$, i.e., the potential outcome under assignment to the observed exposure $A$, (3) $Y(a)$ and $A$ are conditionally independent given $W$, and (4) $g_0(a,W)$ is almost surely positive, then $\theta_0(a) = E_0[Y(a)]$. Hence, under causal conditions (1)--(4), $\theta_0(a)$ can be interpreted as the average potential outcome under the assignment of the entire population to exposure value $a$. The function $a \mapsto E_0[Y(a)]$ is called the \emph{causal dose-response curve} or the \emph{average dose-response function}. This causal identification result has been employed in prior work on causal inference with continuous exposures (e.g.,\ \citealp{robins1986new, gill2001, Kennedy2016ptwisecte,westling2018causal,westling2021causalnull}). These assumptions cannot be verified or tested using the observed data, so their plausibility depends on the particular scientific application.}

As mentioned in the introduction, $\theta_0$ is also of interest outside of causal settings. The outcome regression function $\mu_0$ is the expected outcome given exposure and covariates. Hence, for fixed $a$, $\mu_0(a,W)$ is a random variable representing the expected outcome value given $A=a$ across the distribution of the covariates $W$ in the population, and $\theta_0(a)$ is the mean of this variable. The curve $a \mapsto \theta_0(a)$ depicts how this average conditional mean changes with $a$. Hence, $\theta_0$ is a marginal summary of the multivariate regression function $\mu_0$. \rev{Therefore, obtaining nonparametric inference for $\theta_0$ is a relevant statistical problem even when the causal conditions listed above are implausible or a causal interpretation is not of interest.}

\subsection{Debiased local linear estimator}

We now define our estimator of the covariate-adjusted regression function. We begin with a review of the local linear method proposed by \cite{Kennedy2016ptwisecte} and its key properties.  For an outcome regression function $\mu$, a standardized propensity $g$, and a covariate distribution $Q$, we define the \emph{pseudo-outcome} mapping
\begin{equation}
    \xi_{\mu,g,Q}: (y,a,w) \mapsto \frac{y - \mu(a,w)}{g(a,w)} + \int \mu(a, \bar{w})\, dQ(\bar{w}). \label{eq:pseudo_outcome}
\end{equation}
Theorem~1 of \cite{Kennedy2016ptwisecte} showed that this mapping possesses a double-robust property: it holds that $E_0[\xi_{\mu, g, Q_0}(Y,A,W) \mid A=a_0] = \theta_0(a_0)$ if either $\mu=\mu_0$ or $g=g_0$. \cite{Kennedy2016ptwisecte} thus proposed first constructing estimators $\mu_n$ and $g_n$ of $\mu_0$ and $g_0$, respectively, using the empirical distribution $Q_n$ of the observed covariates to estimate $Q_0$, then regressing the estimated pseudo-outcomes $\xi_{n}(Y_1,A_1,W_1), \dotsc, \xi_{n}(Y_n,A_n,W_n)$ on the observed exposures $A_1, \dotsc, A_n$ using local linear regression, where $\xi_{n} \coloneqq \xi_{\mu_n, g_n, Q_n}$. Specifically, let $K$ be a kernel function (i.e., a symmetric density on $\d{R}$), $h > 0$ be a bandwidth, and $K_{h,a_0}(a) \coloneqq K([a - a_0]/h) / h$. The local linear estimator at a point $a_0$ is then defined as the evaluation at $a_0$ of the weighted ordinary least squares regression of the estimated pseudo-outcomes on intercept and the observed exposures with weights $K_{h,a_0}(A_1), \dotsc, K_{h,a_0}(A_n)$. Mathematically, the local linear estimator can be written as
\begin{align}
    \theta_{n}^{LL}(a_0) \coloneqq  e_1^T \b{D}_{n,h,a_0,1}^{-1} \d{P}_n\left(w_{h,a_0,1} K_{h,a_0} \xi_{n}\right),\label{eq:LL_linear_representation}
\end{align}
where $w_{h,a_0,j}(a) \coloneqq \left(1, [a - a_0] /h, \dotsc, [a-a_0]^j / h^j \right)^T$ and $\b{D}_{n,h,a_0, j} \coloneqq \d{P}_n(w_{h,a_0, j} w^T_{h,a_0, j} K_{h, a_0} )$ for any integer $j \geq 1$.  Other authors have used a similar approach as \cite{Kennedy2016ptwisecte}, but replaced the local linear regression step with an alternative nonparametric regression estimator \citep{westling2018causal,semenova2021debiased,bonvini2022fast}.

As discussed in the introduction, standard approaches to nonparametric regression, including local linear regression, do not yield valid inference when the bandwidth is chosen to minimize mean squared error because the bias of the resulting estimator is of the same order as its standard deviation. Specifically, in their Theorem~3, \cite{Kennedy2016ptwisecte} showed that under suitable conditions, \rev{including that $a_0$ is in the interior of the support of $A$}, $(nh)^{1/2} [ \theta_{n}^{LL}(a_0) - \theta_0(a_0) - h^2 c_2 \theta_0''(a_0) /2]$ converges to a mean-zero normal distribution, where $c_2 \coloneqq \int u^2 K(u)\,du$ and $\theta_0''$ is the second derivative of $\theta_0$. To minimize mean squared error, the bandwidth $h$ should be chosen to balance bias squared and variance, which means choosing $h$ such that $(nh)^{1/2} h^2 = (nh^5)^{1/2}$ converges to a positive, finite constant, or equivalently $h$ proportional to $n^{-1/5}$. Hence, if $\theta_0''(a_0) \neq 0$, then $(nh)^{1/2} \big[ \theta_{n}^{LL}(a_0) - \theta_0(a_0)\big]$ converges to a normal distribution with non-zero mean, implying that confidence intervals centered around $\theta_{n}^{LL}(a_0)$ will be asymptotically anti-conservative.

We propose debiasing the local linear estimator by subtracting an estimator of the bias in the spirit of \cite{Calonico2018}. \rev{We define the debiased estimator as $\theta_{n}^{DB}(a_0) := \theta_{n}^{LL}(a_0) - \tfrac{1}{2}h^2 c_{n,a_0,2} \theta_{n}''(a_0)$, where $c_{n,a_0,2} := e_1^T\b{D}_{n,h,a_0,1}^{-1} \d{P}_n ( \tilde{w}_{h,a_0,1} K_{h,a_0})$ for $\tilde{w}_{h,a_0,1}(a) := w_{h,a_0,1}(a) [(a - a_0) /h ]^2$, and $\theta_{n}''(a_0)$ is a second derivative estimator based on a local quadratic regression with bandwidth $b > 0$. We use $c_{n,a_0,2}$ rather than $c_2$ for proper debiasing on and near the boundary of the support of $A$, since the limiting constant $c_2$ is different on the boundary than the interior. The local quadratic estimator $\theta_{n}''(a_0)$ is the second derivative at $a_0$ of the weighted linear least squares regression of $\xi_n(Y_1, A_1, W_1), \dotsc, \xi_n(Y_n, A_n, W_n)$ on intercept, $A_1, \dotsc, A_n$, and $A_1^2, \dotsc, A_n^2$  with weights $K_{b,a_0}(A_1), \dotsc, K_{b,a_0}(A_n)$. Mathematically, $\theta_{n}''(a_0) := 2b^{-2}e_3^T  \b{D}_{n,b,a_0,2}^{-1} \d{P}_n ( w_{b,a_0,2} K_{b,a_0} \xi_n)$. We can write $\theta_{n}^{DB}(a_0)  = \d{P}_n(\Gamma_{n,a_0} \xi_{n})$, where
 \begin{align}
    \Gamma_{n,a_0}(a) &\coloneqq e_1^T \b{D}_{n,h,a_0,1}^{-1} w_{h,a_0,1}(a) K_{h,a_0}(a) - e_3^T c_{n,a_0,2} (h/b)^2 \b{D}_{n,b,a_0,2}^{-1} w_{b,a_0,2}(a) K_{b,a_0}(a).\label{eq:def_gamma_n}
\end{align}}
To summarize, for given bandwidths $h$ and $b$ and kernel $K$, our proposed estimator is constructed in two steps: (1) construct estimators $\mu_n$ and $g_n$ of $\mu_0$ and $g_0$ respectively, and (2) compute the plug-in estimates of pseudo-outcomes $\xi_n(Y_i, A_i, W_i) = \xi_{\mu_n, g_n, Q_n}(Y_i,A_i,W_i)$ for $i=1, \dotsc,n$,  and regress them on the observed exposures using the bias-corrected local linear estimator \rev{$\theta_{n}^{DB}(a_0) = \d{P}_n(\Gamma_{n,a_0} \xi_{n})$} for each $a_0$. In Section~\ref{section:inference}, we provide conditions on the true data-generating mechanism and on $h$, $b$, $K$, $\mu_n$, and $g_n$, as well as practical guidance for selecting or estimating these quantities.

It may seem that \rev{effective debiasing using an estimator of the second derivative} would require additional smoothness of $\theta_0$, hence violating our stated goal of obtaining the optimal rate of convergence relative to the assumed smoothness of $\theta_0$. This is not the case. As an intuitive explanation, we decompose $(nh)^{1/2}[ \theta_{n}^{DB}(a_0) - \theta_0(a_0)]$ as
\[(nh)^{1/2} \left[ \theta_{n}^{LL}(a_0) - \theta_0(a_0) - \tfrac{1}{2} h^2 c_{n,a_0,2} \theta_0''(a_0)\right] + \tfrac{1}{2} c_{n,a_0,2}(nh^5)^{1/2}\left[ \theta_{n}''(a_0) - \theta_0''(a_0)\right].\]
Under regularity conditions, the first term converges in distribution to a mean-zero limit. Thus, if $nh^5 = \boundeddet(1)$ and $\theta_{n}''(a_0) \inprob \theta_0''(a_0)$, then $(nh)^{1/2} [ \theta_{n}^{DB}(a_0) - \theta_0(a_0) ]$ converges to this same limit. However, perhaps surprisingly, $\theta_{n}''(a_0)$ need not be consistent for $\theta_0''(a_0)$ to achieve good inference using $\theta_{n}(a_0)$. \rev{We show that the variance of $\theta_{n}''(a_0)$ is proportional to $(nb^5)^{-1}$. Hence, if $h/b \longrightarrow \tau > 0$, then the variance of $h^2\theta_{n}''(a_0)$ is of the same order as that of $\theta_{n}^{LL}(a_0)$, so the bias correction \emph{will not} be asymptotically negligible; it will contribute to the asymptotic variance of the estimator as in \cite{Calonico2018}. If in addition $h \propto n^{-1/5}$, then the variance of $\theta_{n}''(a_0)$ does not go to zero, so it is not consistent. However, even in these cases, we show that $(nh)^{1/2} [ \theta_{n}^{DB}(a_0) - \theta_0(a_0) ]$ converges to a mean-zero limit distribution.}

Even if the bias correction is asymptotically negligible, and especially if it is not, accounting for its finite-sample variability is important for achieving good finite-sample inference. As discussed more below,  our variance estimator will account for the variability of $\theta_{n}''(a_0)$, and in particular, we can still achieve valid inference when the variance of $\theta_{n}''(a_0)$ is not going to zero.  We will show that $\theta_{n}''(a_0) - \theta_0''(a_0)$ converges to a mean-zero limit if $b \longrightarrow 0$ and $\theta_{0}''$ is continuous at $a_0$ (and additional technical conditions unrelated to smoothness of $\theta_0$ hold). 

\subsection{Local parameter and its efficient influence function}

We now provide an alternative motivation for our proposed estimator, which also motivates our approach to inference. We show that $\theta_{n}^{DB}(a_0)$ can be considered as a one-step estimator of a debiased smoothed parameter. \rev{We recall that an object indexed by the subscript $P$ indicates the evaluation under a generic probability distribution $P$ in the model, and the subscript $0$ indicates the true value of the object, i.e., evaluated at the true data-generating distribution $P_0$. Hence, $\theta_0(a_0)$ is the evaluation of $\theta_P(a_0)$ at $P = P_0$.}

The parameter mapping $P \mapsto \theta_P(a_0)$ is not \emph{pathwise differentiable} relative to a nonparametric model, meaning that it is not smooth enough as a function of the distribution $P$ to permit $n^{-1/2}$-rate estimation \citep{bickel1982adaptive, pfanzagl1985contributions, bickel1993efficient}. One way to develop inference methods for such a parameter is to consider a sequence of smoothed parameters approaching the parameter of interest, each of which is pathwise differentiable \citep{vanderLaan2018}. Our debiased local linear estimator can be viewed through this lens. For any \rev{distribution $P$ and an} integer $j \geq 1$, we define $\b{D}_{P, h, a_0, j} \coloneqq P( w_{h, a_0, j} w_{h,a_0, j}^T K_{h,a_0})$. We then define the \emph{debiased smoothed parameter mapping} as $\rev{P \mapsto  \theta_{P}^{DB}(a_0) \coloneqq P (\Gamma_{P,a_0} \theta_P) =\int \Gamma_{P,a_0}(a) \theta_P(a) \, dF_P(a)}$, where
\rev{\begin{align}
    \Gamma_{P,a_0}(a)  \coloneqq e_1^T \b{D}_{P,h,a_0,1}^{-1} w_{h,a_0,1}(a) K_{h,a_0}(a) - e_3^T c_{P,h,a_0,2} (h/b)^2 \b{D}_{P,b,a_0,2}^{-1} w_{b,a_0,2}(a) K_{b,a_0}(a),\label{eq:def_gamma_P}
\end{align}}
\rev{for $c_{P,h,a_0,2} := e_1 \b{D}_{P,h,a_0,1}^{-1} P( \tilde{w}_{h,a_0,1} K_{h,a_0})$. We refer to $\theta_{P}^{DB}(a_0)$ as \emph{smoothed} because it is a weighted average of $\theta_P(a)$ for $a$ a neighborhood of $a_0$, with (possibly negative) weights  $\Gamma_{P,a_0}(a) f_P(a)$. We refer to $\theta_P^{DB}(a_0)$ as \emph{debiased} because, as we show in Supplementary Material, if $\theta_P$ is twice continuously differentiable in a neighborhood of $a_0$, $h/b \longrightarrow \tau \in [0,\infty)$, and additional mild conditions hold, then the smoothing bias satisfies $\theta_{P}^{DB}(a_0) - \theta_P(a_0) = \fasterthandet(h^2)$ as $h \longrightarrow 0$. } The asymptotic properties of $\theta_{n}^{DB}(a_0)$ can now be understood through the following decomposition:
\[(nh)^{1/2}\left[\theta_{n}^{DB}(a_0) - \theta_0(a_0)\right] =(nh)^{1/2}\left[\theta_{n}^{DB}(a_0) - \theta_{0}^{DB}(a_0)\right] + (nh)^{1/2}\left[\theta_{0}^{DB}(a_0) - \theta_0(a_0) \right].  \]
Hence, if $nh^5 = \boundeddet(1)$, then $\theta_{0}^{DB}(a_0) - \theta_0(a_0) = \fasterthandet(\{nh\}^{-1/2})$, and so the first-order asymptotic properties of $(nh)^{1/2}\left[\theta_{n}^{DB}(a_0) - \theta_0(a_0)\right]$ are determined by $(nh)^{1/2}\left[\theta_{n}^{DB}(a_0) - \theta_{0}^{DB}(a_0)\right]$. This expression can be studied using semiparametric efficiency theory. The first step in doing so is to derive the efficient influence function of \rev{the functional $P \mapsto \theta_{P}^{DB}(a_0)$}. This is the subject of the following lemma. 
\begin{lemma}[Efficient influence function] 
\label{lm:influence_function}
For each $h$ and $b > 0$ and $a_0 \in \s{A}$, $P \mapsto \theta_{P}^{DB}(a_0)$ is a pathwise differentiable parameter with respect to the model $\s{M}$ consisting of $P$ such that (1) $E_P[Y^2] < \infty$ and (2) there exists $\kappa > 0$ such that $g_P(a,w) \geq \kappa$ for $F_P$-a.e.\ all $a$ such that $|a - a_0| \leq \max\{h,b\}$ and $Q_P$-a.e.\ $w$, and the efficient influence function of \rev{$\theta_{P}^{DB}(a_0)$} relative to this model is 
\begin{align}
    \phi^*_{P, a_0}&=\phi^*_{P, h,b, a_0}: (y,a,w) \mapsto \Gamma_{P, a_0}(a) \xi_{\mu_P,g_P,Q_P}(y,a,w)-\gamma_{P,a_0}(a) \nonumber\\
    &\qquad+ \int \Gamma_{P,a_0}(\bar{a})\left\{\mu_P(\bar{a},w)-\theta_P(\bar{a})\right\}\,dF_P(\bar{a}) , \text{ where}\label{eq:eif_phi} \\
    \gamma_{P,a_0}(a) &:= e_1^T\b{D}^{-1}_{P, h, a_0,1}w_{h,a_0, 1}(a) K_{h, a_0}(a)w^T_{h,a_0, 1}(a)\b{D}^{-1}_{P, h, a_0,1} P \left( w_{h,a_0, 1} K_{h, a_0} \theta_P\right) \nonumber\\
    &\qquad- e_3^T c_{P,h,a_0,2} (h/b)^2 \b{D}_{P, b,a_0,2}^{-1} w_{b,a_0,2}(a) K_{b,a_0}(a)w^T_{b,a_0, 2}(a)\b{D}^{-1}_{P, b, a_0,2} P \left( w_{b,a_0, 2} K_{b, a_0} \theta_P\right)\nonumber \\
    &\qquad- (h/b)^2 e_1^T\b{D}^{-1}_{P,h, a_0,1}\left[ \tilde{w}_{h, a_0, 1}(a)  - w_{h,a_0, 1}(a) w^T_{h,a_0, 1}(a) \b{D}^{-1}_{P,h, a_0,1} P\left( \tilde{w}_{h, a_0, 1} K_{h,a_0} \right)\right] K_{h,a_0}(a) \nonumber \\
        &\qquad\qquad \times e_3^T\b{D}_{P, b,a_0,2}^{-1} P( w_{b,a_0,2} K_{b,a_0} \theta_P) \nonumber.
\end{align}
\end{lemma}
The proof of Lemma~\ref{lm:influence_function} and all other results are provided in Supplementary Material. \rev{Lemma~\ref{lm:influence_function} establishes that the smoothed and debiased parameter mapping $P \mapsto \theta_{P}^{DB}(a_0)$ is pathwise differentiable relative to a nonparametric model. We will use this result below to derive the asymptotic properties of $\theta_{n}^{DB}(a_0)$ and to construct a variance estimator for $\theta_{n}^{DB}(a_0)$.} 

\rev{We note that there may be other smoothed parameters with the same properties as $\theta_{P}^{DB}(a_0)$ that yield an asymptotically mean-zero distribution without undersmoothing---namely, that the parameter mapping is pathwise differentiable for a fixed bandwidth, and that the approximation bias as the bandwidth tends to zero is negligible. Furthermore, other smoothed parameters may result in different asymptotic variances of the resulting estimator of $\theta_0(a_0)$. In particular, Theorem~\ref{thm:pointwise} below demonstrates that the asymptotic variance of our estimator depends on the kernel function $K$ and the ratio $h/b$. Hence, the selection of the precise smoothed parameter impacts the asymptotic variance of the estimator. To the best of our knowledge, there is no precise characterization of the optimality of such approximations. This is an important area of future research.}

\rev{A simple and popular method of estimating a pathwise differentiable parameter is the so-called \emph{one-step construction} \citep{bickel1982adaptive, pfanzagl1982contributions}. For clarity of exposition, we now briefly describe the one-step construction in a general setting. Suppose $\psi : \s{M} \to \d{R}$ is a real-valued functional on a model $\s{M}$ that is pathwise differentiable relative to $\s{M}$ at the true data-generating distribution $P_0$, and it has efficient influence function $\eta_0^*$. Suppose $O_1, \dotsc, O_n$ are drawn IID from $P_0$, $P_n$ is an estimator of $P_0$ based on $O_1, \dotsc, O_n$, $\eta_{P_n}^*$ is the efficient influence function evaluated at $P_n$, and that $\d{P}_n$ is the empirical distribution of $O_1, \dotsc, O_n$. The one-step estimator of $\psi_0$ is then defined as $\psi_n := \psi(P_n) + \d{P}_n \phi_{P_n}^*$, which can be viewed as the plug-in estimator $\psi(P_n)$ plus a term that corrects some of the bias of $\psi(P_n)$. The one-step estimator can be shown to be asymptotically linear with influence function $\eta_0^*$ under conditions on $P_n$ and the true distribution $P_0$.}

\rev{We now demonstrate that the debiased local linear estimator $\theta_{n}^{DB}(a_0)$ defined above can be viewed as a one-step estimator of  $\theta_{0}^{DB}(a_0)$. We define $\phi_{n,a_0}^\circ$ as the estimated efficient influence function obtained by replacing $\mu_P$ and $g_P$ from Lemma~\ref{lm:influence_function} with estimators $\mu_n$ and $g_n$, $Q_P$ and $F_P$ with the empirical distributions $Q_n$ and $F_n$, and $\b{D}_{P, h, a_0,j}$  with $\b{D}_{n, h, a_0,j}$. Due to the cancellation of terms, it holds that
\begin{equation}
    \theta_{n}^{DB}(a_0) = \iint \Gamma_{n,h,h,a_0}(a) \mu_n(a, w) \, dQ_n(w) \, dF_n(a)  + \d{P}_n \phi_{n,a_0}^\circ\label{eq:onestep}.
\end{equation}
This representation is proved in Lemma 3 in Supplementary Material. Since the first term is the plug-in estimator of $\theta_{0}^{DB}(a_0)$ and $\phi_{n,a_0}^\circ$ is the plug-in estimator of the efficient influence function of $\theta_{0}^{DB}(a_0)$ established in Lemma~\ref{lm:influence_function}, \eqref{eq:onestep} represents $\theta_{n}^{DB}(a_0)$ as a one-step  estimator.}

The representation of our debiased estimator as a one-step estimator of a smoothed parameter plays an important role in motivating our approach to inference and our asymptotic results. We would typically expect that one-step estimators of pathwise differentiable parameters are asymptotically linear under appropriate conditions. Similarly, in Theorem~\ref{thm:pointwise}, we will see that a finite-sample version of asymptotic linearity holds for our estimator: \rev{$\theta_{n}^{DB}(a_0) - \theta_0(a_0) = \d{P}_n \phi_{\infty,a_0}^* + \fasterthan(\{nh\}^{-1/2})$, where $\phi_{\infty,a_0}^*$ is the limiting efficient influence function defined precisely below. As with asymptotic linearity, this representation is useful because it reduces the derivation of further asymptotic properties to the study of the linear term $\d{P}_n \phi_{\infty,a_0}^*$. Furthermore, it suggests that the variance of $(nh)^{1/2}[\theta_{n}^{DB}(a_0) - \theta_0(a_0)]$ can be estimated by $\sigma_{n}^2(a_0) := h\d{P}_n (\phi_{n,a_0}^*)^2$, where $\phi_{n,a_0}^*$ is an estimator of $\phi_{\infty, a_0}^*$ also defined below. Importantly, this variance estimator accounts for the contribution of the bias estimator, and in particular it is a consistent estimator of the asymptotic variance even when the bias estimator contributes to the asymptotic variance.}

\subsection{Bandwidth selection}\label{section:bandwidth}

As with the debiased density and regression estimators proposed by \cite{Calonico2018}, our estimator requires the choice of two bandwidths. Many bandwidth selection methods for local polynomial regression can be adapted to our setting. Here, we briefly discuss several strategies for data-driven bandwidth selection. In Section~\ref{section:numerical}, we compare the empirical behavior of the three methods outlined here.

First, we can choose the bandwidths to minimize a cross-validated estimator of the integrated mean squared error (IMSE) of \rev{$\theta_{n}^{DB}$}, as \citet{Kennedy2016ptwisecte} did for the local linear estimator. Specifically, since \rev{$\theta_{n}^{DB}(a_0)$} can be written as a linear smoother of $\xi_n$, a computationally efficient leave-one-out cross-validated estimator of the IMSE of \rev{$\theta_{n}^{DB}$} is given by 
\rev{\begin{align}
    \mathrm{IMSE}_{\text{cv}}(h,b) := \frac{1}{n}\sum_{i=1}^n \left\{\frac{\xi_n(Y_i, A_i, W_i)-\theta_{n}^{DB}(A_i)}{1-\Gamma_{n,A_i}(A_i)/n}\right\}^2. \label{eq:loocv_h=b}
\end{align}}
We refer the reader to Chapter~5.3, and specifically Theorem~5.34, of \cite{wasserman2006all} for additional discussion of this formula. We then define $(h_{\text{cv}}, b_{\text{cv}}) := \argmin_{h, b} \mathrm{IMSE}_{\text{CV}}(h,b)$. Alternatively, we can fix $b$ as a function of $h$ via $h/b = \tau \in (0,\infty)$. For example, we can fix $b = h$ so that $\tau = 1$. We can then optimize the estimated IMSE over $h$ alone, i.e.\ $h_{\text{cv},\tau} := \argmin_{h} \mathrm{IMSE}_{\text{CV}}(h,h/\tau)$.  Fixing $\tau$ removes the need to select $b$, which reduces the search to a one-dimensional space, and guarantees that $h/b = \boundeddet(1)$, which is required by our conditions below. Furthermore, in the setting of regression estimation, \cite{Calonico2018} found that fixing $\tau > 0$ to a positive constant yields improved coverage accuracy. However, the choice of $\tau$ is somewhat arbitrary, and as we will see in Theorem~\ref{thm:pointwise}, $\tau > 0$ also yields an estimator with larger asymptotic variance than that of \citet{Kennedy2016ptwisecte}.

As a second bandwidth selection procedure, we will consider an adaptation of the plug-in method proposed by \cite{Calonico2018} for their debiased local linear estimator of a regression function. The method works by minimizing an estimator of the IMSE of the local linear estimator $\theta_{n}^{LL}$ over $h$. We approximate the large-sample IMSE of $\theta_{n}^{LL}$ with respect to a probability measure $\omega$ as $\mathrm{IMSE}_{\text{plug-in}}(h) := h^4 \int \hat{B}^{LL}(a)^2 d\omega(a)+ (nh)^{-1} \int \hat{V}^{LL}(a) \, d\omega(a)$, where 
\begin{align*}
\hat{B}^{LL}(a_0) &:= e_1^T \b{D}_{n,h_1,a_0,1}^{-1}\frac{1}{n}\sum_{i=1}^n  w_{h_1,a_0, 1}(A_i)K_{h_1,a_0}(A_i) \tfrac{1}{2}\hat\theta_{h_1}''(a_0)\left(\frac{A_i-a_0}{h_1}\right)^2, \text{ and}\\
\hat{V}^{LL}(a_0) &:= e_1^T \b{D}_{n,h_1,a_0,1}^{-1}\left\{\frac{1}{n}\sum_{i=1}^n w_{h_1,a_0,1}(A_i) K_{h_1,a_0}(A_i)^2 \hat\sigma^2(A_i) w^T_{h_1,a_0,1}(A_i) \right\} \b{D}_{n,h_1,a_0,1}^{-1}e_1.
\end{align*}
Here, $h_1$ is a \emph{pilot bandwidth}, $\hat\theta_{h_1}''$ is an estimator of $\theta_0''$, and $\hat\sigma^2(a)$ is a nearest-neighbors estimator of the conditional variance of $\xi_{n}$ given $A = a$. Since by design neither $\hat{B}^{LL}$ nor $\hat{V}^{LL}$ depends on $h$, the bandwidth $h_{\text{plug-in}}$ minimizing $\mathrm{IMSE}_{\text{plug-in}}$ is given explicitly by 
\[h_{\text{plug-in}} := n^{-1/5}\left(\frac{\int \hat{V}^{LL} \, d\omega}{ 4 \int [\hat{B}^{LL}]^2 d\omega}\right)^{1/5}.\]
Finally, the bandwidth $b_{\text{plug-in}}$ of the bias correction is defined as $b_{\text{plug-in}} = h_{\text{plug-in}}/\tau$, where $\tau$ is user-specified. A benefit of this method is that it does not require numerical optimization. A second benefit is that the measure $\omega$ can be chosen based on the range over which it is of interest to estimate $\theta_0$. 

%% file: paper/section3_theory.tex
\section{Asymptotic properties of the proposed methods}\label{section:inference}

\subsection{Pointwise convergence in distribution}

In this section, we study the asymptotic properties of our proposed estimator, and use these properties to derive approaches to pointwise and uniform inference. We first show that $(nh)^{1/2}[\theta_{n}^{DB}(a_0) - \theta_0(a_0)]$ converges in distribution to a normal limit for fixed $a_0$. \rev{We will use this result to show that pointwise $(1-\alpha)$-level Wald-style confidence intervals  are asymptotically valid.}

We begin by introducing technical conditions we will rely upon. We discuss these conditions following the statement of Theorem~\ref{thm:pointwise}. Our first two conditions concern the kernel function and bandwidths, which will be required in all of our results. \
%\begin{description}[style=multiline, labelindent=.9cm, leftmargin=2cm]
\begin{enumerate}[label=\textbf{(A\arabic*)},leftmargin=2cm]
\item \label{cond:bounded_K} The kernel $K$ is a mean-zero, symmetric, nonnegative, and Lipschitz continuous density function with support contained in $[-1,1]$. Additionally, $K$ belongs to the linear span of the functions whose subgraph can be represented as a finite number of Boolean operations among sets of the form $\{(s,u) \in \d{R} \times \d{R} : p(s,u) \leq \varphi(u)\}$ where $p$ is a polynomial and $\varphi$ is an arbitrary real function. 
\item \label{cond:bandwidth} As $n \longrightarrow \infty$, the bandwidths $h = h_n$ and $b = b_n$ satisfy $h_n \longrightarrow 0$,   $nh_n \longrightarrow \infty$, $b_n \longrightarrow 0$, and $\tau_n := h_n/b_n \longrightarrow \tau \in [0,\infty)$.
\end{enumerate}
For some results, we will impose additional restrictions on the rate that $h$ approaches zero.

The next condition restricts the uniform entropy of the class of functions in which the nuisance estimators are assumed to reside. We briefly define uniform entropy; we refer the reader to \cite{vandervaart1996} for additional details. For a class of functions $\s{F}$, a probability measure $Q$ and any $\varepsilon > 0$, the $\varepsilon$ covering number $N(\varepsilon, \s{F}, L_2(Q))$ of $\s{F}$ relative to the $L_2(Q)$ metric is defined as the minimal number of $L_2(Q)$ balls of radius less than or equal to $\varepsilon$ needed to cover $\s{F}$. The uniform $\varepsilon$-entropy of $\s{F}$ is defined as $\sup_{Q} \log N(\varepsilon, \s{F}, L_2(Q))$, where the supremum is taken over all probability measures. We now state the following conditions.
\begin{enumerate}[label=\textbf{(A\arabic*)},leftmargin=2cm]
\setcounter{enumi}{2}
    \item \label{cond:uniform_entropy_nuisances} There exist classes of functions $\s{F}_\mu$ and $\s{F}_g$ such that almost surely for all $n$ large enough, $\mu_0, \mu_n \in \s{F}_\mu$, $g_0,g_n \in \s{F}_g$, and for some constants $C_j\in (0, \infty)$, $V_\mu \in (0,1)$, and $V_g \in (0,2)$:
        \begin{enumerate}
            \item \label{cond:bounded_cond} \rev{$\|\mu\|_\infty \leq C_1$ for all $\mu \in \s{F}_\mu$, and $\|1/g\|_\infty \leq C_2$ and $\|g\|_\infty \leq C_3$ for all $g \in \s{F}_g$}; and
            \item \label{cond:entropy_cond} $ \sup_Q \log N(\varepsilon, \s{F}_\mu, L_2(Q)) \leq C_4 \varepsilon^{-V_\mu}$ and $\sup_Q \log N(\varepsilon, \s{F}_g, L_2(Q)) \leq C_5  \varepsilon^{-V_g}$ for all $\varepsilon > 0$.
        \end{enumerate}
\end{enumerate}

Next, we control the behavior of limiting functions to which nuisance estimators converge. We define the following pseudo-distance for any $P_0$-square integrable functions $\gamma_1, \gamma_2 : \s{A} \times \s{W} \mapsto \mathbb{R}$, $\s{A}_0 \subset \s{A}$, and $\s{S} \subseteq \s{A} \times \s{W}$:
\begin{align*}
    d(\gamma_1, \gamma_2; \s{A}_0, \s{S}) :=  \sup_{a \in \s{A}_0} \left\{E_0 [I_{\s{S}}(a,W)\{\gamma_1(a,W) - \gamma_2(a,W)\}^2]\right\}^{1/2}
\end{align*}
We also define $B_{\varepsilon}(a_0)$ as the closed ball of radius $\varepsilon$ centered at $a_0$. We then state the final two conditions concerning the rate of convergence of the nuisance estimators and properties of the true distribution $P_0$. These conditions are specific to a value $a_0$ because they will be used in the pointwise result.
\begin{enumerate}[label=\textbf{(A\arabic*)},leftmargin=2cm]
\setcounter{enumi}{3}
    \item \label{cond:doubly_robust} There exist $\mu_\infty \in \s{F}_\mu$ and $g_\infty \in \s{F}_g$, $\delta_1 > 0$, and subsets $\s{S}_1, \s{S}_2$ and $\s{S}_3$ of $B_{\delta_1}(a_0)\times \s{W}$ such that $\s{S}_1 \cup \s{S}_2 \cup \s{S}_3 = B_{\delta_1}(a_0) \times \s{W}$ and:
        \begin{enumerate}
            \item $\mu_\infty(a, w) = \mu_0(a, w)$ for all $(a, w) \in \s{S}_1 \cup \s{S}_3$ and $g_\infty(a, w) = g_0(a, w)$ for all $(a,  w) \in \s{S}_2 \cup \s{S}_3$;
            \item $d(\mu_n,\mu_\infty; B_{\delta_1}(a_0),\s{S}_1) = \fasterthan\left( \{nh\}^{-1/2}\right)$, and $d(g_n,g_\infty; B_{\delta_1}(a_0),\s{S}_1) = \fasterthan(1)$;
            \item $d(g_n,g_\infty; B(a_0; \delta_1),\s{S}_2)\} = \fasterthan\left( \{nh\}^{-1/2}\right)$, and $d(\mu_n,\mu_\infty; B(a_0; \delta_1),\s{S}_2) = \fasterthan(1)$; and
            \item $d(\mu_n,\mu_\infty; B_{\delta_1}(a_0),\s{S}_3)d(g_n,g_\infty; B_{\delta_1}(a_0),\s{S}_3) = \fasterthan\left( \{nh\}^{-1/2}\right)$.
        \end{enumerate}
    \item \label{cond:cont_density} It holds that:
    \begin{enumerate}
        \item $\theta_0$ is twice continuously differentiable on $B_{\delta_1}(a_0)$;
        \item $f_0$ is positive and Lipschitz continuous on $B_{\delta_1}(a_0)$; 
        \item there exist $\delta_2 > 0$ and $C_6 < \infty$ such that $E_0[|Y|^{2+\delta_2} \mid A = a, W = w] \leq C_6$ for all $a \in B_{\delta_1}(a_0)$ and $P_0$-almost every $w$ and $E_0[ |Y|^4] < \infty$; and 
        \item $a \mapsto \sigma_0^2(a) \coloneqq E_0\left[ \{\xi_\infty(Y,A,W) - \theta_0(A)\}^2 \mid A = a\right]$ is bounded and continuous on $B_{\delta_1}(a_0)$, where $\xi_\infty \coloneqq \xi_{\mu_\infty, g_\infty, Q_0}$ is the limiting pseudo-outcome.
    \end{enumerate}
\end{enumerate}
Finally, we define the limiting influence function
\begin{align*}
    \phi_{\infty,a_0}^*&: (y,a,w)\mapsto\Gamma_{0,a_0}(a) \xi_\infty(y,a,w) - \gamma_{0,a_0}(a) \\
        &\qquad + \int\Gamma_{0,a_0}(\bar{a}) \left\{ \mu_\infty(\bar{a},w) - \int \mu_\infty(\bar{a}, \bar{w}) \, dQ_0(\bar{w}) \right\} \, dF_0(\bar{a})
\end{align*}
\rev{for $\Gamma_{P_0,a_0} = \Gamma_{0,a_0}$ and $\gamma_{P_0,a_0} = \gamma_{0,a_0}$. We also define our estimator $\phi_{n,a_0}^*$ of $\phi_{\infty, a_0}^*$ as
\begin{align*}
    \phi_{n,a_0}^*&: (y,a,w)\mapsto\Gamma_{n,a_0}(a) \xi_n(y,a,w) - \gamma_{n,a_0}(a) \\
        &\qquad + \int\Gamma_{n,a_0}(\bar{a}) \left\{ \mu_n(\bar{a},w) - \int \mu_n(\bar{a}, \bar{w}) \, dQ_n(\bar{w}) \right\} \, dF_n(\bar{a}), \text{ where} \\
     \gamma_{n,a_0}(a) &:= e_1^T\b{D}^{-1}_{n, h, a_0,1}w_{h,a_0, 1}(a) K_{h, a_0}(a)w^T_{h,a_0, 1}(a)\b{D}^{-1}_{n, h, a_0,1} \d{P}_n\left( w_{h,a_0, 1} K_{h, a_0} \xi_n \right) \,  \\
        &\qquad- e_3^Tc_{n,h,a_0,2} (h/b)^2 \b{D}_{n, b,a_0,2}^{-1} w_{b,a_0,2}(a) K_{b,a_0}(a)w^T_{b,a_0, 2}(a)\b{D}^{-1}_{n, b, a_0,2} \d{P}_n \left( w_{b,a_0, 2} K_{b, a_0} \xi_n\right)\\
        &\qquad + e_1^T (h/b)^2\b{D}^{-1}_{n,h, a_0,1}\left[ \tilde{w}_{h, a_0, 1}(a)  - w_{h,a_0, 1}(a) w^T_{h,a_0, 1}(a) \b{D}^{-1}_{n,h, a_0,1} \d{P}_n\left( \tilde{w}_{h, a_0, 1} K_{h,a_0} \right)\right] K_{h,a_0}(a)  \\
        &\qquad\qquad \times e_3^T\b{D}_{n, b,a_0,2}^{-1} \d{P}_n( w_{b,a_0,2} K_{b,a_0} \xi_n)
\end{align*}
Our variance estimator is then given by $\sigma_{n}^2(a_0) := h\d{P}_n (\phi_{n,a_0}^*)^2$. We note $\phi_{n,a_0}^*$ differs from the plug-in estimator $\phi_{n,a_0}^\circ$ in that $\gamma_{n,a_0}$ uses $\xi_n$ rather than $\mu_n$. We use $\phi_{n,a_0}^*$ rather than $\phi_{n,a_0}^\circ$ for the variance estimator because it is a better estimator when $\mu_n$ is inconsistent, so that $\mu_\infty \neq \mu_0$, due to the appearance of $\gamma_{0,a_0}$ in $\phi_{\infty,a_0}^*$.}

Under the five conditions defined above, we have the following result concerning the pointwise asymptotics of our estimator.
\begin{theorem}\label{thm:pointwise}
If~\ref{cond:bounded_K}--\ref{cond:cont_density} hold, then
$\theta_n^{DB}(a_0) - \theta_0(a_0) = \d{P}_n \phi_{\infty,a_0}^* + \fasterthan\left(\{nh\}^{-1/2} + h^2\right),$ 
and $(nh)^{1/2}\d{P}_n \phi_{\infty, a_0}^*\indist N\left(0, V_{K,\tau} f_0(a_0)^{-1}\sigma_0^2(a_0)\right)$, where 
\begin{align*}
V_{K,\tau} &= \int \left\{ K(u) - \tau^3 c_2\frac{(\tau u)^2 - c_2 }{c_4 - c_2^2}  K(\tau u)\right\}^2 \, du \\
&= c_0^* - 2\tau^3 c_2 \frac{\tau^2 c_{2,\tau}^*- c_2c_{0,\tau}^* }{c_4 - c_2^2} + \tau^5 c_2^2 \frac{ c_4^*-2c_2 c_2^* + c_2^2 c_0^*}{\left( c_4 - c_2\right)^2}
\end{align*}
for $c_j \coloneqq \int u^j K(u)\, du$, $c^*_j \coloneqq \int u^j K^2(u)\, du$ and $c^*_{j, \tau}\coloneqq \int u^j K(u)K(\tau u)\, du$. Hence, if $nh^5 = \boundeddet(1)$, then $(nh)^{1/2} \left[ \theta_{n}^{DB}(a_0) - \theta_0(a_0) \right]$ converges in distribution to this same limit. Furthermore, \rev{$\sigma_{n}^2(a_0) \inprob V_{K,\tau} f_0(a_0)^{-1}\sigma_0^2(a_0)$,} so $(nh)^{1/2} \left[ \theta_{n}^{DB}(a_0) - \theta_0(a_0) \right] / \sigma_{n}(a_0) \indist N(0,1)$.%, \rev{so that the confidence interval given in~\eqref{eq:ci} has asymptotic coverage $1-\alpha$, meaning that $P_0(\theta_n^{DB}(a_0) - (nh)^{-1/2}q_{1-\alpha/2} \sigma_n(a_0) \leq  \theta_0(a_0) \leq \theta_n^{DB}(a_0) + (nh)^{-1/2}q_{1-\alpha/2} \sigma_n(a_0)) \rightarrow 1-\alpha$.}
\end{theorem}

The two most crucial features of Theorem~\ref{thm:pointwise} are that the estimator is centered around $\theta_0(a_0)$, and that the conditions permit the bandwidth to be selected at the optimal rate for estimation. \rev{In particular, the final statement of Theorem~\ref{thm:pointwise} implies that the pointwise $(1-\alpha)$-level  confidence interval given by
\begin{equation}
    \left[\ell_n(a_0), u_n(a_0)\right] := \left[\theta_{n}^{DB}(a_0) - (nh)^{-1/2}q_{1-\alpha/2} \sigma_{n}(a_0), \theta_{n}^{DB}(a_0) + (nh)^{-1/2}q_{1-\alpha/2} \sigma_{n}(a_0)\right] \label{eq:ci}
\end{equation}
has asymptotic coverage level $1-\alpha$ for $\theta_0(a_0)$ without undersmoothing, where $q_p$ is the $p$th quantile of a standard normal distribution. That is, $P_0(\theta_0(a_0) \in [\ell_n(a_0), u_n(a_0)]) \longrightarrow 1-\alpha$.}
The first statement of Theorem~\ref{thm:pointwise} resembles asymptotic linearity but differs in that the influence function \rev{$\phi_{\infty,a_0}^*=\phi_{\infty, h, b,a_0}^*$} changes with $n$ through $h$ and $b$. Nevertheless, the result is useful for suggesting a natural variance estimator and in revealing the process driving the first-order behavior of the estimator.

The limiting variance of our estimator is the same as that of the local linear estimator proposed by \cite{Kennedy2016ptwisecte} up to the constant $V_{K,\tau}$.  When $\tau = 0$, so that $h/ b \longrightarrow 0$, $V_{k,0} = \int K^2 = c_2^*$, which is the same as the constant in \cite{Kennedy2016ptwisecte}. Hence, in this case, the bias correction has no impact on the asymptotic variance of the estimator. \rev{However, we note that identifying the optimal rate of convergence of $h/b$ to 0 requires assuming additional smoothness of $\theta_0$.} When $\tau > 0$, the asymptotic variance of our estimator is a constant factor larger than that of the local linear estimator. When $\tau = 1$, the constant simplifies to 
\[V_{K,1} = \frac{c_0^* c_4^2 - 2c_2 c_4 c_2^* + c_2^2 c_4^*}{(c_4 - c_2^2)^2}.\]
Therefore, if $\tau > 0$, our debiased estimator asymptotically reduces bias at the expense of variance. For the Epanechnikov kernel, $V_{K,1} = 1.25$, while $V_{K,0} = 0.6$. Hence, debiasing approximately doubles the asymptotic variance in this case. However, our variance estimator is consistent even when the bias estimator contributes to the asymptotic variance. \rev{We also note that our variance estimator is not a plug-in estimator of the asymptotic variance established in Theorem~\ref{thm:pointwise}, but is instead based on the estimated influence function of the smoothed and debiased parameter. This is analogous to the fixed-$n$ variance calculations of \cite{Calonico2018}. As in \cite{Calonico2018}, we expect this to improve the finite-sample coverage of our confidence intervals. This is explored more in numerical studies in Section~\ref{section:numerical}.}

We note that we can decompose the linear term as
\begin{align*}
   \d{P}_n\phi_{\infty, a_0}^* 
   &= \d{P}_n \left\{ \Gamma_{0,a_0} (\xi_\infty - \theta_0)\right\} +\d{P}_n\left\{\Gamma_{0,a_0} \theta_0-\gamma_{0,a_0} \right\} + \d{P}_n\left\{ \int\Gamma_{0,a_0} \left( \mu_\infty - \int \mu_\infty \, dQ_0 \right) \, dF_0 \right\}.%.\label{eq:norm_decomposition}
\end{align*}
In the proof of Theorem~\ref{thm:pointwise}, we show that the second and third terms in the above decomposition are $\fasterthan(\{n/h\}^{-1/2})$ and $\bounded(n^{-1/2})$, respectively. Both of these are $\fasterthan(\{nh\}^{-1/2})$, so under the conditions of Theorem~\ref{thm:pointwise}, the simpler representation \rev{$\theta_n^{DB}(a_0) - \theta_0(a_0) = \d{P}_n\left\{ \Gamma_{0,a_0} (\xi_\infty - \theta_0)\right\}   + \fasterthan(\{nh\}^{-1/2})$} holds, and therefore only the first component of \rev{$\phi_{\infty, a_0}^*$} contributes to the limit distribution of the estimator. This suggests \rev{$h\d{P}_n\left\{ \Gamma_{n,a_0} (\xi_n - \theta_n^{DB})\right\}^2$} as an alternative variance estimator. While this variance estimator would still yield asymptotically valid confidence intervals, including the additional asymptotically negligible terms in the variance estimator better captures the finite-sample behavior of the estimator. 

We now discuss the conditions used in Theorem~\ref{thm:pointwise}. Condition~\ref{cond:bounded_K} is a standard condition for kernel smoothing and is satisfied for many common kernel functions. The bounded support condition is technically convenient but may be avoidable. The subgraph requirement of the condition is only used for uniform inference but is relatively mild. We impose this condition so that the class of functions  $\left\{ a\mapsto K\left(\frac{a-a_0}{h}\right) : a_0 \in \d{R}\right\}$ is of VC-type \citep{gine2002}. The condition is satisfied in particular if $K$ is of the form $\phi \circ p$, where $p$ is a polynomial and $\phi$ is a bounded real function of bounded variation, which is the case for many standard kernels including the triangular, Epanechnikov, and truncated Gaussian kernels.

The requirements on $h$ in condition~\ref{cond:bandwidth} are standard in kernel smoothing. They require that the bandwidth goes to zero, so that the estimator properly localizes around $a_0$, but that it goes to zero slower than $n^{-1}$ so that the estimator does not localize too much. In order to ensure that $\fasterthandet(h^2) = \fasterthandet(\{nh\}^{-1/2})$, the convergence in distribution part of the result also requires that $nh^5 = \boundeddet(1)$, which means that $h$ goes to zero at least as fast as $n^{-1/5}$. This permits but does not require undersmoothing. The second part of condition~\ref{cond:bandwidth} requires that the bandwidth $b$ used for estimating the second derivative in the bias correction goes to zero, but that it does not go to zero faster than $h$. %This same condition was required by \cite{Calonico2018} for debiased inference on density and regression functions. In practice, we recommend setting $\tau = 1$ to avoid the need to select the bandwidth $b$.

% as large in both cases. Interestingly, value of constant is smaller for common kernels with infinite support: 0.476 for Gaussian, 0.2425 for logistic. I don't know that the value of the constant translates directly, though, since I think the constant in the asymptotically optimal bandwidth also changes depending on the kernel.}
% The asymptotic variance of the proposed estimator is proportional to $f_0(a_0)^{-1}\sigma_0^2(a_0)$ as with \cite{Kennedy2016ptwisecte}. The additional constant $V_{K,\tau}$ is greater than that of the local linear estimator when $\tau \neq 0$, and this is an expected phenomenon as the result of bias correction. For common choices of kernel, such as the Epanechnikov and the uniform kernel, the variance is approximately twice as large as the original estimator. 

Condition~\ref{cond:uniform_entropy_nuisances} requires that the nuisance estimators be contained in uniformly bounded function classes, and in the case of $g_n$, that the function class be uniformly bounded away from zero as well. Furthermore, \ref{cond:uniform_entropy_nuisances}(b) restricts the uniform entropy of these function classes. The uniform entropy condition for $\s{F}_g$ is standard in empirical process theory since it guarantees that the uniform entropy integral is finite. The uniform entropy condition for $\s{F}_\mu$ is slightly stronger in order to control empirical U-processes associated with the integrated term $\int \mu_n(a, w)\, dQ_n(w)$ in the estimated pseudo-outcomes. %This condition was mistakenly absent in \cite{Kennedy2016ptwisecte}, as noted in \cite{westling2018causal} and \citet{Weng2022ContTrtTest}. %The condition for $\s{F}_g$ is only used to demonstrate that the variance estimator is consistent.
%These conditions could be relaxed to local classes for $a$ in a neighborhood of $a_0$. 

Cross-fitting could be used to avoid the entropy condition \ref{cond:uniform_entropy_nuisances}(b)  \citep{vanderLaan2018, westling2018causal, semenova2021debiased, colangelo2020double}. \rev{Estimators based on cross-fitting have the same asymptotic distribution as those based on the full data, though their finite-sample variance is often larger because the nuisance estimators are based on a smaller training sample.} However, cross-fitting carries a higher computational cost if the nuisance estimators are estimated for each training set, and the nuisance estimators use a smaller training set. Hence, it is of interest to determine whether the results can be obtained under entropy conditions. Furthermore, the conditions for both classes notably do not restrict the nuisance estimator to VC classes; hence, our conditions permit large function classes typically associated with data-adaptive estimators.

Condition~\ref{cond:doubly_robust} is a doubly-robust condition similar to, but slightly more flexible than, that required by \cite{Kennedy2016ptwisecte}. It requires that at least one of $\mu_n$ and $g_n$ be consistent for $\mu_0$ or $g_0$, respectively, in a neighborhood of $a_0$ and for almost all covariate values. It also requires that the product of the rates of convergence of $\mu_n - \mu_0$ and $g_n - g_0$ be faster than $(nh)^{-1/2}$ in order to ensure negligibility of a second-order remainder term. For points at which only one of the nuisance estimators is consistent, that estimator must achieve this rate alone. \rev{Importantly, this assumption does not require $\mu_n$ or $g_n$ to be estimated using parametric models; the required rate of convergence can be attained when $\mu_n$ and $g_n$ are data-adaptive estimators.} If the covariates are low-dimensional, these rates can be guaranteed by many nonparametric estimators. For moderate or high-dimensional covariates, the nuisance estimators need to take advantage of additional smoothness or structure of the true nuisance parameters to ensure these rates of convergence are attainable \citep{bonvini2022fast}. In practice, we recommend leveraging multiple candidate estimators in an ensemble estimator such as SuperLearner \citep{vdlSuperLearner}.% As with condition~\ref{cond:uniform_entropy_nuisances}, the asymptotic behavior of our estimator at $a_0$ is only influenced by the asymptotic behavior of the nuisance estimators in a neighborhood of $a_0$.

Finally, condition~\ref{cond:cont_density} imposes smoothness conditions on features of the true distribution. Most importantly, \ref{cond:cont_density}(a) requires that $\theta_0$ be twice continuously differentiable in a neighborhood of $a_0$. Hence, as in \cite{Calonico2018}, our estimator does not require additional smoothness to yield asymptotically valid inference. Intuition for how this is possible, despite using a third-order local polynomial estimator for the second derivative estimator, was provided in Section~\ref{section:method}. The bandwidth conditions permit the estimator to obtain the optimal rate of convergence relative to its assumed smoothness. Conditions~\ref{cond:cont_density}(c) and~\ref{cond:cont_density}(d) require that $Y$ possesses four bounded moments, and that the conditional distribution of $Y$ given $A$ and $W$ possesses a uniformly bounded $k$th moment for some $k > 2$. They do not require that $Y$ be uniformly bounded.

\rev{We note that conditions~\ref{cond:doubly_robust} and~\ref{cond:cont_density} assume that $a_0$ is in the interior of the support of $F_0$. If $a_0$ is on the boundary of the support, but $f_0(a_0) > 0$ and~\ref{cond:doubly_robust} and~\ref{cond:cont_density} hold in a neighborhood of $a_0$ intersected with the support of $F_0$, then Theorem~\ref{thm:pointwise} continues to hold, except that the constant $V_{K, \tau}$ is different. However, since our variance estimator is based on the influence function, and the influence function is valid for boundary points, our variance estimator is also consistent and the resulting confidence intervals have valid asymptotic coverage for boundary points. This is analogous to the validity of local linear estimators at the boundary (see, e.g., Section 3.2.5 of \citealp{fan1996} and \citealp{Calonico2018}).}

\subsection{Inference on causal effects}\label{sec:effects}

The pointwise results of Theorem~\ref{thm:pointwise} allow us to construct asymptotically valid confidence intervals for $\theta_0(a_0)$ for any $a_0$ for which the conditions hold. However, in many cases, it is also of interest to draw simultaneous inference on a finite collection of values $\{\theta_0(a_1), \dotsc, \theta_0(a_m)\}$. This can be used, for instance, to construct confidence intervals for causal effects of the form $\theta_0(a_2) - \theta_0(a_1)$. For this, joint convergence of the estimator at several points is necessary. The next result demonstrates joint convergence in distribution of our estimator at a finite collection of points.

\begin{theorem}\label{thm:multivariate}
If~\ref{cond:bounded_K}--\ref{cond:cont_density} hold for each $a_0$ in the finite and fixed collection $\{a_1, \dotsc, a_m\}$ and $nh^5 = \boundeddet(1)$, then 
\[ (nh)^{1/2} \begin{pmatrix}\theta_{n}^{DB}(a_1) -\theta_0(a_1)\\ \vdots \\ \theta_{n}^{DB}(a_m) - \theta_0(a_m) \end{pmatrix}\]
%\[(nh)^{1/2}\d{P}_n\begin{pmatrix}\phi_{\infty, h, b,a_1}^* \\ \cdots \\ \phi_{\infty, h, b,a_m}^*  \end{pmatrix}\]
converges in distribution to a mean-zero multivariate normal distribution with diagonal covariance matrix and variances as defined in Theorem~\ref{thm:pointwise}.
\end{theorem}

Theorem~\ref{thm:multivariate} demonstrates that the estimator is asymptotically independent at any two distinct points. \rev{This is because the covariance of the influence functions provided in Lemma~\ref{lm:influence_function}, i.e., $h P_0(\phi_{\infty,a_1}^* \phi_{\infty, a_2}^*)$ converges to zero as $h \longrightarrow 0$ for any $a_1 \neq a_2$}. Intuitively, this is due to the fact that the estimator localizes around each point as the sample size increases. Although this justifies estimating the covariance matrix using a diagonal matrix with variance estimators \rev{$\sigma_{n}^2(a_j)$} on the diagonal, we recommend instead estimating the $(j,k)$ element of the covariance matrix using the estimator \rev{$h \d{P}_n ( \phi_{n,a_j}^*\phi_{n,a_k}^*)$}. This is reminiscent of the influence-function based estimator of the covariance between two or more asymptotically linear estimators. Though the true covariance is going to zero, it is not necessarily zero in finite samples, and we expect this estimator to better capture the finite-sample covariance. In numerical studies, we demonstrate that this estimator can yield substantial finite-sample improvements over an estimator that utilizes asymptotic independence. To obtain inference on $\nu(\theta_0(a_1), \dotsc, \theta_0(a_m))$ for a differentiable function $\nu :\d{R}^m \to \d{R}$, we can combine this covariance estimator with the delta method. Hence, Theorem~\ref{thm:multivariate} enables us to perform asymptotically valid inference on effects of the form $\theta_0(a_1) - \theta_0(a_2)$ without undersmoothing.  

\rev{As with Theorem~\ref{thm:pointwise}, the conditions of Theorem~\ref{thm:multivariate} assume that $\{a_1, \dotsc, a_m\}$ are in the interior of support of $A$, but the result still holds on the boundary, though with different asymptotic variances. Furthermore, our variance estimator is consistent and the resulting confidence intervals have valid asymptotic coverage for causal effects involving boundary points.}

\subsection{Uniform inference}\label{sec:uniform}

We now turn to the uniform behavior of the estimator over a compact set $\s{A}_0$. Our goal is to construct an asymptotically valid uniform confidence band for $\theta_0$ over $\s{A}_0$, \rev{by which we mean random functions $\ell_n^\circ, u_n^\circ : \s{A}_0 \mapsto \d{R}$ such that $P_0( \theta_0(a_0) \in [\ell_n^\circ(a_0), u_n^\circ(a_0)]$ for all $a_0 \in \s{A}_0) \longrightarrow 1-\alpha$.} A standard approach to this problem would be to demonstrate that $\{(nh)^{1/2}[\theta_{n}^{DB}(a_0)  - \theta_0(a_0)] : a_0 \in \s{A}_0\}$ converges weakly as a process to a tight limit in the space $\ell^\infty(\s{A}_0)$ of uniformly bounded functions on $\s{A}_0$ equipped with the supremum norm. However, by the asymptotic independence of $(nh)^{1/2}[\theta_n^{DB}(a_1)  - \theta_0(a_1)] $ and $(nh)^{1/2}[\theta_n^{DB}(a_2)  - \theta_0(a_2)]$ demonstrated in Theorem~\ref{thm:multivariate}, the only possible limit is a white noise process, which is not tight in $\ell^\infty(\s{A}_0)$. \citet{stupfler2016weak} explored this phenomenon in more depth for kernel density estimators. Instead, we will use the finite-sample approximation theory developed in \citet{Chernozhukov_2014}. Using this theory and the first-order representation \rev{$\theta_n^{DB}(a_0) - \theta_{0}(a_0) = \d{P}_n \phi_{\infty,a_0}^* + \fasterthan(\{nh\}^{-1/2})$}, we approximate the distribution of \rev{$\sup_{\s{A}_0} (nh)^{1/2}| \theta_n^{DB} - \theta_0|/ \sigma_{n}$} with that of $\max_{ \s{A}_n} |Z_{n,h,b}|$, where $\s{A}_n$ is a finite subset of $\s{A}_0$ and conditional on the data, $Z_{n,h,b}$ is a multivariate Gaussian random vector on $\s{A}_n$ with covariance given by 
\rev{
\[\n{Cov}\left(Z_{n,h,b}(a_1), Z_{n,h,b}(a_2)\right) = h\d{P}_n \left(\phi_{n,a_1}^* \phi_{n,a_2}^*\right) / \left[\sigma_{n}(a_1)\sigma_{n}(a_2)\right].\]}
Defining $t_{1-\alpha,n}$ as the $1-\alpha$ quantile of the conditional distribution of $\max_{\s{A}_n} |Z_{n,h,b}|$ given the data, the lower and upper limits of our asymptotic $(1-\alpha)$-level confidence band for $\theta_0$ over $a_0 \in \s{A}_0$ are then given by \rev{$\ell_n^\circ(a_0) := \theta_n^{DB}(a_0) - (nh)^{-1/2} t_{1-\alpha,n} \sigma_{n}(a_0)$ and $u_n^\circ(a_0) := \theta_n^{DB}(a_0) + (nh)^{-1/2} t_{1-\alpha,n} \sigma_{n}(a_0)$}. Notably, the limits of this confidence band are proportional to the limits of the pointwise confidence interval defined in~\eqref{eq:ci}.

We now introduce additional conditions we will use. For $\delta > 0$, we define $\s{A}_\delta$ as the $\delta$-enlargement of $\s{A}_0$, that is, the set of $a \in \d{R}$ such that there exists $a_0 \in \s{A}_0$ with $|a - a_0| \leq \delta$.
\begin{enumerate}[label=\textbf{(A\arabic*)},leftmargin=2cm]\setcounter{enumi}{5}
    \item \label{cond:unif_nuisance_rate}The constant $V := \max\{V_\mu, V_g\}$, for $V_\mu$ and $V_g$ defined in~\ref{cond:uniform_entropy_nuisances} and the bandwidth $h$  satisfy $n\left[ h / (\log n)\right]^{\frac{2+V}{2-V}} \longrightarrow \infty$ and $nh^3 \longrightarrow \infty$.
    \item \label{cond:unif_doubly_robust} There exist $\mu_\infty \in \s{F}_\mu$, $g_\infty \in \s{F}_g$, $\delta_3 > 0$ and subsets $\s{S}_1', \s{S}_2'$, and $\s{S}_3'$ of $\s{A}_{\delta_3}\times \s{W}$ such that $\s{S}_1' \cup \s{S}_2' \cup \s{S}_3' = \s{A}_{\delta_3}\times \s{W}$ and:
        \begin{enumerate}
            \item $\mu_\infty(a, w) = \mu_0(a, w)$ for all $(a, w) \in \s{S}_1' \cup \s{S}_3'$ and $g_\infty(a, w) = g_0(a, w)$ for all $(a,  w) \in \s{S}_2' \cup \s{S}_3'$;
            \item $d(\mu_n,\mu_\infty; \s{A}_{\delta_3},\s{S}_1') = \fasterthan\left( \{nh \log n\}^{-1/2}\right)$;
            \item $d(g_n,g_\infty; \s{A}_{\delta_3},\s{S}_2')\} = \fasterthan\left( \{nh \log n\}^{-1/2}\right)$;
            \item $d(\mu_n,\mu_\infty; \s{A}_{\delta_3},\s{S}_3')d(g_n,g_\infty; \s{A}_{\delta_3},\s{S}_3') = \fasterthan\left( \{nh \log n\}^{-1/2}\right)$; and
            \item $d(\mu_n,\mu_\infty; \s{A}_{\delta_3},\s{A} \times \s{W})$ and $d(g_n,g_\infty; \s{A}_{\delta_3},\s{A} \times \s{W})$ are both $\fasterthan\left(h^{\frac{V}{2(2-V)}}\{\log n\}^{-\frac{1}{2-V}}\right)$.
        \end{enumerate}
    \item \label{cond:holder_smooth_theta} It holds that:
    \begin{enumerate}
        \item $\theta_0$ is twice continuously differentiable with H\"{o}lder-continuous second derivative on $\s{A}_{\delta_3}$; 
        \item $f_0$ is Lipschitz continuous and bounded away from 0 and $\infty$ on $\s{A}_{\delta_3}$;
        \item $|Y|$ is $P_0$-almost surely bounded; and
        \item  $a \mapsto E_0\left[ \{\xi_\infty(Y,A,W) - \theta_0(A)\}^2 \mid A = a\right]$ is continuous on $\s{A}_{\delta_3}$.
    \end{enumerate}
\end{enumerate}
We now present a result demonstrating the asymptotic validity of our proposed uniform confidence band. We define $\b{O}_n := (O_1, \dotsc, O_n)$ and $\omega_n := \sup_{a_0 \in \s{A}_0} \min_{a \in \s{A}_n} |a - a_0|$ as the mesh of $\s{A}_n$ in $\s{A}_0$.
\begin{theorem}\label{thm:uniform}
If~\ref{cond:bounded_K}--\ref{cond:uniform_entropy_nuisances} and~\ref{cond:unif_nuisance_rate}--\ref{cond:holder_smooth_theta} hold, and $nh^5 = \boundeddet(1)$, then 
\[\sup_{a_0 \in \s{A}_0}\left|\theta_n^{DB}(a_0) - \theta_0(a_0)\right| =  \bounded\left(\{nh / \log n\}^{-1/2}\right)\]
and $\sup_{u,v \in \s{A}_0} |h\d{P}_n (\phi_{n,u}^* \phi_{n,v}^*) - h P_0( \phi_{\infty, u}^* \phi_{\infty, v}^*) | = \fasterthan(1)$, and if $\omega_n = \fasterthandet\left(h^p\right)$ for some $p > 1$ and $mn^d = \boundeddet(1)$ for some $d \in (0, \infty)$ as well, then
\begin{align*}
    \sup_{t \in \d{R}}\left|  P_0 \left(\sup_{a_0 \in \s{A}_0}\{nh\}^{1/2} \left| \frac{\theta_n^{DB}(a_0) - \theta_0(a_0)}{\sigma_{n}(a_0)} \right| \leq t\right) -  P_0 \left(\max_{a_0 \in \s{A}_n}\left| Z_{n, h, b}(a_0)\right| \leq t \mid \b{O}_n\right) \right| = \fasterthan(1).
\end{align*}
\end{theorem}

\rev{Theorem~\ref{thm:uniform} implies that our proposed uniform confidence band has asymptotically valid coverage. This is because by definition, $P_0 \left(\max_{a_0 \in \s{A}_n}\left| Z_{n, h, b}(a_0)\right| \leq t_{1-\alpha, n} \mid \b{O}_n\right) = 1-\alpha$, and combined with Theorem~\ref{thm:uniform} this implies that 
\begin{align*}
    &P_0\left(\theta_0(a_0) \in [\ell_n^\circ(a_0), u_n^\circ(a_0)] \text{ for all }  a_0 \in \s{A}_0\right) \\
    &\qquad= P_0 \left(\sup_{a_0 \in \s{A}_0}\{nh\}^{1/2} \left| \frac{\theta_n^{DB}(a_0) - \theta_0(a_0)}{\sigma_{n}(a_0)} \right| \leq t_{1-\alpha,n}\right) \longrightarrow 1-\alpha.
\end{align*}}
To prove Theorem~\ref{thm:uniform}, we first use the results of \cite{Chernozhukov_2014} to demonstrate that the distribution of $\sup_{a_0 \in \s{A}_0}\left|\mathbb{G}_n h^{1/2}\phi^*_{\infty, a_0} / \sigma_{\infty}(a_0)\right|$ can be approximated by that of $\sup_{\s{A}_0} |Z_{\infty,h,b}|$, where $Z_{\infty,h,b}$ is a mean-zero Gaussian process on $\s{A}_0$ with covariance
\[\n{Cov}( Z_{\infty,h,b}(a_1), Z_{\infty,h,b}(a_2)) := hP_0(\phi^*_{\infty, a_1 }\phi^*_{\infty, a_2}) / [\sigma_{\infty}(a_1)\sigma_{\infty}(a_2)],\] 
where $\sigma_{\infty}^2(a_0) := hP_0(\phi^*_{\infty,a_0})^2$ and $\d{G}_n := n^{1/2} (\d{P}_n - P_0)$. Notably, the results of \cite{Chernozhukov_2014} cover situations where the empirical process is not converging weakly to a tight limit, which is the case for our process. Since we also establish that $\sup_{a_0 \in \s{A}_0}|\theta_n^{DB}(a_0) - \theta_0(a_0) - \d{P}_n \phi_{\infty,a_0}^*| = \fasterthan\left(\{nh \log n\}^{-1/2}\right)$, this means that the distribution of $(nh)^{1/2}\sup_{\s{A}_0}|\theta_n^{DB} - \theta_0|$ can also be approximated by that of $\sup_{\s{A}_0} |Z_{\infty,h,b}|$. Finally, we  use the results of \cite{chernozhukov2015comparison} to demonstrate that the distribution of $\sup_{\s{A}_0} |Z_{\infty,h,b}|$ can be approximated by that of $\max_{\s{A}_n} |Z_{n,h,b}|$.

In order to show that $\sup_{a_0 \in \s{A}_0}|\theta_n^{DB}(a_0) - \theta_0(a_0) - \d{P}_n \phi_{\infty, a_0}^*| = \fasterthan\left(\{nh \log n\}^{-1/2}\right)$, we need faster rates of convergence for the remainder terms in the first-order expansion of our estimator. For some remainder terms, this is straightforward. However, for the empirical process remainder term $(\d{P}_n - P_0)(\phi_{n,a_0}^* - \phi_{\infty, a_0}^*)$, it is a challenging task because the usual method of demonstrating negligibility of this remainder, namely uniform asymptotic equicontinuity, only yields that it is $\fasterthan\left(\{nh\}^{-1/2}\right)$. In order to achieve the extra $\{\log n\}^{-1/2}$ term in the rate, we use the local maximal inequalities of \cite{van2011local}. By assuming rates of convergence for $\mu_n - \mu_\infty$ and $g_n - g_\infty$ in condition~\ref{cond:unif_doubly_robust}(e), we are able to establish a rate of convergence for $P_0(\phi_{n,a_0}^*- \phi_{\infty, a_0}^*)^2$, which then permits the use of the results of \cite{van2011local}. 

We now discuss the additional conditions used in Theorem~\ref{thm:uniform} beyond those used in Theorem~\ref{thm:pointwise}. Condition~\ref{cond:unif_nuisance_rate}(a) further restricts the exponent of the uniform entropy of the nuisance classes, and is used to control the empirical process remainder using local maximal inequalities as discussed above. If $h$ is chosen at the optimal rate $n^{-1/5}$, then the condition is satisfied if $V_\mu, V_g \in (0,4/3)$. If undersmoothing is employed, then the requirement is stricter. The requirement that $nh^3 \longrightarrow \infty$ can be relaxed somewhat, especially if $V_\mu$ is much less than 1, but not beyond $nh^2 \longrightarrow \infty$. Hence, severe undersmoothing in addition to debiasing is possible without sample-splitting if the nuisance estimators fall in smaller classes, though we expect this is of less interest because a main point of debiasing is to avoid undersmoothing.

Condition~\ref{cond:unif_doubly_robust} is a uniform version of the doubly-robust condition~\ref{cond:doubly_robust}. The rates of convergence of the nuisance estimators on the sets on which they are consistent in parts (b), (c), and (d) are slightly faster in order to establish the needed rate of convergence of the second-order remainder term. Part (e) also requires a  rate of convergence of the nuisance estimators towards their limiting objects. As discussed above, this is also to enable the use of local maximal inequalities. If $h$ is chosen at the optimal rate $n^{-1/5}$ and $V_g < 4/3$ as required by \ref{cond:unif_nuisance_rate}(a), then part (e) is satisfied if $d(\mu_n, \mu_\infty; A_{\delta_3}, \s{A} \times \s{W})$ and $d(g_n, g_\infty; A_{\delta_3}, \s{A} \times \s{W})$ are each $\fasterthan( n^{-1/5})$.

Condition~\ref{cond:holder_smooth_theta} places conditions on the true distribution that are stronger than those required in condition~\ref{cond:cont_density} for pointwise convergence. The assumed smoothness of $\theta_0$ and $f_0$ is standard in the literature (see, e.g., Assumption 3 in \citealp{Calonico2018} and R1 in \citealp{Cheng_2019}), and we notably still do not require more than two derivatives for $\theta_0$. Part (c) requires $Y$ to be uniformly bounded, which is again used to apply the local maximal inequality. However, this could be relaxed to $|Y|$ having an almost surely finite conditional $q$th moment for $q > 4$ at the expense of stronger and more complicated restrictions on the complexity and rate of convergence of the nuisance estimators.

Theorem~\ref{thm:uniform} also requires that the mesh $\omega_n$ of $\s{A}_n$ decrease faster than $h^p$ for some $p > 1$ but that the number of points $m$ in $\s{A}_n$ increase at most at a polynomial rate. The former is to ensure that the distribution of $\sup_{\s{A}_n} | Z_{\infty, h,b}|$ is a good enough approximation to that of $\sup_{\s{A}_0}| Z_{\infty, h,b}|$, while the latter is to ensure that the distribution of $\max_{\s{A}_n} | Z_{n, h,b}|$ is a good enough approximation of $\max_{\s{A}_n} | Z_{\infty, h,b}|$. These conditions can be simultaneously achieved, for instance, with a uniform grid of $n$ points on $\s{A}_0$.

\rev{Conditions~\ref{cond:unif_nuisance_rate} and~\ref{cond:holder_smooth_theta} imply that $\s{A}_0$ does not contain boundary points of $\s{A}$. However, as with Theorems~\ref{thm:pointwise} and~\ref{thm:multivariate}, Theorem~\ref{thm:uniform} continues to hold when $\s{A}_0$ includes boundary points as long as~\ref{cond:unif_nuisance_rate} and~\ref{cond:holder_smooth_theta} hold on $\s{A}_{\delta_3}$ intersected with the support of $F_0$. Hence, if the support of $A$ is compact and the marginal density of $A$ is bounded away from zero on its support, we may be able to construct asymptotically valid uniform confidence bands over the entire support of $A$.}

%% file: paper/section5_sim.tex
\section{Numerical studies}\label{section:numerical}

\subsection{Study design}
In this section, we conduct numerical studies to investigate the finite-sample behavior of the proposed methods. We begin by describing our data-generating process. First, we generate covariates $W \in \d{R}^4$ from a standard multivariate normal distribution. Given $W$, we then generate $A$ from the distribution whose conditional density function is given by $p_0(a \mid w) \coloneqq I_{[0,1]}(a)[\lambda(w) + 2a\{1-\lambda(w)\}]$ for $\lambda(w) \coloneqq 0.1 + 1.8\,  \n{expit}(\beta^Tw)$, where $\n{expit}(x) := 1 / (1 + e^{-x})$. This construction guarantees that $0.1 < p_0(a \mid w) < 1.9$ for all $a \in [0,1]$ and $w \in\d{R}^4$ and that $A$ is marginally $\n{Uniform}(0,1)$. Finally, we generate $Y$ given $A=a$ and $W=w$ as a Bernoulli random variable with mean $\mu_0(a, w) \coloneqq \text{expit}\left(\gamma_1^T \bar{w} + \gamma_2^T \bar{w}a + \gamma_3 a^2 + \gamma_4 T(a)\right),$  where $T(a) \coloneqq \sin(3\pi\{2a-1\}/2) / (a^2+1)$ and $\bar{w} := (1,w)$. %Similar functions to $T(u)$ have been previously used in the nonparametirc regression literature such as \cite{Calonico2018}.
We set $\beta=(-1,-1,1,1)^T$, $\gamma_1=(-1,-1,-1,1,1)^T$, $\gamma_2=(3,-1,-1,1,1)^T$, $\gamma_3 = 1$ and $\gamma_4 = 1$. Figure 17 in Supplementary Material displays $\theta_0$ and its second derivative $\theta_0^{(2)}$ implied by these settings. The resulting curve is non-monotonic and has large second derivatives for some regions of $[0,1]$. 

We simulated 1000 datasets using the above process for each $n \in \{500, 1000, 2500\}$. To illustrate the performance of the proposed procedures when data-adaptive estimators are used for the nuisance functions, we used SuperLearner \citep{vdlSuperLearner}. To estimate $\mu_0$, we used SuperLearner with a library consisting of generalized linear models, multivariate adaptive regression splines, and generalized additive models. To estimate $g_0$, we used the SuperLearner procedure proposed by \citet{munoz2011super} with the same library. To investigate double-robustness, we considered using these estimators but omitting $W_1$ and $W_2$ from the estimation procedure. We then considered three settings: (1) both $\mu_n$ and $g_n$ use all covariates, (2) $\mu_n$ uses all covariates, while $g_n$ only uses $W_3$ and $W_4$, and (3) $g_n$ uses all covariates, while $\mu_n$ only uses $W_3$ and $W_4$.

We estimated $\theta_0$ using both the local linear estimator of \citet{Kennedy2016ptwisecte} and our debiased estimator. For the local linear estimator, we considered \rev{three} bandwidth selection mechanisms. First, we used the leave-one-out cross-validation bandwidth selection proposed in Section~3.5 of \cite{Kennedy2016ptwisecte}, which we refer to as ``Local linear (CV)". Second, we used the plug-in methodology described in Section~\ref{section:bandwidth}, which we refer to as ``Local linear (PI)". \rev{Lastly, we undersmoothed the bandwidth obtained by LOOCV by dividing it by $\log_{10}(n)$, which we refer to as ``Local linear (US)".} For the debiased estimator, we used all three selection procedures described in Section~\ref{section:bandwidth}. We refer to the estimator with LOOCV bandwidths $(h_{\n{cv}}, b_{\n{cv}})$  as ``Debiased (CV)", the estimator with LOOCV bandwidths $(h_{\n{cv}, 1}, h_{\n{cv}, 1})$ selected by minimizing  IMSE$_{\n{cv}}$ over $h$ alone with $b = h$ as ``Debiased (CV, h=b)", and the estimator with bandwidths $(h_{\n{plug-in}}, h_{\n{plug-in}})$ based on the plug-in methodology as ``Debiased (PI)". 

We constructed 95\% confidence intervals for each  \rev{$a_0\in\{0.0, 0.05,0.1,\dots,0.95, 1.0\}$. For the local linear estimator, we constructed confidence intervals based on the influence function proposed by \citet{Kennedy_2016}, For the debiased estimator, we used equation~\eqref{eq:ci}}. \rev{We also considered using a plug-in estimator of the asymptotic variance $V_{K,\tau} f_0(a_0)^{-1}\sigma_0^2(a_0)$ established in Theorem~\ref{thm:pointwise}. We estimated the marginal density $f_0(a_0)$ using a kernel density estimator and estimated $\sigma_0^2(a_0)$ by regressing $\{[\xi_n(Y_i, A_i, W_i) - \theta^{LL}_{n}(A_i)]^2; i=1,\dots, n\}$ on $A_1, \dots, A_n$ using a local linear estimator. We refer to the corresponding confidence intervals as ``Debiased (PI+AV)".} We constructed 95\% confidence intervals for $\theta_0(a_0) - \theta_0(0.5)$ for \rev{$a_0 \in \{0.525, 0.55, \dotsc, 1.0\}$} using our estimator and both variance estimators described Section~\ref{sec:effects}. Finally, we constructed a 95\% confidence band over \rev{$\s{A}_0 = [0,1]$} using the method described in Section~\ref{sec:uniform}. \rev{Notably, we considered the properties of the estimator and confidence intervals both at interior and boundary points, and we considered confidence bands that include the boundary points.}%\rev{To obtain $t_{1-\alpha, n}$, we simulated 10,000 sample paths of the Gaussian process}.

\subsection{Results of numerical studies}
\begin{figure}[t]
    \centering
    \includegraphics[width=6.5in]{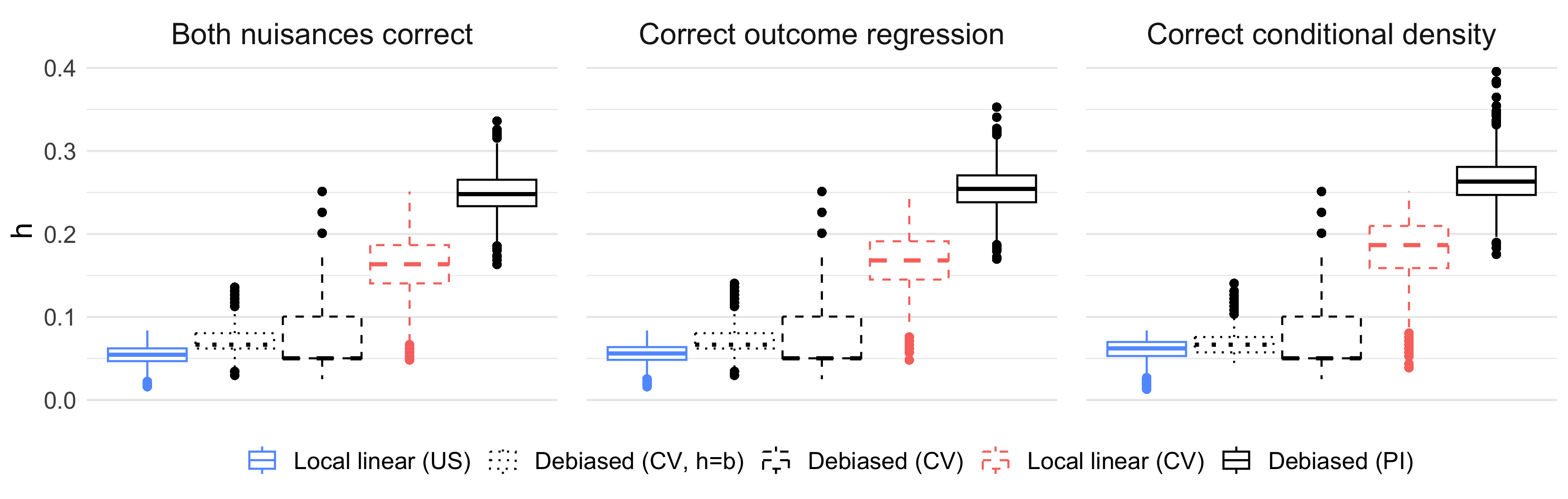}
    \caption{Box plots of the bandwidth $h$ selected by the different  procedures.}
    \label{tab:h-summary}
\end{figure}

\begin{figure}[h!]
\centering
\begin{subfigure}{6.5in}
  \centering
  \includegraphics[width=6.5in]{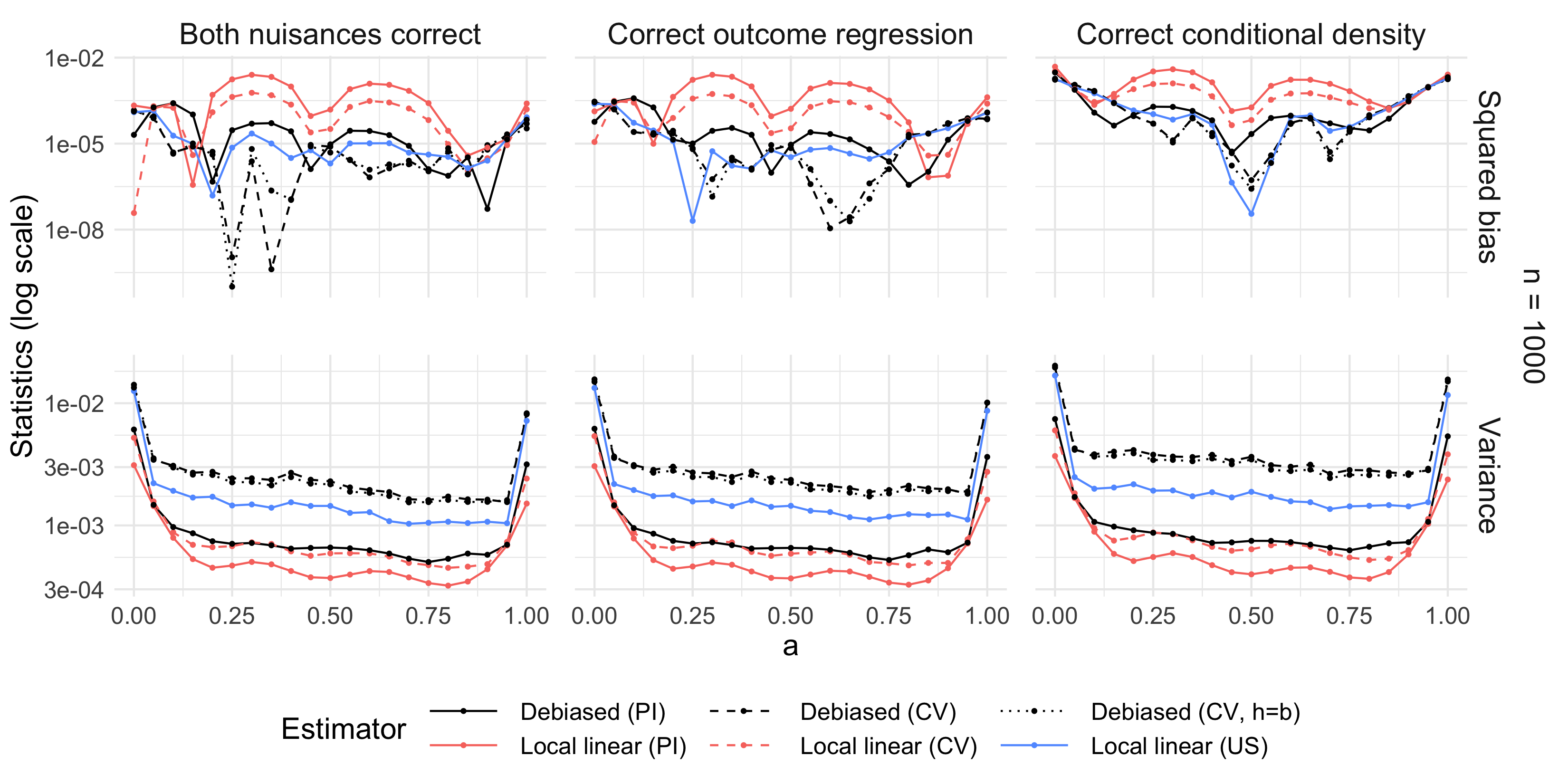}
\end{subfigure}
\caption{Top: squared empirical bias of the estimators on the log scale. The local linear estimator (shown in red) has large bias at points where the magnitude of the second derivative is large. Bottom: empirical variance of the estimators on the log scale.}
\label{fig:bias-variance}
\end{figure}

For ease of viewing, we focus here on the results for $n = 1000$. Results for $n=500$ and $n = 2500$ are provided in Supplementary Material. \rev{Figure~\ref{tab:h-summary} displays box plots of the bandwidths $h$ selected by the different procedures. Undersmoothing tends to select the smallest bandwidth, followed by the LOOCV methods, and the plug-in method tends to select the largest bandwidth.} Figure~\ref{fig:bias-variance} displays the pointwise empirical bias and variance of the two methods. As expected, the magnitude of the bias of both local linear estimators is generally larger than that of all three debiased estimators \rev{unless undersmoothing is employed}. In particular, the local linear estimator has large bias where the second derivative of $\theta_0$ is large. \rev{The bias of the debiased estimators and that of the local linear estimator with undersmoothing are generally comparable.} 

The variance of the local linear estimator is smaller than that of the debiased estimator when using the same bandwidth selection procedure. At interior points, the variance of the local linear estimator is about one-half that of the debiased estimator with the cross-validated bandwidth, which agrees with the constants computed after Theorem~\ref{thm:pointwise}. \rev{The variance of the debiased estimator with the two LOOCV bandwidth selection procedures is comparable.} The variance of both estimators using plug-in bandwidth selection is smaller than when using LOOCV bandwidth selection, which is due to the plug-in bandwidths being larger on average. \rev{The variance of the debiased estimator using the plug-in method has comparable variance with the local linear using LOOCV bandwidth selection, and is much smaller than that of the local linear estimator using undersmoothing. The variance of all estimators is larger at the boundaries and  when the outcome regression model is misspecified.}

Figure~\ref{fig:pointwise_coverage} displays the empirical coverage of pointwise 95\% confidence intervals for $\theta_0(a_0)$. The asymptotic bias of the local linear estimator results in poor coverage at points where the second derivative is large in magnitude \rev{unless undersmoothing is employed}. The debiased estimator has generally good coverage for all bandwidth selection procedures considered \rev{when the influence function-based variance estimator is used. The debiased estimators with LOOCV bandwidth selection have good coverage despite having larger variance because the influence function-based variance estimator accounts for the bandwidth. The debiased methods have good coverage at the boundaries, but slightly lower coverage near the boundaries. This issue is not present when using parametric nuisance estimators (shown in the Supplementary Material), so it could be due to the finite-sample performance of data-adaptive nuisance estimators over the corresponding region. The debiased estimator with confidence intervals based on the plug-in estimator of the asymptotic variance has worse coverage, especially near the boundary, which further illustrates the value of using the influence function-based variance estimator.}

\begin{figure}
    \centering
    \includegraphics[width=6.5in]{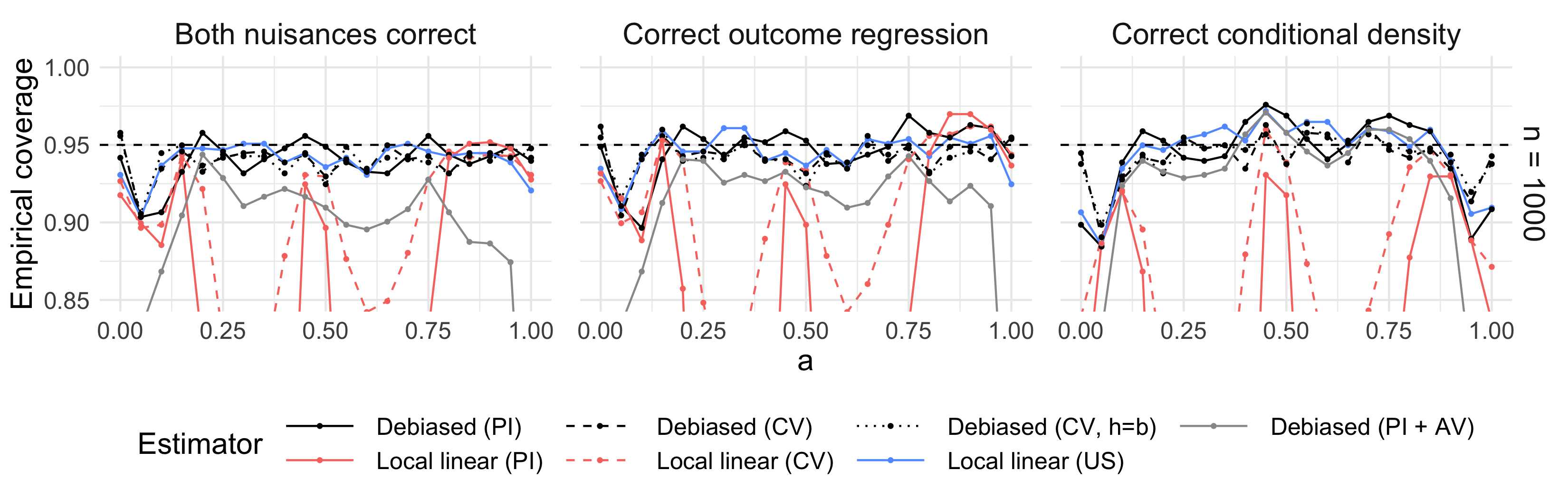}
    \caption{Empirical coverage of 95\% pointwise confidence intervals based on the debiased local linear estimator and the local linear estimator. PI, CV, and US correspond to plug-in, cross-validation, and undersmoothing respectively, and PI + AV corresponds to the plug-in bandwidth with the direct plug-in estimator of the asymptotic variance. }
    \label{fig:pointwise_coverage}
\end{figure}

% \begin{figure}
%     \centering
%     \includegraphics[width=6.5in]{figs/pts_coverage_if_asympt_1000-rev-boundary.png}
%     \caption{Empirical coverage of 95\% pointwise confidence intervals for the debiased local linear estimator using the influence function-based variance estimator and a plug-in estimate of the asymptotic variance. The bandwidth was selected using the plug-in method.}
%     \label{fig:pointwise_coverage_asympt}
% \end{figure}

% \begin{figure}
%     \centering
%     \includegraphics[width=6.5in]{figs/pts_coverage_compare_with_plugin_ml_1000-rev.png}
%     \caption{Empirical coverage of 95\% pointwise confidence intervals for studied methods based on black-box nuisance estimators (SuperLearner)}
% \end{figure}

Figure~\ref{fig:pairwise_coverage} displays the empirical coverage of confidence intervals for the causal effect $\theta_0(a) - \theta_0(0.5)$ based on the debiased estimator and the two variance estimators described following Theorem~\ref{thm:multivariate}. The confidence intervals based on the variance estimator using asymptotic independence (top row) are conservative when $a$ is close to 0.5 because the finite-sample covariance between the estimators is positive when the distance between the  evaluation points is small. When $a$ is further from $0.5$, the confidence intervals have better coverage. The confidence intervals using the influence function-based variance estimator (\rev{second} row) have much better coverage for all values of $a$, especially at large sample sizes, because the variance estimator captures some of the finite-sample covariance between the estimators. \rev{For both approaches, the plug-in method is conservative when the outcome regression is misspecified.}

\begin{figure}[h]
    \centering
    \includegraphics[width=6.5in]{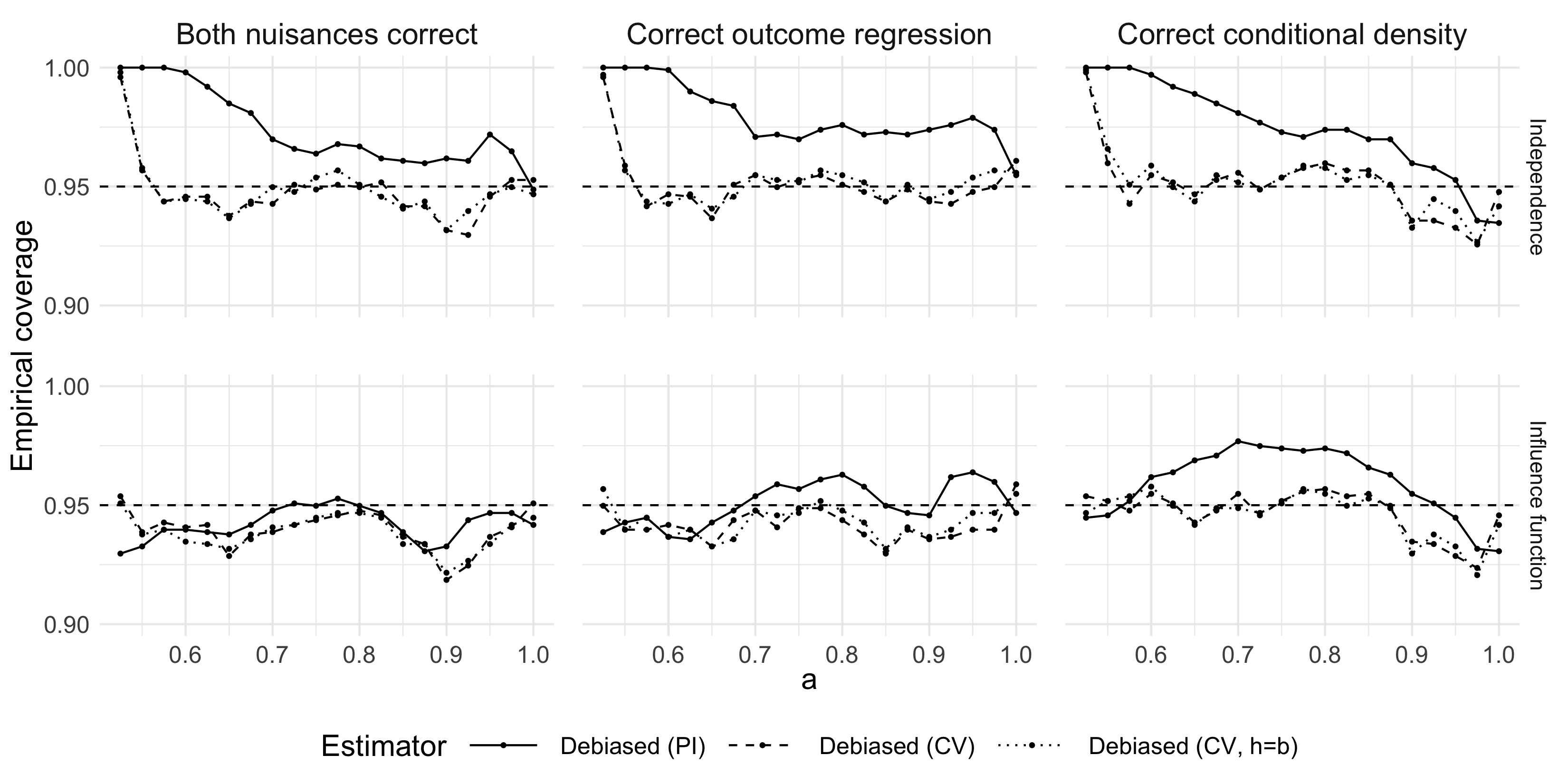}
    \caption{Empirical coverage of 95\% pointwise confidence intervals for $\theta_0(a) - \theta_0(0.5)$ based on the debiased estimator. The intervals in the top row use the sum of two variance estimators; those in the bottom row use the influence function-based variance estimator.}
    \label{fig:pairwise_coverage}
\end{figure}

Finally, Figure~\ref{fig:uniform_coverage_boundary} displays the empirical coverage of the uniform confidence bands. For this simulation, we considered the augmented set of sample sizes $n \in \{500, 750, 1000, 1250, 1500, 2000\}$. The plug-in bandwidth selection method exhibits slight undercoverage for sample sizes less than 1000, but generally performs well. However, both bandwidth selection methods based on cross-validation have serious undercoverage at sample sizes less than 1500. We conjecture that this undercoverage is a result of the bandwidth selected using LOOCV being smaller on average than that selected by the plug-in methodology. Smaller bandwidths result in a process with smaller correlation between points and whose supremum is stochastically larger. This results in a slower rate of convergence for the approximation in Theorem~\ref{thm:uniform}. This further illustrates the benefit of permitting bandwidths to be selected at the optimal rate.

\begin{figure}[h]
    \centering
    \includegraphics[width=6.5in]{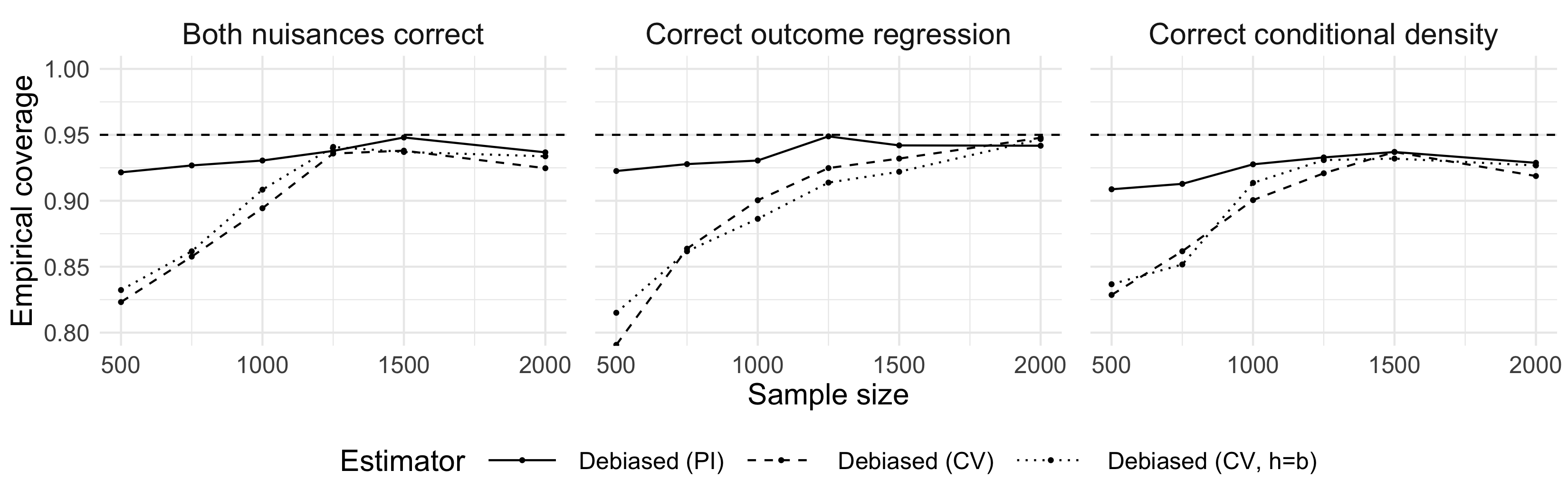}
    \caption{Empirical coverage of $95\%$ uniform confidence bands based on the debiased estimator.}
    \label{fig:uniform_coverage_boundary}
\end{figure}

\rev{An important conclusion of the numerical studies is that the plug-in bandwidth selection method with $h=b$ consistently demonstrates good coverage for all inferential tasks we examined. Additionally, the plug-in method yields smaller variance and narrower confidence intervals than the other debiased methods and the undersmoothed local linear estimator. Based on these findings, we recommend using the plug-in method with  $h=b$ for bandwidth selection.}

\rev{The Supplementary Material contains additional results from the numerical study, including the results presented here for more sample sizes, median width of confidence intervals, and results when the nuisance estimators are based on parametric models.}

%% file: paper/section6_data.tex
\section{Effect of fine particulate matter on cardiovascular mortality}\label{section:application}

Fine particulate matter is a common air pollutant, and its concentration in the United States is regulated by the Environmental Protection Agency under the Clean Air Act. Numerous scientific studies have reported an association between exposure to air pollution and adverse health outcomes. The reduction in the concentration of particulate matter in the atmosphere over the past several decades is considered one of the contributors to the declining cardiovascular mortality rate  \citep{pope2002lung,pope2009fine,correia2013effect,roth2017trends,corrigan2018fine}. 

\citet{wyatt2020contribution} recently conducted an observational study to investigate the association between particulate matter less than 2.5 microns in diameter (PM$_{2.5}$) and cardiovascular mortality rate. Socioeconomic factors are a potential confounding of this relationship because they impact both exposure to PM$_{2.5}$ and cardiovascular mortality. Using data recorded at the county level in the United States between 1990 and 2010, the authors found a positive association between PM$_{2.5}$ and cardiovascular mortality rate after adjusting for socioeconomic characteristics using regression models. %by estimating two regression functions: conditional expectation of PM$_{2.5}$ concentration given socioeconomic characteristics an conditional expectation of the mortality rate given PM$_{2.5}$ concentration and socioeconomic characteristics. In order to characterize a non-linear relationship between PM$_{2.5}$ and the mortality rate, the authors used a generalized additive model with cubic splines and reported the estimated partial response curve in Figure 3A of \citet{wyatt2020contribution}. They found a positive and linear relationship between PM$_{2.5}$ and mortality rate on a national level. They also reported that the dose-response curve may be non-linear after stratifying based on the socioeconomic characteristics.   

% \tw{estimated a (type of) regression adjusting for ... explain more the basic analysis of the study including smoothing spline thing. They found a significant coefficient corresponding to PM$_{2.5}$, and concluded that there is evidence of an association between PM$_{2.5}$ and cardiovascular mortality after adjusting for socioeconomic characteristics.} 

We used the publicly available data compiled by \citet{wyatt2020annual} to estimate the covariate-adjusted relationship between PM$_{2.5}$ and cardiovascular mortality rate using the methods presented here. The data contains information about $n=2132$ counties. Our exposure $A$ was the county-level annual PM$_{2.5}$ (in $\mu$g/m$^3$) averaged over twenty observations from 1990 to 2010, as measured by US Environmental Protection Agency's Community Multi-scale Air Quality modeling system \citep{gan2015assessment}. Our outcome $Y$ was county-level cardiovascular mortality rate (CMR, deaths/100,000 people) in 2010, as measured by the National Center for Health Statistics. Our covariates $W$ consisted of county-level socioeconomic factors based on 1990 and 2000 census data: population, households below the poverty level,  proportion of female-headed households with dependent children, vacant housing units, owner-occupied housing units, median household income, high education rate, unemployment rate, and density of medical facilities. We note that there are very likely to be violations of the stable unit treatment value assumption in this data since air pollution in one county could impact the health of residents in neighboring counties. In addition, the data are unlikely to be independent. It would be of interest in future research to extend the methods presented here to deal with treatment spillover and dependent data. 

%We also adjusted for a ``socioeconomic deprivation index"  computed by \citep{wyatt2020contribution} using a factor analysis.
We estimated the outcome regression using SuperLearner \citep{vdlSuperLearner} with a library consisting of generalized linear models, multivariate adaptive regression splines, generalized additive models, and regression trees. We estimated the conditional density using the version of SuperLearner developed by \citet{munoz2011super} with the same library. We used the Epanechnikov kernel and selected the bandwidth $h$ using the plug-in method discussed in Section~\ref{section:bandwidth} and set $b=h$.  We focus on values of \rev{PM$_{2.5}$ between $2.5 \mu g / m^3$ and $11.5 \mu g / m^3$, which approximately corresponds to the 0.01 and 0.99 quantiles} of the marginal empirical distribution of PM$_{2.5}$, respectively.
%ur analysis differs from the original study largely in two ways. First, we present the integrated dose-response relationship between PM$_{2.5}$ and cardiovascular mortality rate based on the marginal distribution of socioeconomic factors. On the other hand, \cite{wyatt2020contribution} reported the stratified dose-response relationship based on a factor analysis. Second, we do not require any parametric assumptions for associated regression functions and instead used flexible machine learning based methods to estimate nuisance components. 
\begin{figure}[h]
    \centering
    \includegraphics[width=0.8\textwidth]{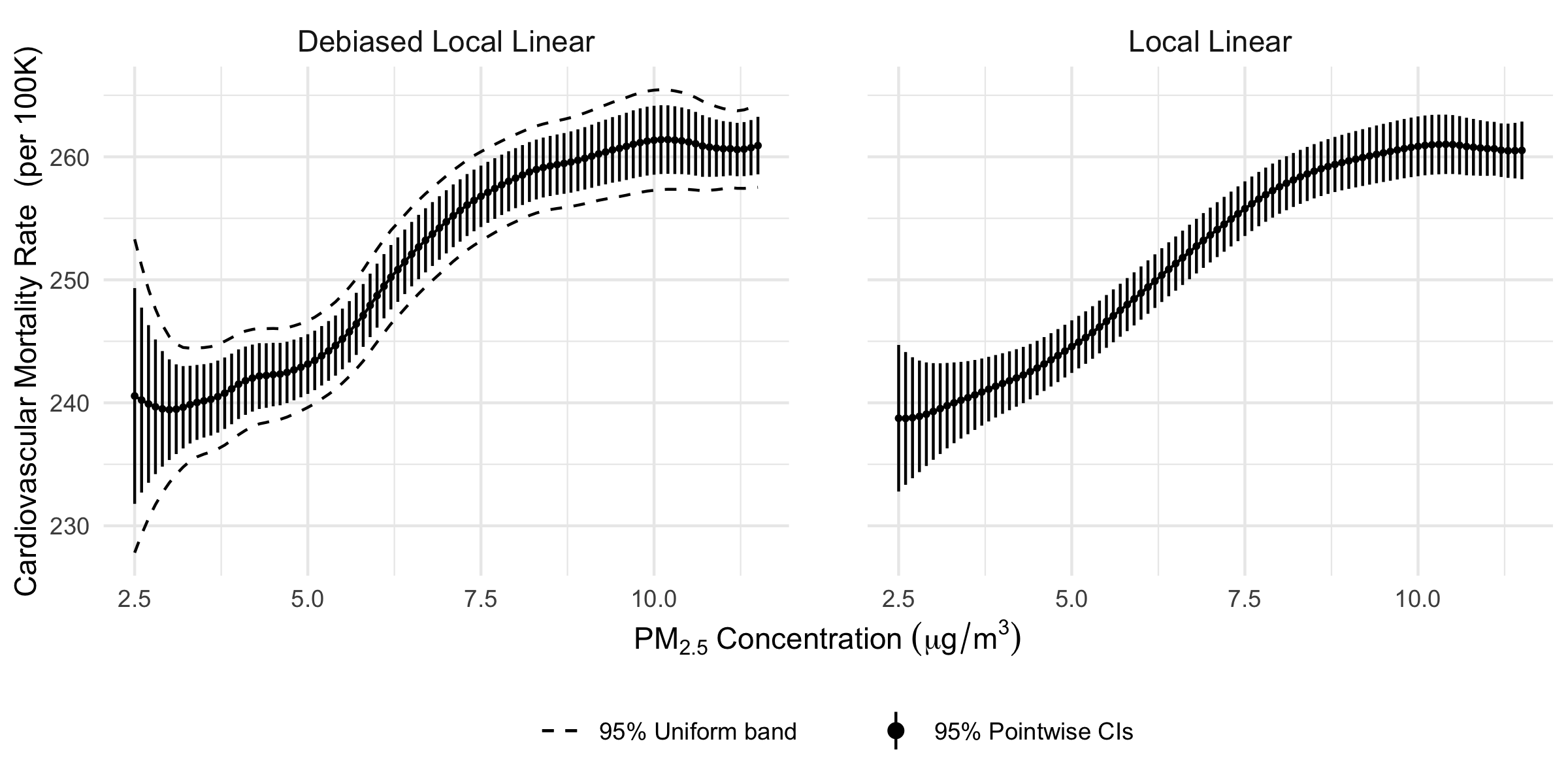}
    \caption{Estimated covariate-adjusted relationship between the concentration of fine particulate matter (PM$_{2.5}$) in the atmosphere and cardiovascular mortality rate at the county level. The regression function is adjusted for socioeconomic factors. 
    The figure shows 95\% pointwise confidence intervals as vertical lines and a 95\% uniform confidence band as dashed lines. \rev{The right panel displays the estimated regression function based on the local linear estimator.}}
    \label{fig:real_data}
\end{figure}

Figure~\ref{fig:real_data} displays the estimated covariate-adjusted regression function using our debiased method. \rev{The result based on the local linear estimator with the plug-in bandwidth selection is also provided for comparison purposes.} Pointwise 95\% confidence intervals are displayed using vertical lines, and a 95\% uniform confidence band is displayed using dashed lines.  The adjusted expected CMR appears to be monotone increasing as a function of PM$_{2.5}$, and its rate of increase is fastest for PM$_{2.5}$ values between roughly 6 and 8 $\mu g / m^3$. Our method shows larger uncertainty for PM$_{2.5}$ levels below 4 $\mu g / m^3$ and above 9 $\mu g / m^3$ due to the low density of PM$_{2.5}$. It is unclear if average CMR increases or plateaus over these regions. We estimate that increasing PM$_{2.5}$ from 6 to 8 $\mu g / m^3$ corresponds to an increase of \rev{10.3 in expected CMR after adjusting for socioeconomic factors (95\% CI: 7.6--12.9)}, and increasing PM$_{2.5}$ from 5 to 9 $\mu g / m^3$ corresponds to an increase of  \rev{17.5 in expected CMR (95\% CI: 14.8--20.2)}. Our conclusions generally agree with those reported in \cite{wyatt2020contribution}, but our flexible approach to estimation and principled approach to inference ensure that the conclusions are not the result of statistical bias. \rev{We also note that in this application, it may be sensible to assume the dose-response function is monotone increasing, so the methods of \citet{westling2018causal} could be used instead.}

%% file: paper/section7_conclusion.tex
\section{Concluding remarks}\label{section:conclusion}

In this article, we studied nonparametric inference for a covariate-adjusted regression function. This problem has wide applications in observational studies when the exposure of interest is continuous. In particular, under appropriate causal assumptions, the covariate-adjusted regression function corresponds to the average counterfactual outcome had all units been assigned to the same fixed exposure level. We presented conditions under which our proposed procedures yield valid pointwise and uniform inference. Our conditions do not require undersmoothing and permit the use of data-adaptive estimators for nuisance functions, and our results do not require more smoothness conditions than the original local linear estimator. 

\rev{Our method requires the choice of two tuning parameters: the bandwidth $h$ of the original local linear estimator and the bandwidth $b$ of the bias correction. We considered several methods of selecting these tuning parameter and compared them in numerical studies. In practice, we suggest choosing $h$ using the plug-in method and setting $b =h$ because this method had the best overall confidence interval and band coverage rates while also maintaining the lowest variance and narrowest median confidence interval length among the debiased methods.}

There are several natural extensions to our work. Cross-fitting the nuisance estimators would remove the empirical process conditions. Debiasing higher-order local polynomial estimators would yield faster rates of convergence under stronger smoothness assumptions. Debiased estimators of derivatives of the parameter could be obtained using similar methodology. Debiasing a higher-order corrected estimator would yield weaker assumptions for the rates of convergence of the nuisance estimators \citep{bonvini2022fast}. Alternative procedures for targeting the smoothed parameter, such as targeted minimum loss-based estimation \citep{van2011targeted} could yield improved finite-sample or asymptotic properties of the methods. Finally, twicing kernels \citep{newey2004twicing,zhang2012twicing} may be an alternative approach to bias correction.

%% file: supp/notation.tex
\clearpage
\section*{Supplementary Material}\vspace{.1in}
\section{Guide to notation}

We recall the data unit is $O := (Y, A,W)$, and takes value in the sample space $\s{Y} \times \s{A} \times \s{W}$, where $\s{Y} \subseteq \d{R}$, $\s{A} \subseteq\d{R}$, and $\s{W} \subseteq \d{R}^d$. We observe $n$ IID replicates of this data unit, which we call $O_1, \dotsc, O_n$, and we define the unknown distribution of $O_i$ as $P_0$.  We index objects and functions by $P$ when they depend on a distribution $P$ generating $O$. We index by $0$ when $P = P_0$. We index by $n$ when the object is random depending on the sample $O_1, \dotsc, O_n$. Finally, we index by $\infty$ when the object is a limit of an object indexed by $n$ that may or may not be equal to its true counterpart depending on $P_0$.

\begin{itemize}
    \item $F_P(a) := P(A \leq a)$ and $f_P(a) := \frac{\partial}{\partial a} F_P(a)$ is the marginal CDF and PDF of $A$ under $P$
    \item $Q_P(\s{X}) := P(W \in \s{X})$ is the marginal distribution of $W$ under $P$
    \item $\mu_P(a,w) := E_P(Y \mid A =a, W=w)$ as the outcome regression function under $P$
    \item $\theta_P(a) := E_P[ \mu_P(a, W)] = \int \mu_P(a, w) \, dQ_P(w)$ is the $G$-computed regression function
    \item $g_P(a,w) := \left[ \frac{\partial}{\partial a} P(A \leq a \mid W = w) \right] / f_P(a)$ is the standardized conditional density of $A$ given $W = w$
    \item $Pf := \int f(o) \, dP(o)$
    \item $\|\beta\|_{P, q} := (P|\beta|^q)^{1/q}$
    \item $\| \beta\|_\s{Z} := \sup_{z \in \s{Z}}|\beta(z)|$
    \item $\d{P}_n$ is the empirical probability measure corresponding to $O_1, \dotsc, O_n$
    \item $\d{G}_n := n^{1/2}(\d{P}_n - P_0)$
    \item $Q_n$ is the marginal empirical distribution of $W_1, \dotsc, W_n$
    \item $F_n$ is the marginal empirical distribution of $A_1, \dotsc, A_n$
    \item $\mu_n$ is an estimator of $\mu_0$ with limit $\mu_\infty$
    \item $g_n$ is an estimator of $g_0$ with limit $g_\infty$
    \item $\psi_P(y,a,w) := \{y - \mu_P(a,w)\} / g_P(a,w)$
    \item $\psi_n(y,a,w) := \{y - \mu_n(a,w)\} / g_n(a,w)$
    \item $\psi_\infty(y,a,w) := \{y - \mu_\infty(a,w)\} / g_\infty(a,w)$
    \item $\xi_P(y,a,w) := \psi_P (y,a,w) + \int \mu_P(a, \bar{w})\, dQ_P(\bar{w})$
    \item $\xi_n(y,a,w) := \psi_n (y,a,w) + \int \mu_n(a, \bar{w})\, dQ_n(\bar{w})$
    \item $\xi_\infty(y,a,w) := \psi_\infty (y,a,w) + \int \mu_\infty(a, \bar{w})\, dQ_0(\bar{w})$
    \item $K_{h,a_0}(a) := K((a - a_0)/h) / h$ for $K$ the kernel function
    \item $c_j := \int u^j K(u)\,du$
    \item $c_j^* := \int u^j K^2(u)\,du$
    \item $c_{j,\tau}^* := \int u^j K(u)K(\tau u)\,du$
    \item $\b{S}_k := (c_{i+j-2})_{1 \leq i,j \leq k}$; i.e.\ the $k \times k$ matrix with $(i,j)$ element $c_{i+j-2}$
    \item $\b{S}_\tau^* := \int (1, u)^T (1, \tau u, [\tau u]^2, [\tau u]^3) K(u) K(\tau u)\, du$
    \item $e_1 := (1,0)^T$, $e_3 := (0,0,1)$;
    \item $w_{h,a_0,j}(a) := \left(1, [a - a_0] /h, \dotsc, [a-a_0]^j / h^j \right)^T$ for an integer $j \geq 1$
    \item $\b{D}_{P, h, a_0, j} := P\left(w_{h,a_0, j} K_{h, a_0} w^T_{h,a_0, j}\right)$ for an integer $j \geq 1$
    \item $\b{D}_{n,h, a_0, j} := \d{P}_n\left(w_{h,a_0, j} K_{h, a_0} w^T_{h,a_0, j}\right)$  for an integer $j \geq 1$
    \item $c_{P,h,a_0,2} := e_1 \b{D}_{P, h, a_0, 1}^{-1} P (\tilde{w}_{h,a_0,1} K_{h,a_0})$ for $\tilde{w}_{h,a_0,1}(a) := w_{h,a_0, 1}(a) [(a - a_0) / h]^2$
    \item $c_{n,h,a_0,2} := e_1 \b{D}_{n, h, a_0, 1}^{-1} \d{P}_n (\tilde{w}_{h,a_0,1} K_{h,a_0})$
    \item $\Gamma_{P, h, b, a_0}(a) := e_1^T \b{D}_{P, h, a_0,1}^{-1} w_{h,a_0,1}(a) K_{h,a_0}(a) - e_3^T c_{P,h,a_0,2}(h/b)^2 \b{D}_{P, b,a_0,2}^{-1} w_{b,a_0,2}(a) K_{b,a_0}(a)$
    \item $\Gamma_{n,h, b, a_0}(a) := e_1^T \b{D}_{n,h,a_0,1}^{-1} w_{h,a_0,1}(a) K_{h,a_0}(a) - e_3^T c_{n,h,a_0,2} (h/b)^2 \b{D}_{n,b,a_0,2}^{-1} w_{b,a_0,2}(a) K_{b,a_0}(a)$
    \item $\theta_{P,h,b}(a_0) := \theta_{P}^{DB}(a_0) = P (\Gamma_{P,h,b,a_0} \theta_P)$
    \item $\theta_{n,h,b}(a_0) := \theta_{n}^{DB}(a_0) = \d{P}_n(\Gamma_{n,h,b,a_0} \xi_n)$
    \item $\begin{aligned}[t]
        \gamma_{P,h,b,a_0}(a) &:= e_1^T\b{D}^{-1}_{P, h, a_0,1}w_{h,a_0, 1}(a) K_{h, a_0}(a)w^T_{h,a_0, 1}(a)\b{D}^{-1}_{P, h, a_0,1} P \left( w_{h,a_0, 1} K_{h, a_0} \theta_P\right) \\
        &\qquad- e_3^T c_{P,h,a_0,2} (h/b)^2 \b{D}_{P, b,a_0,2}^{-1} w_{b,a_0,2}(a) K_{b,a_0}(a)w^T_{b,a_0, 2}(a)\b{D}^{-1}_{P, b, a_0,2} P \left( w_{b,a_0, 2} K_{b, a_0} \theta_P\right) \\
        &\qquad + e_1^T(h/b)^2 \b{D}^{-1}_{P,h, a_0,1}\left[ \tilde{w}_{h, a_0, 1}(a)  - w_{h,a_0, 1}(a) w^T_{h,a_0, 1}(a) \b{D}^{-1}_{P,h, a_0,1} P\left( \tilde{w}_{h, a_0, 1} K_{h,a_0} \right)\right] K_{h,a_0}(a)  \\
        &\qquad\qquad \times e_3^T\b{D}_{P, b,a_0,2}^{-1} P( w_{b,a_0,2} K_{b,a_0} \theta_P)\end{aligned}$
    \item $\begin{aligned}[t]
        \gamma_{n,h,b,a_0}(a) &:= e_1^T\b{D}^{-1}_{n, h, a_0,1}w_{h,a_0, 1}(a) K_{h, a_0}(a)w^T_{h,a_0, 1}(a)\b{D}^{-1}_{n, h, a_0,1} \d{P}_n\left( w_{h,a_0, 1} K_{h, a_0} \xi_n \right) \,  \\
        &\qquad- e_3^Tc_{n,h,a_0,2} (h/b)^2 \b{D}_{n, b,a_0,2}^{-1} w_{b,a_0,2}(a) K_{b,a_0}(a)w^T_{b,a_0, 2}(a)\b{D}^{-1}_{n, b, a_0,2} \d{P}_n \left( w_{b,a_0, 2} K_{b, a_0} \xi_n\right)\\
        &\qquad + e_1^T (h/b)^2\b{D}^{-1}_{n,h, a_0,1}\left[ \tilde{w}_{h, a_0, 1}(a)  - w_{h,a_0, 1}(a) w^T_{h,a_0, 1}(a) \b{D}^{-1}_{n,h, a_0,1} \d{P}_n\left( \tilde{w}_{h, a_0, 1} K_{h,a_0} \right)\right] K_{h,a_0}(a)  \\
        &\qquad\qquad \times e_3^T\b{D}_{n, b,a_0,2}^{-1} \d{P}_n( w_{b,a_0,2} K_{b,a_0} \xi_n)\end{aligned}$
    \item $\begin{aligned}[t]
        \phi^*_{P, h, b, a_0}(y,a,w) &:=\Gamma_{P, h, b, a_0}(a) \xi_P(y,a,w)-\gamma_{P,h,b,a_0}(a)\\
        &\qquad+ \int \Gamma_{P, h, b, a_0}(\bar{a})\left\{\mu_P(\bar{a},w)-\int \mu_P(\bar{a},\bar{w})\, dQ_P(\bar{w})\right\}\,dF_P(\bar{a})\end{aligned}$
    \item $\begin{aligned}[t]
        \phi_{\infty, h,b,a_0}^*(y,a,w) &:= \Gamma_{0,h,b,a_0}(a) \xi_\infty(y,a,w) - \gamma_{0,h,b,a_0}(a) \\
        &\qquad + \int\Gamma_{0,h,b,a_0}(\bar{a}) \left\{ \mu_\infty(\bar{a},w) - \int \mu_\infty(\bar{a}, \bar{w}) \, dQ_0(\bar{w}) \right\} \, dF_0(\bar{a})\end{aligned}$
    \item $\begin{aligned}[t]\phi_{n, h,b,a_0}^*(y,a,w) &:= \Gamma_{n,h,b,a_0}(a) \xi_n(y,a,w) - \gamma_{n,h,b,a_0}(a) \\
    &\qquad + \int\Gamma_{n,h,b,a_0}(\bar{a}) \left\{ \mu_n(\bar{a},w) - \int \mu_n(\bar{a}, \bar{w}) \, dQ_n(\bar{w}) \right\} \, dF_n(\bar{a})\end{aligned}$
    \item $\tau_n := h_n /b_n$
    \item $\sigma_0^2(a) := E_0\left[\left\{\xi_\infty(Y,A, W) - \theta_0(A)\right\}^2 \mid A=a\right]$
    \item $\sigma_{\infty,h,b}^2(a_0) := hP_0\left( \phi_{\infty, h,b,a_0}^*\right)^2$.
    \item $\sigma_{n,h,b}^2(a_0) := h\d{P}_n\left( \phi_{n, h,b,a_0}^*\right)^2$.
    \item $\omega_n := \sup_{a_0 \in \s{A}_0} \inf_{a \in \s{A}_n} |a_0 - a_n|$.
\end{itemize}

Throughout, we use $\lesssim$ to mean ``less than up to a constant not depending on $n$, $h$, or $b$.

%% file: supp/conditions.tex
\clearpage
\section{Conditions}

Here, we re-state the conditions we use in our asymptotic results for convenience. These are the same conditions as stated in the main text.
\begin{enumerate}[label=\textbf{(A\arabic*)},leftmargin=2cm]
\item The kernel $K$ is a mean-zero, symmetric, nonnegative, and Lipschitz continuous density function with support contained in $[-1,1]$. Additionally, $K$ belongs to the linear span of the functions whose subgraph can be represented as a finite number of Boolean operations among sets of the form $\{(s,u) \in \d{R} \times \d{R} : p(s,u) \leq \varphi(u)\}$ where $p$ is a polynomial and $\varphi$ is an arbitrary real function.  %\label{cond:bounded_K} 
\item As $n \longrightarrow \infty$, the bandwidths $h = h_n$ and $b = b_n$ satisfy $h_n \longrightarrow 0$,  $nh_n \longrightarrow \infty$, $b_n \longrightarrow 0$, and $\tau_n := h_n/b_n \longrightarrow \tau \in [0,\infty)$. %\label{cond:bandwidth} 
\end{enumerate}

We define the following pseudo-distance for any $P_0$-square integrable functions $\gamma_1, \gamma_2 : \s{A} \times \s{W} \mapsto \mathbb{R}$, $\s{A}_0 \subset \s{A}$, and $\s{S} \subseteq \s{A} \times \s{W}$:
\begin{align*}
    d(\gamma_1, \gamma_2; \s{A}_0, \s{S}) :=  \sup_{a \in \s{A}_0} \left\{E_0 [I_{\s{S}}(a,W)\{\gamma_1(a,W) - \gamma_2(a,W)\}^2]\right\}^{1/2}
\end{align*}
We also define $B_\varepsilon(a_0)$ as the closed ball of radius $\varepsilon$ centered at $a_0$. 

\begin{enumerate}[label=\textbf{(A\arabic*)},leftmargin=2cm]
\setcounter{enumi}{2}
    \item  There exist classes of functions $\s{F}_\mu$ and $\s{F}_g$ such that almost surely for all $n$ large enough, $\mu_0, \mu_n \in \s{F}_\mu$, $g_0,g_n \in \s{F}_g$, and for some constants $C_j\in (0, \infty)$, $V_\mu \in (0,1)$, and $V_g \in (0,2)$:
        \begin{enumerate}
            \item \label{cond:bounded_cond}  $\|\mu\|_\infty \leq C_1$ for all $\mu \in \s{F}_\mu$, and $\|1/g\|_\infty \leq C_2$ and $\|g\|_\infty \leq C_3$ for all $g \in \s{F}_g$;
            \item \label{cond:entropy_cond} $ \sup_Q \log N(\varepsilon, \s{F}_\mu, L_2(Q)) \leq C_4 \varepsilon^{-V_\mu}$ and $\sup_Q \log N(\varepsilon, \s{F}_g, L_2(Q)) \leq C_5  \varepsilon^{-V_g}$ for all $\varepsilon > 0$.
        \end{enumerate}
      %  \label{cond:uniform_entropy_nuisances}
    \item  There exist $\mu_\infty \in \s{F}_\mu$ and $g_\infty \in \s{F}_g$, $\delta_1 > 0$, and subsets $\s{S}_1, \s{S}_2$ and $\s{S}_3$ of $B_{\delta_1}(a_0)\times \s{W}$ such that $\s{S}_1 \cup \s{S}_2 \cup \s{S}_3 = B_{\delta_1}(a_0)\times \s{W}$ and:
        \begin{enumerate}
            \item $\mu_\infty(a, w) = \mu_0(a, w)$ for all $(a, w) \in \s{S}_1 \cup \s{S}_3$ and $g_\infty(a, w) = g_0(a, w)$ for all $(a,  w) \in \s{S}_2 \cup \s{S}_3$;
            \item $d(\mu_n,\mu_\infty; B_{\delta_1}(a_0),\s{S}_1) = \fasterthan\left( \{nh\}^{-1/2}\right)$, and $d(g_n,g_\infty; B_{\delta_1}(a_0),\s{S}_1) = \fasterthan(1)$;
            \item $d(g_n,g_\infty; B_{\delta_1}(a_0),\s{S}_2)\} = \fasterthan\left( \{nh\}^{-1/2}\right)$, and $d(\mu_n,\mu_\infty; B_{\delta_1}(a_0),\s{S}_2) = \fasterthan(1)$;
            \item $d(\mu_n,\mu_\infty; B_{\delta_1}(a_0),\s{S}_3)d(g_n,g_\infty; B_{\delta_1}(a_0),\s{S}_3) = \fasterthan\left( \{nh\}^{-1/2}\right)$.
        \end{enumerate}
      %  \label{cond:doubly_robust}
    \item It holds that: 
    \begin{enumerate}
        \item $\theta_0$ is twice continuously differentiable on $B_{\delta_1}(a_0)$;
        \item $f_0$ is positive and Lipschitz continuous on $B_{\delta_1}(a_0)$; 
        \item there exist $\delta_2 > 0$ and $C_6 < \infty$ such that $E_0[|Y|^{2+\delta_2} \mid A = a, W = w] \leq C_6$ for all $a \in B_{\delta_1}(a_0)$ and $P_0$-almost every $w$ and $E_0[ |Y|^4] < \infty$; and 
        \item $a \mapsto \sigma_0^2(a) := E_0\left[ \{\xi_\infty(Y,A,W) - \theta_0(A)\}^2 \mid A = a\right]$ is bounded and continuous on $B_{\delta_1}(a_0)$, where $\xi_\infty := \xi_{\mu_\infty, g_\infty, Q_0}$ is the limiting pseudo-outcome.
    \end{enumerate}
    %\label{cond:cont_density} 
\end{enumerate}

Let $\s{A}_0$ be a compact subset of $\s{A}$ over which we wish to perform uniform inference. For $\delta > 0$ define $\s{A}_\delta$ as the $\delta$-enlargement of $\s{A}_0$; that is, the set of $a \in \d{R}$ such that there exists $a_0 \in \s{A}_0$ with $|a - a_0| \leq \delta$. We then assume there exists $\delta_3 > 0$ such that the following hold.

\begin{enumerate}[label=\textbf{(A\arabic*)},leftmargin=2cm]\setcounter{enumi}{5}
    \item The constant $V := \max\{V_\mu, V_g\}$, for $V_\mu$ and $V_g$ defined in~\ref{cond:uniform_entropy_nuisances} and the bandwidth $h$  satisfies $n\left[ h / (\log n)\right]^{\frac{2+V}{2-V}} \longrightarrow \infty$ and $nh^3 \longrightarrow \infty$. %\label{cond:unif_nuisance_rate}
    \item There exist $\mu_\infty \in \s{F}_\mu$, $g_\infty \in \s{F}_g$, $\delta_3 > 0$ and subsets $\s{S}_1', \s{S}_2'$, and $\s{S}_3'$ of $\s{A}_{\delta_3}\times \s{W}$ such that $\s{S}_1' \cup \s{S}_2' \cup \s{S}_3' = \s{A}_{\delta_3}\times \s{W}$ and:
        \begin{enumerate}
            \item $\mu_\infty(a, w) = \mu_0(a, w)$ for all $(a, w) \in \s{S}_1' \cup \s{S}_3'$ and $g_\infty(a, w) = g_0(a, w)$ for all $(a,  w) \in \s{S}_2' \cup \s{S}_3'$;
            \item $d(\mu_n,\mu_\infty; \s{A}_{\delta_3},\s{S}_1') = \fasterthan\left( \{nh \log n\}^{-1/2}\right)$;
            \item $d(g_n,g_\infty; \s{A}_{\delta_3},\s{S}_2')\} = \fasterthan\left( \{nh \log n\}^{-1/2}\right)$;
            \item $d(\mu_n,\mu_\infty; \s{A}_{\delta_3},\s{S}_3')d(g_n,g_\infty; \s{A}_{\delta_3},\s{S}_3') = \fasterthan\left( \{nh \log n\}^{-1/2}\right)$.
            \item $d(\mu_n,\mu_\infty; \s{A}_{\delta_3},\s{A} \times \s{W})$ and $d(g_n,g_\infty; \s{A}_{\delta_3},\s{A} \times \s{W})$ are both $\fasterthan\left(h^{\frac{V}{2(2-V)}}\{\log n\}^{-\frac{1}{2-V}}\right)$
        \end{enumerate}
        %\label{cond:unif_doubly_robust} 
    \item It holds that:
    \begin{enumerate}
        \item $\theta_0$ is twice continuously differentiable with H\"{o}lder-continuous second derivative on $\s{A}_{\delta_3}$; 
        \item $f_0$ is Lipschitz continuous and bounded away from 0 and $\infty$ on $\s{A}_{\delta_3}$;
        \item $|Y|$ is $P_0$-almost surely bounded; and
        \item  $a \mapsto E_0\left[ \{\xi_\infty(Y,A,W) - \theta_0(A)\}^2 \mid A = a\right]$ is continuous on $\s{A}_{\delta_3}$.
    \end{enumerate}
     %\label{cond:holder_smooth_theta} 
\end{enumerate}

%% file: supp/eif.tex
\clearpage

\section{Efficient influence function of the local parameter}

\begin{proof}[\bfseries{Proof of Lemma 1}]
We let $\{P_\varepsilon: |\varepsilon| \leq \delta\}$ be a one-dimensional quadratic mean differentiable path contained in the model $\s{M}$ such that $P_{\epsilon=0} = P_0$. We let $(y,a,w) \mapsto \dot{\ell}_0(y,a,w)$ be the score function of the path at $\varepsilon=0$, which necessarily satisfies $P_0 \dot{\ell}_0 = 0$ and $P_0 \dot{\ell}^2_0 < \infty$. We define the marginal score functions $\dot{\ell}_0(a,w) := E_0[\dot{\ell}_0 \mid A=a, W=w]$,  $\dot{\ell}_0(a) := E_0[\dot{\ell}_0 \mid A=a]$, and $\dot{\ell}_0(w) := E[\dot{\ell}_0 \mid W=w]$, which all have mean zero under $P_0$ by the tower property. We also define the conditional score function $\dot{\ell}_0(y \mid a, w):= \dot{\ell}_0(y,a,w)-\dot{\ell}(a,w)$, which satisfies $E_0[ \dot\ell_0(Y \mid A, W) \mid A = a, W = w] = 0$ for $P_0$-almost every $(a,w)$. For convenience, we denote objects depending on $P_\varepsilon$ with a subscript $\varepsilon$.

To demonstrate that $\phi^*_{0, h, b, a_0}$ is the nonparametric efficient influence function of $\theta_{0,h,b}(a_0) = \theta_0^{DB}(a_0)$, we show that \[\left.\frac{\partial}{\partial \varepsilon} \theta_{\varepsilon, h, b}(a_0) \right|_{\varepsilon=0} = P_0(\phi^*_{0, h, b, a_0} \dot{\ell}_0)\] 
for each fixed $h,b > 0$ and $a_0 \in \s{A}$. We define
\begin{align*}
    \theta_{P,h}(a_0) &:= e_1^T\b{D}^{-1}_{P, h, a_0,1} P_0\left(w_{h,a_0, 1} K_{h, a_0} \theta_P \right) \text{ and} \\ 
    \theta_{P,b}''(a_0) &:= 2 b^{-2} e_3^T\b{D}^{-1}_{P, b, a_0,2} P\left(w_{b,a_0, 2} K_{b, a_0} \theta_P\right),
\end{align*}
so that $\theta_{P, h, b}(a_0)$ can be expressed as
\begin{align*}
    \theta_{P, h, b}(a_0) = \theta_{P,h}(a_0) - \tfrac{1}{2}c_{P, h, a_0,2} h^2 \theta''_{P,b}(a_0).
\end{align*}
Therefore, 
\[\left.\frac{\partial}{\partial \varepsilon} \theta_{\varepsilon, h, b}(a_0) \right|_{\varepsilon=0} = \left.\frac{\partial}{\partial \varepsilon} \theta_{\varepsilon, h}(a_0) \right| _{\varepsilon=0} - \tfrac{1}{2}\left.\frac{\partial}{\partial \varepsilon} c_{\varepsilon, h, a_0,2} \right|_{\varepsilon=0} h^2 \theta''_{P, b}(a_0) -\tfrac{1}{2}c_{P, h, a_0,2} h^2 \left.\frac{\partial}{\partial \varepsilon} \theta''_{\varepsilon, b}(a_0) \right|_{\varepsilon=0}.\]
We now provide the derivation of $\left.\frac{\partial}{\partial \varepsilon}\theta_{\varepsilon, h}(a_0) \right| _{\varepsilon=0}$. We first have by the product rule that
\begin{align*}
    \left. \frac{\partial}{\partial \varepsilon}\theta_{\varepsilon, h}(a_0) \right| _{\varepsilon=0} &= \left. \frac{\partial}{\partial \varepsilon}e_1^T\b{D}^{-1}_{\varepsilon, h, a_0,1}P_\varepsilon \left( w_{h,a_0, 1} K_{h, a_0} \theta_\varepsilon\right)\right|_{\varepsilon=0}\\
    & = \left. \frac{\partial}{\partial \varepsilon}e_1^T\b{D}^{-1}_{\varepsilon, h, a_0,1}\right|_{\varepsilon=0} P_0 \left(w_{h,a_0, 1} K_{h, a_0} \theta_0\right) + e_1^T\b{D}^{-1}_{0, h, a_0,1} \left.\frac{\partial}{\partial \varepsilon}P_\varepsilon \left(w_{h,a_0, 1} K_{h, a_0} \theta_0\right)\right|_{\varepsilon=0}  \\
    &\qquad + e_1^T\b{D}^{-1}_{0, h, a_0,1}\left.\frac{\partial}{\partial \varepsilon}P_0 \left(w_{h,a_0, 1} K_{h, a_0} \theta_\varepsilon \right)\right|_{\varepsilon=0}.
\end{align*}
For the first term, by the definition of $\b{D}_{P, h, a_0}$ as the mean of a fixed and uniformly bounded function and the chain rule, we have
\begin{align*}
    \left. \frac{\partial}{\partial \varepsilon}e_1^T\b{D}^{-1}_{\varepsilon, h, a_0,1}\right|_{\varepsilon=0} &= -e_1^T\b{D}^{-1}_{0,h, a_0,1} \left. \frac{\partial}{\partial \varepsilon}\b{D}_{\varepsilon, h, a_0,1}\right|_{\varepsilon=0}\b{D}^{-1}_{0,h, a_0,1} \\
    &= -e_1^T\b{D}^{-1}_{0,h, a_0,1}P_0\left( w_{h,a_0, 1} K_{h, a_0} w^T_{h,a_0, 1} \dot{\ell}_0 \right)\b{D}^{-1}_{0,h, a_0,1}.
\end{align*}
Similarly, for the second term, we have
\begin{align*}
    \left.\frac{\partial}{\partial \varepsilon}P_\varepsilon \left(w_{h,a_0, 1} K_{h, a_0} \theta_0\right)\right|_{\varepsilon=0}  = P_0\left( w_{h,a_0, 1} K_{h, a_0} \theta_0 \dot{\ell}_0 \right).
\end{align*}
For the third term, we have by definition of $\theta_P$
\begin{align*}
     \left. \frac{\partial}{\partial \varepsilon} P_0 \left( w_{h,a_0, 1} K_{h, a_0}\theta_\varepsilon \right) \right|_{\varepsilon=0} &= \left. \frac{\partial}{\partial \varepsilon} \int w_{h,a_0, 1}(a) K_{h, a_0}(a) \iint y \, dP_\varepsilon(y \mid a, w) \, dQ_\varepsilon(w)\,  dF_0(a) \right|_{\varepsilon=0}\\
    & = \int w_{h,a_0, 1}(a) K_{h, a_0}(a) \iint y \dot{\ell}_{0}(y \mid a, w) \, dP_0(y \mid a, w) \, dQ_0(w) \, dF_0(a) \\
    &\qquad + \int w_{h,a_0, 1}(a) K_{h, a_0}(a) \iint y \, dP_0(y \mid a,w) \dot{\ell}_0(w) \, dQ_0(w)\, dF_0(a).
\end{align*}
By properties of score functions and since $g_0(a,w) > 0$ for $F_0$-a.e.\ $a$ such that $|a - a_0| \leq h$ and $Q_0$-a.e.\ $w$, the first term equals 
\begin{align*}
&\iiint w_{h,a_0, 1}(a) K_{h, a_0}(a) \frac{y}{g_0(a,w)} \dot{\ell}_{0}(y \mid a, w) \, dP_0(y, a, w) \\
&\qquad = \iiint w_{h,a_0, 1}(a) K_{h, a_0}(a) \frac{y-\mu_0(a,w)}{g_0(a,w)} \dot{\ell}_{0}(y,a,w) \, dP_0(y,a,w).
\end{align*}
The second term equals
\begin{align*}
\iint w_{h,a_0, 1}(a) K_{h, a_0}(a)\mu_0(a,w) \, dF_0(a)\dot{\ell}_{0}(w)\, dQ_0(w) &=\int \eta_0(w) \dot{\ell}_{0}(w) \, dQ_0(w)\\
&= \iiint \eta_0(w) \dot\ell_0(y,a,w) \, dP_0(y,a,w),
\end{align*}
where $\eta_0(w) := \int w_{h,a_0, 1}(a) K_{h, a_0}(a) \mu_0(a,w) \,dF_0(a)$. Putting it together, we have that $ \left.\frac{\partial}{\partial \varepsilon}\theta_{\varepsilon, h}(a_0)\right|_{\varepsilon=0}$ equals

\begin{align*}
   &-e_1^T\b{D}^{-1}_{0, h, a_0,1} P_0\left(  w_{h,a_0, 1} K_{h, a_0} w^T_{h,a_0, 1} \dot{\ell}_0\right) \b{D}^{-1}_{0, h, a_0,1} P_0 \left( w_{h,a_0, 1} K_{h, a_0} \theta_0\right) \\
    &\qquad + e_1^T\b{D}^{-1}_{0, h, a_0,1} P_0\left(w_{h,a_0, 1} K_{h, a_0} \theta_0 \dot{\ell}_0\right) \\
    &\qquad+ e_1^T\b{D}^{-1}_{0, h, a_0,1} \iiint \left\{\left[w_{h,a_0, 1}(a) K_{h, a_0}(a) \left\{\frac{y-\mu_0(a,w)}{g_0(a,w)}\right\}+\eta_0(w)\right]\dot{\ell}_{0}(y, a,w)\right\} \, dP_0(y,a,w) \\
    &=  e_1^T\b{D}^{-1}_{0, h, a_0,1} \iiint \left\{\left[w_{h,a_0, 1}(a) K_{h, a_0}(a) \left\{\frac{y-\mu_0(a,w)}{g_0(a,w)} + \theta_0(a) - \gamma_{0,h, a_0}(a) \right\}+\eta_0(w)\right]\dot{\ell}_{0}(y,a,w)\right\} \, dP_0(y,a,w)
\end{align*}
where $\gamma_{0,h, a_0}(a) := w^T_{h,a_0, 1}(a)\b{D}^{-1}_{0, h, a_0,1} P_0 \left( w_{h,a_0, 1} K_{h, a_0} \theta_0\right)$. 
Therefore, the uncentered influence function of $\theta_{P,h}(a_0)$ at $P = P_0$ is 
\[\phi_{0, h,  a_0}(y,a,w) := e_1^T\b{D}^{-1}_{0, h, a_0,1}\left[w_{h,a_0, 1} K_{h, a_0}\left\{\frac{y-\mu_0(a,w)}{g_0(a,w)}+ \theta_0(a)-\gamma_{0,h, a_0}(a)\right\}+\eta_0(w)\right]\]
Noting that 
\[P_0\left(w_{h,a_0, 1} K_{h, a_0}\gamma_{0,h, a_0}\right) = P_0 \eta_0 = P_0\left(w_{h,a_0, 1} K_{h, a_0} \theta_0 \right),\]
the mean of $\phi_{0, h,a_0}$ is
\begin{align*}
    P_0 \phi_{0, h,a_0}  &= e_1^T\b{D}^{-1}_{0, h, a_0,1}P_0\left(w_{h,a_0, 1} K_{h, a_0} \theta_0 \right).
\end{align*}
Thus, the centered influence function of $\theta_{P,h}$ at $P_0$ is 
\begin{align*}
    \phi^*_{0, h, a_0}(y,a,w) &= \phi_{0, h,a_0}(y,a,w)  - P_0 \phi_{0, h, a_0} \\
    &= e_1^T\b{D}^{-1}_{0, h, a_0,1}w_{h,a_0, 1}(a) K_{h, a_0}(a)\left\{\frac{y-\mu_0(a,w)}{g_0(a,w)}+ \theta_0(a)-\gamma_{0,h, a_0}(a)\right\}\\
    &\qquad +e_1^T\b{D}^{-1}_{0, h, a_0,1}\int w_{h,a_0, 1}(a) K_{h, a_0}(a) \{\mu_0(a,w)-\theta_0(a)\} \,dF_0(a)
\end{align*}
By a similar argument, the influence function of $\theta_{P,b}''(a_0)$ at $P_0$ is 
\begin{align*}
    \phi^*_{0, b, a_0}(y,a,w) &= 2b^{-2}e_3^T\b{D}^{-1}_{0, b, a_0,2}w_{b,a_0, 2}(a) K_{b, a_0}(a)\left\{\frac{y-\mu_0(a,w)}{g_0(a,w)}+ \theta_0(a)-\gamma_{0,b, a_0}(a)\right\}\\
    &\qquad +2b^{-2}e_3^T\b{D}^{-1}_{0, b, a_0,2}\int w_{b,a_0, 2}(a) K_{b, a_0}(a) \{\mu_0(a,w)-\theta_0(a)\} \,dF_0(a)
\end{align*}
where $\gamma_{0,b, a_0}(a) := w^T_{b,a_0, 2}(a)\b{D}^{-1}_{0, b, a_0,2} P_0 \left( w_{b,a_0, 2} K_{b, a_0} \theta_0\right)$. We next have
\begin{align*}
     \left.\frac{\partial}{\partial \varepsilon} c_{\varepsilon, h,a_0,2}\right|_{\varepsilon=0} &=  \left.\frac{\partial}{\partial \varepsilon} e_1^T\b{D}^{-1}_{\varepsilon, h, a_0,1} 
 P_\varepsilon \left( \tilde{w}_{h, a_0, 1} K_{h,a_0} \right) \right|_{\varepsilon=0} \\
 &=\left.\frac{\partial}{\partial \varepsilon} e_1^T\b{D}^{-1}_{\varepsilon, h, a_0,1} \right|_{\varepsilon=0}   P_0 \left( \tilde{w}_{h, a_0, 1} K_{h,a_0} \right)+ e_1^T\b{D}^{-1}_{0, h, a_0,1} 
\left.\frac{\partial}{\partial \varepsilon}  P_\varepsilon \left( \tilde{w}_{h, a_0, 1} K_{h,a_0} \right) \right|_{\varepsilon=0} \\
&=-e_1^T\b{D}^{-1}_{0,h, a_0,1}P_0\left( w_{h,a_0, 1} K_{h, a_0} w^T_{h,a_0, 1} \dot{\ell}_0 \right)\b{D}^{-1}_{0,h, a_0,1} P_0 \left( \tilde{w}_{h, a_0, 1} K_{h,a_0} \right)+ e_1^T\b{D}^{-1}_{0, h, a_0,1} 
  P_0\left( \tilde{w}_{h, a_0, 1} K_{h,a_0}  \dot\ell_0 \right).
\end{align*}
Hence, the uncentered influence function of $c_{P,h, a_0,2}$ is
\begin{align*}
   \gamma_{0,h,a_0}^c(a)&=-e_1^T\b{D}^{-1}_{0,h, a_0,1} w_{h,a_0, 1}(a) K_{h, a_0}(a) w^T_{h,a_0, 1}(a) \b{D}^{-1}_{0,h, a_0,1} P_0 \left( \tilde{w}_{h, a_0, 1} K_{h,a_0} \right)+ e_1^T\b{D}^{-1}_{0, h, a_0,1} 
   \tilde{w}_{h, a_0, 1}(a) K_{h,a_0}(a) \\
   &= e_1^T\b{D}^{-1}_{0,h, a_0,1}\left[ \tilde{w}_{h, a_0, 1}(a)  - w_{h,a_0, 1}(a) w^T_{h,a_0, 1}(a) \b{D}^{-1}_{0,h, a_0,1} P_0 \left( \tilde{w}_{h, a_0, 1} K_{h,a_0} \right)\right] K_{h,a_0}(a).
\end{align*}
This function has mean zero, so it is the nonparametric efficient influence function of $c_{P,h, a_0,2}$ at $P = P_0$.

Putting the three pieces together, we have
\begin{align*}
    &\phi^*_{0, h, b, a_0}(y,a,w) = \phi^*_{0, h, a_0}(y,a,w) - \tfrac{1}{2}c_{0,h, a_0,2} h^2  \phi^*_{0, b, a_0}(y,a,w) - \tfrac{1}{2}\gamma_{0,h,a_0}^c(a) h^2 \theta_{0,b}''(a_0) \\
    &\qquad =\Gamma_{0, h, b, a_0}(a)\left\{\frac{y-\mu_0(a,w)}{g_0(a,w)}+ \theta_0(a)\right\} - \gamma_{0,h,b,a_0}(a)+ \int \Gamma_{0, h, b, a_0}(a)\{\mu_0(a,w)-\theta_0(a) \}\,dF_0(a).
\end{align*}
\end{proof}

%% file: supp/decomp.tex
\clearpage

\section{Decomposition of the estimator}

In this section, we present a decomposition of the estimator into six remainder terms. Before we present the decomposition, we provide a supporting Lemma showing that the pseudo-outcome function possesses a double robustness property under~\ref{cond:doubly_robust}(a). %We note that this double robustness property is slightly more general than the double robustness property established in~\cite{Kennedy_2016}.
\begin{lemma} \label{lm:double_robustness} 
If~\ref{cond:doubly_robust}(a) holds, then 
\begin{align*}
    \theta_0(a) = E_0 \left\{ \frac{Y-\mu_\infty(A,W)}{g_\infty(A, W)} + \int \mu_\infty(A, w) \, dQ_0(w) \, \Big| \,  A=a\right\}
\end{align*}
for $F_0$-almost every $a \in B_{\delta_1}(a_0)$.
\end{lemma}
\begin{proof}[\bfseries{Proof of Lemma~\ref{lm:double_robustness}}]
By the tower property and since $dP_0(w \mid A = a) = g_0(a,w) \, dQ_0(w)$, we have for any $a$ in the support of $A$ that 
\begin{align*}
   E_0 \left\{ \xi_{\infty}(Y, A, W) \mid A = a\right\}  &= E_0 \left\{ \frac{Y-\mu_\infty(A,W)}{g_\infty(A, W)} + \int \mu_\infty(A, w) \, dQ_0(w) \, \Big| \, A=a \right\} \\
    &= \int \frac{\mu_0(a, w)-\mu_\infty(a,w)}{g_\infty(a_0,w)}g_0(a,w) \, dQ_0(w) + \int \mu_\infty(a, w) \, dQ_0(w)\\
    & =\int \{\mu_0(a, w)-\mu_\infty(a,w)\}\left\{\frac{g_0(a,w)}{g_\infty(a,w)}-1\right\} \, dQ_0(w) + \int \mu_0(a, w) \, dQ_0(w) \\
    & =\int \{\mu_0(a, w)-\mu_\infty(a,w)\}\left\{\frac{g_0(a,w)}{g_\infty(a,w)}-1\right\} \, dQ_0(w) + \theta_0(a).
\end{align*}
Hence, 
\begin{align*}
    &\int_{B_{\delta_1}(a_0)} \left| E_0 \left\{ \xi_{\infty}(Y, A, W) \mid A = a\right\} - \theta_0(a)\right| \, dF_0(a)\\
    &\qquad= \int_{B_{\delta_1}(a_0)} \left|\int \{\mu_0(a, w)-\mu_\infty(a,w)\}\left\{\frac{g_0(a,w)}{g_\infty(a,w)}-1\right\} \, dQ_0(w) \right| \, dF_0(a) \\
    &\qquad \leq \iint_{B_{\delta_1}(a_0) \times \s{W}} \left|\{\mu_0(a, w)-\mu_\infty(a,w)\}\left\{\frac{g_0(a,w)}{g_\infty(a,w)}-1\right\}\right| \, d(F_0 \times Q_0)(a,w) \\
    &\qquad= \sum_{j=1}^3 \iint_{\s{S}_j} \left|\{\mu_0(a, w)-\mu_\infty(a,w)\}\left\{\frac{g_0(a,w)}{g_\infty(a,w)}-1\right\}\right| \, d(F_0 \times Q_0)(a,w).
\end{align*}
By~\ref{cond:doubly_robust}(a), at least one of the two terms inside the integral is zero for each of the three integrals over $\s{S}_1$, $\s{S}_2$, and $\s{S}_3$. Hence, the above expression is zero, and the claim follows.
\end{proof}

%We note that we cannot directly use $\Gamma_{0,h,b,a_0}$ in practice since it depends on $P_0$ through $\b{D}_{0, h, a_0,1}$ and $\b{D}_{0,b, a_0}$.
We now provide the decomposition of the estimator that we will use throughout our results. The decomposition involves a leading empirical mean, $\d{P}_n \phi^*_{\infty, h, b, a_0}$, which drives the first-order asymptotic behavior of the estimator, and six remainder terms. We separate the remainder into six terms because the methods of controlling these remainders are conceptually distinct, as we will see in the ensuing results.
\begin{lemma}\label{lm:first-order-decomposition}
If~\ref{cond:bounded_K} and~\ref{cond:doubly_robust}(a) hold, then the following expansion of the estimator holds for all $h \in (0,\delta_1)$:
\[ \theta_{n,h,b}(a_0) - \theta_{0}(a_0) = \d{P}_n \phi^*_{\infty, h, b, a_0} + \sum_{j=1}^6\mathrel{R}_{n,h,b, a_0, j},\]
where
\begin{align*}
    \mathrel{R}_{n, h,b,a_0, 1}&:= \theta_{0,h,b}(a_0) - \theta_{0}(a_0),\\
    \mathrel{R}_{n, h,b,a_0, 2}&:=(\d{P}_n - P_0)\left[ \left\{\Gamma_{0,h,b,a_0} \left(\psi_n + \int \mu_n \, dQ_0 \right) + \int \Gamma_{0,h,b,a_0} \mu_n \, dF_0 \right\} \right. \\
    &\qquad\qquad\qquad\qquad\left.-\left\{\Gamma_{0,h,b,a_0} \left(\psi_\infty + \int \mu_\infty \, dQ_0 \right) + \int \Gamma_{0,h,b,a_0} \mu_\infty \, dF_0 \right\} \right],\\  
    \mathrel{R}_{n, h,b,a_0, 3}&:=(\d{P}_n - P_0) \left\{ \left(\Gamma_{n,h,b,a_0} - \Gamma_{0,h,b,a_0}\right)\left( \psi_n + \int \mu_n \, dQ_0 \right) + \int \left(\Gamma_{n,h,b,a_0} - \Gamma_{0,h,b,a_0}\right) \mu_n \, dF_0 \right\},\\ 
    \mathrel{R}_{n, h,b,a_0, 4}&:=\iint\Gamma_{n,h,b,a_0}\{\mu_n-\mu_0\}\left\{1-\frac{g_0}{g_n}\right\}\,dF_0\,dQ_0,\\
    \mathrel{R}_{n, h,b, a_0, 5}&:=\iint \Gamma_{n,h,b,a_0}(a)\mu_n(a,w)\,d(Q_n-Q_0)(w)\,d(F_0-F_n)(a), \text{ and}\\
    \mathrel{R}_{n, h,b, a_0, 6} &:= e_1^T \b{D}^{-1}_{0,h,a_0,1}\left( \b{D}_{0,h,a_0,1} - \b{D}_{n,h,a_0,1}\right)\left( \b{D}^{-1}_{n,h,a_0,1} - \b{D}^{-1}_{0,h,a_0,1}\right) P_0 \left( w_{h,a_0, 1} K_{1, a_0} \theta_0\right) \\
    &\qquad- c_{0,h,a_0,2} \tau_n^2 e_3^T \b{D}^{-1}_{0, b,a_0,2}\left( \b{D}_{0,b,a_0,2} - \b{D}_{n,b,a_0,2} \right)\left( \b{D}^{-1}_{n,b,a_0,2} - \b{D}^{-1}_{0, b,a_0,2}\right) P_0 \left( w_{b,a_0, 2} K_{b, a_0} \theta_0\right) \\
    &\qquad -\tau_n^2 e_1^T    \left(\b{D}^{-1}_{n,h,a_0,1} - \b{D}^{-1}_{0,h,a_0,1} \right) (\d{P}_n - P_0)(\tilde{w}_{h,a_0,1} K_{h,a_0}) e_3^T \b{D}^{-1}_{0, b,a_0,2}P_0\left(w_{b,a_0,2} K_{b,a_0} \theta_0 \right) \\
    &\qquad -\tau_n^2  e_1^T   \left(\b{D}^{-1}_{n,h,a_0,1} - \b{D}^{-1}_{0,h,a_0,1} \right) \b{D}_{0,h,a_0,1}^{-1}\left( \b{D}_{0,h,a_0,1} -  \b{D}_{n,h,a_0,1}\right) P_0( \tilde{w}_{h,a_0,1} K_{h,a_0})  e_3^T \b{D}^{-1}_{0, b,a_0,2}P_0\left(w_{b,a_0,2} K_{b,a_0} \theta_0 \right) \\
    &\qquad - \tau_n^2  (c_{n,h,a_0,2} - c_{0,h,a_0,2}) e_3^T\left(\b{D}^{-1}_{n,b,a_0,2}-\b{D}^{-1}_{0,b,a_0,2}\right) P_0\left(w_{b,a_0,2} K_{b,a_0} \theta_0 \right)
\end{align*}
\end{lemma}
\begin{proof}[\bfseries{Proof of Lemma~\ref{lm:first-order-decomposition}}]
It is straightforward to see that $P_0 \gamma_{0,h,b,a_0}  = P_0 \left( \Gamma_{0,h,b,a_0} \theta_0\right)$.
%\begin{align*}
    %P_0 \gamma_{0,h,b,a_0} &= e_1^T \b{D}^{-1}_{0, h, a_0,1} P_0 \left( w_{h,a_0, 1} K_{h, a_0}w_{h,a_0, 1}^T \right)\b{D}^{-1}_{0, h, a_0,1} P_0 \left( w_{h,a_0, 1} K_{h, a_0} \theta_0\right) \\
   % &\qquad- e_3^T c_2 \tau^2 \b{D}_{0,b, a_0}^{-1} P_0 \left( w_{b,a_0,3} K_{b,a_0}  w_{b,a_0,3}^T \right)\b{D}_{0,b, a_0}^{-1} P_0\left(  w_{b,a_0,3} K_{b,a_0} \theta_0 \right) \\
   % &= e_1^T \b{D}^{-1}_{0, h, a_0,1} \b{D}_{0, h, a_0,1} \b{D}^{-1}_{0, h, a_0,1} P_0 %\left( w_{h,a_0, 1} K_{h, a_0} \theta_0\right) - e_3^T c_2 \tau^2 \b{D}_{0,b, a_0}^{-1} \b{D}_{b,a_0,3}\b{D}_{0,b, a_0}^{-1} P_0\left(  w_{b,a_0,3} K_{b,a_0} \theta_0 \right) \\
   % &= e_1^T  \b{D}^{-1}_{0, h, a_0,1} P_0 \left( w_{h,a_0, 1} K_{h, a_0} \theta_0\right) - e_3^T c_2 \tau^2 \b{D}_{0,b, a_0}^{-1} P_0\left(  w_{b,a_0,3} K_{b,a_0} \theta_0 \right) \\
   % &= P_0 \left( e_1^T  \b{D}^{-1}_{0, h, a_0,1} w_{h,a_0, 1} K_{h, a_0} \theta_0 - e_3^T c_2 \tau^2 \b{D}_{0,b, a_0}^{-1}  w_{b,a_0,3} K_{b,a_0} \theta_0 \right) \\
   % &= P_0 \left( \Gamma_{0,h,b,a_0} \theta_0\right).
%\end{align*}
Since~\ref{cond:doubly_robust}(a) holds by assumption, Lemma \ref{lm:double_robustness} implies that $E_0\left\{\xi_\infty(Y,A,W) \mid A =a\right\} = \theta_0(a)$ for $F_0$-almost every $a \in B_{\delta_1}(a_0)$, and since~\ref{cond:bounded_K} holds, for all $h \leq \delta_1$, $\Gamma_{0,h,b,a_0}(a) = 0$ for $a$ such that $|a - a_0| > \delta_1$. Thus,  
\[P_0 \left( \Gamma_{0,h,b,a_0} \xi_\infty\right) = E_0\left[ \Gamma_{0,h,b,a_0}(A) E_0 \left\{ \xi_\infty(Y,A,W) \mid A  \right\} \right] =E_0\left[ \Gamma_{0,h,b,a_0}(A) \theta_0(A) \right] = P_0\left(\Gamma_{0,h,b,a_0} \theta_0 \right). \]
Therefore,
\begin{align*}
    P_0 \phi_{\infty, h,b,a_0}^* &= P_0 \left( \Gamma_{0,h,b,a_0} \xi_\infty\right) - P_0 \gamma_{0,h,b,a_0} + \iint \Gamma_{0,h,b,a_0} \mu_\infty \, dF_0 \, dQ_0 - \iint \Gamma_{0,h,b,a_0} \mu_\infty\, dF_0 \, dQ_0\\
    &= P_0\left(\Gamma_{0,h,b,a_0} \theta_0 \right) - P_0\left(\Gamma_{0,h,b,a_0} \theta_0 \right) = 0.
\end{align*}
We now define 
% \begin{align*}
%      \gamma_{n,h,b,a_0}(a)&:=e_1^T\b{D}^{-1}_{n,h, a_0,1}w_{h,a_0, 1}(a) K_{h, a_0}(a)w_{h,a_0, 1}^T(a) \b{D}^{-1}_{n,h, a_0,1}\d{P}_n \left( w_{h,a_0,1} K_{h,a_0} \xi_n \right) \\
%      &\qquad- e_3^T c_2 \tau_n^2 \b{D}_{n,b,a_0,3}^{-1} w_{b,a_0,3}(a) K_{b,a_0}(a)w_{b,a_0,3}^T(a) \b{D}_{n,b,a_0,3}^{-1}\d{P}_n \left( w_{b,a_0,3} K_{b,a_0} \xi_n \right),
% \end{align*}
\begin{align*}
     \gamma_{n,h,b,a_0}^\circ(a)&:=e_1^T\b{D}^{-1}_{n,h, a_0,1}w_{h,a_0, 1}(a) K_{h, a_0}(a)w_{h,a_0, 1}^T(a) \b{D}^{-1}_{n,h, a_0,1}\d{P}_n \left( w_{h,a_0,1} K_{h,a_0} \theta_{n}\right) \\
     &\qquad- e_3^T c_{n,h,a_0,2}\tau_n^2 \b{D}_{n,b,a_0,2}^{-1} w_{b,a_0,3}(a) K_{b,a_0}(a)w_{b,a_0,2}^T(a) \b{D}_{n,b,a_0,2}^{-1}\d{P}_n \left( w_{b,a_0,2} K_{b,a_0} \theta_{n} \right)\\
     &\qquad + \tau_n^2 e_1^T\b{D}^{-1}_{n,h, a_0,1}\left[ \tilde{w}_{h, a_0, 1}(a)  - w_{h,a_0, 1}(a) w^T_{h,a_0, 1}(a) \b{D}^{-1}_{n,h, a_0,1} \d{P}_n\left( \tilde{w}_{h, a_0, 1} K_{h,a_0} \right)\right] K_{h,a_0}(a)  \\
     &\qquad\qquad \times e_3^T\b{D}_{n, b,a_0,2}^{-1} \d{P}_n( w_{b,a_0,2} K_{b,a_0} \theta_n),
\end{align*}
where $\theta_n : a \mapsto \int \mu_n(a,w) \, dQ_n(w)$, and we note that 
% \begin{align*}
%     \d{P}_n  \gamma_{n,h,b,a_0} &=e_1^T\b{D}^{-1}_{n,h, a_0,1} \d{P}_n \left(w_{h,a_0, 1} K_{h, a_0}w_{h,a_0, 1}^T \right) \b{D}^{-1}_{n,h, a_0,1}\d{P}_n \left( w_{h,a_0,1} K_{h,a_0} \xi_n \right) \\
%     &\qquad- e_3^T c_2 \tau_n^2 \b{D}_{n,b,a_0,3}^{-1} \d{P}_n \left( w_{b,a_0,3} K_{b,a_0} w_{b,a_0,3}^T\right) \b{D}_{n,b,a_0,3}^{-1}\d{P}_n \left( w_{b,a_0,3} K_{b,a_0} \xi_n \right)\\
%     &= e_1^T\b{D}^{-1}_{n,h, a_0,1} \b{D}_{n,h, a_0,1}  \b{D}^{-1}_{n,h, a_0,1}\d{P}_n \left( w_{h,a_0,1} K_{h,a_0} \xi_n \right) \\
%     &\qquad- e_3^T c_2 \tau_n^2 \b{D}_{n,b,a_0,3}^{-1} \b{D}_{n,b,a_0,3} \b{D}_{n,b,a_0,3}^{-1}\d{P}_n \left( w_{b,a_0,3} K_{b,a_0} \xi_n \right)\\
%     &= e_1^T\b{D}^{-1}_{n,h, a_0,1} \d{P}_n \left( w_{h,a_0,1} K_{h,a_0} \xi_n \right)- e_3^T c_2 \tau_n^2 \b{D}_{n,b,a_0,3}^{-1} \d{P}_n \left( w_{b,a_0,3} K_{b,a_0} \xi_n \right)\\
%     &= \d{P}_n \left( e_1^T\b{D}^{-1}_{n,h, a_0,1}  w_{h,a_0,1} K_{h,a_0} \xi_n - e_3^T c_2 \tau_n^2 \b{D}_{n,b,a_0,3}^{-1} w_{b,a_0,3} K_{b,a_0} \xi_n \right)\\
%     &= \d{P}_n \left( \Gamma_{n,h, b, a_0} \xi_n\right).
% \end{align*}
\begin{align*}
    \d{P}_n  \gamma_{n,h,b,a_0}^\circ &=e_1^T\b{D}^{-1}_{n,h, a_0,1} \d{P}_n \left(w_{h,a_0, 1} K_{h, a_0}w_{h,a_0, 1}^T \right) \b{D}^{-1}_{n,h, a_0,1}\d{P}_n \left( w_{h,a_0,1} K_{h,a_0} \theta_n \right) \\
    &\qquad- e_3^T c_{n,h,a_0,2} \tau_n^2 \b{D}_{n,b,a_0,2}^{-1} \d{P}_n \left( w_{b,a_0,2} K_{b,a_0} w_{b,a_0,2}^T\right) \b{D}_{n,b,a_0,2}^{-1}\d{P}_n \left( w_{b,a_0,2} K_{b,a_0} \theta_n \right)\\
    &\qquad+ \tau_n^2 e_1^T\b{D}^{-1}_{n,h, a_0,1}\left[ \d{P}_n \left( \tilde{w}_{h, a_0, 1}  K_{h,a_0}\right) - \d{P}_n\left(w_{h,a_0, 1} w^T_{h,a_0, 1} K_{h,a_0}\right) \b{D}^{-1}_{n,h, a_0,1} \d{P}_n\left( \tilde{w}_{h, a_0, 1} K_{h,a_0} \right)\right]  \\
     &\qquad\qquad \times e_3^T\b{D}_{n, b,a_0,2}^{-1} \d{P}_n( w_{b,a_0,2} K_{b,a_0} \theta_n) \\
    &= e_1^T\b{D}^{-1}_{n,h, a_0,1} \b{D}_{n,h, a_0,1}  \b{D}^{-1}_{n,h, a_0,1}\d{P}_n \left( w_{h,a_0,1} K_{h,a_0} \theta_n \right) \\
    &\qquad- e_3^T c_{n,h,a_0,2} \tau_n^2 \b{D}_{n,b,a_0,2}^{-1} \b{D}_{n,b,a_0,2} \b{D}_{n,b,a_0,2}^{-1}\d{P}_n \left( w_{b,a_0,2} K_{b,a_0} \theta_n \right)\\
    &\qquad+ \tau_n^2 e_1^T\b{D}^{-1}_{n,h, a_0,1}\left[ \d{P}_n \left( \tilde{w}_{h, a_0, 1}  K_{h,a_0}\right) -  \b{D}_{n,h, a_0,1}\b{D}^{-1}_{n,h, a_0,1} \d{P}_n\left( \tilde{w}_{h, a_0, 1} K_{h,a_0} \right)\right] e_3^T\b{D}_{n, b,a_0,2}^{-1} \d{P}_n( w_{b,a_0,2} K_{b,a_0} \theta_n) \\
    &= e_1^T\b{D}^{-1}_{n,h, a_0,1} \d{P}_n \left( w_{h,a_0,1} K_{h,a_0} \theta_n \right)- e_3^T  c_{n,h,a_0,2}\tau_n^2 \b{D}_{n,b,a_0,2}^{-1} \d{P}_n \left( w_{b,a_0,2} K_{b,a_0} \theta_n \right)\\
    &= \d{P}_n \left( e_1^T\b{D}^{-1}_{n,h, a_0,1}  w_{h,a_0,1} K_{h,a_0} \theta_n - e_3^T c_{n,h,a_0,2}  \tau_n^2 \b{D}_{n,b,a_0,2}^{-1} w_{b,a_0,2} K_{b,a_0} \theta_n \right)\\
    &= \d{P}_n \left( \Gamma_{n,h, b, a_0} \theta_n\right).
\end{align*}
We then define the plug-in influence function estimator $\phi_{n, h,b,a_0}^\circ$ as
\begin{align*}
   \phi_{n, h,b,a_0}^\circ(y,a,w) &:= \Gamma_{n,h,b,a_0}(a) \xi_n(y,a,w) -  \gamma_{n,h,b,a_0}^\circ(a) + \int\Gamma_{n,h,b,a_0}(\bar{a}) \left\{ \mu_n(\bar{a},w) - \int \mu_n(\bar{a}, \bar{w}) \, dQ_n(\bar{w}) \right\} \, dF_n(\bar{a}), 
\end{align*}
and we note that
% \begin{align*}
%     \d{P}_n \phi_{n, h,b,a_0}^* &=  \d{P}_n \left( \Gamma_{n,h, b, a_0} \xi_n\right) - \d{P}_n \gamma_{n,h,b,a_0} + \iint\Gamma_{n,h,b,a_0}  \mu_n \, dF_n \, dQ_n - \iint\Gamma_{n,h,b,a_0} \mu_n \, dQ_n \, dF_n \\
%     &= \d{P}_n \gamma_{n,h,b,a_0} - \d{P}_n \gamma_{n,h,b,a_0}= 0.
% \end{align*}
\begin{align*}
    \d{P}_n \phi_{n, h,b,a_0}^\circ &=  \d{P}_n \left( \Gamma_{n,h, b, a_0} \xi_n\right) - \d{P}_n \gamma_{n,h,b,a_0}^\circ + \iint\Gamma_{n,h,b,a_0}  \mu_n \, dF_n \, dQ_n - \iint\Gamma_{n,h,b,a_0} \mu_n \, dQ_n \, dF_n \\
    &= \d{P}_n \Gamma_{n,h,b,a_0} (\xi_n - \theta_n) =  \d{P}_n \Gamma_{n,h,b,a_0}\psi_n = \theta_{n,h,b}(a_0) - \d{P}_n \Gamma_{n,h,b,a_0}\theta_n.
\end{align*}
Therefore, $\theta_{n,h,b}(a_0) = \d{P}_n \Gamma_{n,h,b,a_0}\theta_n + \d{P}_n \phi_{n, h,b,a_0}^\circ$, which establishes the one-step representation of the debiased estimator. By adding and subtracting terms and using the derivations above, we can now write
% \begin{align*}
%     \theta_{n,h,b}(a_0) - \theta_0(a_0) &= \d{P}_n \phi_{\infty, h,b,a_0}^*  + (\d{P}_n - P_0)\left(\phi_{n, h,b,a_0}^* -  \phi_{\infty, h,b,a_0}^*\right) + \left\{ \theta_{n,h,b}(a_0) - \theta_{0,h,b}(a_0) + P_0 \phi_{n, h,b,a_0}^* \right\}  \\
%     &\qquad+ \left\{ \theta_{0,h,b}(a_0) - \theta_0(a_0)\right\}.
% \end{align*}
\begin{align*}
    \theta_{n,h,b}(a_0) - \theta_0(a_0) &= \d{P}_n \phi_{\infty, h,b,a_0}^*  + (\d{P}_n - P_0)\left(\phi_{n, h,b,a_0}^\circ -  \phi_{\infty, h,b,a_0}^*\right) + \left\{ \d{P}_n \Gamma_{n,h,b,a_0}\theta_n - \theta_{0,h,b}(a_0) + P_0 \phi_{n, h,b,a_0}^\circ \right\}  \\
    &\qquad+ \left\{ \theta_{0,h,b}(a_0) - \theta_0(a_0)\right\}.
\end{align*}
We note that the last summand in braces equals $R_{n,h,b,a_0,1}$. This is a standard first-order expansion of the estimator $\theta_{n,h,b}(a_0)$ of the pathwise differentiable parameter $\theta_{0,h,b}(a_0)$ for fixed $h$ and $b$, though our analysis in subsequent results will consider the case where $h$ and $b$ go to zero as $n$ grows.

We now enter a calculation showing that the second and third summands above equal $\sum_{j=2}^6 R_{n,h,b,a_0,j}$. For convenience, in this derivation we omit the subscripts $h$, $b$, and $a_0$ when it is clear. We have
% \begin{align}
%     &(\d{P}_n - P_0)\left(\phi_{n, h,b,a_0}^* -  \phi_{\infty, h,b,a_0}^*\right) + \left\{ \theta_{n,h,b}(a_0) - \theta_{0,h,b}(a_0) + P_0 \phi_{n, h,b,a_0}^* \right\} \nonumber\\
%     &\qquad=(\d{P}_n - P_0) \left( \left\{ \Gamma_n \psi_n + \Gamma_n \int \mu_n \, dQ_n - \gamma_n + \int \Gamma_n \mu_n \, dF_n\right\} \right.\nonumber \\
%     &\qquad\qquad\qquad\qquad \left.- \left\{ \Gamma_0 \psi_\infty + \Gamma_0 \int \mu_\infty \, dQ_0 - \gamma_0 + \int \Gamma_0 \mu_\infty \, dF_0\right\} \right) \nonumber\\
%     &\qquad\qquad + \left\{ \d{P}_n ( \Gamma_n \xi_n) - P_0 (\Gamma_0 \theta_0) + P_0 \left(\Gamma_n \psi_n + \Gamma_n \int \mu_n \, dQ_n - \gamma_n + \int \Gamma_n \mu_n \, dF_n - \iint \Gamma_n \mu_n\,dF_n \, dQ_n \right) \right\}\nonumber \\
%     &\qquad= (\d{P}_n - P_0) (\Gamma_n \psi_n - \Gamma_0 \psi_\infty ) + \left\{- (\d{P}_n - P_0) \gamma_n + \d{P}_n( \Gamma_n \xi_n ) - P_0 \gamma_n \right\} + \left\{ (\d{P}_n - P_0) \gamma_0 + P_0( \Gamma_n \theta_0) - P_0 ( \Gamma_0 \theta_0)\right\}\nonumber\\
%     &\qquad\qquad + \iint \Gamma_n \mu_n \left\{dQ_n \, d(F_n - F_0) + dF_n \, d(Q_n - Q_0) + dQ_n \, dF_0 + dF_n \, dQ_0 - dF_n \, dQ_n - dQ_0 \, dF_0\right\}\nonumber \\
%     &\qquad\qquad - \iint \Gamma_0 \mu_\infty \left\{ dQ_0 \, d(F_n - F_0) + dF_0 \, d(Q_n - Q_0) \right\}\nonumber \\
%     &\qquad\qquad + \left\{ P_0 (\Gamma_n \psi_n) + \iint \Gamma_n \mu_n \, dQ_0 \, dF_0 - P_0 (\Gamma_n \theta_0) \right\}\label{eq:inf_fn_expansion}
% \end{align}
\begin{align}
    &(\d{P}_n - P_0)\left(\phi_{n, h,b,a_0}^\circ -  \phi_{\infty, h,b,a_0}^*\right) + \left\{ \d{P}_n \Gamma_{n,h,b,a_0}\theta_n - \theta_{0,h,b}(a_0) + P_0 \phi_{n, h,b,a_0}^\circ \right\} \nonumber\\
    &\qquad=(\d{P}_n - P_0) \left( \left\{ \Gamma_n \psi_n + \Gamma_n \int \mu_n \, dQ_n - \gamma_n^\circ + \int \Gamma_n \mu_n \, dF_n\right\} \right.\nonumber \\
    &\qquad\qquad\qquad\qquad\qquad \left.- \left\{ \Gamma_0 \psi_\infty + \Gamma_0 \int \mu_\infty \, dQ_0 - \gamma_0 + \int \Gamma_0 \mu_\infty \, dF_0\right\} \right) \nonumber\\
    &\qquad\qquad + \left\{ \d{P}_n ( \Gamma_n \theta_n) - P_0 (\Gamma_0 \theta_0) + P_0 \left(\Gamma_n \psi_n + \Gamma_n \int \mu_n \, dQ_n - \gamma_n^\circ + \int \Gamma_n \mu_n \, dF_n - \iint \Gamma_n \mu_n\,dF_n \, dQ_n \right) \right\}\nonumber \\
    &\qquad= (\d{P}_n - P_0) (\Gamma_n \psi_n - \Gamma_0 \psi_\infty ) + \left\{- (\d{P}_n - P_0) \gamma_n^\circ + \d{P}_n( \Gamma_n \theta_n ) - P_0 \gamma_n^\circ \right\} + \left\{ (\d{P}_n - P_0) \gamma_0 + P_0( \Gamma_n \theta_0) - P_0 ( \Gamma_0 \theta_0)\right\}\nonumber\\
    &\qquad\qquad + \iint \Gamma_n \mu_n \left\{dQ_n \, d(F_n - F_0) + dF_n \, d(Q_n - Q_0) + dQ_n \, dF_0 + dF_n \, dQ_0 - dF_n \, dQ_n - dQ_0 \, dF_0\right\}\nonumber \\
    &\qquad\qquad - \iint \Gamma_0 \mu_\infty \left\{ dQ_0 \, d(F_n - F_0) + dF_0 \, d(Q_n - Q_0) \right\}\nonumber \\
    &\qquad\qquad + \left\{ P_0 (\Gamma_n \psi_n) + \iint \Gamma_n \mu_n \, dQ_0 \, dF_0 - P_0 (\Gamma_n \theta_0) \right\}\label{eq:inf_fn_expansion}
\end{align}
% \textcolor{red}{
% \begin{align}
%     &(\d{P}_n - P_0)\left(\phi_{n, h,b,a_0}^* -  \phi_{\infty, h,b,a_0}^*\right) + \left\{ \d{P}_n \Gamma_{n,h,b,a_0}\theta_n - \theta_{0,h,b}(a_0) + P_0 \phi_{n, h,b,a_0}^* \right\} \nonumber\\
%     &\qquad=(\d{P}_n - P_0) \left( \left\{ \Gamma_n \psi_n + \Gamma_n \int \mu_n \, dQ_n - \gamma_n + \int \Gamma_n \mu_n \, dF_n\right\} \right.\nonumber \\
%     &\qquad\qquad\qquad\qquad \left.- \left\{ \Gamma_0 \psi_\infty + \Gamma_0 \int \mu_\infty \, dQ_0 - \gamma_\infty + \int \Gamma_0 \mu_\infty \, dF_0\right\} \right) \nonumber\\
%     &\qquad\qquad + \left\{ \d{P}_n ( \Gamma_n \theta_n) - P_0 (\Gamma_0 \theta_0) + P_0 \left(\Gamma_n \psi_n + \Gamma_n \int \mu_n \, dQ_n - \gamma_n + \int \Gamma_n \mu_n \, dF_n - \iint \Gamma_n \mu_n\,dF_n \, dQ_n \right) \right\}\nonumber \\
%     &\qquad= (\d{P}_n - P_0) (\Gamma_n \psi_n - \Gamma_0 \psi_\infty ) + \left\{- (\d{P}_n - P_0) \gamma_n + \d{P}_n( \Gamma_n \theta_n ) - P_0 \gamma_n \right\} + \left\{ (\d{P}_n - P_0) \gamma_\infty + P_0( \Gamma_n \theta_0) - P_0 ( \Gamma_0 \theta_0)\right\}\nonumber\\
%     &\qquad\qquad + \iint \Gamma_n \mu_n \left\{dQ_n \, d(F_n - F_0) + dF_n \, d(Q_n - Q_0) + dQ_n \, dF_0 + dF_n \, dQ_0 - dF_n \, dQ_n - dQ_0 \, dF_0\right\}\nonumber \\
%     &\qquad\qquad - \iint \Gamma_0 \mu_\infty \left\{ dQ_0 \, d(F_n - F_0) + dF_0 \, d(Q_n - Q_0) \right\}\nonumber \\
%     &\qquad\qquad + \left\{ P_0 (\Gamma_n \psi_n) + \iint \Gamma_n \mu_n \, dQ_0 \, dF_0 - P_0 (\Gamma_n \theta_0) \right\}\label{eq:inf_fn_expansion}
% \end{align}}
We note that we added and subtracted the terms $\iint \Gamma_n \mu_n \, dQ_0 \, dF_0$ and $P_0 \Gamma_n \theta_0$ above.  We address each of the summands in equation~\eqref{eq:inf_fn_expansion} in turn. First, we write
\[ (\d{P}_n - P_0) (\Gamma_n \psi_n - \Gamma_0 \psi_\infty ) = (\d{P}_n - P_0)\left\{(\Gamma_n  -\Gamma_0) \psi_n\right\} + (\d{P}_n - P_0)\left\{\Gamma_0 (\psi_n- \psi_\infty )\right\}.\]
Next, we note that since $\d{P}_n \gamma_n^\circ = \d{P}_n (\Gamma_n \theta_n)$,
$ - (\d{P}_n - P_0) \gamma_n^\circ + \d{P}_n( \Gamma_n \theta_n ) - P_0 \gamma_n^\circ = 0.$
Next, defining $\tilde{D}_{0,h} := P_0(\tilde{w}_{h,a_0,1}K_{h,a_0})$ and $\tilde{D}_{n,h} := \d{P}_n(\tilde{w}_{h,a_0,1}K_{h,a_0})$ we have
\begin{align*}
    (\d{P}_n - P_0)& \gamma_0+ P_0( \Gamma_n \theta_0) - P_0 ( \Gamma_0 \theta_0)\\
    &=   e_1^T  \b{D}^{-1}_{0,h}\left(\b{D}_{n,h} - \b{D}_{0,h }\right) \b{D}^{-1}_{0,h}P_0 \left(w_{h,a_0, 1} K_{h, a_0} \theta_0 \right) \\
    &\qquad - e_3^T c_{0,2} \tau_n^2 \b{D}^{-1}_{0, b}(\b{D}_{n,b} - \b{D}_{0,b}) \b{D}^{-1}_{0, b} P_0\left(w_{b,a_0,2} K_{b,a_0} \theta_0 \right) \\
    &\qquad +  e_1^T \tau_n^2  \b{D}^{-1}_{0,h} \left(\tilde{D}_{n,h} -  \b{D}_{n,h} \b{D}^{-1}_{0,h} \tilde{D}_{0,h} \right) e_3^T \b{D}^{-1}_{0, b}P_0\left(w_{b,a_0,2} K_{b,a_0} \theta_0 \right) \\
    &\qquad + e_1^T  \left(\b{D}^{-1}_{n,h} - \b{D}^{-1}_{0,h}\right) P_0 \left( w_{h,a_0, 1} K_{h, a_0} \theta_0\right) - e_3^T c_{0,2} \tau_n^2 \left(\b{D}^{-1}_{n,b} - \b{D}^{-1}_{0, b}\right)  P_0\left(w_{b,a_0,2} K_{b,a_0} \theta_0 \right) \\
    &\qquad - \tau_n^2  (c_{n,2} - c_{0,2}) e_3^T\b{D}^{-1}_{0,b}  P_0\left(w_{b,a_0,2} K_{b,a_0} \theta_0 \right)- \tau_n^2  (c_{n,2} - c_{0,2}) e_3^T\left(\b{D}^{-1}_{n,b}-\b{D}^{-1}_{0,b}\right) P_0\left(w_{b,a_0,2} K_{b,a_0} \theta_0 \right)\\
    &= e_1^T  \left[ \b{D}^{-1}_{0,h}\b{D}_{n,h}(\b{D}^{-1}_{0,h} - \b{D}^{-1}_{n,h}) +  (\b{D}^{-1}_{n,h} - \b{D}^{-1}_{0,h})\right] P_0\left(w_{h,a_0, 1} K_{h, a_0} \theta_0\right) \\
    &\qquad  - e_3^T c_{0,2} \tau_n^2 \left[ \b{D}^{-1}_{0, b}\b{D}_{n,b}(\b{D}^{-1}_{0, b}- \b{D}^{-1}_{n,b})+(\b{D}^{-1}_{n,b} - \b{D}^{-1}_{0, b}) \right]P_0 \left( w_{b,a_0,2} K_{b,a_0} \theta_0 \right) \\
    &\qquad + e_1^T \tau_n^2   \left(\b{D}^{-1}_{0,h} \tilde{D}_{n,h} -  \b{D}^{-1}_{0,h}\b{D}_{n,h} \b{D}^{-1}_{0,h} \tilde{D}_{0,h} - \b{D}_{n,h}^{-1} \tilde{D}_{n,h} +  \b{D}_{0,h}^{-1} \tilde{D}_{0,h} \right) e_3^T \b{D}^{-1}_{0, b}P_0\left(w_{b,a_0,2} K_{b,a_0} \theta_0 \right) \\
    &\qquad - \tau_n^2  (c_{n,2} - c_{0,2}) e_3^T\left(\b{D}^{-1}_{n,b}-\b{D}^{-1}_{0,b}\right) P_0\left(w_{b,a_0,2} K_{b,a_0} \theta_0 \right) \\
    &= e_1^T \b{D}^{-1}_{0,h} ( \b{D}_{0,h} - \b{D}_{n,h})(\b{D}^{-1}_{n,h} - \b{D}^{-1}_{0,h})  P_0 \left( w_{h,a_0, 1} K_{h, a_0} \theta_0 \right) \\
    &\qquad - e_3^T c_{0,2} \tau_n^2 \b{D}^{-1}_{0, b}(\b{D}_{0, b} - \b{D}_{n,b})(\b{D}^{-1}_{n,b} - \b{D}^{-1}_{0, b})P_0\left(w_{b,a_0,2} K_{b,a_0} \theta_0 \right) \\
    &\qquad - e_1^T \tau_n^2   \left(\b{D}^{-1}_{n,h} - \b{D}^{-1}_{0,h} \right) \left[ \left( \tilde{D}_{n,h} -\tilde{D}_{0,h} \right) + \b{D}_{0,h}^{-1}\left( \b{D}_{0,h} -  \b{D}_{n,h}\right) \tilde{D}_{0,h}   \right] e_3^T \b{D}^{-1}_{0, b}P_0\left(w_{b,a_0,2} K_{b,a_0} \theta_0 \right) \\
    &\qquad - \tau_n^2  (c_{n,2} - c_{0,2}) e_3^T\left(\b{D}^{-1}_{n,b}-\b{D}^{-1}_{0,b}\right) P_0\left(w_{b,a_0,2} K_{b,a_0} \theta_0 \right),
\end{align*}
which equals $R_{n,h,b,a_0,6}$. We address the next two terms in equation~\eqref{eq:inf_fn_expansion} together:
\begin{align*}
    \iint \Gamma_n \mu_n &\left\{ dQ_n \, d(F_n - F_0) + dF_n \, d(Q_n - Q_0) + dQ_n \, dF_0 + dF_n \, dQ_0 - dF_n \, dQ_n - dQ_0 \, dF_0\right\} \\
    &\qquad - \iint \Gamma_0 \mu_\infty \left\{ dQ_0 \, d(F_n - F_0) + dF_0 \, d(Q_n - Q_0) \right\} \\
    &=\iint \Gamma_n \mu_n \left\{ dQ_n \, dF_n - dQ_0 \, dF_0\right\} - \iint \Gamma_0 \mu_n \left\{ dQ_0 \, d(F_n - F_0) + dF_0 \, d(Q_n - Q_0) \right\}\\
    &\qquad+ \iint \Gamma_0 (\mu_n- \mu_\infty) \left\{ dQ_0 \, d(F_n - F_0) + dF_0 \, d(Q_n - Q_0) \right\} \\
    &=\iint \Gamma_n \mu_n \left\{ dQ_n \, dF_n - dQ_0 \, dF_0 -dQ_0 \, d(F_n - F_0) - dF_0 \, d(Q_n - Q_0)\right\}\\
    &\qquad+ \iint \Gamma_0 (\mu_n- \mu_\infty) \left\{dQ_0 \, d(F_n - F_0) + dF_0 \, d(Q_n - Q_0) \right\}\\
    &\qquad + \iint (\Gamma_n - \Gamma_0) \mu_n \left\{ dQ_0 \, d(F_n - F_0) + dF_0 \, d(Q_n - Q_0) \right\} .
\end{align*}
We then note that
\begin{align*}
    \iint \Gamma_n \mu_n &\left\{ dQ_n \, dF_n - dQ_0 \, dF_0 -dQ_0 \, d(F_n - F_0) - dF_0 \, d(Q_n - Q_0)\right\}\\
    &= \iint \Gamma_n \mu_n \left\{ dQ_n \, dF_n  -dQ_0 \, dF_n - dF_0 \, dQ_n + dQ_0 \, dF_0\right\} \\
    &= \iint \Gamma_n \mu_n \, d(Q_n - Q_0) \, d(F_n - F_0),
\end{align*}
which equals $R_{n,h,b,a_0,5}$. Finally, we have
\begin{align*}
    P_0 (\Gamma_n \psi_n) + \iint \Gamma_n \mu_n \, dQ_0 \, dF_0 - P_0 (\Gamma_n \theta_0) &= P_0\left (\Gamma_n \frac{\mu_0 - \mu_n}{g_n}\right) + \iint \Gamma_n (\mu_n - \mu_0) \, dQ_0 \, dF_0 \\
    &= \iint \Gamma_n\left\{\frac{g_0}{g_n} (\mu_0 - \mu_n) + (\mu_n - \mu_0)\right\} \, dQ_0 \, dF_0 \\
    &= \iint \Gamma_n\left(1 - \frac{g_0}{g_n}\right) (\mu_n - \mu_0) \, dQ_0 \, dF_0,
\end{align*}
which equals  $R_{n,h,b,a_0,4}$. Putting these derivations back into equation~\eqref{eq:inf_fn_expansion}, we have
\begin{align*}
    (\d{P}_n - P_0)&\left(\phi_{n, h,b,a_0}^* -  \phi_{\infty, h,b,a_0}^*\right) + \left\{ \theta_{n,h,b}(a_0) - \theta_{0,h,b}(a_0) + P_0 \phi_{n, h,b,a_0}^* \right\} \\
    &= (\d{P}_n - P_0)\left\{ (\Gamma_n  -\Gamma_0) \psi_n\right\} + (\d{P}_n - P_0)\left\{\Gamma_0 (\psi_n- \psi_\infty )\right\} \\
    &\qquad+ \iint \Gamma_0 (\mu_n- \mu_\infty) \left\{ dQ_0 \, d(F_n - F_0) + dF_0 \, d(Q_n - Q_0) \right\} \\
    &\qquad + \iint (\Gamma_n - \Gamma_0) \mu_n \left\{ dQ_0 \, d(F_n - F_0) + dF_0 \, d(Q_n - Q_0) \right\} \\
    &\qquad + R_{n,h,b,a_0,4} + R_{n,h,b,a_0,5} + R_{n,h,b,a_0,6} \\
    &=  (\d{P}_n - P_0)\left[ \left\{\Gamma_0 \left(\psi_n + \int \mu_n \, dQ_0 \right) + \int \Gamma_0 \mu_n \, dF_0 \right\} -\left\{\Gamma_0 \left(\psi_\infty + \int \mu_\infty \, dQ_0 \right) + \int \Gamma_0 \mu_\infty \, dF_0 \right\} \right]  \\
    &\qquad + (\d{P}_n - P_0)\left\{ (\Gamma_n  -\Gamma_0) \left(\psi_n + \int \mu_n \, dQ_0\right) + \int(\Gamma_n - \Gamma_0) \mu_n \, dF_0 \right]\\
    &\qquad+ R_{n,h,b,a_0,4} + R_{n,h,b,a_0,5} + R_{n,h,b,a_0,6},
\end{align*}
which equals $\sum_{j=2}^6 R_{n,h,b,a_0,j}$.
\end{proof}

%% file: supp/theorem1.tex
\clearpage

\section{Proof of Theorems}

\begin{proof}[\bfseries{Proof of Theorem~1}]
By Lemma~\ref{lm:first-order-decomposition}, we have
\[ \theta_{n,h,b}(a_0) - \theta_0(a_0) =  \d{P}_n \phi_{\infty, h, b,a_0}^* + \sum_{j=1}^6 R_{n,h,b,a_0,j}.\]
By Lemma~\ref{lm:R1}, $R_{n,h,b,a_0,1} = \fasterthandet\left( h^2\right)$. By Corollary~\ref{cor:R2R3} and Lemmas~\ref{lm:R4},~\ref{lemma:R5}, and~\ref{lm:R6}, $R_{n,h,b,a_0,j} = \fasterthan\left( \{nh\}^{-1/2}\right)$ for each $j \in \{2, \dotsc, 6\}$. This establishes the first claim. By Lemma~\ref{lm:lindeberg_feller_CLT}, $(nh)^{1/2}\d{P}_n \phi_{\infty, h, b,a_0}^*$ converges in distribution to the claimed limit distribution. If $n h^{5} = \boundeddet(1)$, then $(nh)^{1/2} R_{n,h,b,a_0,1} = \fasterthandet\left(\{nh^5\}^{1/2}\right) = \fasterthandet(1)$, so
\[ (nh)^{1/2} \left[\theta_{n,h,b}(a_0) - \theta_0(a_0)\right] =   (nh)^{1/2}\d{P}_n \phi_{\infty, h, b,a_0}^* + \fasterthan(1),\]
so $(nh)^{1/2} \left[\theta_{n,h,b}(a_0) - \theta_0(a_0)\right]$ converges to this same limit. For the final statement, by Lemma~\ref{lemma:covar}, $\sigma_{n,h,b}^2(a_0) = h\d{P}_n ( \phi_{n,h,b,a_0}^*)^2$ satisfies $\sigma_{n,h,b}^2(a_0) - h P_0( \phi_{\infty,h,b,a_0}^*)^2 = \fasterthan(1)$. By Lemma~\ref{lm:lindeberg_feller_CLT}, $ h P_0( \phi_{\infty,h,b,a_0}^*)^2$ converges to the variance of the limit distribution. Hence, $\sigma_{n,h,b}(a_0)$ converges in probability to the standard deviation of the limit distribution, so the final statement follows by Slutsky's theorem.
\end{proof}

\begin{proof}[\bfseries{Proof of Theorem~2}]
By Theorem~1 and since $nh^5 = \boundeddet(1)$, we have
\[ (nh)^{1/2} \begin{pmatrix}\theta_{n,h,b}(a_1) -\theta_0(a_1)\\ \vdots \\ \theta_{n,h,b}(a_m) - \theta_0(a_m) \end{pmatrix} = (nh)^{1/2} \d{P}_n\begin{pmatrix}\phi_{\infty, h, b,a_1}^* \\ \vdots \\ \phi_{\infty, h, b,a_m}^*  \end{pmatrix} + \fasterthan(1).\]
The result follows by Lemma~\ref{lm:joint_convergence} and Slutsky's theorem. 
\end{proof}

%We recall that $\{Z_{\infty,h,b}(a_0) : a_0\in\s{A}_0\}$ is defined as a mean-zero Gaussian process with covariance function $(u,v) \mapsto \Sigma_{\infty,h,b}(u,v) := P_0(\phi^*_{\infty, h, b, u }\phi^*_{\infty, h, b, v}) / [\sigma_{\infty,h,b}(u)\sigma_{\infty,h,b}(v)]$. We also recall that $\s{A}_n$ is a deterministic grid in $\s{A}_0$ with mesh $\omega_n := \sup_{a_0 \in \s{A}_0} \inf_{a' \in \s{A}_n} |a_0 - a'|$.

\begin{proof}[\bfseries{Proof of Theorem~3}]
By Lemma~\ref{lm:first-order-decomposition}, we have
\[ \sup_{a\in\mathcal{A}_0}\left|\theta_{n,h,b}(a) - \theta_0(a)\right| =  \sup_{a\in\mathcal{A}_0}\left|\d{P}_n \phi_{\infty, h, b,a}^* + \sum_{j=1}^6 R_{n,h,b,a,j}\right|.\]
By Lemma~\ref{lm:R1}, $\sup_{a\in\mathcal{A}_0}\left|R_{n,h,b,a,1}\right| = \boundeddet(h^{2+\delta_4})$ for some $\delta_4 > 0$. By Corollary~\ref{cor:supR2R3} and Lemmas~\ref{lm:R4}, \ref{lemma:R5} and \ref{lm:R6}, $\sup_{a\in\mathcal{A}_0} \left|R_{n,h,b,a,j} \right|= \fasterthan\left(\{nh\log n\}^{-1/2}\right)$ for $j \in \{2, \dots\, 6\}$. We now write
\begin{align*}
\sup_{a\in\mathcal{A}_0} \left| \d{P}_n \phi_{\infty, h, b,a}^* \right| &= (nh)^{-1/2} \sup_{a\in\mathcal{A}_0} \left|\sigma_{\infty,h,b}(a) \d{G}_n \frac{h^{1/2}\phi_{\infty, h, b,a}^*}{\sigma_{\infty,h,b}(a)} \right| \\
&\leq (nh)^{-1/2} \sup_{a\in\mathcal{A}_0} \left|\sigma_{\infty,h,b}(a) \right| \sup_{a\in\mathcal{A}_0} \left| \d{G}_n \frac{h^{1/2}\phi_{\infty, h, b,a}^*}{\sigma_{\infty,h,b}(a)} \right| \\
&= (nh)^{-1/2} \sup_{a\in\mathcal{A}_0} \left|\sigma_{\infty,h,b}(a) \right| \left[ \sup_{a\in\mathcal{A}_0}  \left| \d{G}_n \frac{h^{1/2}\phi_{\infty, h, b,a}^*}{\sigma_{\infty,h,b}(a)} \right| -  \sup_{a\in\mathcal{A}_0} \left| Z_{\infty, h,b}(a)\right| \right]\\
&\qquad + (nh)^{-1/2} \sup_{a\in\mathcal{A}_0} \left|\sigma_{\infty,h,b}(a) \right| \sup_{a\in\mathcal{A}_0} \left| Z_{\infty, h,b}(a)\right|.
\end{align*}
By Lemma~\ref{lm:unif_bounded_variance}, $\sup_{a\in\mathcal{A}_0} \left|\sigma_{\infty,h,b}(a) \right| = \boundeddet(1)$, by Lemma~\ref{lm:subgaussian}, $\sup_{a\in\mathcal{A}_0} \left| Z_{\infty, h,b}(a)\right| = \bounded(\{\log h^{-1}\}^{1/2} )$, which is $\bounded(\{\log n\}^{1/2})$ since $nh \longrightarrow \infty$, and by Lemma~\ref{lm:gaussian_approx}, $\sup_{a\in\mathcal{A}_0} \left|\d{G}_n \frac{h^{1/2}\phi_{\infty, h, b,a}^*}{\sigma_{\infty,h,b}(a)} \right| -  \sup_{a\in\mathcal{A}_0} \left| Z_{\infty, h,b}(a)\right| = \fasterthan(1)$. Hence, 
\[\sup_{a\in\mathcal{A}_0} \left| \d{P}_n \phi_{\infty, h, b,a}^* \right| = \bounded\left( \left\{ nh / \log n\right\}^{-1/2}\right). \]
This proves the first claim. The second claim follows by Lemma~\ref{lemma:covar}. 

For the final claim, by the triangle inequality, we have
\begin{align*}
 &\sup_{t \in \d{R}}\left|  P_0 \left(\sup_{a_0 \in \s{A}_0}(nh)^{1/2} \left| \frac{\theta_{n,h,b}(a_0) - \theta_0(a_0)}{\sigma_{n,h,b}(a_0)} \right| \leq t\right) -  P_0 \left(\max_{a_0 \in \s{A}_n}\left| Z_{n, h, b}(a_0)\right| \leq t\mid \b{O}_n\right) \right| \\
 &\qquad \leq  \sup_{t \in \d{R}}\left|  P_0 \left(\sup_{a_0 \in \s{A}_0}(nh)^{1/2} \left| \frac{\theta_{n,h,b}(a_0) - \theta_0(a_0)}{\sigma_{n,h,b}(a_0)} \right| \leq t\right) -  P_0 \left(\sup_{a_0 \in \s{A}_0}\left| Z_{\infty, h, b}(a_0)\right| \leq t\right) \right| \\
 &\qquad\qquad +  \sup_{t \in \d{R}}\left| P_0 \left(\sup_{a_0 \in \s{A}_0}\left| Z_{\infty, h, b}(a_0)\right| \leq t\right) -  P_0 \left(\max_{a_0 \in \s{A}_n}\left| Z_{\infty, h, b}(a_0)\right| \leq t\right) \right| \\
 &\qquad\qquad +  \sup_{t \in \d{R}}\left| P_0 \left(\max_{a_0 \in \s{A}_n}\left| Z_{\infty, h, b}(a_0)\right| \leq t\right) -  P_0 \left(\max_{a_0 \in \s{A}_n}\left| Z_{n, h, b}(a_0)\right| \leq t\mid \b{O}_n\right) \right|.
\end{align*}
The first term on the right hand side is $\fasterthandet(1)$ by Lemma~\ref{lemma:sup_empirical_process}. The second term on the right hand side is $\fasterthandet(1)$ by Lemma~\ref{lemma:finite_approx}. The last term on the right hand side is $\fasterthan(1)$ by Lemma~\ref{lemma:approx_process}. 
\end{proof}

%% file: supp/CLT.tex
\clearpage

\section{Analysis of the leading term}

We first establish the following result on the behavior of $\b{D}^{-1}_{0, h, a_0,1}$ and $\b{D}^{-1}_{0, b, a_0,3}$ as $h, b \longrightarrow 0$.
\begin{lemma}\label{lm:D0_altform}
If conditions~\ref{cond:bounded_K} and~\ref{cond:cont_density} hold, then $\b{D}^{-1}_{0, h, a_0,1} = f_0(a_0)^{-1} \b{S}_{2}^{-1} + \boundeddet(h)$ and $c_{0, h, a_0,2} = c_2 + \boundeddet(h)$ as $h \longrightarrow 0$, and $\b{D}^{-1}_{0, b, a_0,2} = f_0(a_0)^{-1} \b{S}_{3}^{-1} + \boundeddet(b)$ as $b \longrightarrow 0$.

If conditions~\ref{cond:bounded_K} and~\ref{cond:holder_smooth_theta} hold, then $\sup_{a_0 \in \s{A}_0} \left\| \b{D}^{-1}_{0, h, a_0,1} - f_0(a_0)^{-1} \b{S}_{2}^{-1}\right\|_{\infty} = \boundeddet(h)$, $\sup_{a_0 \in \s{A}_0} | c_{0, h, a_0,2} -c_2| = \boundeddet(h)$, and $\sup_{a_0 \in \s{A}_0} \left\| \b{D}^{-1}_{0, h, a_0,2} - f_0(a_0)^{-1} \b{S}_{3}^{-1}\right\|_{\infty} = \boundeddet(h)$ as $h \longrightarrow 0$.
\end{lemma}
\begin{proof}[\bfseries{Proof of Lemma~\ref{lm:D0_altform}}] The proof is analogous to that of Section 3.7 of \cite{fan1996} and many others. By definition of $\b{D}_{0, h, a_0,1}$ and the change of variables $u = (t-a_0)/h$,
%  Lemma 1 of~\cite{fan1994robust}, 
\begin{align*}
   \b{D}_{0, h, a_0,1}[j,k] &= \int \left(\frac{t - a_0}{h}\right)^{j+k-2}h^{-1}K\left(\frac{t - a_0}{h}\right) \, dF_0(t) = \int u^{j+k-2}K(u) f_0(a_0+uh)\, du.
\end{align*}
By~\ref{cond:cont_density}, $f_0(a_0+uh) = f_0(a_0) + \boundeddet(uh)$, and since $K$ has support $[-1,1]$ by~\ref{cond:bounded_K}, we then have
\[\int u^{j+k-2}K(u) f_0(a_0+uh)\, du = f_0(a_0)\int u^{j+k-2}K(u)\, du + \boundeddet(h) = f_0(a_0) c_{j+k-2}+\boundeddet(h).\]
% If $a_0$ is on the lower boundary of the support of $f_0$, then $f_0(a_0 + uh) = 0$ for $u < 0$, so
% \[\int_{-1}^1 u^{j+k-2}K(u) f_0(a_0+uh)\, du = f_0(a_0)\int_0^1 u^{j+k-2}K(u)\, du + \boundeddet(h) = f_0(a_0) c_{j+k-2}^b+\boundeddet(h).\]
Hence, $\b{D}_{0, h, a_0,1} = f_0(a_0) \b{S}_2 + \boundeddet(h)$, so $\b{D}_{0, h, a_0,1}^{-1} = f_0(a)^{-1}\b{S}_{2}^{-1} + \boundeddet(h)$ since $f(a_0) > 0$. The proof for $\b{D}_{0, b, a_0, 2}$ is analogous. By the same logic, we have $P_0 (\tilde{w}_{h,a_0,1} K_{h,a_0}) = f_0(a_0)(c_2, c_3)^T + \boundeddet(h)$. Hence,
\begin{align*}
c_{0,h,a_0,2} &= e_1^T\b{D}_{0, h, a_0,1}^{-1} P_0 (\tilde{w}_{h,a_0,1} K_{h,a_0}) = e_1^T\left[f_0(a_0)^{-1} \b{S}_2^{-1} + \boundeddet(h)\right] \left[ f_0(a_0)(c_2, 0)^T + \boundeddet(h)\right] = e_1^T\b{S}_2^{-1} (c_2, 0)^T  + \boundeddet(h) \\
&= c_2 + \boundeddet(h).
\end{align*}
For the uniform result, we have for each $j,k \in \{1,2\}$,
\begin{align*}
    \sup_{a_0 \in \s{A}_0}\left|\b{D}_{0, h, a_0,1}[j,k] - f_0(a_0)c_{j+k+2} \right|  &= \sup_{a_0 \in \s{A}_0}\left|\int u^{j+k-2}K(u) \left[ f_0(a_0+uh) - f_0(a_0)\right]\, du  \right|\\
    &\leq \sup_{a_0 \in \s{A}_0} \int \left|u\right|^{j+k-2}K(u)\left| f_0(a_0+uh) - f_0(a_0) \right|\, du \\
    &\leq Ch 
\end{align*}
for some $C < \infty$ because $f_0$ is Lipschitz on $\s{A}_{\delta_3}$ and $K$ is uniformly bounded with compact support. The result follows, and a similar argument yields the results for $c_{0,h,a_0,2}$ and $\b{D}_{0,b,a_0,2}^{-1}$.
\end{proof}

\begin{lemma} \label{lm:lindeberg_feller_CLT} If \ref{cond:bounded_K}--\ref{cond:bandwidth},~\ref{cond:doubly_robust}(a), and~\ref{cond:cont_density} hold, then
\[(nh)^{1/2}\d{P}_n \phi^*_{\infty, h, b, a_0} \indist \mathcal{N}\left(0, V_{ K, \tau}f_0(a_0)^{-1} \sigma_0^2(a_0)\right),\]
where
\[ V_{K,\tau} :=\int \left\{ K(u) - \tau^3 c_2\frac{(\tau u)^2 - c_2 }{c_4 - c_2^2}  K(\tau u)\right\}^2 \, du =  c_0^* - 2\tau^3 c_2 \frac{\tau^2 c_{2,\tau}^*- c_2c_{0,\tau}^* }{c_4 - c_2^2} + \tau^5 c_2^2 \frac{ c_4^*-2c_2 c_2^* + c_2^2 c_0^*}{\left( c_4 - c_2\right)^2}\]
is a positive, finite constant for all $\tau \in (0, \infty)$ and kernels $K$ satisfying~\ref{cond:bounded_K}. Furthermore, $hP_0(\phi^*_{\infty, h, b, a_0})^2$ converges to $V_{ K, \tau}\sigma_0^2(a_0)/ f_0(a_0)$ as $h \longrightarrow 0$.
\end{lemma}
\begin{proof}[\bfseries{Proof of Lemma~\ref{lm:lindeberg_feller_CLT}}]
By adding and subtracting $\d{P}_n\Gamma_{0,h,b,a_0} \theta_0$, we can rewrite $\d{P}_n \phi^*_{\infty, h, b, a_0}$ as
\begin{align*}
   \d{P}_n\phi_{\infty, h,b,a_0}^* 
   &= \d{P}_n \left\{ \Gamma_{0,h,b,a_0} (\xi_\infty - \theta_0)\right\} +\d{P}_n\left\{\Gamma_{0,h,b,a_0} \theta_0-\gamma_{0,h,b,a_0} \right\} + \d{P}_n\left\{ \int\Gamma_{0,h,b,a_0} \left( \mu_\infty - \int \mu_\infty \, dQ_0 \right) \, dF_0 \right\} %.\label{eq:norm_decomposition}
\end{align*}
We show that only the first term contributes to the limit distribution; the remaining two terms are each  $\fasterthan(\{nh\}^{-1/2})$. We use the Lyapunov central limit theorem for triangular arrays to demonstrate that the first term is asymptotically normal. The Lyapunov CLT implies that if for each $n$, $X_{n,1}, X_{n,2}, \dots X_{n,n}$ are IID, mean-zero random variables that satisfy (1) $\lim_{n\to\infty}\n{Var}(X_{n,i}) = \Sigma > 0$, and  (2) there exists $\delta > 0$ such that $n^{-\delta / 2}E|X_{n,i}|^{2+\delta} = \fasterthandet(1)$, then $n^{-1/2}\sum^n_{i=1} X_{n,i}$ converges in distribution to $\mathcal{N}(0, \Sigma)$. We can write $(nh)^{1/2}\d{P}_n \left\{ \Gamma_{0,h,b,a_0} (\xi_\infty - \theta_0)\right\} = n^{-1/2}\sum_{i=1}^n X_{n,i}$ for
\[ X_{n,i} := h^{1/2}\Gamma_{0,h,b,a_0}(A_i)\left\{\xi_\infty(Y_i, A_i, W_i) - \theta_0(A_i) \right\}. \]
Since~\ref{cond:doubly_robust}(a) holds, $E_0\left[\xi_\infty(Y,A,W) \mid A =a\right] = \theta_0(a)$ for $F_0$-almost every $a \in B_{\delta_1}(a_0)$ by Lemma~\ref{lm:double_robustness}, and by~\ref{cond:bounded_K}, $\Gamma_{0,h,b,a_0}(a) = 0$ for all $a \notin B_{\delta_1}(a_0)$ and $h \leq \delta_1$. Hence, for all $n$ large enough, $E_0[ X_{n,i}] = 0$ for all $i$. We thus have
\begin{align*}
    \n{Var}\left(X_{n,i}\right) &= h E_0\left[\Gamma_{0,h,b,a_0}(A)^2\left\{\xi_\infty(Y, A, W) - \theta_0(A)\right\}^2\right] \\
    &=h  \int  \Gamma_{0,h,b,a_0}(a)^2\sigma_0^2(a) \, dF_0(a) \\
    &=h \int \left\{ e_1^T \b{D}_{0, h, a_0,1}^{-1} w_{h,a_0,1}(a) K_{h,a_0}(a) - e_3^T c_{0,h,a_0,2}\tau_n^2 \b{D}_{0, b,a_0,2}^{-1} w_{b,a_0,2}(a) K_{b,a_0}(a) \right\}^2 \sigma_0^2(a) f_0(a) \, da \\
    &=  \int \left\{ e_1^T \b{D}_{0, h, a_0,1}^{-1} (1,u)^T K(u) - e_3^T c_{0,h,a_0,2}  \tau_n^3 \b{D}_{0, b,a_0,2}^{-1} (1, \{\tau_n u\}, \{\tau_n u\}^2)^T K(\tau_n u) \right\}^2 \sigma_0^2(a_0 + uh) f_0(a_0 + uh) \, du.
\end{align*}
By~\ref{cond:cont_density}, $\sigma_0$ and $f_0$ are continuous at $a_0$. Also, $\b{D}_{0,h,a_0,1}^{-1}$, $c_{0,h,a_0,2}$,  and $\b{D}_{0,b,a_0,2}^{-1}$ converge to $f_0(a_0)^{-1} \b{S}_2^{-1}$, $c_2$, and $f_0(a_0)^{-1} \b{S}_3^{-1}$, respectively, as $b,h \longrightarrow 0$ by Lemma~\ref{lm:D0_altform}, and $\tau_n \longrightarrow \tau \in (0,\infty)$. The preceding display thus converges to 
\begin{align*}
&\sigma_0^2(a_0) f_0(a_0)^{-1}\int \left\{ e_1^T \b{S}_{2}^{-1}(1,u)^T K(u) - e_3^T c_2  \tau^3 \b{S}_3^{-1} (1, \{\tau u\}, \{\tau u\}^2)^T K(\tau u) \right\}^2 \, du.
\end{align*}
Now, $e_1^T \b{S}_{2}^{-1}(1,u)^T = 1$, and using the block structure of $\b{S}_3$, we find that $e_3^T \b{S}_3^{-1} =\frac{1}{c_4 - c_2^2} \left(-c_2, 0, 1 \right)$. Therefore, we can simplify the above to $\sigma_0^2(a_0) f_0(a_0)^{-1} V_{K,\tau}$ for
\begin{align*}
V_{K,\tau} &= \int \left\{ K(u) - \tau^3 c_2\frac{(\tau u)^2 - c_2 }{c_4 - c_2^2}  K(\tau u)\right\}^2 \, du \\
&= \int  K(u)^2 \, du - 2\tau^3 c_2 \frac{\int (\tau u)^2  K(u) K(\tau u) \, du - c_2 \int K(u) K(\tau u) \, du }{c_4 - c_2^2} \\
&\qquad + \tau^6 c_2^2 \frac{ \int ( \tau u)^4 K(\tau u)^2 \,du-2c_2 \int (\tau u)^2 K(\tau u)^2 \, du + c_2^2 \int K(\tau u)^2 \, du}{\left( c_4 - c_2\right)^2} \\
&= c_0^* - 2\tau^3 c_2 \frac{\tau^2 c_{2,\tau}^*- c_2c_{0,\tau}^* }{c_4 - c_2^2} + \tau^5 c_2^2 \frac{ c_4^*-2c_2 c_2^* + c_2^2 c_0^*}{\left( c_4 - c_2\right)^2}.
\end{align*}
% \begin{align*}
%     &\sigma_0^2(a_0) f_0(a_0)^{-1} \left\{ e_1^T \b{S}_{2}^{-1} \int (1,u)^T(1,u) K(u)^2 \,du \, \b{S}_{2}^{-1} e_1 \right.\\
%     &\qquad\qquad - 2c_2 \tau^3 e_1^T \b{S}_{2}^{-1} \int  (1,u)^T (1, \tau u, \{\tau u\}^2, \{\tau u\}^3)  K(u)K(\tau u) \,du  \, \b{S}_{4}^{-1} e_3\\
%     &\qquad\qquad \left.+ c_2^2 \tau^5 e_3^T \b{S}_{4}^{-1}\int  (1, u, u^2, u^3)^T (1, u, u^2, u^3) K(u)^2\, du \, \b{S}_{4}^{-1} e_3\right\} \\
%     &\qquad= \frac{\sigma_0^2(a_0)}{f_0(a_0)}\left(e_1^T \b{S}_{2}^{-1} \b{S}^*_{2}\b{S}_{2}^{-1}e_1 - 2c_2 \tau^3e_1^T \b{S}_{2}^{-1} \b{S}_\tau^*\b{S}_{4}^{-1} e_3 + c_2^2 \tau^5  e_3^T \b{S}_{4}^{-1}\b{S}^*_{4}\b{S}_{4}^{-1} e_3\right).
% \end{align*}
% \tw{simplify integral in first line using block structure again.}
% We note that $e_1^T\b{S}_{2}^{-1} \b{S}^*_{2}\b{S}_{2}^{-1}e_1 = c_0^*$. Using the block structure of $\b{S}_4$ and $\b{S}_\tau^*$, we find that 
% \[ e_1^T \b{S}_{2}^{-1} \b{S}_\tau^*\b{S}_{4}^{-1} e_3 = \frac{ c_{0,\tau}^*  c_2- \tau^2 c_{2,\tau}^*}{c_2^2 - c_4}.\]
% Similarly, using the block structure of $\b{S}_4$ and $\b{S}_4^*$, we find that
% \[ e_3^T \b{S}_{4}^{-1}\b{S}^*_{4}\b{S}_{4}^{-1} e_3 = \frac{c_0^* c_2^2 -2c_2 c_2^* + c_4^*}{(c_2^2 - c_4)^2}.\]
Hence, $\n{Var}(X_{n,i})$ converges to $\sigma_0^2(a_0) f_0(a_0)^{-1} V_{K,\tau}$ as claimed. If $\tau = 1$, then the above simplifies to
%\[e_1^T \b{S}_{2}^{-1} \b{S}_\tau^*\b{S}_{4}^{-1} e_3  =  \frac{c_0^* c_2 - c_2^*}{c_2^2 - c_4}.\]
%Thus,
\[ V_{K,1} =  \int \left\{\frac{c_4 - c_2 u^2 }{c_4 - c_2^2} K(u)  \right\}^2\, du = \frac{c_0^* c_4^2 - 2c_2 c_4 c_2^* + c_2^2 c_4^*}{(c_4 - c_2^2)^2}. \]
Clearly, $V_{K, \tau} \geq 0$, with equality if and only if the expression in the integral is zero identically. When $\tau \neq 1$, the differing supports of $K(u)$ and $K(\tau u)$ guarantees that this is not the case, and if $\tau = 1$ then the expression is zero if and only if $c_2 = c_4 = 0$, which is not the case. Hence, $V_{K, \tau} > 0$ for any $\tau \in (0, \infty)$. Furthermore, by the boundedness of $K$ and since $c_2^2 < c_4$ by Jensen's inequality, $V_{K, \tau}$ is finite for every $\tau \in (0,\infty)$.

For the second condition of the Lyapunov CLT, we first note that
\begin{align*}
    \sup_{|a - a_0| < \delta_1} E_0\left[\left|\xi_\infty(Y, A, W)-\theta_0(A)\right|^{2+\delta_2} \mid A = a \right] \leq C. 
\end{align*}
for some $C < \infty$ since $\mu_\infty$ is uniformly bounded and $g_\infty$ is uniformly bounded away from zero by~\ref{cond:doubly_robust}(a) and $E_0(|Y|^{2+\delta_2} \mid A=a, W=w)$ is uniformly bounded over $a \in B_{\delta_1}(a_0)$ and $w \in \s{W}$ by~\ref{cond:cont_density}. Therefore, for all $n$ large enough,
\begin{align*}
    E|X_{n,i}|^{2+\delta_2} &= h^{1 + \delta_2 / 2}\int \left|\Gamma_0(a)\right|^{2 + \delta_2} E_0\left[ \left|\xi_\infty(y,a,w)- \theta_0(a)\right|^{2+\delta_2} \mid A = a\right] f_0(a) \, da \\
    &\leq C h^{1 + \delta_2 / 2}\int \left|\Gamma_0(a)\right|^{2 + \delta_2} f_0(a) \, da.
\end{align*}
By the triangle inequality, we then have
\begin{align*}
    \left\{ \int \left|\Gamma_0(a)\right|^{2 + \delta_2} f_0(a) \, da \right\}^{1/(2+\delta_2)} &\leq \left\{ \int \left|e_1^T \b{D}_{0,h,a_0,1}^{-1} w_{h,a_0,1}(a) K_{h,a_0}(a)\right|^{2+\delta_2} f_0(a) \, da\right\}^{1/(2 + \delta_2)}\\
    &\qquad+ \left\{ \int \left|e_3^T c_{0,h,a_0,2} \tau^2 \b{D}_{0,b,a_0,2}^{-1} w_{b,a_0,2}(a) K_{b,a_0}(a)\right|^{2+\delta_2} f_0(a) \, da \right\}^{1/(2 + \delta_2)}\\
    &= \left\{ \int \left|e_1^T \b{D}_{0,h,a_0,1}^{-1} (1,u)^T K(u)\right|^{2+\delta_2} h^{-(1 + \delta_2)}f_0(a_0 + uh) \, du\right\}^{1/(2 + \delta_2)}\\
    &\qquad+ \left\{ \int \left|e_3^T c_{0,h,a_0,2} \tau^2 \b{D}_{0,b,a_0,2}^{-1} (1, u, u^2)^T K(u)\right|^{2+\delta_2} h^{-(1 + \delta_2)} f_0(a_0 + uh) \, du \right\}^{1/(2 + \delta_2)}.
\end{align*}
By Lemma~\ref{lm:D0_altform},~\ref{cond:bounded_K},~\ref{cond:bandwidth}, and~\ref{cond:cont_density}, the preceding display is bounded up to a constant by $h^{-(1+\delta_2) / (2+\delta_2)}$. Hence, for all $n$ large enough, $E|X_{n,i}|^{2+\delta_2}$ is bounded up to a constant by $h^{-\delta_2/2}$. Therefore, 
\[n^{-\delta_2/2}E|X_{n,i}|^{2+\delta_2} = \boundeddet\left( \left\{ nh\right\}^{-\delta_2 / 2}\right) = \fasterthandet(1) \]
since $nh \longrightarrow \infty$ by assumption.

Next, we show that $\d{P}_n (\Gamma_{0,h,b,a_0} \theta_0-\gamma_{0,h,b,a_0}) = \bounded(\left\{n/h\right\}^{-1/2})$, which implies that it is $\fasterthan(\left\{nh\right\}^{-1/2})$. We note that since $P_0(\Gamma_{0,h,b,a_0} \theta_0) = P_0\gamma_{0,h,b,a_0}$, $E_0 \left[\d{P}_n (\Gamma_{0,h,b,a_0} \theta_0-\gamma_{0,h,b,a_0}) \right] = 0$ and $\n{Var}\left[\d{P}_n \left(\Gamma_{0,h,b,a_0} \theta_0-\gamma_{0,h,b,a_0}\right)\right] = n^{-1}P_0\left(\Gamma_{0,h,b,a_0} \theta_0-\gamma_{0,h,b,a_0}\right)^2$. Thus, it is sufficient to show that $P_0(\Gamma_{0,h,b,a_0} \theta_0-\gamma_{0,h,b,a_0})^2 = \boundeddet(h)$.
We have
\begin{align*}
    & P_0(\Gamma_{0,h,b,a_0} \theta_0-\gamma_{0,h,b,a_0})^2\\
    &\qquad\leq 2 \int \left [ e_1^T \b{D}_{0,h,a_0}^{-1} w_{h,a_0,1}(a) K_{h,a_0}(a)\left\{\theta_0(a)-w_{h,a_0, 1}(a)^T \b{D}^{-1}_{0,h, a_0} P_0 \left(w_{h,a_0, 1} K_{h, a_0} \theta_0\right)\right\}\right]^2  f_0(a) \, da\\
    &\qquad\qquad + 2\int \left[e_3^T c_{0,h,a_0,2}  \tau_n^2 \b{D}_{0,b,a_0,2}^{-1} w_{b,a_0,2}(a) K_{b,a_0}(a)\left\{\theta_0(a) - w_{b,a_0, 2}(a)^T \b{D}^{-1}_{b, a_0, 2} P_0 \left(w_{b,a_0, 2} K_{b, a_0} \theta_0\right)\right\}\right]^2 f_0(a) \,da\\
    &\qquad\qquad + 2 \int \left[e_1^T \tau_n^2  \b{D}_{0,h,a_0}^{-1} \left\{ \tilde{w}_{h,a_0,1}(a) - w_{h,a_0,1}(a)w_{h,a_0,1}(a)^T\b{D}_{0,h,a_0}^{-1} P_0(\tilde{w}_{h,a_0,1} K_{h,a_0})   \right\} K_{h,a_0}(a) \right]^2 f_0(a) \, da  \\
    &\qquad\qquad\qquad \times \left[ e_3^T \b{D}_{0,b,a_0,2}^{-1} P_0 (w_{b,a_0,2} K_{b,a_0} \theta_0) \right]^2\\
    &\qquad=  2 h^{-1}\int \left [ e_1^T \b{D}_{0,h,a_0}^{-1} (1,u)^T K(u)\left\{\theta_0(a_0 + uh)-(1,u)^T \b{D}^{-1}_{0,h, a_0} P_0 \left(w_{h,a_0, 1} K_{h, a_0} \theta_0\right)\right\}\right]^2  f_0(a_0 + uh) \, du\\
    &\qquad\qquad + 2b^{-1}\int \left[e_3^T c_{0,h,a_0,2}  \tau_n^2 \b{D}_{0,b,a_0,2}^{-1} \b{u}_2^T K(u)\left\{\theta_0(a_0 + ub)-  \b{u}_2^T \b{D}^{-1}_{b, a_0, 2} P_0 \left(w_{b,a_0, 2} K_{b, a_0} \theta_0\right)\right\}\right]^2 f_0(a_0 + ub) \,du \\
    &\qquad\qquad + 2 h^{-1} \int \left[e_1^T \tau_n^2  \b{D}_{0,h,a_0}^{-1} \left\{ (u^2, u^3)^T - (1,u) (1,u)^T \b{D}_{0,h,a_0}^{-1} P_0(\tilde{w}_{h,a_0,1} K_{h,a_0})   \right\} K(u) \right]^2 f_0(a_0 + uh) \, du  \\
    &\qquad\qquad\qquad \times \left[ e_3^T \b{D}_{0,b,a_0,2}^{-1} P_0 (w_{b,a_0,2} K_{b,a_0} \theta_0) \right]^2\\
\end{align*}
% \kt{Not sure if this works yet but I am trying 
% \begin{align*}
%     &\left[e_1^T \tau_n^2  \b{D}_{0,h,a_0}^{-1} \left\{ (u^2, u^3)^T - (1,u) (1,u)^T \b{D}_{0,h,a_0}^{-1} P_0(\tilde{w}_{h,a_0,1} K_{h,a_0})   \right\} K(u) \right] \\
%     &\qquad = \left[ \left\{e_1^T \tau_n^2  \b{D}_{0,h,a_0}^{-1}  (u^2, u^3)^TK(u)  - e_1^T \tau_n^2  \b{D}_{0,h,a_0}^{-1} (1,u) (1,u)^T \b{D}_{0,h,a_0}^{-1} P_0(\tilde{w}_{h,a_0,1} K_{h,a_0})K(u)    \right\} \right]
% \end{align*}
% \begin{align*}
%     &\b{D}_{0,h,a_0}^{-1} P_0(\tilde{w}_{h,a_0,1} K_{h,a_0}) \\
%     &\qquad= \b{D}_{0,h,a_0}^{-1} P_0(\tilde{w}_{h,a_0,1} K_{h,a_0}) +  f(a_0)^{-1}\b{S}_2^{-1} P_0(\tilde{w}_{h,a_0,1} K_{h,a_0}) -  f(a_0)^{-1}\b{S}_2^{-1} P_0(\tilde{w}_{h,a_0,1} K_{h,a_0})\\
%     &\qquad= (\b{D}_{0,h,a_0}^{-1} -  f(a_0)^{-1}\b{S}_2^{-1} )P_0(\tilde{w}_{h,a_0,1} K_{h,a_0}) +   f(a_0)^{-1}\b{S}_2^{-1} P_0(\tilde{w}_{h,a_0,1} K_{h,a_0})\\
%     &\qquad= (\b{D}_{0,h,a_0}^{-1} -  f(a_0)^{-1}\b{S}_2^{-1} )P_0(\tilde{w}_{h,a_0,1} K_{h,a_0}) +   \b{S}_2^{-1} \int (u^2, u^3)^T K(u) \frac{f_0(a_0 + uh)}{f_0(a_0)}\, du\\
%     &\qquad= (\b{D}_{0,h,a_0}^{-1} -  f(a_0)^{-1}\b{S}_2^{-1} )P_0(\tilde{w}_{h,a_0,1} K_{h,a_0}) +   \b{S}_2^{-1} \int (u^2, u^3)^T K(u) \left\{\frac{f_0(a_0 + uh)}{f_0(a_0)}-1\right\}\, du \\
%     &\qquad\qquad + \b{S}_2^{-1} (c_2, 0)^T
% \end{align*}
% }
where $\b{u}_2$ denotes the vector $(1,u,u^2)$. For the first term, we define $\beta_{h,a_0,1} := (\theta_0(a_0), h\theta^{(1)}_0(a_0))^T$. We then write
\begin{align*}
    &(1,u)^T \b{D}^{-1}_{0,h, a_0} P_0 \left(w_{h,a_0, 1} K_{h, a_0} \theta_0\right) - \theta_0(a_0) - uh \theta_0^{(1)}(a_0) \\
    &\qquad= (1,u)^T \b{D}^{-1}_{0,h, a_0,1} P_0 \left(w_{h,a_0, 1} K_{h, a_0} \theta_0\right) - (1,u)^T \b{D}^{-1}_{0,h, a_0,1} \b{D}_{0,h, a_0,1}\beta_{h,a_0,1} \\
    &\qquad= (1,u)^T \b{D}^{-1}_{0,h, a_0,1} \left[ P_0 \left(w_{h,a_0, 1} K_{h, a_0} \theta_0\right) - P_0 \left(w_{h,a_0, 1} K_{h, a_0} w_{h,a_0, 1}^T\right)\beta_{h,a_0,1} \right]\\
    &\qquad= (1,u)^T \b{D}^{-1}_{0,h, a_0,1} P_0 \left(w_{h,a_0, 1} K_{h, a_0}\left[ \theta_0 - w_{h,a_0, 1}^T\beta_{h,a_0,1} \right]\right)\\
    &\qquad= (1,u)^T \b{D}^{-1}_{0,h, a_0,1} \int w_{h,a_0, 1}(a) K_{h, a_0}(a)\left[ \theta_0(a) - w_{h,a_0, 1}(a)^T\beta_{h,a_0,1} \right] f_0(a) \, da\\
    &\qquad= (1,u)^T \b{D}^{-1}_{0,h, a_0,1} \int (1,v)^T K(v)\left[ \theta_0(a_0 + vh) -(1,v)^T\beta_{h,a_0,1} \right] f_0(a_0 + vh) \, dv\\
    &\qquad= (1,u)^T \b{D}^{-1}_{0,h, a_0,1} \int (1,v)^T K(v)\left[ \theta_0(a_0 + vh) - \theta_0(a_0) -vh  \theta_0^{(1)}(a_0) \right] f_0(a_0 + vh) \, dv\\
    &\qquad= (1,u)^T \b{D}^{-1}_{0,h, a_0,1} \fasterthandet(h),
\end{align*}
where we have also used~\ref{cond:bounded_K} and~\ref{cond:cont_density}. Therefore,
\begin{align*}
    &h^{-1}\int \left [ e_1^T \b{D}_{0,h,a_0,1}^{-1} (1,u)^T K(u)\left\{\theta_0(a_0 + uh)-(1,u)^T \b{D}^{-1}_{0,h, a_0,1} P_0 \left(w_{h,a_0, 1} K_{h, a_0} \theta_0\right)\right\}\right]^2  f_0(a_0 + uh) \, du\\
    &\qquad = h^{-1}\int \left [ e_1^T \b{D}_{0,h,a_0,1}^{-1} (1,u)^T K(u)\left\{\left[ \theta_0(a_0 + uh)- \theta_0(a_0) - uh \theta_0^{(1)}(a_0)\right]  \right. \right. \\
    &\qquad\qquad \qquad \left.\left. - \left[ (1,u)^T \b{D}^{-1}_{0,h, a_0,1} P_0 \left(w_{h,a_0, 1} K_{h, a_0} \theta_0\right)-\theta_0(a_0) - uh \theta_0^{(1)}(a_0)\right]\right\} \right]^2  f_0(a_0 + uh) \, du \\
    &\qquad = h^{-1}\int \left [ e_1^T \b{D}_{0,h,a_0,1}^{-1} (1,u)^T K(u) \fasterthandet(h) \right]^2 f_0(a_0 + uh) \, du \\
    &\qquad = \fasterthandet(h),
\end{align*}
using Lemma~\ref{lm:D0_altform},~\ref{cond:bounded_K},~\ref{cond:bandwidth}, and~\ref{cond:cont_density}.

For the second term, we similarly define $\beta_{b,a_0,2} := (\theta_0(a_0), b\theta^{(1)}_0(a_0), 0)^T$. By an identical calculation, we can then show that
\begin{align*}
    \b{u}_2^T \b{D}^{-1}_{0,b, a_0, 2} P_0 \left(w_{b,a_0, 2} K_{b, a_0} \theta_0\right) - \theta_0(a_0) - ub \theta_0^{(1)}(a_0) = \b{u}_2^T \b{D}^{-1}_{0,b, a_0, 2} \fasterthandet(b)
\end{align*}
and hence
\begin{align*}
    b^{-1}\int \left[e_3^T c_{0,h,a_0,2}  \tau_n^2 \b{D}_{0,b,a_0,2}^{-1} \b{u}_2^T K(u)\left\{\theta_0(a_0 + ub)-  \b{u}_2^T \b{D}^{-1}_{b, a_0, 2} P_0 \left(w_{b,a_0, 2} K_{b, a_0} \theta_0\right)\right\}\right]^2 f_0(a_0 + ub) \,du = \fasterthandet(b).
\end{align*}
Note that we can even show that this term is $\fasterthandet(b^3)$ using a second-order Taylor expansion, though this is unnecessary for the purposes of our proof.

A similar derivation as in the proof of Lemma~\ref{lm:D0_altform} yields $P_0 (w_{b,a_0,2} K_{b,a_0} \theta_0) = \b{S}_2 e_1 \theta_0(a_0) f_0(a_0) + \boundeddet(b)$. Hence, using Lemma~\ref{lm:D0_altform}, we have
\begin{align*}
    \left[ e_3^T \b{D}_{0,b,a_0,2}^{-1} P_0 (w_{b,a_0,2} K_{b,a_0} \theta_0) \right]^2 &= \left[ \left\{ e_3^Tf_0(a_0)^{-1} \b{S}_2^{-1}+\boundeddet(b)\right\} \left\{\b{S}_2 e_1 \theta_0(a_0) f_0(a_0) + \boundeddet(b)\right\} \right]^2\\
    &= \left[ e_3^T\b{S}_2^{-1}\b{S}_2 e_1 \theta_0(a_0) +\boundeddet(b)  \right]^2 \\
    &= \boundeddet(b^2).
\end{align*}
In addition, it is straightforward to see that 
\begin{align*}
    &\int \left[e_1^T \tau_n^2  \b{D}_{0,h,a_0}^{-1} \left\{ (u^2, u^3)^T - (1,u) (1,u)^T \b{D}_{0,h,a_0}^{-1} P_0(\tilde{w}_{h,a_0,1} K_{h,a_0})   \right\} K(u) \right]^2 f_0(a_0 + uh) \, du  = \boundeddet(1).
    % &\qquad =\int \left[e_1^T \tau_n^2  \b{D}_{0,h,a_0}^{-1} \left\{ (u^2, u^3)^T - (1,u) (1,u)^T \{f_0(a_0)^{-1}\b{S}_1^{-1}+\boundeddet(h)\}\{f_0(a_0)  (c_2,0)^T + \boundeddet(h)\}   \right\} K(u) \right]^2 f_0(a_0 + uh) \, du\\
    % &\qquad \le \int \left[e_1^T \tau_n^2  (f_0(a_0)^{-1}\b{S}_2^{-1} + \boundeddet(h) )\left\{ (u^2, u^3)^T - (c_2,0)^T  + \boundeddet(h)   \right\} K(u) \right]^2 f_0(a_0 + uh) \, du \\
    % &\qquad \le 2\tau_n^4f_0(a_0)^{-1} \int u^4 K^2(u)f_0(a_0 + uh) \, du+ 2\tau_n^4 f_0(a_0)^{-1}\int c_2^2 K^2(u) f_0(a_0 + uh) \, du + \boundeddet(h^2)\\
    % &\qquad= 2\tau_n^4 \int u^4 K^2(u) \, du+ 2\tau_n^4 \int c_2^2 K^2(u) \, du + \boundeddet(h)
\end{align*}
 We now have $P_0(\Gamma_{0,h,b,a_0} \theta_0-\gamma_{0,h,b,a_0})^2 = \boundeddet(h)$.

Finally, we claim that 
\[\d{P}_n\left\{ \int\Gamma_0 \left( \mu_\infty - \int \mu_\infty \, dQ_0 \right) \, dF_0 \right\} = \bounded\left(n^{-1/2}\right),\]
which implies in particular that it is $\fasterthan\left( \{nh\}^{-1/2}\right)$. Since the function inside the empirical mean has mean zero, it is sufficient to show that 
\[ \int \left[\int\Gamma_0 \left( \mu_\infty - \int \mu_\infty \, dQ_0 \right) \, dF_0 \right]^2 \, dQ_0 = \boundeddet(1).\]
Since $\mu_\infty$ is uniformly bounded by~\ref{cond:doubly_robust}(a), we have
\begin{align*}
     \int \left[\int\Gamma_0 \left( \mu_\infty - \int \mu_\infty \, dQ_0 \right) \, dF_0 \right]^2 \, dQ_0 &\leq \int \left[\int\Gamma_0 \mu_\infty  \, dF_0 \right]^2 \, dQ_0\\
     &\leq \int \left[\int \left| e_1^T\b{D}_{0,h,a_0}^{-1} w_{h,a_0,1}(a) K_{h,a_0}(a) \mu_\infty(a, w)\right| \, dF_0(a)  \right. \\
    &\qquad \left. + \int  \left| e_3^T c_{0,h,a_0,2} \tau_n^2 \b{D}_{0,b,a_0,2}^{-1} w_{b,a_0,2}(a) K_{b,a_0}(a)\mu_\infty(a, w) \right| \, dF_0(a) \right]^2 \, dQ_0(w)\\
    &\lesssim \left[\int \left| e_1^T\b{D}_{0,h,a_0,1}^{-1}(1,u)^T K(u) f_0(a_0 + uh)\right|  \, du \right. \\
    &\qquad \left. + \int \left|e_3^T c_{0,h,a_0,2} \tau_n^2 \b{D}_{0,b,a_0,2}^{-1} (1,u,u^2)^T K(u) f_0(a_0 + uh) \right| \, du \right]^2.
\end{align*}
This latter expression is $\boundeddet(1)$ by Lemma~\ref{lm:D0_altform},~\ref{cond:bounded_K},~\ref{cond:bandwidth}, and~\ref{cond:cont_density}.

For the final statement of the result, using the same decomposition of $\phi^{*}_{\infty, h, b, a_0}$ as above and the fact that $E_0[ \xi_\infty \mid A = a] = \theta_0(a)$ for $a \in B_{\delta_1}(a_0)$, we can write
\begin{align*}
   \sigma_{\infty,h,b}^2(a_0) &= hP_0{\phi^{*2}_{\infty, h, b, a_0}}\\
   &= h P_0\left[ \Gamma_{0, h,b,a_0}^2 \left( \xi_\infty - \theta_0\right)^2\right] + h P_0 \left( \Gamma_{0, h,b,a_0} \theta_0 -\gamma_{0,h,b,a_0}\right)^2 + h P_0 \left[\int \Gamma_{0, h,b,a_0} (\mu_\infty -  \smallint \mu_\infty \, dQ_0) \, dF_0 \right]^2 \\
   &\qquad + 2h P_0 \left[  \Gamma_{0, h,b,a_0}\left( \xi_\infty  - \theta_0\right) \int \Gamma_{0, h,b,a_0} (\mu_\infty -  \smallint \mu_\infty \, dQ_0) \, dF_0 \right] \\
   &\qquad + 2h P_0 \left[ \left( \Gamma_{0, h,b,a_0} \theta_0 -\gamma_{0,h,b,a_0}\right)\int \Gamma_{0, h,b,a_0} (\mu_\infty -  \smallint \mu_\infty \, dQ_0) \, dF_0 \right].
\end{align*}
We have showed that $h P_0\left[ \Gamma_{0, h,b,a_0}^2 \left( \xi_\infty - \theta_0\right)^2\right]$ converges to $V_{K, \tau}\sigma_0^2(a_0) / f_0(a_0)$ as $h \longrightarrow 0$. We have also showed that $h P_0 \left( \Gamma_{0, h,b,a_0} \theta_0 -\gamma_{0,h,b,a_0}\right)^2 = \boundeddet(h^2)$ and $h P_0 \left[\int \Gamma_{0, h,b,a_0} (\mu_\infty -  \smallint \mu_\infty \, dQ_0) \, dF_0 \right]^2 = \boundeddet(h)$, so the statement follows.

% \begin{align*}
%   h \sigma_{\infty,h,b}^2(a_0) &= h P_0\left[ \Gamma_{0, h,b,a_0}^2 \left( \xi_\infty - \theta_0\right)^2\right] + h P_0 \left( \Gamma_{0, h,b,a_0} \theta_0 -\gamma_{0,h,b,a_0}\right)^2 + h P_0 \left[\int \Gamma_{0, h,b,a_0} (\mu_\infty -  \smallint \mu_\infty \, dQ_0) \, dF_0 \right]^2 \\
%   &\qquad + 2h P_0 \left[  \Gamma_{0, h,b,a_0}\left( \xi_\infty  - \theta_0\right) \int \Gamma_{0, h,b,a_0} (\mu_\infty -  \smallint \mu_\infty \, dQ_0) \, dF_0 \right] \\
%   &\qquad + 2h P_0 \left[ \left( \Gamma_{0, h,b,a_0} \theta_0 -\gamma_{0,h,b,a_0}\right)\int \Gamma_{0, h,b,a_0} (\mu_\infty -  \smallint \mu_\infty \, dQ_0) \, dF_0 \right].
% \end{align*}
% We have showed that $h P_0\left[ \Gamma_{0, h,b,a_0}^2 \left( \xi_\infty - \theta_0\right)^2\right]$ converges to $V_{K, \tau}\sigma_0^2(a_0) / f_0(a_0)$ as $h \to 0$. We have also showed that $h P_0 \left( \Gamma_{0, h,b,a_0} \theta_0 -\gamma_{0,h,b,a_0}\right)^2 = \fasterthandet(h^2)$ and $h P_0 \left[\int \Gamma_{0, h,b,a_0} (\mu_\infty -  \smallint \mu_\infty \, dQ_0) \, dF_0 \right]^2 = \boundeddet(h)$, so the statement follows.
\end{proof}

\begin{lemma}\label{lm:joint_convergence}
If the assumptions of Lemma~\ref{lm:lindeberg_feller_CLT} hold for each $a_0$ in $\{a_1, \dotsc, a_m\}$, then 
\[(nh)^{1/2}\d{P}_n\begin{pmatrix}\phi_{\infty, h, b,a_1}^* \\ \vdots \\ \phi_{\infty, h, b,a_m}^*  \end{pmatrix}\]
converges in distribution to a mean-zero multivariate normal distribution with zero off-diagonal elements and diagonal elements $V_{ K, \tau}f_0(a_1)^{-1} \sigma_0^2(a_1), \dotsc, V_{ K, \tau}f_0(a_m)^{-1} \sigma_0^2(a_m)$. Furthermore, $hP_0(\phi^*_{\infty, h, b, a_j}\phi^*_{\infty, h, b, a_k})$ converges to zero for $j \neq k$.
\end{lemma}
\begin{proof}[\bfseries{Proof of Lemma~\ref{lm:joint_convergence}}]
We focus on the case of $m = 2$ for simplicity; the proof for $m > 2$ is entirely analogous. We will show that for any $t_1, t_2 \in \d{R}$,
\[t_1(nh)^{1/2}\d{P}_n \phi^*_{\infty, h, b, a_1} + t_2(nh)^{1/2}\d{P}_n \phi^*_{\infty, h, b, a_2} \indist t_1^2 V_{ K, \tau}f_0(a_1)^{-1} \sigma_0^2(a_1) Z_1+ t_2^2 V_{ K, \tau}f_0(a_2)^{-1} \sigma_0^2(a_2)Z_2\]
where $Z_1$ and $Z_2$ are independent standard normal random variables. By the Cramer-Wold device, this yields the result. First, by the derivations in the proof of Lemma~\ref{lm:lindeberg_feller_CLT}, 
\begin{align*}
    t_1\d{P}_n \phi^*_{\infty, h, b, a_1} &= \d{P}_n \left\{t_1 \Gamma_{0,h,b,a_1} (\xi_\infty - \theta_0)\right\} + \fasterthan(\{nh\}^{-1/2}) \qquad\text{and}\\
    t_2\d{P}_n \phi^*_{\infty, h, b, a_2} &= \d{P}_n \left\{t_2 \Gamma_{0,h,b,a_2} (\xi_\infty - \theta_0)\right\} + \fasterthan(\{nh\}^{-1/2}).
\end{align*}
Hence, 
\[t_1(nh)^{1/2}\d{P}_n \phi^*_{\infty, h, b, a_1} + t_2(nh)^{1/2}\d{P}_n \phi^*_{\infty, h, b, a_2} = (nh)^{1/2}\d{P}_n \left\{(t_1 \Gamma_{0,h,b,a_1} +t_2 \Gamma_{0,h,b,a_2})(\xi_\infty - \theta_0)\right\} + \fasterthan(\{nh\}^{-1/2}).\]
As in the proof of Lemma~\ref{lm:lindeberg_feller_CLT}, we demonstrate the convergence of the expression on the right using the Lyapunov CLT for triangular arrays. We define $X_{n,1}, X_{n,2}, \dotsc, X_{n,n}$ as 
\[ X_{n,i} := h^{1/2}\left\{t_1\Gamma_{0,h,b,a_1}(A_i)+t_2\Gamma_{0,h,b,a_2}(A_i)\right\}\left\{\xi_\infty(Y_i, A_i, W_i) - \theta_0(A_i) \right\}. \]
We can then write $(nh)^{1/2}\d{P}_n \left\{(t_1 \Gamma_{0,h,b,a_1} +t_2 \Gamma_{0,h,b,a_2})(\xi_\infty - \theta_0)\right\}$ as $n^{-1/2}\sum_{i=1}^n \tilde{X}_{n,i}$. Hence, the claim follows by showing that $n^{-1/2}\sum_{i=1}^n X_{n,i}$ converges in distribution to the claimed limit, which we do using the Lyapunov CLT.

The main condition of the Lyapunov CLT that differs from Lemma~\ref{lm:lindeberg_feller_CLT} is convergence of the variance. By linearity of expectation and since $P_0\phi^*_{\infty, h, b, a_1}=P_0\phi^*_{\infty, h, b, a_2}=0$, we have $E_0[X_{n,i}]=0$ for all $i$. We also have 
\begin{align*}
    \n{Var}\left(X_{n,i}\right) &= h E_0\left[\left\{t_1\Gamma_{0,h,b,a_1}(A)+t_2\Gamma_{0,h,b,a_2}(A)\right\}^2\left\{\xi_\infty(Y, A, W) - \theta_0(A)\right\}^2\right] \\
    &=ht_1^2 E_0\left[\Gamma_{0,h,b,a_1}(A)^2\left\{\xi_\infty(Y, A, W) - \theta_0(A)\right\}^2\right]\\
    &\qquad + ht_2^2 E_0\left[\Gamma_{0,h,b,a_2}(A)^2\left\{\xi_\infty(Y, A, W) - \theta_0(A)\right\}^2\right]\\
    &\qquad +  2ht_1t_2 E_0\left[\Gamma_{0,h,b,a_1}(A)\Gamma_{0,h,b,a_2}(A)\left\{\xi_\infty(Y, A, W) - \theta_0(A)\right\}^2\right].
\end{align*}
The first two terms converge to $t_1^2 V_{ K, \tau}f_0(a_1)^{-1} \sigma_0^2(a_1)$ and $t_2^2 V_{ K, \tau}f_0(a_2)^{-1} \sigma_0^2(a_2)$ by Lemma~\ref{lm:lindeberg_feller_CLT}. We can write out the third term as 
\begin{align*}
    &E_0\left[\Gamma_{0,h,b,a_1}(A)\Gamma_{0,h,b,a_2}(A)\left\{\xi_\infty(Y, A, W) - \theta_0(A)\right\}^2\right] \\
    &\qquad =\int \left\{ e_1^T \b{D}_{0, h, a_1,1}^{-1} w_{h,a_1,1}(a) K_{h,a_1}(a)e_1^T \b{D}_{0, h, a_2,1}^{-1} w_{h,a_2,1}(a) K_{h,a_2}(a) \right.\\
    &\qquad\qquad\left.- e_3^T c_{0,h,a_0,2}\tau_n^2 \b{D}_{0, b,a_1,3}^{-1} w_{b,a_1,3}(a) K_{b,a_1}(a)e_1^T \b{D}_{0, h, a_2,1}^{-1} w_{h,a_2,1}(a) K_{h,a_2}(a) \right.\\
    &\qquad\qquad\left.-e_1^T \b{D}_{0, h, a_1,1}^{-1} w_{h,a_1,1}(a) K_{h,a_1}(a) e_3^T c_{0,h,a_0,2}\tau_n^2 \b{D}_{0, b,a_2,2}^{-1} w_{b,a_2,2}(a) K_{b,a_2}(a)\right.\\
    &\qquad\qquad\left.+e_3^T c_{0,h,a_0,2}\tau_n^2 \b{D}_{0, b,a_1,2}^{-1} w_{b,a_1,2}(a) K_{b,a_1}(a)e_3^T c_{0,h,a_0,2}\tau_n^2 \b{D}_{0, b,a_2,2}^{-1} w_{b,a_2,2}(a) K_{b,a_2}(a)\right\}\sigma_0^2(a) f_0(a) \, da.
\end{align*}
The four summands in the above display involve the products $K_{h,a_1}K_{h,a_2}$, $K_{h,a_1}K_{b,a_2}$, $K_{b,a_1}K_{h,a_2}$, and $K_{b,a_1}K_{b,a_2}$. Since the support of $K$ is contained in $[-1, 1]$, each of these products is zero for all $h,b$ small enough. Specifically, if $\max\{h, b\} < |a_1 - a_2|/2$, then $\{a: K_{h,a_1}(a) > 0\} \cap \{a: K_{h,a_2}(a) > 0\} = \{a: |a - a_1| \leq h\} \cap \{a: |a - a_2| \leq h\} = \emptyset$. Therefore, the variance converges to $t_1^2 V_{ K, \tau}f_0(a_1)^{-1} \sigma_0^2(a_1) + t_2^2 V_{ K, \tau}f_0(a_2)^{-1} \sigma_0^2(a_2)$. Each term is bounded away from zero by Lemma~\ref{lm:lindeberg_feller_CLT}. The remainder of the conditions of the Lyapunov CLT can be checked using the same derivations as in Lemma~\ref{lm:lindeberg_feller_CLT} and using the triangle inequality.
\end{proof}

%% file: supp/R1.tex
\clearpage

\section{Analysis of remainder term $R_{n,h,b,a_0,1}$}

\begin{lemma}\label{lm:R1}
If~\ref{cond:bounded_K} and~\ref{cond:cont_density} hold, then $\theta_{0,h,b}(a_0) - \theta_0(a_0) = \fasterthandet(h^2)$.  If~\ref{cond:bounded_K} and~\ref{cond:holder_smooth_theta} hold, then for some $\delta_4 > 0$, $sup_{a_0\in\s{A}_0}\left| \theta_{0,h,b}(a_0) - \theta_0(a_0)\right| = \boundeddet(h^{2+\delta_4})$.
\end{lemma}
\begin{proof}[\bfseries{Proof of Lemma~\ref{lm:R1}}]
By \ref{cond:cont_density}, the second derivative of $a \mapsto \theta_0(a)$ is  continuous in a neighborhood of $a=a_0$, so by the mean value form of Taylor's theorem, for each $a$ in a neighborhood of $a_0$, there exists $a_*$ (depending on $a$) between $a_0$ and $a$ such that 
\[\theta_0(a)-\theta_0(a_0) =\theta_0'(a_0)(a-a_0) + \tfrac{1}{2}\theta_0''(a_*)(a-a_0)^2.\]
We then have
\begin{align*}
    \theta_{0,h,b}(a_0) &= \int \Gamma_{0,h,b,a_0}(a) \theta_0(a) f_0(a) \, da \\
    &= \int \Gamma_{0,h,b,a_0}(a) \left[ \theta_0(a_0) + \theta_0'(a_0)(a-a_0) + \tfrac{1}{2}\theta_0''(a_0)(a-a_0)^2 \right] f_0(a) \, da \\
    &\qquad + \tfrac{1}{2} \int  \Gamma_{0,h,b,a_0}(a) \left[ \theta_0''(a_*)- \theta_0''(a_0)\right] (a-a_0)^2 f_0(a) \, da.
\end{align*}
We show that the first term equals $\theta_0(a_0)$. We have
\begin{align*}
   &\int \Gamma_{0,h,b,a_0}(a) \left[ \theta_0(a_0) + \theta_0'(a_0)(a-a_0) + \tfrac{1}{2}\theta_0''(a_0)(a-a_0)^2 \right] f_0(a) \, da \\
   &\qquad = e_1^T \b{D}_{0,h,a_0,1}^{-1} \int  w_{h,a_0, 1}(a) K_{h,a_0}(a) \left[ \theta_0(a_0) + \theta_0'(a_0)(a-a_0) + \tfrac{1}{2}\theta_0''(a_0)(a-a_0)^2 \right] f_0(a) \, da \\
   &\qquad\qquad - c_{0,h,a_0,2} (h/b)^2 e_3^T \b{D}_{0,b,a_0,2}^{-1} \int  w_{b,a_0, 2}(a) K_{b,a_0}(a) \left[ \theta_0(a_0) + \theta_0'(a_0)(a-a_0) + \tfrac{1}{2}\theta_0''(a_0)(a-a_0)^2 \right] f_0(a) \, da.
\end{align*}
By the definitions of $w_{h,a_0,1}$, $\b{D}_{0,h,a_0,1}$, and $c_{0,h,a_0,2}$, we then have
\begin{align*}
     &e_1^T \b{D}_{0,h,a_0,1}^{-1} \int  w_{h,a_0, 1}(a) K_{h,a_0}(a) \left[ \theta_0(a_0) + \theta_0'(a_0)(a-a_0) + \tfrac{1}{2}\theta_0''(a_0)(a-a_0)^2 \right] f_0(a) \, da \\
     &\qquad = e_1^T \b{D}_{0,h,a_0,1}^{-1} \int  w_{h,a_0, 1}(a) K_{h,a_0}(a) w_{h,a_0,1}(a)^T f_0(a) \,da \, \left[ \theta_0(a_0) , h \theta_0'(a_0) \right]\\
     &\qquad \qquad + \tfrac{1}{2}h^2\theta_0''(a_0) \int  w_{h,a_0, 1}(a) K_{h,a_0}(a) \left[ (a-a_0)/h \right]^2 f_0(a) \, da \\
     &\qquad = e_1^T \b{D}_{0,h,a_0,1}^{-1} \b{D}_{0,h,a_0,1} \, \left[ \theta_0(a_0) ,\, h \theta_0'(a_0) \right] + \tfrac{1}{2}h^2\theta_0''(a_0) c_{0,h,a_0,2} \\
     &\qquad = \theta_0(a_0)  + \tfrac{1}{2}h^2\theta_0''(a_0) c_{0,h,a_0,2}.
\end{align*}
Similarly, 
\begin{align*}
    &e_3^T \b{D}_{0,b,a_0,2}^{-1} \int  w_{b,a_0, 2}(a) K_{b,a_0}(a) \left[ \theta_0(a_0) + \theta_0'(a_0)(a-a_0) + \tfrac{1}{2}\theta_0''(a_0)(a-a_0)^2 \right] f_0(a) \, da \\
    &\qquad= e_3^T \b{D}_{0,b,a_0,2}^{-1} \int  w_{b,a_0, 2}(a) K_{b,a_0}(a) w_{b,a_0,2}(a)^T f_0(a) \, da \left[ \theta_0(a_0),\, \theta_0'(a_0),\, \tfrac{1}{2}\theta_0''(a_0) \right] \\
    &\qquad = e_3^T \b{D}_{0,b,a_0,2}^{-1} \b{D}_{0,b,a_0,2} \left[ \theta_0(a_0),\, b\theta_0'(a_0),\, \tfrac{1}{2}b^2\theta_0''(a_0) \right] \\
    &\qquad= \tfrac{1}{2}b^2\theta_0''(a_0).
\end{align*}
Thus, 
\begin{align*}
   &\int \Gamma_{0,h,b,a_0}(a) \left[ \theta_0(a_0) + \theta_0'(a_0)(a-a_0) + \tfrac{1}{2}\theta_0''(a_0)(a-a_0)^2 \right] f_0(a) \, da \\
   &\qquad =  \theta_0(a_0)  + \tfrac{1}{2}h^2\theta_0''(a_0) c_{0,h,a_0,2} - c_{0,h,a_0,2} (h/b)^2 \tfrac{1}{2}b^2\theta_0''(a_0) \\
   &\qquad = \theta_0(a_0).
\end{align*}
Hence,
\[ \theta_{0,h,b}(a_0) - \theta_0(a_0) = \tfrac{1}{2} h^2\int  \Gamma_{0,h,b,a_0}(a) \left[ \theta_0''(a_*)- \theta_0''(a_0)\right] [(a-a_0) / h]^2 f_0(a) \, da.\]
This is $\fasterthandet(h^2)$ by continuity of $\theta_0''$ at $a_0$.

% By adding and subtracting terms and noting that $e_3^T \b{S}_{3}^{-1} (c_2, 0, c_4)^T = 1$, we then have
% \begin{align*}
%     \theta_{0,b}''(a_0) &= e_3^T \b{D}_{0,b,a_0,2}^{-1}  \int (u^2, u^3, u^4)^T K(u) \theta''_0(a_0 + u_* b) f_0(a_0 + ub) \, du\\
%     &= \theta_0''(a_0) + \theta_0''(a_0) e_3^T \left[  \b{D}_{0,b,a_0,2}^{-1} - \b{S}_{3}^{-1} f_0(a_0)^{-1}  \right] (c_2, 0, c_4, 0)^T f_0(a_0)\\
%     &\qquad +  e_3^T \b{D}_{0,b,a_0,2}^{-1} \int (u^2, u^3, u^4)^T K(u)\left[ \theta''_0(a_0 + u_* b) - \theta_0''(a_0)\right] f_0(a_0) \, du \\
%     &\qquad + e_3^T \b{D}_{0,b,a_0,2}^{-1}  \int (u^2, u^3, u^4)^T K(u) \theta''_0(a_0 + u_* b) \left[f_0(a_0 + ub) - f_0(a_0)\right] \, du.
% \end{align*}
% As above, Lemma~\ref{lm:D0_altform}, continuity of $\theta_0''$, and continuity of $f_0$ then imply that $\theta''_{0,b}(a_0) = \theta_0''(a_0) + \fasterthandet(1)$ as $b \longrightarrow 0$. 

For the uniform statement, using the above derivations, since $K$ is uniformly bounded, we have
\begin{align*}
&\sup_{a_0 \in \s{A}_0} \left|\theta_{0,h,b}(a_0) - \theta_0(a_0)  \right| \\
&\qquad\lesssim h^2 \left[ \sup_{a_0 \in \s{A}_0} \left\| \b{D}_{0,h,a_0,1}^{-1}\right\|_{\infty} + c_{0,h,a_0,2}  \sup_{a_0 \in \s{A}_0}\left\| \b{D}_{0,b,a_0,2}^{-1}\right\|_{\infty} \right]\sup_{a_0 \in \s{A}_0} \sup_{|a_1 - a_0| \leq h}\left| \theta_0''(a_1) - \theta_0''(a_0)\right| \sup_{a_0 \in \s{A}_0} |f_0(a_0) |.
\end{align*}
By Lemma~\ref{lm:D0_altform} and~\ref{cond:holder_smooth_theta}, this is $\fasterthandet(h^{2 + \delta_4})$ for $\delta_4$ the H\"{o}lder exponent of $\theta_0''$.

% \begin{align*}
%     \theta_{0,h,b}(a_0) - \theta_0(a_0)   &= \tfrac{1}{2} h^2  e_1^T \left[ \b{D}_{0,h,a_0,1}^{-1}  - \b{S}_2^{-1} f_0(a_0)^{-1} \right]  (c_2,0)^T f_0(a_0) \theta_0''(a_0)   \\
%     &\qquad + \tfrac{1}{2} h^2 e_1^T \b{D}_{0,h,a_0,1}^{-1}\int (u^2,u^3)^T K(u) \left[ \theta_0''(a_0 + u_* h) - \theta_0''(a_0)\right] f_0(a_0) \, du \\
%     &\qquad + \tfrac{1}{2} h^2 e_1^T \b{D}_{0,h,a_0,1}^{-1}\int (u^2,u^3)^T K(u) \theta_0''(a_0 + u_* h) \left[ f_0(a_0 + uh) - f_0(a_0)\right] \, du \\
%     &\qquad - \theta_0''(a_0) e_3^T \left[  \b{D}_{0,b,a_0,3}^{-1} - \b{S}_{4}^{-1} f_0(a_0)^{-1}  \right] (c_2, 0, c_4, 0)^T f_0(s)\\
%     &\qquad -  e_3^T \b{D}_{0,b,a_0,3}^{-1} \int (u^2, u^3, u^4, u^5)^T K(u)\left[ \theta''_0(a_0 + u_* b) - \theta_0''(a_0)\right] f_0(a_0) \, du \\
%     &\qquad - e_3^T \b{D}_{0,b,a_0,3}^{-1}  \int (u^2, u^3, u^4, u^5)^T K(u) \theta''_0(a_0 + u_* b) \left[f_0(a_0 + ub) - f_0(a_0)\right] \, du.
% \end{align*}
% By Lemma~\ref{lm:D0_altform}, the first and fourth lines on the right hand side are both $\boundeddet(h^3)$ uniformly over $a_0 \in \s{A}_0$. By the smoothness of $\theta_0''$ imposed by~\ref{cond:holder_smooth_theta}, the second and fifth lines are both $\boundeddet(h^{2 + \delta_4})$ uniformly over $a_0$, where $\delta_4$ is the H\"{o}lder exponent of $\theta_0''$. By the smoothness of $f_0$ imposed by~\ref{cond:holder_smooth_theta} the third and sixth lines are $\boundeddet(h^3)$. The result follows.
\end{proof}

%% file: supp/R2R3.tex
\clearpage

\section{Analysis of remainder terms $R_{n,h,b,a_0,2}$ and $R_{n,h,b,a_0,3}$}

Before presenting results regarding the remainder terms $R_{n,h,b,a_0,2}$ and $R_{n,h,b,a_0,3}$, we state a Lemma regarding the rate of convergence of $\b{D}_{n, h, a_0,1} -  \b{D}_{0, h, a_0,1}$.

\begin{lemma}\label{lm:Dmatrix} Denote by $\|\cdot\|_\infty$ the element-wise maximum norm. If~\ref{cond:bounded_K} and~\ref{cond:cont_density} hold and $nh \longrightarrow \infty$, then $\|\b{D}_{n, h, a_0,1} -  \b{D}_{0, h, a_0,1}\|_\infty$ and $\|\b{D}^{-1}_{n,h, a_0,1} -  \b{D}^{-1}_{0, h, a_0,1}\|_\infty$ are $\bounded(\{nh\}^{-1/2})$, $\|\b{D}_{n, b, a_0,2} - \b{D}_{0, b, a_0,2}\|_\infty$   and $\|\b{D}^{-1}_{n,b, a_0,2} -  \b{D}^{-1}_{0, b, a_0,2}\|_\infty$ are $\bounded(\{nb\}^{-1/2})$, and $c_{n,h,a_0,2} - c_{0,h,a_0,2} = \bounded(\{nh\}^{-1/2})$.
\end{lemma}
\begin{proof}[\bfseries{Proof of Lemma~\ref{lm:Dmatrix}}]
We first note that for all $(j,k)$, $E_0\left(\b{D}_{n, h, a_0,1}[j,k]\right) = \b{D}_{0, h, a_0,1}[j,k]$. 
%We first note that $\b{D}_{n, h, a_0,1}[i,j] \geq 0$ by \ref{cond:bounded_K} and thus the variance of $\b{D}_{n, h, a_0,1}[i,j]$ can be bounded by $E_0\left(\b{D}_{n, h, a_0,1}[i,j]^2\right)$ \tw{This is always true, not just for non-negative RVs}. 
Using the change of variables $u = (a - a_0) /h$, we have
\begin{align*}
    \n{Var}\left(\b{D}_{n, h, a_0,1}[j,k]\right) &= \n{Var}\left\{n^{-1}\sum_{i=1}^n \left(\frac{A_i - a_0}{h}\right)^{j+k-2}h^{-1}K\left(\frac{A_i - a_0}{h}\right) \right\}\\
    &=(nh)^{-2}\sum_{i=1}^n \n{Var}\left\{\left(\frac{A_i - a_0}{h}\right)^{j+k-2}K\left(\frac{A_i - a_0}{h}\right)\right\}\\
    & \leq (nh)^{-2}\sum_{i=1}^nE_0\left[\left\{\left(\frac{A_i - a_0}{h}\right)^{j+k-2}K\left(\frac{A_i - a_0}{h}\right)\right\}^2\right]\\
    &= n^{-1}h^{-2}\int\left\{ \left(\frac{a - a_0}{h}\right)^{j+k-2}K\left(\frac{a - a_0}{h}\right) \right\}^2 \, dF_0(a)\\
    &= (nh)^{-1}\int u^{2(j+k-2)}K^2\left(u\right)  f_0(a_0 + uh)\, du.
\end{align*}
% \begin{align*}
%     E_0\left(\b{D}_{n, h, a_0,1}[j,k]^2\right) &= E_0\left[\left\{n^{-1}\sum_{i=1}^n \left(\frac{A_i - a_0}{h}\right)^{j+k-2}h^{-1}K\left(\frac{A_i - a_0}{h}\right) \right\}^2\right]\\
%     &= n^{-1}\int\left\{ \left(\frac{a - a_0}{h}\right)^{j+k-2}h^{-1}K\left(\frac{a - a_0}{h}\right) \right\}^2 \, dF_0(a)\\
%     &= n^{-1}\int u^{2(j+k-2)}h^{-1}K^2\left(u\right)  f_0(a_0 + uh)\, du.
% \end{align*}
Since $f_0(a_0+uh) = f_0(a_0) + \boundeddet(uh)$ by~\ref{cond:cont_density} and $K$ is uniformly bounded with compact support by~\ref{cond:bounded_K}, we have
\begin{align*}
    (nh)^{-1}\int u^{2(j+k-2)}K^2\left(u\right)  f_0(a_0 + uh)\, du &= (nh)^{-1} f_0(a_0) \left\{\int u^{2(j+k-2)}K^2\left(u\right)\, du + \boundeddet(h)\right\} \\
    &= (nh)^{-1} f_0(a_0)c_{2(i+j-2), 2} + \fasterthandet(\{nh\}^{-1})
\end{align*}
for $1 \leq j,k,\leq 2$. By Chebyshev's inequality, we then have $\b{D}_{n, h, a_0,1}[j,k] - \b{D}_{0, h, a_0,1}[j,k] =\bounded(\{nh\}^{-1/2})$. Since this holds for all $1 \leq j,k,\leq 2$, we conclude that $\|\b{D}_{n, h, a_0,1} - \b{D}_{0, h, a_0,1}\|_\infty =\bounded(\{nh\}^{-1/2})$. The result for $\|\b{D}_{n, b, a_0,2} - \b{D}_{0, b, a_0,2}\|_\infty =\bounded(\{nb\}^{-1/2})$ can be shown analogously. 

To show that $\|\b{D}_{n, h, a_0,1}^{-1} - \b{D}_{0, h, a_0,1}^{-1}\|_\infty = \bounded(\{nh\}^{-1})$, we first write $\b{D}_{n, h, a_0,1}^{-1}-\b{D}_{0, h, a_0,1}^{-1} = \b{D}^{-1}_{n, h, a_0,1}(\b{D}_{0, h, a_0,1}-\b{D}_{n, h, a_0,1})\b{D}^{-1}_{0, h, a_0,1}$. Denoting by $\|\cdot\|_{1}$ the $L_1$ matrix operator norm, it follows from the fact that the $L_1$ norm is sub-multiplicative,
\begin{align*}
\|\b{D}_{n, h, a_0,1}^{-1}-\b{D}_{0, h, a_0,1}^{-1}\|_\infty &= \| \b{D}^{-1}_{n, h, a_0,1}(\b{D}_{0, h, a_0,1}-\b{D}_{n, h, a_0,1})\b{D}^{-1}_{0, h, a_0,1}\|_\infty \\
&\leq C \| \b{D}^{-1}_{n, h, a_0,1}(\b{D}_{0, h, a_0,1}-\b{D}_{n, h, a_0,1})\b{D}^{-1}_{0, h, a_0,1}\|_1 \\
&\leq C \|\b{D}^{-1}_{n, h, a_0,1} \|_1 \|\b{D}_{0, h, a_0,1}-\b{D}_{n, h, a_0,1} \|_1 \|\b{D}^{-1}_{0, h, a_0,1}\|_1.
\end{align*}
We then note that $f_0(a_0) > 0$ implies that $\|\b{D}^{-1}_{0, h, a_0,1}\|_1 = \boundeddet(1)$ by Lemma~\ref{lm:D0_altform}. Since $nh \longrightarrow \infty$, Lemma~\ref{lm:D0_altform} also implies that $\b{D}_{n, h, a_0,1}^{-1} \inprob f_0(a_0)^{-1} \b{S}_2^{-1}$, so by the continuous mapping theorem, $\|\b{D}^{-1}_{n, h, a_0,1}\|_1 \inprob  \|f_0(a_0)^{-1} \b{S}_2^{-1}\|_1$ and hence $\|\b{D}^{-1}_{n, h, a_0,1}\|_1 = \bounded(1)$.  We conclude that $\|\b{D}^{-1}_{n, h, a_0,1} - \b{D}^{-1}_{0, h, a_0,1}\|_\infty =\bounded(\{nh\}^{-1/2})$. The derivation above holds for finite dimensional matrices, and thus a similar argument yields that $\|\b{D}^{-1}_{n, b, a_0,2} - \b{D}^{-1}_{0, b, a_0,2}\|_\infty =\bounded(\{nb\}^{-1/2})$.
% \kt{I found this in page 103 of \cite{fan1996}; $(A + hB)^{-1} = A^{-1} - hA^{-1}BA^{-1} + \boundeddet(h^2)$}
% \kt{The similar one I could find was Sherman–Morrison formula; $(A+uv^T)^{-1} = A^{-1} + \frac{A^{-1}uv^TA^{-1}}{1 + v^TA^{-1}u}$ but assume rank one perturbation}\kt{Since we only 
% show $\bounded$ element-wise, I am not sure we can use something like $||\cdot||_F$}
% \begin{align*}
%     \b{D}_{n, h, a_0,1}^{-1} = \left[\b{D}_{0, h, a_0,1} + \bounded(\{nh\}^{-1/2})\right]^{-1} = \b{D}_{0, h, a_0,1}^{-1} + \bounded(\{nh\}^{-1/2})
% \end{align*}

% Finally, by the continuous mapping theorem and assuming $f_0(a_0) > 0$ \tw{how does this follow?}, we conclude $\b{D}_{n,h, a_0}^{-1} - \b{D}_{0,h, a_0}^{-1} = \bounded(\{nh\}^{-1/2})$ as claimed. The proof for $\b{D}_{0,b, a_0}$ is analogous thus omitted.

Finally, we have
\begin{align*}
    c_{n,h,a_0,2} - c_{0,h,a_0,2} &= e_1^T \left(\b{D}_{n, h, a_0,1}^{-1} - \b{D}_{0, h, a_0,1}^{-1}\right) \d{P}_n \left( \tilde{w}_{h,a_0,1} K_{h,a_0} \right) + e_1^T \b{D}_{0, h, a_0,1}^{-1} (\d{P}_n - P_0) \left( \tilde{w}_{h,a_0,1} K_{h,a_0} \right).
\end{align*}
Both terms are $\bounded(\{nh\}^{-1/2})$ by the calculations above.
\end{proof}

We establish rates of convergence of $R_{n,h,b,a_0,2}$ and $R_{n,h,b,a_0,3}$ using empirical process theory. Before presenting supporting lemmas, we review key definitions and notation used in empirical process theory. For a class of functions $\s{H}$ defined on a domain $\s{X}$, a probability measure $Q$ on $\s{X}$, and any $\varepsilon > 0$, the $\varepsilon$-covering number $N(\varepsilon,\s{H}, L_2(Q))$ of $\s{H}$ relative to the $L_2(Q)$ metric is the smallest number of $L_2(Q)$-balls of radius less than or equal to $\varepsilon$ needed to cover $\s{H}$. A function $H$ on $\s{X}$ is called an \emph{envelope function} for $\s{H}$ if $\sup_{\eta \in \s{H}}|\eta(x)| \leq H(x)$ for all $x \in \s{X}$. %We then define the \emph{uniform entropy integral} of $\s{H}$ relative to $H$ as:
%\begin{align}\label{cond:uniform_entropy_condition}
   % J(\delta, \s{H}) := \sup_Q \int^\delta_0  \left\{1 + \log N(\varepsilon \|H \|_{Q,2}, \s{H}, L_2(Q)) \right\}^{1/2}\,  d\varepsilon.
%\end{align}

%Suppose an arbitrary index set $\s{T}$ is equipped with a semi-metric $\rho$. A sequence of stochastic processes $\{V_n : \s{T} \to \d{R}:n=1,2, \dots\}$ indexed by $\s{T}$  is called \emph{asymptotically uniformly $\rho$-equicontinuous} \citep{pollard1984, vandervaart1996} if for every $\varepsilon_1, \varepsilon_2 >0$, there exists $\delta> 0$ such that 
%\begin{align*}
    %\limsup_{n\to\infty} P \left\{ \sup_{\rho(s,t)<\delta} |V_n(s) - V_n(t)| > \varepsilon_1 \right\} < \varepsilon_2.
%\end{align*}
%If $t_n$ is a sequence of possibly random elements in $\s{T}$ such that $\rho(t_n, t) \inprob 0$ for some $t \in \s{T}$, then asymptotic $\rho$-equicontinuity of $V_n$ implies that $V_n(t_n) - V_n(t) \inprob 0$.

We first consider the sequence of stochastic processes $\{V_n(\lambda) : \lambda \in \Lambda\}$, for $V_n(\lambda) := \d{G}_n \eta_{h,a_0, j,\lambda}$, where
\[\eta_{h,a_0, j,\lambda}(y,a,w) := h^{1/2} \left(\frac{a-a_0}{h}\right)^{j-1}K_{h,a_0}(a) \lambda(y, a, w).\]
We note that $\eta$ depends on $n$ through $h$. We also make use of the following semi-metric on $\Lambda$:
\[ \rho_{a_0,\varepsilon}(\lambda_1, \lambda_2) := \sup_{|a - a_0| < \varepsilon} \left(E_0 \left[ \left\{ \lambda_1(Y,A,W) - \lambda_2(Y,A,W) \right\}^2 \mid A = a \right] \right)^{1/2}.\] 
We then have the following lemma.
\begin{lemma}\label{lemma:equicontinuity} Suppose~\ref{cond:bounded_K},~\ref{cond:bandwidth}, and~\ref{cond:cont_density} hold, and that $\Lambda$ is a class of functions with envelope $L$ satisfying:
\begin{enumerate}
    \item $\sup_{|a - a_0| < \delta_1} E_0[ L^{2+\delta_2} \mid A = a] < \infty$ for some $\delta_1, \delta_2 > 0$; and
    %\item $\sup_{|a - a_0| < \delta} E_0[ L^2 I(L > M) \mid A = a] = \fasterthandet(1)$ as $M \longrightarrow \infty$; and
    \item $\int_0^1 \left \{ \sup_Q \log N\left(\varepsilon \| L\|_{Q,2}, \Lambda, L_2(Q) \right) \right\}^{1/2} \, d\varepsilon < \infty$.
\end{enumerate} 
Then for each $j \in \{1,2,3,4\}$, $\sup_{\lambda \in \Lambda} |\d{G}_n \eta_{h,a_0, j,\lambda}| = \bounded(1)$, and for any possibly random sequences $\lambda_{n1}, \lambda_{n2} \in \Lambda$ such that $\rho_{a_0, \varepsilon}(\lambda_{n1}, \lambda_{n2})=\fasterthan(1)$ for some $\varepsilon > 0$,  $\d{G}_n \{\eta_{h,a_0, j, \lambda_{n1}}-\eta_{h,a_0, j, \lambda_{n2}}\}= \fasterthan(1)$, so $\d{G}_n\left\{ h^{1/2} \Gamma_{0, h, b, a_0}(\lambda_{1n} - \lambda_{2n}) \right\}= \fasterthan(1)$ as well.
\end{lemma}
\begin{proof}[\bfseries{Proof of Lemma~\ref{lemma:equicontinuity}}]
We first establish properties of the class $\s{H}_{h,a_0,j} := \{ \eta_{h,a_0, j,\lambda} : \lambda \in \Lambda\}$. Since $L$ is an envelope for $\Lambda$,
\[ H_{h,a_0,j}(y,a,w) := h^{1/2}  \left(\frac{a-a_0}{h}\right)^{j-1}K_{h,a_0}(a) L(y, a, w)\]
is an envelope for $\s{H}_{h,a_0,j}$. By the tower property and a change of variables, we then have for all $h < \delta_1$
\begin{align*}
    P_0 H_{h,a_0,j}^2 &= \int   h  \left(\frac{a-a_0}{h}\right)^{2(j-1)}K_{h,a_0}(a)^2 E_0 \left[ L^2 \mid A = a \right] \, f_0(a) \, da \\
    &= \int u^{2(j-1)}K(u)^2 E_0 \left[ L^2 \mid A = a_0 + uh \right] \, f_0(a_0 + uh) \, du \\
    &\leq \sup_{|a -a_0| <\delta_1} E_0 \left[ L^2 \mid A = a \right] \int u^{2(j-1)}K(u)^2 \, f_0(a_0 + uh) \, du. 
\end{align*}
The last expression above is $\boundeddet(1)$ by assumption,~\ref{cond:bounded_K}, and~\ref{cond:cont_density}.

Next, we study the uniform entropy of the class $\s{H}_{h,a_0,j}$. We clearly have $\s{H}_{h,a_0,j} = \s{G}_{h,a_0,j} \Lambda$, where $\s{G}_{h,a_0,j}$ is the class consisting of the single function $a \mapsto h^{1/2} [(a - a_0) /h]^j K_{h,a_0}(a)$. We then have by, e.g.\ Lemma~5.1 of~\cite{vanderVaartvanderLaan2006}, that 
\[ \sup_Q  N\left(\varepsilon \| H_{h,a_0,j}\|_{Q,2}, \s{H}_{h,a_0,j}, L_2(Q) \right) \leq \sup_Q N\left(\varepsilon \| L\|_{Q,2}, \Lambda, L_2(Q) \right)\]
for any $\varepsilon > 0$.  We can now establish the first claim of Lemma~\ref{lemma:equicontinuity}. By Theorem~2.14.1 of \cite{vandervaart1996}, we have
\begin{align*}
    E_0 \left[ \sup_{\lambda \in \Lambda} \left| \d{G}_n \eta_{h,a_0,h, \lambda} \right| \right] &= E_0 \left[ \sup_{\eta \in \s{H}_{h,a_0,j}} \left| \d{G}_n \eta \right| \right] \\
    &\lesssim \left\{ P_0 H_{h, a_0,j}^2 \right\}^{1/2} \sup_Q \int_0^1 \left\{1 +  \log N\left(\varepsilon \| H_{h,a_0,j}\|_{Q,2}, \s{H}_{h,a_0,j}, L_2(Q) \right)  \right\}^{1/2} \, d\varepsilon \\
    &\leq \boundeddet(1)  \int_0^1 \left\{ \sup_Q \log N\left(\varepsilon \| L\|_{Q,2}, \Lambda, L_2(Q) \right) \right\}^{1/2} \, d\varepsilon.
\end{align*}
The integral is finite by assumption, which establishes the claim.

We will establish the second statement by showing that the sequence of processes is \emph{asymptotically uniformly $\rho_{a_0,\varepsilon}$-equicontinuous} \citep{pollard1984, vandervaart1996}.  To do so, we will use Theorem 2.11.1 (see also Theorem 2.11.22) of~\cite{vandervaart1996}, which  implies that $\{ V_n: n = 1,2,\dots\}$ is asymptotically uniformly $\rho_{a_0, \varepsilon}$-equicontinuous if the following conditions hold:
\begin{itemize}
\item[(a)] $P_0 H_{h,a_0,j}^2 = \boundeddet(1)$; 
\item[(b)] $P_0\left[ H_{h,a_0,j}^2 I\{H_{h,a_0,j}>\varepsilon n^{1/2}\} \right] = \fasterthandet(1)$ for every $\varepsilon > 0$;
\item[(c)] $\sup_{\rho_{a_0, \varepsilon} (\lambda_1,\lambda_2)<\delta_n}P_0\{\eta_{h,a_0,j,\lambda_1}-\eta_{h,a_0,j,\lambda_2}\}^2 = \fasterthandet(1)$ for every $\delta_n\longrightarrow 0$; and
\item[(d)] for every $\delta_n=\fasterthandet(1)$, $\sup_Q \int_{0}^{\delta_n}\left\{\log N(\varepsilon \|H_{h,a_0,j}\|_{Q,2}, \s{H}_{h,a_0,j}, L_2(Q)) \right\}^{1/2}\,d\varepsilon = \fasterthandet(1)$.
\end{itemize}
We prove that $\{\d{G}_n \eta_{h,a_0, j, \lambda} : \lambda \in \Lambda\}$ is $\rho_{a_0, \varepsilon}$-equicontinuous by establishing these conditions, and the result follows.

%% (b)
We showed condition (a) above. For (b), we have again by the tower property and a change of variables
\begin{align*}
    P_0 \left\{ H_{h,a_0,j}^2 I\left(H_{h,a_0,j} > \varepsilon n^{1/2}  \right) \right\} &= \int   h  \left(\frac{a-a_0}{h}\right)^{2(j-1)}K_{h,a_0}(a)^2 E_0 \left[ L^2 I\left(H_{h,a_0,j} > \varepsilon n^{1/2}  \right) \mid A = a \right] \, f_0(a) \, da \\
    &=  \int  u^{2(j-1)}K(u)^2 E_0 \left[ L^2 I\left(H_{h,a_0,j} > \varepsilon n^{1/2}  \right) \mid A = a_0 + uh \right] \, f_0(a_0 + uh) \, du.
\end{align*}
Now we note that since the kernel $K$ is uniformly bounded by, say $\bar{K} < \infty$,  and has support contained in $[-1,1]$,
\[ I\left(H_{h,a_0,j} > \varepsilon n^{1/2}  \right) = I\left( h^{1/2} [(A - a_0) / h]^{j-1} K_{h,a_0}(A) L > \varepsilon n^{1/2}  \right) \leq  I\left( \bar{K}L > \varepsilon (nh)^{1/2}  \right) .\]
Hence, for all $h < \delta_1$,
\begin{align*}
    P_0 \left\{ H_{h,a_0,j}^2 I\left(H_{h,a_0,j} > \varepsilon n^{1/2}  \right) \right\} &\leq \bar{K}^2  \int_{-1}^1   E_0 \left[ L^2 I\left( L > \varepsilon (nh)^{1/2} / \bar{K}  \right) \mid A = a_0 + uh \right] \, f_0(a_0 + uh) \, du \\
    &\lesssim \sup_{|a - a_0| < \delta_1}E_0 \left[ L^2 I\left( L > \varepsilon (nh)^{1/2} / \bar{K}  \right) \mid A = a  \right].
\end{align*}
Then by H\"{o}lder's inequality,
\begin{align*}
    E_0 \left[ L^2 I\left( L > \varepsilon (nh)^{1/2} / \bar{K}  \right) \mid A = a  \right] 
    &\leq\left\{E_0\left[L^{2+\delta_2} \mid A = a\right]\right\}^{2 / (2 + \delta_2)} \left\{ P_0\left( L > \varepsilon (nh)^{1/2} / \bar{K} \mid A = a \right)\right\}^{\delta_2 / (2+\delta_2)},
\end{align*}
and by Markov's inequality,
\begin{align*}
    P_0\left( L > \varepsilon (nh)^{1/2} / \bar{K} \mid A = a \right) &\leq  E_0 \left[ L^{2+\delta_2} \mid A = a\right]  / \left[ \varepsilon (nh)^{1/2} / \bar{K} \right]^{2+\delta_2}.
\end{align*}
Thus, 
\begin{align*}
    \sup_{|a - a_0| < \delta_1}E_0 \left[ L^2 I\left( L > \varepsilon (nh)^{1/2} / \bar{K}  \right) \mid A = a  \right] &\leq \left[ \varepsilon (nh)^{1/2} / \bar{K} \right]^{-\delta_2} \sup_{|a - a_0| < \delta_1}E_0\left[L^{2+\delta_2} \mid A = a\right].
\end{align*}
Since $nh \longrightarrow \infty$, this latter expression tends to zero by assumption for every $\varepsilon > 0$, which verifies (b).

%% (c) 
For (c), by a similar calculation we have for all $h$ small enough that
\begin{align*}
    P_0 \left\{ \eta_{h,a_0, j,\lambda_1} - \eta_{h,a_0, j,\lambda_2}\right\}^2 &= \int h \left(\frac{a-a_0}{h}\right)^{2(j-1)}K_{h,a_0}(a)^2 E_0\left[ \left\{\lambda_1 - \lambda_2\right\}^2 \mid A = a\right] f_0(a) \, da \\
    &= \int  u^{2(j-1)}K(u)^2 E_0\left[ \left\{\lambda_1 - \lambda_2\right\}^2 \mid A = a_0 + uh\right] f_0(a_0 + uh) \, du \\
    &\lesssim C \sup_{|a - a_0| < \varepsilon} E_0\left[ \left\{\lambda_1 - \lambda_2\right\}^2 \mid A = a\right] \\
    &= \rho_{a_0, \varepsilon}(\lambda_1, \lambda_2)^2.
\end{align*}
Hence, 
\[ \sup_{\rho(\lambda_1,\lambda_2)<\delta_n}P_0\{\eta_{h,a_0, j,\lambda_1} - \eta_{h,a_0, j,\lambda_2}\}^2 \lesssim  \delta_n^2, \]
so that (c) is satisfied.

%% (d)
For (d), using the entropy bound established above, we have
\[ \sup_Q \int_0^\delta \left\{ \log  N\left(\varepsilon \| H_{h,a_0,j}\|_{Q,2}, \s{H}_{h,a_0,j}, L_2(Q) \right) \right\}^{1/2} \, d\varepsilon \leq \int_0^\delta \left \{ \sup_Q \log N\left(\varepsilon \| L\|_{Q,2}, \Lambda, L_2(Q) \right) \right\}^{1/2} \, d\varepsilon.\]
Since the integral on the right is convergent by assumption, the expression converges to zero as $\delta \to 0$, which verifies (d). 

For the final statement, by Lemma~\ref{lm:D0_altform}, we can write
\begin{align*}
    h^{1/2}\Gamma_{0,h, b, a_0}\lambda  &= \sum_{j=1}^2 C_{h,a_0,j} \eta_{h,a_0,j,\lambda} + \sum_{j=1}^3 C_{h,b,a_0,j}' \eta_{b,a_0,j,\lambda}
\end{align*}
for some constants $C_{h,a_0,j} = \boundeddet(1)$ and $C_{h,b,a_0,j}' = \boundeddet(1)$. By the result above, we then have 
\begin{align*}
    \d{G}_n\left\{ h^{1/2} \Gamma_{0, h, b, a_0}(\lambda_{1n} - \lambda_{2n}) \right\}  &= \sum_{j=1}^2 C_{h,a_0,j} \d{G}_n\left\{ \eta_{h,a_0,j,\lambda_{n1}} - \eta_{h,a_0,j,\lambda_{n2}} \right\} + \sum_{j=1}^3 C_{h,b,a_0,j}' \d{G}_n\left\{ \eta_{b,a_0,j,\lambda_{n1}} - \eta_{b,a_0,j,\lambda_{n2}} \right\} \\
    &= \sum_{j=1}^2 \boundeddet(1) \fasterthan(1)+\sum_{j=1}^3 \boundeddet(1) \fasterthan(1) = \fasterthan(1).
\end{align*}
\end{proof}

Next, we consider the sequence of stochastic processes $\{\bar{V}_n(\lambda) : \lambda \in \Lambda\}$, for $\bar{V}_n(\lambda) := \d{G}_n \bar\eta_{h,a_0, j,\lambda}$, where
\[\bar\eta_{h,a_0, j,\lambda}(y,a,w) := \int h^{1/2}\left(\frac{a-a_0}{h}\right)^{j-1}K_{h,a_0}(a) \lambda(y, a, w) \, dF_0(a).\]
%We then state the following lemma concerning $\bar{V}_n$, and its application to the process $\{\d{G}_n \left( \int h^{1/2}\Gamma_{0,h,b,a_0} \lambda \, dF_0 \right) : \lambda \in \Lambda\}$. 

%Next, we consider the sequence the following sequence of stochastic processes $\left\{ W_{n, h,a_0,j}(\omega) : \omega \in \Omega\right\}$ indexed by a set of functions $\Omega :=\{\omega : \s{A}\times\s{W} \mapsto \d{R}\}$ uniformly bounded by some constant $C < \infty$ such that $|\omega| \leq C$ for all $\omega \in \Omega$:
%\[W_{n, h,a_0, j}(\omega) = \int \omega(a, w) \left(\frac{a-a_0}{h}\right)^{j-1}K_{h,a_0}(a)\, dF_0(a)\]

\begin{lemma}\label{lemma:integral_process} Suppose~\ref{cond:bounded_K},~\ref{cond:bandwidth}, and~\ref{cond:cont_density} hold, and that $\Lambda$ is a class of functions uniformly bounded by $L < \infty$  and satisfying $ \sup_Q \log N\left(\varepsilon, \Lambda, L_2(Q) \right) \leq  C\varepsilon^{-V}$ for some $C < \infty$ and $V \in (0,2)$. Then $\sup_{\lambda \in \Lambda} \left| \d{G}_n \bar\eta_{h,a_0,j,\lambda} \right| = \fasterthan(h^{(2-s)/(2s)})$ for every $s \in (V \vee 1,2)$ and each $j \in\{ 1,2,\ldots\}$. Hence, $\sup_{\lambda \in \Lambda} \left| \d{G}_n \left( \int h^{1/2} \Gamma_{0,h,b,a_0} \lambda \, dF_0\right) \right| = \fasterthan(h^{(2-s)/(2s)})$ for every $s \in (V,2)$ as well.
%for any possibly random sequences $\lambda_{1n}$ and $\lambda_{2n} \in \Lambda$ such that $\rho_{a_0, \varepsilon}(\lambda_{1n}, \lambda_{2n})=\fasterthan(1)$, $\d{G}_n \left\{\bar\eta_{h,a_0, j, \lambda_{1n}}-\bar\eta_{h,a_0, j, \lambda_{2n}}\right\}= \fasterthan(1)$ for each $j \in\{ 1,2,3,4\}$.
\end{lemma}
\begin{proof}[\bfseries{Proof of Lemma~\ref{lemma:integral_process}}]
We let $\bar{\s{H}}_{h,a_0,j} := \{\bar\eta_{h,a_0, j,\lambda} : \lambda \in \Lambda\}$.  We use Theorem~2.14.1 of \cite{vandervaart1996}. %We show the sequence $\{\d{G}_n \bar\eta_{h,a_0, j,\lambda} : \lambda \in \Lambda\}$ is $\rho_{a_0, \varepsilon}$-asymptotically equicontinuous using Theorem 2.11.1 of~\cite{vandervaart1996} by verifying conditions (a)--(d) presented in the proof of Lemma \ref{lm:stochastic_equicontinuity1}. 
We first bound the uniform entropy of $\bar{\s{H}}_{h,a_0,j}$. Lemma~5.1 of \cite{vanderVaartvanderLaan2006} implies that the class of functions $\s{F}_{h,a_0,j} := \left\{(y,a,w) \mapsto h^{1/2}\left(\frac{a-a_0}{h}\right)^{j-1}K_{h,a_0}(a) \lambda(y,a,w) : \lambda \in \Lambda\right\}$ satisfies 
\[ \sup_{Q} N\left( \varepsilon \|F_{h,a_0,j} \|_{Q,2}, \s{F}_{h,a_0,j}, L_2(Q)\right) \leq \sup_Q N\left( \varepsilon L, \Lambda, L_2(Q)\right)\]
relative to the envelope function $F_{h,a_0,j} : a \mapsto L h^{1/2} |(a -a_0)/h|^jK_{h,a_0}(a)$. We note that for $s \geq 1$,
\[\left[ \int F_{h,a_0,j}(a)^s \, dF_0(a) \right]^{1/s} = L h^{1/2+(1-s)/s}\left[ \int |u|^{s(j-1)} K(u)^s f_0(a_0 + uh) \,du \right]^{1/s}  \lesssim h^{(2-s)/(2s)} \]
by~\ref{cond:bounded_K} and~\ref{cond:cont_density}. Hence, $\bar{H}_{h,s} := C h^{(2-s)/(2s)}$ is an envelope for $\bar{\s{H}}_{h,a_0,j}$ for some $C <\infty$. By Lemma~5.2 of \cite{vanderVaartvanderLaan2006} (with $r = t = 2$), the above bound for the uniform entropy of $\s{F}_{h,a_0,j}$, and the assumed bound for the uniform entropy of $\Lambda$, we then have
\begin{align*}
     \sup_{Q} \log N\left( \varepsilon \|\bar{H}_{h,s} \|_{Q,2}, \s{H}_{h,a_0,j}, L_2(Q)\right) &\leq \sup_{Q} \log N\left( \varepsilon^{2/s} \|F_{h,a_0,j} \|_{Q,2} / 2^{r/s}, \s{F}_{h,a_0,j}, L_2(Q)\right) \\
     &\leq \sup_Q \log N\left(\varepsilon^{2/s} L/ 2^{r/s}, \Lambda, L_2(Q)\right) \\
     &\lesssim \varepsilon^{-2V/s}.
\end{align*}
Theorem~2.14.1 of \cite{vandervaart1996} then gives
\begin{align*}
    E_0 \left[ \sup_{\bar\eta \in \bar{\s{H}}_{h,a_0,j}} \left| \d{G}_n \bar\eta \right| \right] &\lesssim \| \bar{H}_{h,s} \|_{P,2} \int_0^1 \left\{ 1 +  \sup_{Q} \log N\left( \varepsilon \|\bar{H}_{h,s} \|_{Q,2}, \s{H}_{h,a_0,j}, L_2(Q)\right) \right\}^{1/2} \, d\varepsilon   \\
    &\lesssim  h^{(2-s)/(2s)}\int_0^1 \left\{ 1 +  \varepsilon^{-2V/s} \right\}^{1/2} \, d\varepsilon.
\end{align*}
The integral is finite so long as $V/s < 1$. Thus, 
\[ E_0 \left[ \sup_{\bar\eta \in \bar{\s{H}}_{h,a_0,j}} \left| \d{G}_n \bar\eta \right| \right] = \boundeddet\left(h^{(2-s)/(2s)}\right) \]
for all $s > V$. Furthermore, for any $s > V$, there exists $s' \in (V, s)$, so that 
\[ E_0 \left[ \sup_{\bar\eta \in \bar{\s{H}}_{h,a_0,j}} \left| \d{G}_n \bar\eta \right| \right] = \boundeddet\left(h^{(2-s')/(2s')}\right) = \fasterthandet\left(h^{(2-s)/(2s)}\right). \]
The final statement regarding $\{\d{G}_n \left( \int h^{1/2} \Gamma_{0,h,b,a_0} \lambda \, dF_0\right) :\lambda \in \Lambda\}$ follows by decomposing $\Gamma_{0,h,b,a_0}$ and using Lemma~\ref{lm:Dmatrix}, as in the proof of Lemma~\ref{lemma:equicontinuity}.
\end{proof}

We now use these results to establish rates of convergence of the remainder terms $R_{n,h,b,a_0,2}$ and $R_{n,h,b,a_0,3}$.

\begin{corollary}\label{cor:R2R3}
If \ref{cond:bounded_K}--\ref{cond:cont_density} hold, then $R_{n,h,b,a_0,2} = \fasterthan\left(\{nh\}^{-1/2}\right)$ and $R_{n,h,b,a_0,3} = \bounded\left(\{nh\}^{-1}\right)$.
\end{corollary}
\begin{proof}[\bfseries{Proof of Corollary~\ref{cor:R2R3}}]
We can write
\begin{align}
    (nh)^{1/2}R_{n, h,b,a_0, 2} &= \d{G}_n\left\{ h^{1/2} \Gamma_{0,h,b,a_0} \left( \psi_n - \psi_\infty \right) \right\} + \d{G}_n\left\{ h^{1/2} \Gamma_{0,h,b,a_0} \left( \int \mu_n \, dQ_0 - \int \mu_\infty \, dQ_0 \right) \right\}\nonumber \\
    &\qquad + \d{G}_n \left\{ \int h^{1/2} \Gamma_{0,h,b,a_0} \mu_n \, dF_0 \right\} - \d{G}_n \left\{ \int h^{1/2} \Gamma_{0,h,b,a_0} \mu_\infty \, dF_0 \right\} \label{eq:R2_decomp}
\end{align}
and 
\begin{align}
    (nh)^{1/2}R_{n, h,b,a_0, 3} &= \d{G}_n\left\{ h^{1/2} \left(\Gamma_{n,h,b,a_0}- \Gamma_{0,h,b,a_0}\right) \psi_n \right\} +   \d{G}_n\left\{ h^{1/2} \left(\Gamma_{n,h,b,a_0}- \Gamma_{0,h,b,a_0}\right)  \int \mu_n \, dQ_0\right\} \nonumber \\
    &\qquad + \d{G}_n \left\{ \int h^{1/2} \left(\Gamma_{n,h,b,a_0}- \Gamma_{0,h,b,a_0}\right) \mu_n \, dF_0 \right\}  \label{eq:R3_decomp}
\end{align}
We address the first two terms in each expansion using Lemma~\ref{lemma:equicontinuity}. For the first term, we set $\Lambda := \{ o \mapsto [y - \mu(a,w)] / g(a,w): \mu \in \s{F}_\mu, g \in \s{F}_g\}$. By~\ref{cond:uniform_entropy_nuisances}, an envelope function for $\Lambda$ is $L(y,a,w) = C + C' |y|$ for some $C, C' < \infty$. Hence, by the triangle inequality,
\begin{align*}
    \sup_{|a - a_0| < \delta_1} E_0\left[ L^{2+\delta_2} \mid A = a\right] &= \sup_{|a - a_0| < \delta_1} E_0 \left[ \left\{ C + C'|Y|\right\}^{2+\delta_2} \mid A = a\right] \\
    &\leq \left[C + C'\sup_{|a - a_0| < \delta_1}  \left( E_0 \left\{ |Y|^{2+\delta_2} \mid A = a\right\}\right)^{1/(2+\delta_2)} \right]^{2+\delta_2},
\end{align*}
which is finite by~\ref{cond:cont_density}. Therefore, condition 1 of Lemma~\ref{lemma:equicontinuity} holds. For condition 2 of Lemma~\ref{lemma:equicontinuity}, we note that permanence properties of uniform entropy numbers and uniform entropy integrals (see, e.g., Theorem~2.10.20 of \citealp{vandervaart1996} and Lemma~5.1 of \citealp{vanderVaartvanderLaan2006}) in conjunction with~\ref{cond:uniform_entropy_nuisances} implies that $\Lambda$ possesses finite uniform entropy integral. Finally, we show that $\rho_{a_0,\delta_1}(\psi_n, \psi_\infty) = \fasterthan(1)$. By Jensen's inequality,
\[  E_0\left( Y^2 \mid A = a, W = w\right) \leq \left[ E_0\left( |Y|^{2+\delta_2} \mid A = a, W =w \right) \right]^{2 / (2 + \delta_2)}\]
which is uniformly bounded for $a \in B_{\delta_1}(a_0)$ and $P_0$-almost every $w$ by~\ref{cond:cont_density}. We also note that by~\ref{cond:uniform_entropy_nuisances} and~\ref{cond:doubly_robust}, for any functions $\gamma_1, \gamma_2$ of $(a,w)$, 
\begin{align*}
    \sup_{|a -a_0| < \delta_1} \left\{ E_0\left[ \left\{ \gamma_1(a,W) - \gamma_2(a,W)\right\}^2 \mid A = a\right] \right\}^{1/2} &= \sup_{|a -a_0| < \delta_1} \left\{ E_0\left[ g_0(a,W)\left\{ \gamma_1(a,W) - \gamma_2(a,W)\right\}^2 \right] \right\}^{1/2} \\
    &\lesssim \sup_{|a -a_0| < \delta_1} \left\{ E_0\left[ \left\{ \gamma_1(a,W) - \gamma_2(a,W)\right\}^2 \right] \right\}^{1/2} \\
    &= \sup_{|a -a_0| < \delta_1} \left\{ E_0\left[ I_{B_{\delta_1}(a_0) \times \s{W}}(a,W)\left\{ \gamma_1(a,W) - \gamma_2(a,W)\right\}^2 \right] \right\}^{1/2} \\
    &\leq \sum_{j=1}^3 \sup_{|a -a_0| < \delta_1} \left\{ E_0\left[ I_{\s{S}_j}(a,W)\left\{ \gamma_1(a,W) - \gamma_2(a,W)\right\}^2 \right] \right\}^{1/2} \\
    &= \sum_{j=1}^3 d\left(\gamma_1, \gamma_2; B_{\delta_1}(a_0), \s{S}_j\right).
\end{align*}
We then have  
\begin{align*}
    \rho_{a_0,\delta_1}(\psi_n, \psi_\infty) &= \sup_{|a - a_0| <\delta_1} \left\{ E_0 \left[ \left( \psi_n - \psi_0 \right)^2 \mid A =a \right] \right\}^{1/2} \\
    &= \sup_{|a - a_0| <\delta_1} \left\{ E_0 \left[ \left( \{ 1/(g_n g_\infty)\} \left\{Y - \mu_n \right\} \left\{ g_\infty - g_n\right\} - \{1/g_\infty\}\{\mu_n - \mu_\infty\} \right)^2 \mid A =a \right] \right\}^{1/2} \\ 
    &\lesssim \sup_{|a - a_0| <\delta_1} \left\{ E_0 \left[Y^2 \left( g_n - g_\infty\right)^2 \mid A =a \right] \right\}^{1/2} + \sup_{|a - a_0| <\delta_1} \left\{ E_0 \left[\mu_n^2 \left( g_n - g_\infty\right)^2 \mid A =a \right] \right\}^{1/2} \\
    &\qquad + \sup_{|a - a_0| <\delta_1} \left\{ E_0 \left[\left(\mu_n - \mu_\infty\right)^2 \mid A =a \right] \right\}^{1/2} \\
    &\lesssim \sup_{|a - a_0| <\delta_1} \left\{ E_0 \left[ E_0\left(Y^2 \mid A =a, W\right) \left( g_n - g_\infty\right)^2 \mid A =a \right] \right\}^{1/2} +  \sum_{j=1}^3 d\left(g_n, g_\infty; B_{\delta_1}(a_0), \s{S}_j\right) \\
    &\qquad + \sum_{j=1}^3 d\left(\mu_n, \mu_\infty; B_{\delta_1}(a_0), \s{S}_j\right)  \\
    &\lesssim \sum_{j=1}^3 d\left(g_n, g_\infty; B_{\delta_1}(a_0), \s{S}_j\right)  + \sum_{j=1}^3 d\left(\mu_n, \mu_\infty; B_{\delta_1}(a_0), \s{S}_j\right).
\end{align*}
Hence, by~\ref{cond:doubly_robust}, $\rho_{a_0,\delta_1}(\psi_n, \psi_\infty) = \fasterthan(1)$. Lemma~\ref{lemma:equicontinuity} thus implies that  $\d{G}_n\{ h^{1/2} \Gamma_{0,h,b,a_0} (\psi_n - \psi_\infty)\} = \fasterthan(1)$, which addresses the first term in equation~\eqref{eq:R2_decomp}.  For the first term in equation~\eqref{eq:R3_decomp}, we note that
\begin{align*}
    \d{G}_n\left\{ h^{1/2} \left(\Gamma_{n,h,b,a_0}- \Gamma_{0,h,b,a_0}\right) \psi_n \right\} &= e_1^T \left( \b{D}_{n,h,a_0,1}^{-1} - \b{D}_{0,h,a_0,1}^{-1}\right)  \d{G}_n\left\{ h^{1/2} w_{h,a_0,1} K_{h,a_0} \psi_n \right\} \\
    &\qquad -  c_{0,h,a_0,2} (h/b)^2 e_3^T\left( \b{D}_{n,h,a_0,2}^{-1} - \b{D}_{0,h,a_0,2}^{-1}\right)  \d{G}_n\left\{ h^{1/2} w_{b,a_0,2} K_{b,a_0} \psi_n \right\} \\
    &\qquad -  \left(c_{n,h,a_0,2} - c_{0,h,a_0,2}\right) (h/b)^2  e_3^T\b{D}_{n,h,a_0,2}^{-1}  \d{G}_n\left\{ h^{1/2} w_{b,a_0,2} K_{b,a_0} \psi_n \right\}.
\end{align*}
Since we have established that the conditions of Lemma~\ref{lemma:equicontinuity} hold for the class $\Lambda$ that $\psi_n$ falls in, Lemma~\ref{lemma:equicontinuity} implies that all elements of $\d{G}_n\left\{ h^{1/2} w_{h,a_0,1} K_{h,a_0} \psi_n \right\}$ and $\d{G}_n\left\{ h^{1/2} w_{b,a_0,2} K_{b,a_0} \psi_n \right\}$ are $\bounded(1)$. In conjunction with Lemma~\ref{lm:Dmatrix}, we then have
\[  \d{G}_n\left\{ h^{1/2} \left(\Gamma_{n,h,b,a_0}- \Gamma_{0,h,b,a_0}\right) \psi_n \right\} = \bounded\left( \{nh\}^{-1/2} \right).\]

%Now, by~\ref{cond:uniform_entropy_nuisances} and~\ref{cond:cont_density},
%\begin{align*}
%E_0\left\{Y^{2+\delta_2} \mid A = a \right\} &= E_0 \left\{ E_0\left[Y^{2+\delta_} \mid A = a,W \right]g_0(a, W)\right\} \leq C_3 E_0\left\{  E_0\left[Y^{2+\delta} \mid A = a,W \right] \right\}
%\end{align*}

For the second term in equation~\eqref{eq:R2_decomp}, we use Lemma~\ref{lemma:equicontinuity} with $\Lambda := \left\{ a \mapsto \int \mu(a, w) \, dQ_0(w) : \mu \in \s{F}_\mu \right\}$, which is a uniformly bounded class with envelope $L = C_1$ by~\ref{cond:uniform_entropy_nuisances}. Hence, condition (1)  of Lemma~\ref{lemma:equicontinuity} hold trivially for this class. For condition (2), \ref{cond:uniform_entropy_nuisances} and Lemma~5.2 of \cite{vanderVaartvanderLaan2006} (with $r = s = t = 2$) together imply that condition (2) of Lemma~\ref{lemma:equicontinuity} is satisfied. Therefore, if $\rho_{a_0, \delta_1}(\lambda_n, \lambda_\infty) = \fasterthan(1)$ for $\lambda_n : a \mapsto \int \mu_n(a,w) \, dQ_0(w)$ and $\lambda_\infty : a \mapsto \int \mu_\infty(a,w) \, dQ_0(w)$, Lemma~\ref{lemma:equicontinuity} implies that \[\d{G}_n\left\{ h^{1/2} \Gamma_{0,h,b,a_0} \left( \int \mu_n \, dQ_0 - \int \mu_\infty \, dQ_0 \right) \right\} = \fasterthan(1).\] 
Since $\lambda_n$ and $\lambda_\infty$ are functions just of $a$, we have
\begin{align*}
    \rho_{a_0, \delta_1}(\lambda_n, \lambda_\infty)^2 &= \sup_{|a -a_0|<\delta_1} \left\{ \int \left[ \mu_n(a, w) - \mu_\infty(a,w) \right] \, dQ_0(w) \right\}^2 \leq  \sup_{|a -a_0|<\delta_1} \int \left[ \mu_n(a, w) - \mu_\infty(a,w) \right]^2 \, dQ_0(w),
\end{align*}
which is $\fasterthan(1)$ by~\ref{cond:doubly_robust} as noted above. A similar argument yields that the second term of~\eqref{eq:R3_decomp} is $\bounded\left(\{nh\}^{-1/2}\right)$.

For the third and fourth terms in equation~\eqref{eq:R2_decomp}, we use Lemma~\ref{lemma:integral_process} with $\Lambda = \s{F}_\mu$. The conditions of the lemma are satisfied by~\ref{cond:uniform_entropy_nuisances}, so that 
\[ \sup_{\mu \in \s{F}_\mu} \left| \d{G}_n \left\{ \int h^{1/2} \Gamma_{0,h,b,a_0} \mu \, dF_0 \right\} \right| = \fasterthan\left(h^{(2-s) /(2s)} \right)\]
for every $s \in (V_\mu,2)$. Since $\mu_n, \mu_\infty \in \s{F}_\mu$ almost surely for all $n$ large enough, the result follows.

For the third term in equation~\eqref{eq:R3_decomp}, we note that 
\begin{align*}
    \d{G}_n\left\{ \int h^{1/2} \left(\Gamma_{n,h,b,a_0}- \Gamma_{0,h,b,a_0}\right) \mu_n  \, dF_0 \right\} &= e_1^T \left( \b{D}_{n,h,a_0,1}^{-1} - \b{D}_{0,h,a_0,1}^{-1}\right)  \d{G}_n\left\{ \int h^{1/2} w_{h,a_0,1} K_{h,a_0} \mu_n\, dF_0 \right\} \\
    &\quad -  c_{0,h,a_0,2} (h/b)^2e_3^T\left( \b{D}_{n,b,a_0,2}^{-1} - \b{D}_{0,b,a_0,2}^{-1}\right)  \d{G}_n\left\{ \int  h^{1/2} w_{h,a_0,2} K_{h,a_0} \mu_n \, dF_0 \right\} \\
    &\quad-  (c_{0,h,a_0,2}  - c_{0,h,a_0,2}) (h/b)^2e_3^T \b{D}_{n,b,a_0,2}^{-1}  \d{G}_n\left\{ \int  h^{1/2} w_{h,a_0,2} K_{h,a_0} \mu_n \, dF_0 \right\} \\
\end{align*}
Each of the terms in $\d{G}_n\left\{ \int h^{1/2} w_{h,a_0,1} K_{h,a_0} \mu_n\, dF_0 \right\}$ and $\d{G}_n\left\{ \int h^{1/2} w_{b,a_0,3} K_{b,a_0} \mu_n \, dF_0\right\}$ are $\fasterthan\left(h^{(2-s) /(2s)} \right)$ by Lemma~\ref{lemma:integral_process}, so they are also $\bounded(1)$. In conjunction with Lemma~\ref{lm:Dmatrix}, we then have 
\[\d{G}_n\left\{ \int h^{1/2} \left(\Gamma_{n,h,b,a_0}- \Gamma_{0,h,b,a_0}\right) \mu_n  \, dF_0 \right\} = \bounded\left( \{nh\}^{-1/2}\right),\]
as claimed.
\end{proof}

We now turn to uniform control of $R_{n,h,b,s,2}$ and $R_{n,h,b,,3}$ over $a_0 \in \s{A}_0$. Let $\s{F}$ be a class of measurable functions equipped with a measurable envelope $F$. We say that $\s{F}$ is \emph{VC type} with envelope $F$ if there exists a constant $V>0$ such that $\sup_Q N(\varepsilon\|F\|_{Q, 2}, \s{F}, L_2(Q)) \lesssim \varepsilon^{-V}$ when $\varepsilon < 1$ and the supremum is taken over all probability measures. 
\begin{lemma}\label{lm:VC_class_K}
If~\ref{cond:bounded_K} holds, then $\{a \mapsto (\tfrac{a-a_0}{h})^{j-1}K(\tfrac{a-a_0}{h}) : a_0\in\s{A}_0\}$ is VC type for each fixed $h > 0$ and positive integer $j$.
\end{lemma}
\begin{proof}[\bfseries{Proof of Lemma~\ref{lm:VC_class_K}}]
By \ref{cond:bounded_K}, $K$ is supported on $[-1,1]$ and is uniformly bounded. Hence, 
\[ \left(\frac{a-a_0}{h}\right)^{j-1}K\left(\frac{a-a_0}{h}\right) = I\left(\left|\frac{a-a_0}{h}\right|\leq 1\right)\left(\frac{a-a_0}{h}\right)^{j-1}K\left(\frac{a-a_0}{h}\right),\]
which implies that the class in question is contained in the product of the two classes of functions
\begin{align*}
    \left\{ K\left(\frac{a-a_0}{h}\right): a_0 \in \s{A}_0\right\} \text{ and }\left\{\left(\frac{a-a_0}{h}\right)^{j-1}I\left(\left|\frac{a-a_0}{h}\right|\leq 1\right)  : a_0 \in \s{A}_0\right\}.
\end{align*}
The first class is VC type by \ref{cond:bounded_K} and the results of \cite{gine2002}. For the second class, we first rewrite this class as 
\[\left\{ \phi \circ p_j\left(\frac{a-a_0}{h}\right) : a_0 \in \s{A}_0\right\},\]
where where $\phi := u \mapsto u I (|u| \leq 1)$ and $p_j := u \mapsto u^j$. Since $\phi$ is a bounded real function of bounded variation and $p_j$ is a polynomial, \cite{gine2002} implies that this class is also VC. Since the product of VC classes is also VC, (e.g., Lemma 5.1 of~\citealp{vanderVaartvanderLaan2006}), the result follows.
\end{proof}

\begin{corollary}\label{cor:phi_infty_VC}
For $a_0\in\s{A}_0$ and $h,b > 0$, consider the following function:
\[\eta_{h,b,a_0} : (y,a,w) \mapsto h^{1/2}\sigma^{-1}_{\infty,h,b}(a_0)\phi_{\infty, h, b,a_0}^*(y,a,w).\]
If~\ref{cond:bounded_K}  holds, then the class of functions $\s{H}_{h,b} := \{\eta_{h,b,a_0} : a_0\in\s{A}_0\}$ is VC type.
\end{corollary}
\begin{proof}[\bfseries{Proof of Corollary \ref{cor:phi_infty_VC}}]
For each $a_0 \in \s{A}_0$, $\eta_{h,b,a_0}$ is defined as
\begin{align*}
    h^{1/2}\sigma^{-1}_{\infty,h,b}(a_0)\left[\Gamma_{0,h,b,a_0}(a) \xi_\infty(y,a,w) - \gamma_{0,h,b,a_0}(a) + \int\Gamma_{0,h,b,a_0}(\bar{a}) \left\{ \mu_\infty(\bar{a},w) - \int \mu_\infty(\bar{a}, \bar{w}) \, dQ_0(\bar{w}) \right\} \, dF_0(\bar{a})\right].
\end{align*}
Any collection of constants $\{c_{a_0} : a_0 \in \s{A}_0\}$ is a VC class with VC index 1 because no set of size 2 can be shattered. This implies that $\{\sigma_{\infty,h,b}(a_0) : a_0 \in \s{A}_0\}$, $\{P_0\left( w_{h,a_0, 1} K_{h, a_0} \theta_0\right) : a_0\in\s{A}_0\}$, $\{\b{D}^{-1}_{0, h, a_0, 1}[i,j] : a_0\in\s{A}_0\}$ for each $(i,j)$, and $\{\iint\Gamma_{0,h,b,a_0}(\bar{a}) \mu_\infty(\bar{a}, \bar{w}) \, dQ_0(\bar{w}) \, dF_0(\bar{a}) : a_0\in\s{A}_0\}$ are all VC type.  We note that by permanence properties of uniform covering numbers, sums and products of VC classes of functions are also VC type. It thus remains to show that $\{\Gamma_{0,h,b,a_0}(a) \xi_\infty(y,a,w) : a_0\in\s{A}_0\}$, $\{\gamma_{0,h,b,a_0}(a) : a_0\in\s{A}_0\}$ and $\{\int\Gamma_{0,h,b,a_0}(\bar{a}) \mu_\infty(\bar{a},w) \, dF_0(\bar{a}) : a_0\in\s{A}_0\}$ are VC-type.

We recall that
\begin{align*}
    \Gamma_{0,h,b,a_0}(a)  = e_1^T \b{D}_{0,h,a_0,1}^{-1} w_{h,a_0,1}(a) K_{h,a_0}(a) - e_3^T c_{0,h,a_,2} (h/b)^2 \b{D}_{0,b,a_0,2}^{-1} w_{b,a_0,2}(a) K_{b,a_0}(a).
\end{align*} 
By Lemma~\ref{lm:VC_class_K}, each entry of $\{w_{h,a_0,1}(a) K_{h,a_0}(a) : a_0\in\s{A}_0\}$ and $\{w_{b,a_0,2}(a) K_{b,a_0}(a) : a_0\in\s{A}_0\}$ are VC type. Combined with the discussion of sets of constants above, this implies that $\{\Gamma_{0,h,b,a_0}(a) : a_0 \in \s{A}_0\}$ is VC type for each $h,b >0$.  By Lemma 2.6.18 of \cite{vandervaart1996}, element-wise products of a fixed function and a VC class of functions are also VC type. Hence, it follows that $\{\Gamma_{0,h,b,a_0}(a) \xi_\infty(y,a,w) : a_0\in\s{A}_0\}$ is also VC type. Similarly, $\{\Gamma_{0,h,b,a_0}(a) \mu_\infty(a,w) : a_0\in\s{A}_0\}$ is VC type. Lemma~5.2 of \cite{vanderVaartvanderLaan2006} (with $r=s=t=2$) then implies that $\{\int\Gamma_{0,h,b,s}(\bar{a}) \mu_\infty(\bar{a},w) \, dF_0(\bar{a}) : a_0\in\s{A}_0\}$ is also VC type.

Finally, we recall that
\begin{align*}
    \gamma_{0,h,b,a_0}(a) &= e_1^T\b{D}^{-1}_{0, h, a_0,1}w_{h,a_0, 1}(a) K_{h, a_0}(a)w^T_{h,s, 1}(a)\b{D}^{-1}_{0, h, a_0,1} P_0 \left( w_{h,a_0, 1} K_{h, a_0} \theta_0\right) \\
        &\qquad- c_2 (h/b)^2 e_3^T \b{D}_{0, b,a_0,2}^{-1} w_{b,a_0,2}(a) K_{b,a_0}(a)w^T_{b,a_0, 2}(a)\b{D}^{-1}_{0, b, a_0,2} P_0 \left( w_{b,a_0, 2} K_{b, a_0} \theta_0\right) \\
        &\qquad + (h/b)^2 e_1^T \b{D}^{-1}_{0, h, a_0,1} \left[ \tilde{w}_{h,a_0,1}(a)  - w_{h,a_0,1}(a) w_{h,a_0,1}(a)^T  \b{D}_{0, h, a_0,1} P_0 (\tilde{w}_{h,a_0,1} K_{h,a_0}) \right] K_{h,a_0}(a) \\
        &\qquad \qquad \times e_3^T \b{D}_{0, b,a_0,2}^{-1}P_0 \left( w_{b,a_0, 2} K_{b, a_0} \theta_0\right)
\end{align*} 
As in Lemma~\ref{lm:VC_class_K}, each entry of $\{w_{h,a_0, 1}(a) K_{h, a_0}(a)w^T_{h,a_0, 1}(a) : a_0\in\s{A}_0\}$, $\{w_{b,a_0,2}(a) K_{b,a_0}(a)w^T_{b,a_0, 2}(a) : a_0\in\s{A}_0\}$, and $\{\tilde{w}_{h,a_0, 1}(a) K_{h, a_0}(a)w^T_{h,a_0, 1}(a) : a_0\in\s{A}_0\}$ are VC type. Thus by the permanence property of VC-type function classes, we conclude that $\{\gamma_{0,h,b,a_0}(a) : a_0\in\s{A}_0\}$ is also VC type.
\end{proof}
\begin{lemma}\label{lm:Dmatrix_unif} 
If \ref{cond:bounded_K} and~\ref{cond:holder_smooth_theta} hold, then $\sup_{a_0\in \s{A}_0}\|\b{D}_{n, h, a_0,1} -  \b{D}_{0, h, a_0,1}\|_\infty$ and $\sup_{a_0\in \s{A}_0}\|\b{D}_{n, b, a_0,2} - \b{D}_{0, b, a_0,2}\|_\infty$ are $\bounded(\{nh/\log h^{-1}\}^{-1/2})$. If in addition $nh / \log h^{-1} \longrightarrow \infty$, then $\sup_{a_0\in \s{A}_0}\|\b{D}^{-1}_{n,h, a_0,1} -  \b{D}^{-1}_{0, h, a_0,1}\|_\infty$, $\sup_{a_0\in \s{A}_0}\|\b{D}^{-1}_{n,b, a_0,2} -  \b{D}^{-1}_{0, b, a_0,2}\|_\infty$, and $\sup_{a_0 \in \s{A}_0} | c_{n,h,a_0,2} - c_{0,h,a_0,2}|$  are $\bounded(\{nh/\log h^{-1}\}^{-1/2})$.
\end{lemma}
\begin{proof}[\bfseries{Proof of Lemma~\ref{lm:Dmatrix_unif}}]
For all $1 \leq j,k\leq 2$, we can write
\begin{align*}
    \sup_{a_0 \in \s{A}_0}\left|\b{D}_{n, h, a_0,1}[j,k] - \b{D}_{0, h, a_0,1}[j,k]\right| &= h^{-1}\sup_{a_0 \in \s{A}_0}\left|\int \left(\frac{a-a_0}{h}\right)^{j+k-2} K\left(\frac{a-a_0}{h}\right)\, \left\{dF_n(a) - \, dF_0(a)\right\}\right| \\
    &= n^{-1/2} h^{-1} \sup_{f \in \s{F}_{j+k-2,h}} \left| \d{G}_n f\right|
\end{align*}
for
\[ \s{F}_{j+k-2,h} := \left\{a \mapsto \left(\frac{a-a_0}{h}\right)^{j+k-2} K\left(\frac{a-a_0}{h}\right)\, : \, a_0 \in \s{A}_0 \right\}.\]
By \ref{cond:bounded_K} and Lemma~\ref{lm:VC_class_K}, $\s{F}_{j+k-2,h}$ is a uniformly bounded VC class. For all $a_0 \in \s{A}_0$, 
\begin{align*}
    \n{Var} \left\{\left(\frac{A-a_0}{h}\right)^{j+k-2} K\left(\frac{A-a_0}{h}\right)\right\} &\leq E_0 \left[\left\{\left(\frac{A-a_0}{h}\right)^{j+k-2} K\left(\frac{A-a_0}{h}\right)\right\}^2\right]\\
    &=\int \left\{\left(\frac{a-a_0}{h}\right)^{j+k-2} K\left(\frac{a-a_0}{h}\right)\right\}^2 f_0(a) \, da\\
    &= \int h u^{2(j+k-2)} K(u)^2 f_0(uh+a_0)\, du \\
    &\leq  h\|f_0\|_\infty c_{2(j+k-2)}.
\end{align*} 
Additionally, we have that $\sup_{a_0\in\s{A}_0}|\left(\frac{a-a_0}{h}\right)^{j+k-2} K\left(\frac{a-a_0}{h}\right)| \leq \|K\|_\infty$ by \ref{cond:bounded_K}. Hence, $\s{F}_{j+k-2, h}$ satisfies the conditions of Theorem~2.1 of \cite{gine2002} with $\sigma^2$ proportional to $h$, which implies that that 
\[\sup_{f \in \s{F}_{j+k-2,h}} \left| \d{G}_n f\right| = \bounded\left( \left\{ h \log h^{-1} \right\}^{1/2} \right).\]
Hence, 
\begin{align*}
    \sup_{a_0 \in \s{A}_0}\left|\b{D}_{n, h, a_0,1}[j,k] - \b{D}_{0, h, a_0,1}[j,k]\right| = \bounded\left(\left\{\frac{\log h^{-1}}{nh}\right\}^{1/2}\right)
\end{align*}
for each $1 \leq j,k\leq 2$. We conclude that $\sup_{a_0\in \s{A}_0}\|\b{D}_{n, h, a_0,1} -  \b{D}_{0, h, a_0,1}\|_\infty$ and $\sup_{a_0\in \s{A}_0}\|\b{D}_{n, h, a_0,2} -  \b{D}_{0, h, a_0,2}\|_\infty$ are both $\bounded(\{nh/\log h^{-1}\}^{-1/2})$.  The results for the inverse matrices follow along the lines of Lemma~\ref{lm:Dmatrix}, where we use the assumption that $nh / \log h^{-1} \longrightarrow \infty$ to conclude that $\sup_{a_0\in \s{A}_0}\|\b{D}_{n, h, a_0,1} -  \b{D}_{0, h, a_0,1}\|_\infty = \fasterthan(1)$.

Finally, we can write
\begin{align*}
    \sup_{a_0 \in \s{A}_0} \left| c_{n,h,a_0,2} - c_{0,h,a_0,2}\right| &\leq  \sup_{a_0 \in \s{A}_0} \left| e_1^T \left(\b{D}_{n, h, a_0,1}^{-1} - \b{D}_{0, h, a_0,1}^{-1}\right) \right|  \sup_{a_0 \in \s{A}_0} \left|\d{P}_n \left(\tilde{w}_{h,a_0,1} K_{h,a_0} \right)  \right| \\
    &\qquad + \sup_{a_0 \in \s{A}_0} \left| e_1^T  \b{D}_{0, h, a_0,1}^{-1}\right| \sup_{a_0 \in \s{A}_0} \left| \left(\d{P}_n - P_0\right) \left(\tilde{w}_{h,a_0,1} K_{h,a_0} \right) \right|.
\end{align*}
The first term is $\bounded(\{nh/\log h^{-1}\}^{-1/2})$ by the result for $\sup_{a_0\in \s{A}_0}\|\b{D}_{n, h, a_0,1}^{-1} -  \b{D}_{0, h, a_0,1}^{-1}\|_\infty$, and the second term is $\bounded(\{nh/\log h^{-1}\}^{-1/2})$ by a similar calculation to that above.
\end{proof}

As above, we define
\[\eta_{h,a_0, j,\lambda}(y,a,w) :=  h^{1/2}\left(\frac{a-a_0}{h}\right)^{j-1}K_{h,a_0}\left(a\right) \lambda(y, a, w),\]
and we let $\lambda_\infty$ be a fixed function in $\Lambda$. 
We also define the following semi-metric on $\Lambda$:
\[ \rho_{\s{A}_\delta}(\lambda_1, \lambda_2) := \sup_{a_0 \in \s{A}_\delta} \left(E_0 \left[ \left\{ \lambda_1(Y,A,W) - \lambda_2(Y,A,W) \right\}^2 \mid A = a_0 \right] \right)^{1/2},\] 
where $\s{A}_\delta := \{a : \exists a_0 \in \s{A}_0, |a-a_0| \leq\delta\}$. For each $r > 0$, we then define $\Lambda_r := \{\lambda \in \Lambda: \rho_{\s{A}_\delta}(\lambda, \lambda_\infty) \leq r\}$ and $\s{H}_{h, j, r} := \{\eta_{h,a_0, j,\lambda} - \eta_{h,a_0,j,\lambda_\infty} : a_0 \in \s{A}_0, \lambda \in \Lambda_r\}$. If $\Lambda$ is equipped with an envelope function $L$, then $H_{h} := 2h^{-1/2}\|K\|_\infty L$ is an envelope for $\s{H}_{h, j, r}$. The next Lemma controls the uniform entropy of $\s{H}_{h,j,r}$ in terms of that of $\Lambda$.
\begin{lemma} \label{lm:uniform_entropy_lambda_r} Suppose $\Lambda$ has envelope function $L$ with $P L^2 < \infty$ and such that $\sup_Q \log N(\varepsilon \|L\|_{Q,2}, \Lambda, L_2(Q)) \leq C \varepsilon^{-V}$ for some $C, V < \infty$. If  \ref{cond:bounded_K} also holds, then 
\[\sup_{Q} \log N(\varepsilon \|H_{h}\|_{Q, 2}, \s{H}_{h,j,r}, L_2(Q)) \leq C' \varepsilon^{-V}\]
for all $\varepsilon \leq 1$ and some $C' < \infty$.
\end{lemma}
\begin{proof}[\bfseries{Proof of Lemma~\ref{lm:uniform_entropy_lambda_r}}]
For each $h > 0$ and $j \in \{1,\dotsc, 4\}$, we define
\begin{align*}
    \s{K}_{h,j} := \left\{\left(\frac{a-a_0}{h}\right)^{j-1}K\left(\frac{a-a_0}{h}\right)  : a_0 \in \s{A}_0\right\},
\end{align*}
which is uniformly bounded by $\|K\|_\infty$. We have that $\s{H}_{h,j,r}$ is then contained in the product of $h^{-1/2}\s{K}_{h,j}$ and $\Lambda - \lambda_\infty$, which has envelope $2L$. Lemma~5.1 of~\cite{vanderVaartvanderLaan2006} implies that 
\begin{align*}
    &\sup_Q  \log N\left(\varepsilon \| H_{h,j,r}\|_{Q,2}, \s{H}_{h,j,r}, L_2(Q) \right) \leq \sup_Q \log N\left(\varepsilon h^{-1/2}\| K\|_{\infty}, h^{-1/2}\s{K}_{h,j}, L_2(Q) \right) + \sup_Q \log N\left(\varepsilon \|L\|_{Q,2}, \Lambda, L_2(Q) \right) \label{eq1:product_covering_H}
\end{align*}
for any $\varepsilon > 0$. By Lemma~\ref{lm:VC_class_K}, $\s{K}_{h,j}$ is VC type, which implies by Theorem~2.6.7 of \cite{vandervaart1996} that 
\[\sup_Q \log N\left(\varepsilon h^{-1/2}\| K\|_{\infty}, h^{-1/2}\s{K}_{h,j}, L_2(Q) \right) = \sup_Q \log N\left(\varepsilon \| K\|_{\infty}, \s{K}_{h,j}, L_2(Q) \right)\lesssim\log \varepsilon^{-1},\]
which is bounded up to a constant by $\varepsilon^{-V}$. 
\end{proof}

\begin{lemma}\label{lm:local_process}
Suppose $\Lambda$ is a class of functions uniformly bounded by $L <\infty$ and satisfying $\sup_Q \log N\left(\varepsilon, \Lambda, L_2(Q) \right) \leq  C\varepsilon^{-V}$ for some $C < \infty$ and $V \in (0,2)$ such that $n \left[ h /(\log n)\right]^{\frac{2+V}{2-V}} \longrightarrow \infty$. If $r = r_n > 0$ is a sequence satisfying 
\[r=\fasterthandet\left(h^{\frac{V}{2(2-V)}}\{\log n\}^{-\frac{1}{2-V}}\right)\]
and \ref{cond:bounded_K} holds, then
\[E_0\left[\sup_{\zeta\in \s{H}_{h,j,r}} \left|\d{G}_n \zeta \right|\right] = \fasterthandet(\{\log n\}^{-1/2})\]
for each $j \in \{1,2,3,4\}$. Consequently, if the above conditions hold and $\lambda_n \in \Lambda$ is a sequence of possibly random functions satisfying 
\[ \rho_{\s{A}_\delta}(\lambda_n, \lambda_\infty) = \fasterthan\left(h^{\frac{V}{2(2-V)}}\{\log n\}^{-\frac{1}{2-V}}\right)\]
for some $\delta > 0$ and~\ref{cond:holder_smooth_theta} also holds, then
\[\sup_{a_0\in\s{A}_0}\left|\d{G}_n\left\{h^{1/2}\Gamma_{0,h,b,a_0}(\lambda_{n} -\lambda_{\infty})\right\} \right| =\fasterthan\left(\{\log n\}^{-1/2}\right).\]
\end{lemma}
\begin{proof}[\bfseries{Proof of Lemma~\ref{lm:local_process}}]
For each $\eta_{h,a_0,j,\lambda}$ such that $h \leq \delta$, we have 
\begin{align*}
    P_0 (\eta_{h,a_0,j,\lambda} - \eta_{h,a_0,j,\lambda_\infty})^2 &= h^{-1}\int \left\{ \left(\frac{a-a_0}{h}\right)^{j-1}K\left(\frac{a-a_0}{h}\right) (\lambda-\lambda_\infty)(y, a, w)\right\}^2\, dP_0(y,a,w) \\
    &= h^{-1}\int  \left(\frac{a-a_0}{h}\right)^{2(j-1)}\left\{K\left(\frac{a-a_0}{h}\right)\right\}^2 E_0\left\{ (\lambda-\lambda_\infty)^2 \mid A = a\right\} f_0(a) \, da \\
    &= \int u^{2(j-1)}K(u)^2 E_0\{(\lambda-\lambda_\infty)^2 \mid A=a_0+uh\} f_0(a_0+uh) \, du \\
    &\leq \|K\|_\infty^2 \sup_{a\in \s{A}_{\delta}}  E_0\{(\lambda-\lambda_\infty)^2 \mid A=a\} \|f_0\|_\infty \\
    &=  \|K\|_\infty^2 \rho_{\s{A}_\delta}(\lambda,\lambda_\infty)^2 \|f_0\|_\infty.
\end{align*}
Hence, for all $\zeta \in \s{H}_{h,j,r}$ with $h \leq \delta$, we have $P_0 \zeta^2 \leq  r^2 \|K\|_\infty^2  \|f_0\|_\infty$. Since $\Lambda$ is uniformly bounded by $L$, $\s{H}_{h,j,r}$ is uniformly bounded by $H_{h,j,r} =  2h^{-1/2}L\|K\|_\infty$, and hence $P_0 H^2_{h,j,r} = 4h^{-1}L^2\|K\|_\infty^2$. Therefore, $P_0 \zeta^2 \leq C^2 hr^2 P_0 H^2_{h,j,r}$ for all $\zeta \in \s{H}_{h,j,r}$ with $h \leq \delta$, where $C = \|f_0\|_\infty^{1/2} / (2L)$ does not depend on $h$ or $r$. Theorem~2.1 of \cite{van2011local} (applied to $\s{F} = \s{H}_{h,j,r} / H_{h,j,r}$, which is uniformly bounded by 1) then implies that for all $n$ large enough,
\begin{align*}
    E_0 \left[\sup_{\zeta \in \s{H}_{h,j,r}} |\d{G}_n \zeta|\right] \lesssim h^{-1/2}J(Crh^{1/2}, \s{H}_{h,j,r}, L_2)\left\{1 + \frac{J(Crh^{1/2}, \s{H}_{h,j,r}, L_2)}{C^2hr^2n^{1/2}}\right\},
\end{align*}
where
\begin{align*}
    J(x, \s{H}_{h,j,r},L_2) &:= \sup_Q \int_0^x \sqrt{1 + \log  N\left(\varepsilon \| H_{h,j,r}\|_{Q,2}, \s{H}_{h,j,r}, L_2(Q) \right)} \, d\varepsilon \lesssim \int_0^x \varepsilon^{-V/2} \, d\varepsilon  \lesssim x^{1-V/2}
\end{align*}
for all $x \in (0,1]$ by Lemma~\ref{lm:uniform_entropy_lambda_r}. Therefore,
\begin{align*}
    \{\log n\}^{1/2}E_0 \left[\sup_{\zeta \in \s{H}_{h,j,r}} |\d{G}_n \zeta|\right] &\lesssim \{\log n\}^{1/2} h^{-1/2} r^{\frac{2 - V}{2}} h^{\frac{2-V}{4}}\left\{ 1 + \frac{r^{\frac{2 - V}{2}} h^{\frac{2-V}{4}}}{hr^2n^{1/2}}\right\} \\
    &= \left\{ r \left[\log n\right]^{\frac{1}{2-V}}  h^{-\frac{V}{2(2-V)}} \right\}^{\frac{2-V}{2}} + \left\{r \left[\log n\right]^{-\frac{1}{2V}} n^{\frac{1}{2V}} h^{\frac{1+V}{2V}} \right\}^{-V} \\
    &= \left\{ r s \right\}^{\frac{2-V}{2}} + \left\{r t \right\}^{-V},
    %&= \left\{ r \left[\log n\right]^{\frac{1}{2-V}}  h^{-\frac{V}{2(2-V)}} \right\}^{\frac{2-V}{2}}\left\{1 + \left[rh^{1/2}n^{\frac{1}{2+V}}\right]^{-\frac{2+V}{2}}\right\}.
\end{align*}
where $s := \left[\log n\right]^{\frac{1}{2-V}}  h^{-\frac{V}{2(2-V)}}$ and $t := \left[\log n\right]^{-\frac{1}{2V}} n^{\frac{1}{2V}} h^{\frac{1+V}{2V}}$. We note that 
\[t/s = \left\{n \left[ h /(\log n)\right]^{\frac{2+V}{2-V}} \right\}^{\frac{1}{2V}} \longrightarrow \infty \]
by assumption. We now define $\bar{r} := \max\{ r, [st]^{-1/2}\}$. Since $r \leq \bar{r}$, we then have 
\[\sup_{\zeta \in \s{H}_{h,j,r}} |\d{G}_n \zeta| \leq \sup_{\zeta \in \s{H}_{h,j,\bar{r}}} |\d{G}_n \zeta|\]
almost surely by the increasing nature of the sets $\s{H}_{h,j,r}$. Using the above bound applied to $\bar{r}$, we also have 
\begin{align*}
    \{\log n\}^{1/2}E_0 \left[\sup_{\zeta \in \s{H}_{h,j,\bar{r}}} |\d{G}_n \zeta|\right] &\lesssim  \left\{ \bar{r} s \right\}^{\frac{2-V}{2}} + \left\{\bar{r} t \right\}^{-V},
    %&= \left\{ r \left[\log n\right]^{\frac{1}{2-V}}  h^{-\frac{V}{2(2-V)}} \right\}^{\frac{2-V}{2}}\left\{1 + \left[rh^{1/2}n^{\frac{1}{2+V}}\right]^{-\frac{2+V}{2}}\right\}.
\end{align*}
By assumption, $rs = \fasterthandet(1)$, and $[st]
^{-1/2}s = [s/t]^{1/2}$, which tends to zero since $t/s \longrightarrow \infty$ as noted above. Therefore, $\bar{r} s = \fasterthandet(1)$. In addition, $\bar{r}t \geq [st]^{-1/2}t = [t/s]^{1/2} \longrightarrow \infty$. Hence, $\left\{ \bar{r} s \right\}^{\frac{2-V}{2}} + \left\{\bar{r} t \right\}^{-V} = \fasterthandet(1)$. Putting it together, we have 
\[ \{\log n\}^{1/2}E_0 \left[\sup_{\zeta \in \s{H}_{h,j,r}} |\d{G}_n \zeta|\right] \leq \{\log n\}^{1/2}E_0 \left[\sup_{\zeta \in \s{H}_{h,j,\bar{r}}} |\d{G}_n \zeta|\right] \lesssim  \left\{ \bar{r} s \right\}^{\frac{2-V}{2}} + \left\{\bar{r} t \right\}^{-V} = \fasterthandet(1),\]
which proves the first claim.

For the second claim, as in the proof of Lemma~\ref{lemma:equicontinuity},  we can write 
\[h^{1/2}\Gamma_{0,h, b, a_0}\lambda  = \sum_{j=1}^2 C_{h,a_0,j} \eta_{h,a_0,j,\lambda} + \sum_{j=1}^3 C_{h,b,a_0,j}' \eta_{b,a_0,j,\lambda}\]
for some constants $C_{h,a_0,j}$ and $C_{b,a_0,j}$ such that $\sup_{a_0\in\s{A}_0} |C_{h,a_0,j}| = \boundeddet(1)$  and $\sup_{a_0\in\s{A}_0} |C_{h,b,a_0,j}'| = \boundeddet(1)$ by the uniform statements of Lemma~\ref{lm:D0_altform}. Thus, %we conclude that 
\begin{align*}
    \sup_{a_0\in\s{A}_0}\left|\d{G}_n\left\{ h^{1/2} \Gamma_{0, h, b, a_0}(\lambda_{n} - \lambda_{\infty}) \right\}\right|  &\leq \sum_{j=1}^2 \sup_{a_0\in\s{A}_0}\left|C_{h,a_0,j}\right| \sup_{a_0\in\s{A}_0}\left|\d{G}_n \eta_{h,a_0,j,\lambda_{n}} \right| +  \sum_{j=1}^3 \sup_{a_0\in\s{A}_0}\left|C_{h,b,a_0,j}'\right| \sup_{a_0\in\s{A}_0}\left|\d{G}_n \eta_{b,a_0,j,\lambda_{n}} \right| \\
    &= \boundeddet(1) \sum_{j=1}^2 \sup_{a_0\in\s{A}_0}\left|\d{G}_n \eta_{h,a_0,j,\lambda_{n}} \right| + \boundeddet(1)\sum_{j=1}^3 \sup_{a_0\in\s{A}_0}\left|\d{G}_n \eta_{b,a_0,j,\lambda_{n}} \right|.
\end{align*}
Hence, if we can show that $\sup_{a_0\in\s{A}_0}\left|\d{G}_n \eta_{h,a_0,j,\lambda_{n}} \right|  = \fasterthan\left(\{\log n\}^{-1/2}\right)$ for each $j$, then the claim follows.

For any $\nu, \gamma > 0$, we have
\begin{align*}
    P_0\left( \sup_{a_0\in\s{A}_0}\left|\d{G}_n \eta_{h,a_0,j,\lambda_{n}} \right| > \nu / \{\log n\}^{1/2}\right) &\leq P_0\left( \sup_{a_0\in\s{A}_0}\left|\d{G}_n \eta_{h,a_0,j,\lambda_{n}} \right| > \nu / \{\log n\}^{1/2}, \rho_{\s{A}_\delta}(\lambda_n, \lambda_\infty) < \gamma / s\right)\\
    &\qquad + P_0 \left( \rho_{\s{A}_\delta}(\lambda_n, \lambda_\infty) \geq \gamma / s \right).
\end{align*}
Now, $\rho_{\s{A}_\delta}(\lambda_n, \lambda_\infty) < \gamma / s$ implies that $\eta_{h,a_0,j,\lambda_{n}} \in \s{H}_{h,j,\gamma / s}$, so by Markov's inequality,
\begin{align*}
    P_0\left( \sup_{a_0\in\s{A}_0}\left|\d{G}_n \eta_{h,a_0,j,\lambda_{n}} \right| > \nu / \{\log n\}^{1/2}, \rho_{\s{A}_0, \delta}(\lambda_n, \lambda_\infty) < \gamma / s\right) &\leq P_0\left( \sup_{\zeta\in\s{H}_{h,j,\gamma / a}}\left|\d{G}_n \zeta\right| > \nu / \{\log n\}^{1/2}\right) \\
    &\leq \nu^{-1} \{\log n\}^{1/2}E_0\left[ \sup_{\zeta\in\s{H}_{h,j,\gamma / s}}\left|\d{G}_n \zeta\right|\right].
\end{align*}
Applying the same the same technique as used above, we then have 
\begin{align*}
     \{\log n\}^{1/2}E_0\left[ \sup_{\zeta\in\s{H}_{h,j,\gamma / s}}\left|\d{G}_n \zeta\right|\right] &\lesssim \left[ \max\{\gamma / s (st)^{-1/2}\} s\right]^{\frac{2-V}{2}} + \left[\max\{\gamma / s, (st)^{-1/2}\} t \right]^{-V} \\
     &= \left[ \max\{\gamma, (s/t)^{1/2}\} \right]^{\frac{2-V}{2}} + \left[\max\{\gamma (t/ s), (t/s)^{1/2}\} \right]^{-V}.
\end{align*}
Therefore, for any $\nu, \gamma > 0$, we have
\begin{align*}
     P_0\left( \sup_{a_0\in\s{A}_0}\left|\d{G}_n \eta_{h,a_0,j,\lambda_{n}} \right| > \nu / \{\log n\}^{1/2}\right) &\lesssim   \nu^{-1} \left[ \max\{\gamma, (s/t)^{1/2}\} \right]^{\frac{2-V}{2}} + \nu^{-1}\left[\max\{\gamma (t/ s), (t/s)^{1/2}\} \right]^{-V}\\
     &\qquad + P_0 \left( \rho_{\s{A}_\delta}(\lambda_n, \lambda_\infty) \geq \gamma / s \right)
\end{align*}
Since $s/t \longrightarrow 0$, $\rho_{\s{A}_\delta}(\lambda_n, \lambda_\infty) = \fasterthan(s^{-1})$, and $\gamma$ was arbitrary, for any fixed $\nu > 0$, we can choose $\gamma$ to make the above expression as small as we like for all $n$ large enough. This implies that 
\[P_0\left( \sup_{a_0\in\s{A}_0}\left|\d{G}_n \eta_{h,a_0,j,\lambda_{n}} \right| > \nu / \{\log n\}^{1/2}\right) \longrightarrow 0\]
for any $\nu > 0$, so that $\sup_{a_0\in\s{A}_0}\left|\d{G}_n \eta_{h,a_0,j,\lambda_{n}} \right|  = \fasterthan\left(\{\log n\}^{-1/2}\right)$ for each $j \in \{1, 2, \dots\}$, which completes the proof.

\end{proof}
Next, we define
\[\bar\eta_{h,a_0, j,\lambda}(y,a, w) :=\int \eta_{h,a_0, j,\lambda}(y,a, w) \, dF_0(a) =  \int h^{1/2} \left(\frac{a-a_0}{h}\right)^{j-1}K_{h,a_0}(a) \lambda(y, a, w) \, dF_0(a).\]
The next lemma provides a rate of convergence for $\sup_{s\in\s{S},\lambda \in \Lambda} \left| \d{G}_n \bar\eta_{h,s,j,\lambda} \right|$, which is sufficient to obtain a rate of convergence of $\sup_{a_0 \in \s{A}_0} |R_{n,h,b,s,2}|$ and $\sup_{a_0 \in \s{A}_0} |R_{n,h,b,s,3}|$. We note that we could obtain an even faster rate of convergence for $\sup_{s\in\s{S},\lambda \in \Lambda} \left| \d{G}_n \left\{\bar\eta_{h,s,j,\lambda} - \bar\eta_{h,s,j,\lambda_\infty}\right\} \right|$ using similar techniques to those above for $\sup_{s\in\s{S},\lambda \in \Lambda} \left| \d{G}_n \left\{\eta_{h,s,j,\lambda} - \eta_{h,s,j,\lambda_\infty}\right\} \right|$. However, the integration over $a$ in $\bar\eta_{h,s,j,\lambda}$ is sufficient to obtain a rate of convergence without localizing around $\lambda_\infty$, which provides a sufficient rate for our purposes.
\begin{lemma}\label{lm:local_integrated_process}
Suppose $\Lambda$ is a class of functions uniformly bounded by $L < \infty$ and satisfying $\sup_Q \log N\left(\varepsilon, \Lambda, L_2(Q) \right) \leq  C\varepsilon^{-V}$ for some $C < \infty$ and $V \in (0,2)$. If~\ref{cond:bounded_K} holds, then $\sup_{a_0\in\s{A}_0,\lambda \in \Lambda} \left| \d{G}_n \bar\eta_{h,a_0,j,\lambda} \right| = \bounded\left(h^{\frac{1-V}{2}} + \left\{ n h^{1 + 2V}\right\}^{-1/2}\right)$ for every $j \in\{ 1,2,\dots\}$. Consequently, if~\ref{cond:holder_smooth_theta} holds as well, then \[\sup_{a_0\in\s{A}_0, \lambda \in \Lambda} \left| \d{G}_n \left\{ \int h^{1/2} \Gamma_{0,h,b,a_0} \lambda \, dF_0\right\} \right| = \bounded\left(h^{\frac{1-V}{2}} + \left\{ n h^{1 + 2V}\right\}^{-1/2}\right).\]
\end{lemma}
\begin{proof}[\bfseries{Proof of Lemma~\ref{lm:local_integrated_process}}]
We consider the class of functions $\bar{\s{H}}_{h, j} := \{\bar{\eta}_{h,a_0, j,\lambda} : a_0\in\s{A}_0, \lambda\in\Lambda\}$. We equip this class with the envelope function $H_{h} = h^{-1/2} L\|K\|_{\infty}$. We also define $\s{H}_{h,j} := \{\eta_{h,a_0, j,\lambda} : a_0\in\s{A}_0, \lambda\in\Lambda\}$, and we note that $H_{h}$ is also an envelope for $\s{H}_{h,j}$. Hence, by Lemma~5.2 of \cite{vanderVaartvanderLaan2006} with $s=t=r=2$ and  Lemma~\ref{lm:uniform_entropy_lambda_r}, we have
\[\sup_{Q} \log N\left(\varepsilon H_{h}, \bar{\s{H}}_{h,j}, L_2(Q)\right) \leq \sup_{Q} \log N\left(\varepsilon H_{h}/2, \s{H}_{h,j}, L_2(Q)\right) \lesssim \varepsilon^{-V}. \]
Therefore,
\begin{align*}
    J(x, \bar{\s{H}}_{h,j},L_2) &:= \sup_Q \int_0^x \sqrt{1 + \log  N\left(\varepsilon  \bar{H}_{h}, \bar{\s{H}}_{h,j}, L_2(Q) \right)} \, d\varepsilon \lesssim x^{1-V/2}
\end{align*}
for all $x \in (0,1]$. We also have
\begin{align*}
    P_0\bar{\eta}_{h,a_0, j,\lambda}^2 &\leq  h^{-1} L^2\left\{ \int \left|\frac{a-a_0}{h}\right|^{j-1}K\left(\frac{a-a_0}{h}\right)  \, dF_0(a) \right\}^2 \\
    &= h L^2 \left\{ \int \left|u\right|^{j-1}K\left(u\right) f_0(a_0 + uh) \, du \right\}^2 \\
    &\leq h L^2 \|f_0\|_\infty^2 \|K\|_\infty^2
\end{align*}
for every $a_0 \in \s{A}_0$ and $\lambda \in \Lambda$. Hence,  $P_0\bar{\eta}_{h,a_0, j,\lambda}^2 \leq \|f_0\|_\infty^2 h^2 P_0 H_h^2$ for every  $\bar{\eta}_{h,a_0, j,\lambda} \in \bar{\s{H}}_{h,j}$. By Theorem~2.1 of \cite{van2011local}, we then have 
\begin{align*}
    E_0 \left[ \sup_{\zeta \in \bar{\s{H}}_{h,j}} |\d{G}_n \zeta| \right] &\lesssim  h^{-1/2} J\left(\|f_0\|_\infty h, \bar{\s{H}}_{h,j},L_2\right) \left\{ 1 + \frac{J\left(\|f_0\|_\infty h, \bar{\s{H}}_{h,j},L_2\right)}{n^{1/2} \|f_0\|_\infty^2 h^2} \right\} \\
    &\lesssim  h^{\frac{1-V}{2}} \left\{ 1 + n^{-1/2} h^{-\frac{2 + V}{2}}\right\} \\
    &=  h^{\frac{1-V}{2}} + \left\{ n  h^{1+2V}\right\}^{-1/2}.
\end{align*}
The second claim follows as in the proof of Lemma~\ref{lm:local_process}. %We can write $h^{1/2}\Gamma_{0,h, b, a_0}\lambda  = \sum_{j=1}^4 C_{h,a_0,j} \eta_{h,a_0,j,\lambda}$ for some constants $C_{h,a_0,j}$ such that $\sup_{a_0\in\s{A}_0} |C_{h,a_0,j}| = \boundeddet(1)$ by the uniform statements of Lemma~\ref{lm:D0_altform}. Thus,
\end{proof}

\begin{lemma}\label{lm:bounded_supGn} Suppose $\Lambda$ is a class of functions uniformly bounded by $L < \infty$ such that $\sup_Q \log N\left(\varepsilon, \Lambda, L_2(Q) \right) \leq  C\varepsilon^{-V}$ for some $C < \infty$ and $V \in (0,2)$. If~\ref{cond:bounded_K} holds, then for each $j \in \{1,2,\dots\}$, $\sup_{a_0\in\s{A}_0, \lambda\in\Lambda}\d{G}_n\left|\eta_{h,a_0, j,\lambda}\right| = \bounded\left(h^{-1/2}\right)$.
\end{lemma}
\begin{proof}[\bfseries{Proof of Lemma~\ref{lm:bounded_supGn}}]
Since $L$ is an envelope for $\Lambda$, $ h^{-1/2}\|K\|_\infty L$ is an envelope for $\{\eta_{h,a_0, j,\lambda} : a_0\in\s{A}_0, \lambda\in\Lambda\}$  by~\ref{cond:bounded_K}. By Lemma~\ref{lm:uniform_entropy_lambda_r}, the uniform entropy integral of this class is finite. Hence, by Theorem 2.14.1 of \cite{vandervaart1996}, we have 
\begin{align*}
    E_0 \left[ \sup_{a_0\in\s{A}_0, \lambda\in\Lambda}\d{G}_n\left|\eta_{h,a_0, j,\lambda}\right| \right] &\lesssim h^{-1/2},
\end{align*}
and the result follows.
\end{proof}

\begin{corollary}\label{cor:supR2R3}
If \ref{cond:bounded_K}--\ref{cond:uniform_entropy_nuisances},~\ref{cond:unif_nuisance_rate},~\ref{cond:unif_doubly_robust}(d), and~\ref{cond:holder_smooth_theta} hold and $nh^p \longrightarrow 0$ for some $p>0$, then $\sup_{a_0\in\s{A}_0}\left| R_{n,h,b,a_0,2}\right| = \fasterthan\left(\{nh\log n\}^{-1/2}\right)$ and $\sup_{a_0\in\s{A}_0}\left| R_{n,h,b,a_0,3}\right| = \fasterthan\left(\{n h\log n\}^{-1/2}\right)$.
\end{corollary}
\begin{proof}[\bfseries{Proof of Corollary~\ref{cor:supR2R3}}]
As in the proof of Corollary~\ref{cor:R2R3}, we can write
\begin{align}
    \sup_{a_0\in\s{A}_0}\left| \{nh\}^{1/2}R_{n, h,b,a_0, 2}\right| &\leq \sup_{a_0\in\s{A}_0}\left| \d{G}_n\left\{ h^{1/2} \Gamma_{0,h,b,a_0} \left( \psi_n - \psi_\infty \right) \right\}\right|\\
    &\qquad+ \sup_{a_0\in\s{A}_0}\left| \d{G}_n\left\{ h^{1/2} \Gamma_{0,h,b,a_0} \left( \int \mu_n \, dQ_0 - \int \mu_\infty \, dQ_0 \right) \right\}\right|\nonumber \\
    &\qquad + \sup_{a_0\in\s{A}_0}\left| \d{G}_n \left\{ \int h^{1/2} \Gamma_{0,h,b,a_0} \mu_n \, dF_0 \right\}\right| + \sup_{a_0\in\s{A}_0}\left|\d{G}_n \left\{ \int h^{1/2} \Gamma_{0,h,b,a_0} \mu_\infty \, dF_0 \right\} \right|.\label{eq:supR2_decomp}
\end{align}
For the first term, we use Lemma~\ref{lm:local_process} with $\Lambda := \{ o \mapsto [y - \mu(a,w)] / g(a,w): \mu \in \s{F}_\mu, g \in \s{F}_g\}$. By~\ref{cond:uniform_entropy_nuisances} and~\ref{cond:holder_smooth_theta}, $\Lambda$ is uniformly bounded by some $L < \infty$. By~\ref{cond:uniform_entropy_nuisances} and permanence properties of uniform entropy integrals, $\sup_Q \log N(\varepsilon, \Lambda, L_2(Q)) \leq C \varepsilon^{-V}$ for $V = \max\{V_\mu, V_g\} \in (0,2)$, and by~\ref{cond:unif_nuisance_rate}, $n[ h / (\log n)]^{\frac{2+V}{2-V}} \longrightarrow \infty$. Using a similar argument to that used in the proof of Corollary~\ref{cor:R2R3} and using the assumption that $|Y|$ is bounded almost surely, we can show that
\begin{align*}
    \rho_{\s{A}_{\delta_3}}(\psi_n, \psi_\infty) \lesssim d(g_n, g_\infty; \s{A}_{\delta_3}, \s{A} \times \s{W})  + d(\mu_n, \mu_\infty;\s{A}_{\delta_3}, \s{A} \times \s{W}).
\end{align*}
Hence, by~\ref{cond:unif_doubly_robust}(d), $\rho_{\s{A}_{\delta_3}}(\psi_n, \psi_\infty)=\fasterthan\left(h^{\frac{V}{2(2-V)}}\{\log n\}^{-\frac{1}{2-V}}\right)$. Thus, the conditions of Lemma~\ref{lm:local_process} hold, and it follows that $\sup_{a_0\in\s{A}_0}\left| \d{G}_n\{ h^{1/2} \Gamma_{0,h,b,a_0} (\psi_n - \psi_\infty)\}\right| = \fasterthan(\{\log n\}^{-1/2})$. 

For the second term in equation~\eqref{eq:supR2_decomp}, we also use Lemma~\ref{lm:local_process} with $\Lambda = \left\{ a \mapsto \int \mu(a, w) \, dQ_0(w) : \mu \in \s{F}_\mu \right\}$. This class is uniformly bounded since $\s{F}_\mu$ is uniformly bounded by~\ref{cond:uniform_entropy_nuisances}. By~\ref{cond:uniform_entropy_nuisances} and Lemma~5.2 of \cite{vanderVaartvanderLaan2006}, $\sup_Q \log N(\varepsilon, \Lambda, L_2(Q)) \lesssim \varepsilon^{-V_\mu} \lesssim \varepsilon^{-V}$, and by~\ref{cond:unif_nuisance_rate}  $n[ h / (\log n)]^{\frac{2+V}{2-V}} \longrightarrow \infty$. By Jensen's inequality and~\ref{cond:unif_doubly_robust}(d),
\[\rho_{\s{A}_{\delta_3}}\left(\int \mu_n \, dQ_0, \int \mu_\infty \, dQ_0\right) \leq d(\mu_n, \mu_\infty;\s{A}_{\delta_3}, \s{A} \times \s{W}) =\fasterthan\left(h^{\frac{V}{2(2-V)}}\{\log n\}^{-\frac{1}{2-V}}\right).\] 
The conditions of Lemma~\ref{lm:local_process} are satisfied, and it follows that
\[\sup_{a_0\in\s{A}_0}\left|\d{G}_n\left\{ h^{1/2} \Gamma_{0,h,b,a_0} \left( \int \mu_n \, dQ_0 - \int \mu_\infty \, dQ_0 \right) \right\}\right| = \fasterthan\left(\{\log n\}^{-1/2}\right).\] 
For the third and fourth terms in equation~\eqref{eq:supR2_decomp}, we use Lemma~\ref{lm:local_integrated_process} with $\Lambda = \s{F}_\mu$. The conditions of the lemma are satisfied by~\ref{cond:uniform_entropy_nuisances}, so that 
\[ \{\log n\}^{1/2}\sup_{a_0\in\s{A}_0, \mu \in \s{F}_\mu} \left| \d{G}_n \left\{ \int h^{1/2} \Gamma_{0,h,b,a_0} \mu \, dF_0 \right\} \right| = \bounded\left(\{\log n\}^{1/2}h^{\frac{1-V_\mu}{2}} + \left\{ nh^{1 + 2V_\mu} / \log n\right\}^{-1/2} \right).\]
By~\ref{cond:unif_nuisance_rate}, $nh^3 \longrightarrow \infty$, which implies that $ nh^{1 + 2V_\mu} / \log n \longrightarrow \infty$ since $V_\mu \in (0,1)$. By assumption, $nh^p \longrightarrow 0$ for some $p > 0$, and since $V_\mu \in (0,1)$, this implies that $\{\log n\}^{1/2}h^{\frac{1-V_\mu}{2}} = \fasterthandet(1)$ as well. Hence,
\[ \sup_{a_0\in\s{A}_0, \mu \in \s{F}_\mu} \left| \d{G}_n \left\{ \int h^{1/2} \Gamma_{0,h,b,a_0} \mu \, dF_0 \right\} \right| = \fasterthan\left(\{\log n\}^{-1/2} \right).\]
We have now shown that every term of equation~\eqref{eq:supR2_decomp} is $\fasterthan\left(\{\log n\}^{-1/2} \right)$, so we conclude that $\sup_{a_0 \in \s{A}_0}\left|R_{n, h,b,a_0, 2}\right| = \fasterthan\left( \left\{nh\log n\right\}^{-1/2}\right)$.

We can similarly decompose $R_{n, h,b,a_0, 3}$ as
\begin{align}
    \sup_{a_0\in\s{A}_0}\left| (nh)^{1/2}R_{n, h,b,a_0, 3}\right| &\leq \sup_{a_0\in\s{A}_0}\left|\d{G}_n\left\{ h^{1/2} \left(\Gamma_{n,h,b,a_0}- \Gamma_{0,h,b,a_0}\right) \psi_n \right\}\right| \nonumber \\
    &\qquad+   \sup_{a_0\in\s{A}_0}\left|\d{G}_n\left\{ h^{1/2} \left(\Gamma_{n,h,b,a_0}- \Gamma_{0,h,b,a_0}\right)  \int \mu_n \, dQ_0\right\} \right|\nonumber \\
    &\qquad + \sup_{a_0\in\s{A}_0}\left|\d{G}_n \left\{ \int h^{1/2} \left(\Gamma_{n,h,b,a_0}- \Gamma_{0,h,b,a_0}\right) \mu_n \, dF_0 \right\}\right|.  \label{eq:sup_R3_decomp}
\end{align}
For the first term in equation~\eqref{eq:sup_R3_decomp}, we note that
\begin{align*}
    \sup_{a_0 \in \s{A}_0}\left|\d{G}_n\left\{ h^{1/2} \left(\Gamma_{n,h,b,a_0}- \Gamma_{0,h,b,a_0}\right) \psi_n \right\}\right|  &=  \sup_{a_0 \in \s{A}_0}\left|e_1^T \left( \b{D}_{n,h,a_0,1}^{-1} - \b{D}_{0,h,a_0,1}^{-1}\right)  \d{G}_n\left\{ h^{1/2} w_{h,a_0,1} K_{h,a_0} \psi_n \right\}\right. \\
    &\qquad -  c_{0,h,a_0,2} (h/b)^2 e_3^T\left( \b{D}_{n,b,a_0,2}^{-1} - \b{D}_{0,b,a_0,2}^{-1}\right)  \d{G}_n\left\{ h^{1/2} w_{b,a_0,2} K_{b,a_0} \psi_n \right\}\\
     &\qquad\left. -  (c_{n,h,a_0,2} - c_{0,h,a_0,2}) (h/b)^2 e_3^T\b{D}_{n,b,a_0,2}^{-1} \d{G}_n\left\{ h^{1/2} w_{b,a_0,2} K_{b,a_0} \psi_n \right\}\right| \\
    &\lesssim \sup_{a_0 \in \s{A}_0}\left\|  \b{D}_{n,h,a_0,1}^{-1} - \b{D}_{0,h,a_0,1}^{-1}\right\|_{\infty}  \sup_{a_0 \in \s{A}_0}\left| \d{G}_n\left\{ h^{1/2} w_{h,a_0,1} \b{1}^T K_{h,a_0} \psi_n \right\}\right| \\
    &\qquad +  \sup_{a_0 \in \s{A}_0}\left\| \b{D}_{n,b,a_0,2}^{-1} - \b{D}_{0,b,a_0,2}^{-1}\right\|_\infty  \sup_{a_0 \in \s{A}_0}\left| \d{G}_n\left\{ h^{1/2} w_{b,a_0,2} \b{1}^T K_{b,a_0} \psi_n \right\}\right| \\
    &\qquad +  \sup_{a_0 \in \s{A}_0}\left| c_{n,h,a_0,2}- c_{0,h,a_02}\right|  \sup_{a_0 \in \s{A}_0}\left| \d{G}_n\left\{ h^{1/2} w_{b,a_0,2} \b{1}^T K_{b,a_0} \psi_n \right\}\right|.
\end{align*}
By~\ref{cond:unif_nuisance_rate}, $nh^3 \longrightarrow \infty$, which implies that $nh / \log h^{-1} \longrightarrow \infty$. Hence, by  Lemma~\ref{lm:Dmatrix_unif}, $\sup_{a_0 \in \s{A}_0}\left\|  \b{D}_{n,h,a_0,1}^{-1} - \b{D}_{0,h,a_0,1}^{-1}\right\|_{\infty}$, $\sup_{a_0 \in \s{A}_0}\left\| \b{D}_{n,b,a_0,2}^{-1} - \b{D}_{0,b,a_0,2}^{-1}\right\|_\infty$, and $\sup_{a_0 \in \s{A}_0}\left| c_{n,h,a_0,2}- c_{0,h,a_02}\right|$ are all $\bounded( \{ nh / \log h^{-1} \}^{-1/2})$. In addition, by Lemma~\ref{lm:bounded_supGn}, we have that $\sup_{a_0 \in \s{A}_0}\left| \d{G}_n\left\{ h^{1/2} w_{h,a_0,1}^T \b{1} K_{h,a_0} \psi_n \right\}\right|$ and $\sup_{a_0 \in \s{A}_0}\left| \d{G}_n\left\{ h^{1/2} w_{b,a_0,2}^T \b{1} K_{b,a_0} \psi_n \right\}\right|$ are both $\bounded(h^{-1/2})$.  We then have
\[  \{ \log n\}^{1/2}\sup_{a_0 \in \s{A}_0}\left|\d{G}_n\left\{ h^{1/2} \left(\Gamma_{n,h,b,a_0}- \Gamma_{0,h,b,a_0}\right) \psi_n \right\}\right| = \bounded\left( \left\{\frac{nh^2}{\log h^{-1} \log n} \right\}^{-1/2} \right).\]
Since $nh^3 \longrightarrow \infty$ and $nh^p \longrightarrow 0$ for some $p > 0$, $\frac{nh^2}{\log h^{-1} \log n} \longrightarrow \infty$. The second term in equation~\eqref{eq:sup_R3_decomp} can be addressed in the same way using Lemma~5.2 of \cite{vanderVaartvanderLaan2006}.

%Now, by~\ref{cond:uniform_entropy_nuisances} and~\ref{cond:cont_density},
%\begin{align*}
%E_0\left\{Y^{2+\delta_2} \mid A = a \right\} &= E_0 \left\{ E_0\left[Y^{2+\delta_} \mid A = a,W \right]g_0(a, W)\right\} \leq C_3 E_0\left\{  E_0\left[Y^{2+\delta} \mid A = a,W \right] \right\}
%\end{align*}

We can similarly bound the third term in equation~\eqref{eq:sup_R3_decomp} up to a constant by
\begin{align*}
    & \sup_{a_0\in\s{A}_0}\left\| \b{D}_{n,h,a_0,1}^{-1} - \b{D}_{0,h,a_0,1}^{-1}\right\|_\infty\sup_{a_0\in\s{A}_0}\left|  \d{G}_n\left\{ \int h^{1/2} w_{h,a_0,1}^T \b{1} K_{h,a_0} \mu_n\, dF_0 \right\}\right| \\
    &\qquad + \sup_{a_0\in\s{A}_0}\left\| \b{D}_{n,b,a_0,2}^{-1} - \b{D}_{0,b,a_0,2}^{-1}\right\|_\infty  \sup_{a_0\in\s{A}_0}\left|\d{G}_n\left\{ \int  h^{1/2} w_{b,a_0,2}^T \b{1} K_{b,a_0} \mu_n \, dF_0 \right\}\right| \\
    &\qquad +  \sup_{a_0\in\s{A}_0}\left| c_{n,h,a_0,2} - c_{0,h,a_0,2}\right| \sup_{a_0\in\s{A}_0}\left|\d{G}_n\left\{ \int  h^{1/2} w_{b,a_0,2}^T \b{1} K_{b,a_0} \mu_n \, dF_0 \right\}\right|
\end{align*}
By Lemmas~\ref{lm:Dmatrix_unif} and~\ref{lm:local_integrated_process}, we then have 
\begin{align*}
    &\{\log n\}^{1/2}\sup_{a_0\in\s{A}_0}\left|\d{G}_n \left\{ \int h^{1/2} \left(\Gamma_{n,h,b,a_0}- \Gamma_{0,h,b,a_0}\right) \mu_n \, dF_0 \right\}\right|\\
    &\qquad= \bounded\left( \left\{\frac{nh}{\log h^{-1} \log n}\right\}^{-1/2} \left\{ h^{(1-V_\mu) / 2} + \left(nh^{1+2V_\mu}\right)^{-1/2} \right\}\right) \\
    &\qquad= \bounded\left( \left\{\frac{nh^{V_\mu}}{\log h^{-1} \log n}\right\}^{-1/2} + \left\{\frac{nh^{1 + V_\mu}}{ \left(\log h^{-1} \log n \right)^{1/2}}\right\}^{-1}\right).
\end{align*}
By~\ref{cond:unif_nuisance_rate}, both terms are $\fasterthandet(1)$.
\end{proof}

%% file: supp/R4.tex
\clearpage 

\section{Analysis of remainder term $R_{n,h,b,a_0,4}$}

\begin{lemma}\label{lm:R4}
If~\ref{cond:bounded_K}--\ref{cond:doubly_robust} hold, then $R_{n, h,b,a_0, 4} = \fasterthan( \{nh\}^{-1/2})$. If~\ref{cond:bounded_K}--\ref{cond:uniform_entropy_nuisances},~\ref{cond:unif_doubly_robust}, and~\ref{cond:holder_smooth_theta} hold, then $\sup_{a_0 \in \s{A}_0}\left|R_{n, h,b,a_0, 4}\right| = \fasterthan\left( \{nh\log n\}^{-1/2}\right)$.
\end{lemma}
\begin{proof}[\bfseries{Proof of Lemma~\ref{lm:R4}}]
We note that $\Gamma_{n,h,b,a_0}(a) = 0$ for all $a$ such that $|a - a_0| > \max\{h,b\}$ by~\ref{cond:bounded_K}. Hence, for $\max\{h,b\} \leq \delta_1$, $\Gamma_{n,h,b,a_0}(a) = 0$ for $a \notin B_{\delta_1}(a_0)$. Therefore, using~\ref{cond:doubly_robust}, for all $\max\{h,b\} \leq \delta_1$, we can write
\begin{align*}
    \left|\mathrel{R}_{n, h,b,a_0, 4}\right| &= \left| \iint_{B_{\delta_1}(a_0) \times \s{W}} \Gamma_{n,h,b,a_0}(a) \left\{\mu_n(a,w) - \mu_0(a,w)\right\}\left\{1-\frac{g_0(a,w)}{g_n(a,w)}\right\}\,dF_0(a)\,dQ_0(w) \right| \\
    &= \left| \sum_{j=1}^3\iint_{\s{S}_j} \Gamma_{n,h,b,a_0}(a) \left\{\mu_n(a,w) - \mu_0(a,w)\right\}\left\{1-\frac{g_0(a,w)}{g_n(a,w)}\right\}\,dF_0(a)\,dQ_0(w) \right| \\
    &\leq \sum_{j=1}^3 \iint_{\s{S}_j} \left|\Gamma_{n,h,b,a_0}(a)\right| \left|\mu_n(a,w) - \mu_0(a,w)\right|\left|1-\frac{g_0(a,w)}{g_n(a,w)}\right|\,dF_0(a)\,dQ_0(w).
\end{align*}
Now we note that for any uniformly bounded function $\lambda : \s{A} \times \s{W} \to \d{R}$ and $\s{S} \subset \s{A} \times \s{W}$, we have
\begin{align*}
    &\iint_S \left| w_{h,a_0,j}(a) K_{h,a_0}(a) \lambda(a,w)\right| \, dF_0(a) \, dQ_0(w)\\
    &\qquad=  \int_{-1}^1 \left|(1, u, \dotsc, u^j)^T \right| K(u) \int I_{\s{S}}(a_0 +uh,w)  \left|\lambda(a_0 +uh, w)\right| \, dQ_0(w) f_0(a_0 + uh) \, du \\
    &\qquad\lesssim \sup_{|a - a_0| \leq h}\int I_{\s{S}} (a,w) \left|\lambda(a, w)\right| \, dQ_0(w).
\end{align*}
Hence, for all $n$ large enough, we have 
\begin{align*}
    &\iint_{\s{S}_j} \left|\Gamma_{n,h,b,a_0}(a)\right| \left|\mu_n(a,w) - \mu_0(a,w)\right|\left|1-\frac{g_0(a,w)}{g_n(a,w)}\right|\,dF_0(a)\,dQ_0(w) \\
    &\qquad \lesssim \left[\left\|\b{D}_{n,h,a_0,1}^{-1}\right\|_{\infty}  +\left\|\b{D}_{n,h,a_0,2}^{-1}\right\|_{\infty}\right] \sup_{|a - a_0| < \delta_1} E_0 \left\{ I_{\s{S}_j}(a,W) \left|\mu_n(a,w) - \mu_0(a,w)\right|\left|1-\frac{g_0(a,W)}{g_n(a,W)}\right| \right\} \\
    &\qquad \lesssim \bounded(1) d\left(\mu_n, \mu_0; B_{\delta_1}(a_0), \s{S}_j\right)d\left(g_n, g_0; B_{\delta_1}(a_0), \s{S}_j\right),
\end{align*}
where for the last inequality we used Lemma~\ref{lm:Dmatrix}, the Cauchy-Schwartz inequality, and the uniform boundedness of $1/g_n$ guaranteed by~\ref{cond:uniform_entropy_nuisances}. We now use~\ref{cond:doubly_robust}  to see that
\begin{align*}
   d\left(\mu_n, \mu_0; B_{\delta_1}(a_0), \s{S}_1\right) &=  \fasterthan\left( \{nh\}^{-1/2}\right), \\
   d\left(g_n, g_0; B_{\delta_1}(a_0), \s{S}_2\right) &=  \fasterthan\left( \{nh\}^{-1/2}\right), \text{ and}\\
   d\left(\mu_n, \mu_0; B_{\delta_1}(a_0), \s{S}_3\right) d\left(g_n, g_0; B_{\delta_1}(a_0), \s{S}_3\right) &= \fasterthan\left(\{nh\}^{-1/2}\right).
\end{align*}
Thus, $R_{n, h,b,a_0, 4} = \fasterthan\left(\{nh\}^{-1/2}\right)$.

For the uniform statement, we note that $\Gamma_{n,h,b,a_0}(a) = 0$ for all $a_0 \in \s{A}_0$, $a \notin \s{A}_{\varepsilon_3}$, and $\max\{h,b\} < \varepsilon_3$. Hence, for $\max\{h,b\} < \varepsilon_3$, we can write
\begin{align*}
    \sup_{a_0 \in \s{A}_0}\left|R_{n, h,b,a_0, 4}\right| &= \sup_{a_0 \in \s{A}_0}\left| \iint_{\s{A}_3 \times \s{W}} \Gamma_{n,h,b,a_0}(a) \left\{\mu_n(a,w) - \mu_0(a,w)\right\}\left\{1-\frac{g_0(a,w)}{g_n(a,w)}\right\}\,dF_0(a)\,dQ_0(w) \right| \\
    &\leq \sum_{j=1}^3 \sup_{a_0 \in \s{A}_0}\iint_{\s{S}_j'} \left|\Gamma_{n,h,b,a_0}(a)\right| \left|\mu_n(a,w) - \mu_0(a,w)\right|\left|1-\frac{g_0(a,w)}{g_n(a,w)}\right|\,dF_0(a)\,dQ_0(w).
\end{align*}
Using the derivations above, we then have
\begin{align*}
    &\sup_{a_0 \in \s{A}_0}\iint_{\s{S}_j'} \left|\Gamma_{n,h,b,a_0}(a)\right| \left|\mu_n(a,w) - \mu_0(a,w)\right|\left|1-\frac{g_0(a,w)}{g_n(a,w)}\right|\,dF_0(a)\,dQ_0(w) \\
    &\qquad \lesssim \sup_{a_0\in\s{A}_0}\left[\left\|\b{D}_{n,h,a_0,1}^{-1}\right\|_{\infty}  +\left\|\b{D}_{n,h,a_0,2}^{-1}\right\|_{\infty}\right] \sup_{a_0\in\s{A}_{\varepsilon_3}} E_0 \left\{ I_{\s{S}_j'}(a,W) \left|\mu_n(a,w) - \mu_0(a,w)\right|\left|1-\frac{g_0(a,W)}{g_n(a,W)}\right| \right\} \\
    &\qquad \lesssim \bounded(1) d\left(\mu_n, \mu_0; \s{A}_{\varepsilon_3}, \s{S}_j'\right)d\left(g_n, g_0; \s{A}_{\varepsilon_3}, \s{S}_j'\right),
\end{align*}
where the last inequality uses Lemma~\ref{lm:Dmatrix_unif} and the uniform boundedness of $1/g_n$ guaranteed by~\ref{cond:uniform_entropy_nuisances}. We now use the faster rates guaranteed by~\ref{cond:unif_doubly_robust} to see that
\begin{align*}
   d(\mu_n, \mu_0; \s{A}_{\varepsilon_3}, \s{S}_1') &=  \fasterthan\left( \{nh \log n\}^{-1/2}\right), \\
   d(g_n, g_0; \s{A}_{\varepsilon_3}, \s{S}_2') &=  \fasterthan\left( \{nh \log n\}^{-1/2}\right), \text{ and}\\
   d(\mu_n, \mu_0; \s{A}_{\varepsilon_3}, \s{S}_3') d_2(g_n, g_0; \s{A}_{\varepsilon_3}, \s{S}_3') &= \fasterthan\left( \{nh \log n\}^{-1/2}\right).
\end{align*}
Thus, $\sup_{a_0\in\s{A}_0}\left|R_{n, h,b,a_0, 4}\right| = \fasterthan\left( \{nh \log n\}^{-1/2}\right)$.
\end{proof}

%% file: supp/R5.tex
\clearpage

\section{Analysis of remainder term $R_{n,h,b,a_0,5}$}

\begin{lemma}\label{lemma:Uprocess_general}
Suppose $X_1, X_2, \dots$ is a sequence of IID random variables on sample space $\s{X}$ with marginal distribution $P_0$. Let $\d{P}_n$ be the empirical distribution corresponding to $(X_1, \dotsc, X_n)$, and $\s{F}$ be a collection of measurable functions from $\s{X} \times \s{X}$ to $\d{R}$ with envelope function $F$. Then
\begin{align*}
    \sup_{f \in \s{F}} \left| \iint f(x_1, x_2) \, d(\d{P}_n - P_0)(x_1) \, d(\d{P}_n - P_0)(x_2)  \right| &\lesssim n^{-1} \left\| F \right\|_{P_0 \times P_0,2} \int_0^1 \left[ 1 + \log \sup_Q N(\varepsilon \|F\|_{Q,2}, \s{F}, L_2(Q)) \right] \, d\varepsilon  \\
    &\quad + n^{-3/2}\|F \|_{P_0 \times P_0,2}  \int_0^1\left[ 1 + \log\sup_Q  N\left(\varepsilon \| F\|_{Q,2}/2,  \s{F}, L_2(Q) \right) \right]^{1/2} \, d\varepsilon.
\end{align*}
\end{lemma}
\begin{proof}[\bfseries{Proof of Lemma~\ref{lemma:Uprocess_general}}]
For each $f \in \s{F}$, define the symmetrized, centered version of $f$ as
\begin{align*}
    f^\circ(x_1, x_2) &:= f(x_1, x_2) + f(x_2, x_1) - \int \left[ f(u,x_1) + f(u, x_2) + f(x_1, u) + f(x_2, u) \right] \, dP_0(u) + 2 \iint f(u,v) \, dP_0(u) \, dP_0(v).
\end{align*}
We note that $f^\circ$ is symmetric in its arguments, meaning $f^\circ(x_1, x_2) = f^\circ(x_2, x_1)$ for all $x_1, x_2 \in \s{X}$, and $\int f(x_1, u) \, dP_0(u) = 0$ for all $x_1\in \s{X}$. We let $\s{F}^\circ := \{f^\circ : f \in \s{F}\}$, and we note that an envelope function for $\s{F}^\circ$ is $F^\circ$ for
\begin{align*}
    F^\circ(x_1, x_2) &:= F(x_1, x_2) + F(x_2, x_1) + \int \left[ F(u,x_1) + F(u, x_2) + F(x_1, u) + F(x_2, u) \right] \, dP_0(u) + 2 \iint F(u,v) \, dP_0(u) \, dP_0(v).
\end{align*}
By adding and subtracting terms, we can write
\begin{align*}
    \iint f(x_1, x_2) \, d(\d{P}_n - P_0)(x_1) \, d(\d{P}_n - P_0)(x_2)  &= \frac{1}{2n^2} \sum_{\substack{i,j = 1\\ i\neq j}}^n f^\circ(X_i, X_j) - \frac{1}{n} \iint f(x_1, x_2) \, dP_0(x_1) \, d(\d{P}_n - P_0)(x_2) \\
    &\qquad -  \frac{1}{n} \iint f(x_1, x_2) \, d(\d{P}_n - P_0)(x_1) \, dP_0(x_2) \\
    &\qquad - \frac{1}{n}\iint f(x_1, x_2) \, dP_0(x_1) \, dP_0(x_2)
\end{align*}
By Lemma~3 of \cite{westling2018causal}, we have
\begin{align*}
    E_0\left[\frac{1}{2n^2} \sup_{f \in \s{F}^\circ} \left| \sum_{\substack{i,j = 1\\ i\neq j}}^n f^\circ(X_i, X_j) \right| \right] \lesssim \frac{1}{n} \left\| F^\circ\right\|_{P_0 \times P_0, 2} \int_0^1 \left[ 1 + \log \sup_Q N(\varepsilon \|F^\circ\|_{Q,2}, \s{F}^\circ, L_2(Q)) \right] \, d\varepsilon.
\end{align*}
By the triangle inequality and Jensen's inequality, $\left\| F^\circ \right\|_{P_0 \times P_0,2} \leq 8 \left\| F \right\|_{P_0 \times P_0,2}$. We also note that by definition of $f^\circ$, $\s{F}^\circ$ is contained in the sum of $\s{F}$ and the following classes:
\begin{align*}
    \s{F}_r &:= \left\{(x_1, x_2) \mapsto f(x_2, x_1) : f \in \s{F} \right\},& \bar{\s{F}}_{11} &:=  \left\{ x_1 \mapsto -\int f(x_1, u) \, dP_0(u) : f \in \s{F}\right\},\\
    \bar{\s{F}}_{12} &:=  \left\{ x_1 \mapsto -\int f(u, x_1) \, dP_0(u) : f \in \s{F}\right\},& \bar{\s{F}}_{21} &:=  \left\{ x_2 \mapsto -\int f(x_2, u) \, dP_0(u) : f \in \s{F}\right\},\\
    \bar{\s{F}}_{22} &:=  \left\{ x_2 \mapsto -\int f(u, x_2) \, dP_0(u) : f \in \s{F}\right\},& \text{and } \s{F}_m &:=  \left\{ 2\iint f(u,v) \, dP_0(u) \, dP_0(v) : f \in \s{F}\right\}.
\end{align*}
By Lemma~5.1 of \cite{vanderVaartvanderLaan2006}, $\log \sup_Q N(\varepsilon \|F^\circ\|_{Q,2}, \s{F}^\circ, L_2(Q))$ is bounded by the sum of the uniform entropies of each of the above classes. The uniform entropy of $\s{F}_r$ is the same as that of $\s{F}$ because for any measure $Q$ on $\s{X} \times \s{X}$, 
\[\int [f_1(x_1, x_2) - f_2(x_1, x_2)]^2 \, dQ(x_1, x_2) = \int [f_1(x_2, x_1) - f_2(x_2, x_1)]^2 \, dQ_r(x_1, x_2),\]
where $Q_r$ is defined as $dQ_r(x_1, x_2) := dQ(x_2, x_1)$.  We equip $\bar{\s{F}}_{11}$ with the envelope function 
\[\bar{F}_{11} : x_1 \mapsto \left[\int F(x_1, u)^2 \, dP_0(u) \right]^{1/2}.\]
We note that $ \| \bar{F}_{11}\|_{P_0,2} = \|F \|_{P_0 \times P_0,2}$. By Lemma~5.2 of \cite{vanderVaartvanderLaan2006} (with $r=s=t=2$), we have
\[ \sup_Q  N\left(\varepsilon \| \bar{F}_{11}\|_{Q,2},  \bar{\s{F}}_{11}, L_2(Q) \right) \leq \sup_Q  N\left(\varepsilon \| F\|_{Q,2} /2,  \s{F}, L_2(Q) \right).\]
Identical results hold for $\bar{F}_{12}$, $\bar{F}_{21}$, and $\bar{F}_{22}$. For $\s{F}_m$ equipped with envelope $\| F\|_{P_0 \times P_0, 2}$, we have 
\[\sup_Q  N\left(\varepsilon \| F\|_{P_0 \times P_0, 2},  \s{F}_m, L_2(Q) \right) \leq \sup_Q  N\left(\varepsilon \| F\|_{P_0 \times P_0, 2},  \s{F}, L_2(Q) \right)\]
since 
\[\left\|\int f_1 \, d(P_0 \times P_0) - \int f_2 \, d(P_0 \times P_0)\right\|_{Q,2} = \left| \int f_1 \, d(P_0 \times P_0) - \int f_2 \, d(P_0 \times P_0)\right| \leq \| f_1 - f_2\|_{P_0 \times P_0, 2}\]
for any probability measure $Q$. Therefore, we have
\begin{align*}
    &\left\| F^\circ\right\|_{P_0 \times P_0, 2} \int_0^1 \left[ 1 + \log \sup_Q N(\varepsilon \|F^\circ\|_{Q,2}, \s{F}^\circ, L_2(Q)) \right] \, d\varepsilon \lesssim \left\| F \right\|_{P_0 \times P_0,2} \int_0^1 \left[ 1 + \log \sup_Q N(\varepsilon \|F\|_{Q,2}, \s{F}, L_2(Q)) \right] \, d\varepsilon
\end{align*}
Next, we have 
\begin{align*}
    E_0\left[\sup_{f \in \s{F}} n^{-1} \left| \iint f(x_1, x_2) \, dP_0(x_1) \, d(\d{P}_n - P_0)(x_2) \right| \right] &= n^{-3/2} E_0\left[\sup_{\bar{f} \in \bar{\s{F}}_{22}} \left| \d{G}_n \bar{f}\right| \right].
\end{align*}
By Theorem~2.14.1 of \cite{vandervaart1996},
\begin{align*}
     E_0\left[\sup_{\bar{f} \in \bar{\s{F}}_{22}} \left| \d{G}_n \bar{f}\right| \right] &\lesssim \| \bar{F}_{22}\|_{P_0,2} \int_0^1\left[ 1 + \log\sup_Q  N\left(\varepsilon \| \bar{F}_{22}\|_{Q,2},  \bar{\s{F}}_{22}, L_2(Q) \right) \right]^{1/2} \, d\varepsilon.
\end{align*}
Hence,
\begin{align*}
    E_0\left[\sup_{f \in \s{F}} n^{-1} \left| \iint f(x_1, x_2) \, dP_0(x_1) \, d(\d{P}_n - P_0)(x_2) \right| \right] \lesssim n^{-3/2}\|F \|_{P_0 \times P_0,2}  \int_0^1\left[ 1 + \log\sup_Q  N\left(\varepsilon \| F\|_{Q,2}/2,  \s{F}, L_2(Q) \right) \right]^{1/2} \, d\varepsilon.
\end{align*}  
By an identical argument, we also have
\begin{align*}
    E_0\left[\sup_{f \in \s{F}} n^{-1} \left| \iint f(x_1, x_2) \, d(\d{P}_n - P_0)(x_1) \, dP_0(x_2)\right| \right] \lesssim n^{-3/2}\|F \|_{P_0 \times P_0,2}  \int_0^1\left[ 1 + \log\sup_Q  N\left(\varepsilon \| F\|_{Q,2}/2,  \s{F}, L_2(Q) \right) \right]^{1/2} \, d\varepsilon.
\end{align*} 
Finally, 
\[\sup_{f \in \s{F}} n^{-1}\left| \iint f(x_1, x_2) \, dP_0(x_1) \, dP_0(x_2) \right| \leq n^{-1} \|F \|_{P_0 \times P_0,2}.\]
Putting together the pieces yields the result.
\end{proof}

\begin{lemma}\label{lemma:R5}
If~\ref{cond:bounded_K}--\ref{cond:uniform_entropy_nuisances}, and~\ref{cond:cont_density} hold, then $\mathrel{R}_{n, h,b, a_0, 5} = \bounded(\{nh^{1/2}\}^{-1})$. If~\ref{cond:bounded_K}, \ref{cond:uniform_entropy_nuisances}, and~\ref{cond:holder_smooth_theta} hold and $nh / \log h^{-1} \longrightarrow \infty$, then $\sup_{a_0 \in \s{A}_0}|R_{n, h,b, a_0, 5}| = \bounded(\{nh \}^{-1})$.
\end{lemma}
\begin{proof}[\bfseries{Proof of Lemma~\ref{lemma:R5}}]
For the pointwise statement, we write
\begin{align*}
    \mathrel{R}_{n, h,b, a_0, 5} &= \iint \Gamma_{n,h, b, a_0}(a)\mu_n(a,w)\,d(Q_n-Q_0)(w)\,d(F_n-F_0)(a) \\
    &= e_1^T \b{D}_{n,h,a_0,1}^{-1} \iint w_{h,a_0,1}(a) K_{h,a_0}(a) \mu_n(a,w)\,d(Q_n-Q_0)(w)\,d(F_n-F_0)(a)\\
    &\qquad -  c_{n,h,a_0,2} \tau_n^2 e_3^T\b{D}_{n,b,a_0,2}^{-1} \iint w_{b,a_0,2}(a) K_{b,a_0}(a) \mu_n(a,w)\,d(Q_n-Q_0)(w)\,d(F_n-F_0)(a).
\end{align*}
By Lemmas~\ref{lm:D0_altform} and~\ref{lm:Dmatrix}, $\left\|\b{D}_{n,h,a_0,1}^{-1}\right\|_\infty$ and $\left\|\b{D}_{n,h,a_0,2}^{-1}\right\|_\infty$ are both $\bounded(1)$. Hence, there exist constants $C_{n,j} = \bounded(1)$ such that 
\begin{align*}
    \left|R_{n, h,b, a_0, 5}\right| &\leq \sum_{j=0}^4 C_{n,j} \left| \iint [(a-a_0) / h]^j K_{h,a_0}(a) \mu_n(a,w)\,d(Q_n-Q_0)(w)\,d(F_n-F_0)(a)\right|.
\end{align*}
For each $j$, we then write
\begin{align*}
   &\left| \iint [(a-a_0) / h]^j K_{h,a_0}(a) \mu_n(a,w)\,d(Q_n-Q_0)(w)\,d(F_n-F_0)(a)\right| \\
   &\qquad=   \left| \iint f_{n,h,a_0,j}(a_1, w_1, a_2, w_2)\,d(\d{P}_n-P_0)(a_1,w_1)\,d(\d{P}_n-P_0)(a_2,w_2)\right| \\
   &\qquad\leq \sup_{f \in \s{F}_{h,a_0,j}}  \left| \iint f(a_1, w_1, a_2, w_2)\,d(\d{P}_n-P_0)(a_1,w_1)\,d(\d{P}_n-P_0)(a_2,w_2)\right|,
\end{align*}
where
\begin{align*}
f_{n,h,a_0,j}(a_1, w_1, a_2, w_2) &:= [(a_1-a_0) / h]^j K_{h,a_0}(a_1) \mu_n(a_1,w_2) \text{ and} \\
\s{F}_{h,a_0,j} &= \left\{(a_1, w_1, a_2, w_2) \mapsto [(a_1-a_0) / h]^j K_{h,a_0}(a_1) \mu(a_1,w_2) : \mu \in \s{F}_\mu  \right\}.
\end{align*}
By~\ref{cond:bounded_K} and~\ref{cond:uniform_entropy_nuisances}, an envelope function $F_{h,a_0}$ for $\s{F}_{h,a_0,j}$ is given by a constant times $K_{h,a_0}(a_1)$, and $\|F_{h,a_0}\|_{(P_0 \times P_0),2} \lesssim h^{-1/2}$ by the standard change of variables. In addition, by~\ref{cond:uniform_entropy_nuisances}, \[ \sup_Q \log N\left(\varepsilon \|F_{h,a_0}\|_{(P_0 \times P_0),2}, \s{F}_{h,a_0,j}, L_2(Q)\right) \lesssim \varepsilon^{-V_\mu},\]
where $V_\mu \in (0,1)$. Hence, by Lemma~\ref{lemma:Uprocess_general}, 
\begin{align*}
&E_0 \left[ \sup_{f \in \s{F}_{h,a_0,j}}  \left| \iint f(a_1, w_1, a_2, w_2)\,d(\d{P}_n-P_0)(a_1,w_1)\,d(\d{P}_n-P_0)(a_2,w_2)\right| \right] \lesssim n^{-1} h^{-1/2}
\end{align*}
for each $j$, and the pointwise result follows.

For the uniform result, by Lemmas~\ref{lm:D0_altform} and~\ref{lm:Dmatrix_unif}, $\sup_{a_0 \in\s{A}_0}\left\|\b{D}_{n,h,a_0,1}^{-1}\right\|_\infty$ and $\sup_{a_0 \in\s{A}_0}\left\|\b{D}_{n,h,a_0,2}^{-1}\right\|_\infty$ are both $\bounded(1)$. Hence, there exist constants $C_{n,j}' = \bounded(1)$ such that 
\begin{align*}
    \sup_{a_0 \in \s{A}_0}\left|R_{n, h,b, a_0, 5}\right| &\leq \sum_{j=0}^2 C_{n,j}' \sup_{a_0 \in \s{A}_0}\left| \iint [(a-a_0) / h]^j K_{h,a_0}(a) \mu_n(a,w)\,d(Q_n-Q_0)(w)\,d(F_n-F_0)(a)\right|.
\end{align*}
As above, for each $j$ we then write
\begin{align*}
   &\sup_{a_0\in\s{A}_0}\left| \iint [(a-a_0) / h]^j K_{h,a_0}(a) \mu_n(a,w)\,d(Q_n-Q_0)(w)\,d(F_n-F_0)(a)\right| \\
   &\qquad=  \sup_{a_0 \in \s{A}_0} \left| \iint f_{n,h,a_0,j}(a_1, w_1, a_2, w_2)\,d(\d{P}_n-P_0)(a_1,w_1)\,d(\d{P}_n-P_0)(a_2,w_2)\right| \\
   &\qquad\leq \sup_{f \in \s{F}_{h,j}}  \left| \iint f(a_1, w_1, a_2, w_2)\,d(\d{P}_n-P_0)(a_1,w_1)\,d(\d{P}_n-P_0)(a_2,w_2)\right|,
\end{align*}
where
\begin{align*}
\s{F}_{h,j} &= \left\{(a_1, w_1, a_2, w_2) \mapsto [(a_1-a_0) / h]^j K_{h,a_0}(a_1) \mu(a_1,w_2) : \mu \in \s{F}_\mu, a_0 \in \s{A}_0  \right\}.
\end{align*}
By~\ref{cond:bounded_K} and~\ref{cond:uniform_entropy_nuisances}, $\s{F}_{h,j}$ is uniformly bounded up to a constant by $h^{-1}$, and 
by Lemma~\ref{lm:uniform_entropy_lambda_r} and~\ref{cond:uniform_entropy_nuisances}, $\s{F}_{h,j}$ has uniform entropy bounded up to a constant by $\varepsilon^{-V_\mu}$ relative to this envelope. Thus, by Lemma~\ref{lemma:Uprocess_general}, 
\begin{align*}
&E_0 \left[ \sup_{f \in \s{F}_{h,j}}  \left| \iint f(a_1, w_1, a_2, w_2)\,d(\d{P}_n-P_0)(a_1,w_1)\,d(\d{P}_n-P_0)(a_2,w_2)\right| \right] \lesssim (nh)^{-1}
\end{align*}
for each $j$, and the result follows.
\end{proof}

%% file: supp/R6.tex
\clearpage

\section{Analysis of remainder term $R_{n,h,b,a_0,6}$}

\begin{lemma}\label{lm:R6}
If~\ref{cond:bounded_K},~\ref{cond:bandwidth}, and~\ref{cond:cont_density} hold, then $|R_{n, h,b, a_0, 6}| = \bounded(\{nh\}^{-1})$. If~\ref{cond:bounded_K},~\ref{cond:bandwidth}, and~\ref{cond:holder_smooth_theta} hold and  $nh/\log h^{-1} \longrightarrow \infty$, then $\sup_{a_0 \in \s{A}_0}|R_{n, h,b,a_0, 6}| = \bounded( \{nh/\log h^{-1}\}^{-1})$.
\end{lemma}
\begin{proof}[\bfseries{Proof of Lemma~\ref{lm:R6}}]
We note that by~\ref{cond:bounded_K} and~\ref{cond:cont_density}, $\left| P_0 \left( w_{h,a_0, 1} K_{h, a_0} \theta_0\right) \right|$, $\left| P_0 \left( w_{b,a_0, 2} K_{b, a_0} \theta_0\right)\right|$, and $\left| P_0 \left( \tilde{w}_{h,a_0, 2} K_{h, a_0} \right)\right|$ are $\boundeddet(1)$, and by Lemma~\ref{lm:D0_altform}, $\left\| \b{D}^{-1}_{0,h, a_0,1} \right\|_{\infty}$ and $\left\| \b{D}^{-1}_{0,b, a_0,2} \right\|_\infty$ are also $\boundeddet(1)$. Hence,
\begin{align*}
    \left|R_{n, h,b, a_0, 6}\right| &\lesssim \left\| \b{D}_{0,h, a_0,1} - \b{D}_{n,h, a_0,1}\right\|_{\infty} \left\| \b{D}^{-1}_{n,h, a_0,1} - \b{D}^{-1}_{0,h, a_0,1}\right\|_\infty \\
    &\qquad+\left\| \b{D}_{0,b, a_0,2} - \b{D}_{n,b, a_0,2} \right\|_\infty \left\| \b{D}^{-1}_{n,b, a_0,2} - \b{D}^{-1}_{0,b, a_0,2}\right\|_\infty \\
    &\qquad +  \left\| \b{D}^{-1}_{n,h, a_0,1} - \b{D}^{-1}_{0,h,a_0,1}\right\|_\infty \left\| (\d{P}_n- P_0) (\tilde{w}_{h,a_0,1} K_{h,a_0}) \right\|_\infty  \\
    &\qquad +  \left\| \b{D}^{-1}_{n,h, a_0,1} - \b{D}^{-1}_{0,h,a_0,1}\right\|_\infty\left\| \b{D}_{0,h, a_0,1} - \b{D}_{n,h, a_0,1}\right\|_{\infty} \\
    &\qquad + \left|c_{n,h,a_0,2} - c_{0,h,a_0,2}\right| \left\| \b{D}^{-1}_{n,b, a_0,2} - \b{D}^{-1}_{0,b, a_0,2}\right\|_\infty 
\end{align*}
By Lemma~\ref{lm:Dmatrix}, each of the differences is $\bounded(\{nh\}^{-1/2})$. Thus,  $R_{n, h,b, a_0, 6} = \bounded(\{nh\}^{-1})$. For the uniform statement, by~\ref{cond:bounded_K} and \ref{cond:holder_smooth_theta}, $\sup_{a_0\in\s{A}_0}\left| P_0 \left( w_{h,a_0, 1} K_{h, a_0} \theta_0\right) \right|$, $\sup_{a_0\in\s{A}_0}\left| P_0 \left( w_{b,a_0, 2} K_{b, a_0} \theta_0\right)\right|$, and $\sup_{a_0 \in \s{A}_0}\left| P_0 \left( \tilde{w}_{h,a_0, 2} K_{h, a_0} \right)\right|$ are all $\boundeddet(1)$. By Lemma~\ref{lm:D0_altform} and~\ref{cond:holder_smooth_theta}, $\sup_{a_0\in\s{A}_0}\left\| \b{D}^{-1}_{0,h, a_0,1} \right\|_{\infty}$ and $\sup_{a_0\in\s{A}_0}\left\| \b{D}^{-1}_{0,b, a_0,2} \right\|_\infty$ are also $\boundeddet(1)$. By Lemma \ref{lm:Dmatrix_unif}, the differences are all $\bounded(\{nh/\log h^{-1}\}^{-1/2})$ uniformly over $a_0 \in \s{A}_0$. We thus conclude $\sup_{a_0\in\s{A}_0}\left|R_{n, h,b, a_0, 6}\right|=\bounded(\{nh/\log h^{-1}\}^{-1})$.
\end{proof}

%% file: supp/covar.tex
\clearpage

\section{Analysis of the covariance estimator}
\begin{lemma}\label{lemma:covar}
If~\ref{cond:bounded_K}--\ref{cond:cont_density} hold for $a_0 = u$ and $a_0 = v$, then
\begin{align*}
    h\d{P}_n( \phi_{n,h,b,u}^* \phi_{n,h,b,v}^*) - hP_0( \phi_{\infty,h,b,u}^* \phi_{\infty,h,b,v}^*) = \fasterthan(1).
\end{align*}
If~\ref{cond:bounded_K}--\ref{cond:uniform_entropy_nuisances} and~\ref{cond:unif_nuisance_rate}--\ref{cond:holder_smooth_theta} hold and $nh^5 = \boundeddet(1)$, then
\begin{align*}
    \sup_{u,v \in \s{A}_0} \left| h\d{P}_n( \phi_{n,h,b,u}^* \phi_{n,h,b,v}^*) - hP_0( \phi_{\infty,h,b,u}^* \phi_{\infty,h,b,v}^*) \right|= \bounded(n^{-p})
\end{align*}
for some $p > 0$.
\end{lemma}
\begin{proof}[\bfseries{Proof of Lemma~\ref{lemma:covar}}]
Analysis of this expression uses many of the techniques developed throughout this document. Hence, for brevity, we omit some of the details in this proof.

For convenience, we define
\begin{align*}
    \eta_{n,h,b,u}(w) &:= \int \Gamma_{n,h,b,u}(a)\left[\mu_n(a,w) - \smallint \mu_n \, dQ_n\right] \, dF_n(a), \text{ and} \\
    \eta_{0,h,b,u}(w) &:= \int \Gamma_{0,h,b,u}(a)\left[\mu_\infty(a,w) - \smallint \mu_\infty \, dQ_0\right] \, dF_0(a).
\end{align*}
Up to terms with $u$ and $v$ swapped, we then expand $h\d{P}_n( \phi_{n,h,b,u}^* \phi_{n,h,b,v}^*) - hP_0( \phi_{\infty,h,b,u}^* \phi_{\infty,h,b,v}^*)$ as
\begin{align}
 &n^{-1/2}\d{G}_n\left(h\Gamma_{n,h,b,u} \Gamma_{n,h,b,v}\xi_n^2\right)+ n^{-1/2}\d{G}_n\left( h \gamma_{n,h,b,u} \gamma_{n,h,b,v}\right) -  n^{-1/2}\d{G}_n \left(h\Gamma_{n,h,b,u}\gamma_{n,h,b,v} \xi_n \right)\nonumber \\
 &\qquad + P_0\left[h\left(\Gamma_{n,h,b,u} - \Gamma_{0,h,b,u}\right)\left( \Gamma_{n,h,b,v}\xi_n - \gamma_{n,h,b,v}\right) \xi_n \right]+P_0\left[h \Gamma_{0,h,b,u} \left\{ \Gamma_{0,h,b,v}\left(\xi_n + \xi_\infty\right) - \gamma_{0,h,b,v} \right\} \left\{ \xi_n - \xi_\infty \right\}\right] \nonumber \\
 &\qquad - P_0\left[ h\left( \Gamma_{0,h,b,u} \xi_n - \gamma_{0,h,b,u}\right) \left(\gamma_{n,h,b,v} - \gamma_{0,h,b,v}\right) \right] \nonumber   \\
  &\qquad + \d{P}_n\left(  h \eta_{n,h,b,u}\eta_{n,h,b,v}\right) + \d{P}_n \left( h \Gamma_{n,h,b,u} \eta_{n,h,b,v} \xi_n \right) - \d{P}_n \left( h \gamma_{n,h,b,u} \eta_{n,h,b,v}  \right)  \nonumber\\ 
 &\qquad - P_0 \left(h\eta_{0,h,b,u}\eta_{0,h,b,v} \right) - P_0 \left(h\Gamma_{0,h,b,u}\eta_{0,h,b,v} \xi_\infty \right) +P_0\left(h\gamma_{0,h,b,u} \eta_{0,h,b,v}\right).  \label{eq:emp_process_covar}
\end{align}
For the first term in~\eqref{eq:emp_process_covar}, by expanding $\Gamma_{n,h,b,u} \Gamma_{n,h,b,v}$ and Lemmas~\ref{lm:D0_altform} and~\ref{lm:Dmatrix}, $\d{G}_n\left(h\Gamma_{n,h,b,u} \Gamma_{n,h,b,v}\xi_n^2\right)$ can be decomposed into $\bounded(1)$ times terms of the form 
\begin{align*}
\d{G}_n\left(h w_{h,u,j} w_{h,v, k} K_{h,u} K_{b,v}\xi_n^2\right).
\end{align*}
By the bounded fourth moment of $Y$ and condition~\ref{cond:uniform_entropy_nuisances}, the class
\[ \left\{h w_{h,u,j} w_{h,v, k} K_{h,u} K_{b,v}\xi_n^2 : \mu_n \in \s{F}_\mu, g_n \in \s{F}_g\right\}\]
has finite uniform entropy integral and an envelope that is a square-integrable  function times $h^{-1}$. Hence, by Theorem~2.14.1 of \cite{vandervaart1996}, $\d{G}_n\left(h\Gamma_{n,h,b,u} \Gamma_{n,h,b,v}\xi_n^2\right) = \bounded(h^{-1})$. Similarly, by conditions~\ref{cond:bounded_K} and~\ref{cond:uniform_entropy_nuisances}, the class
\[ \left\{h w_{h,u,j} w_{h,v, k} K_{h,u} K_{b,v}\xi_n^2 : \mu_n \in \s{F}_\mu, g_n \in \s{F}_g, u \in \s{A}_0, v \in \s{A}_0\right\}\]
possesses finite uniform entropy integral and under~\ref{cond:holder_smooth_theta}, an envelope that is bounded up to a constant by $h^{-1}$. Therefore, by Lemmas~\ref{lm:D0_altform} and~\ref{lm:Dmatrix_unif}, $\sup_{u, v \in \s{A}_0}|\d{G}_n\left(h\Gamma_{n,h,b,u} \Gamma_{n,h,b,v}\xi_n^2\right)|= \bounded(h^{-1})$ as well. 
Using an analogous argument, we can show that $\d{G}_n\left( h \gamma_{n,h,b,u} \gamma_{n,h,b,v}\right)$ and $\d{G}_n \left(h\Gamma_{n,h,b,u} \gamma_{h,h,b,v}\xi_n \right)$, which are the second and third terms in~\eqref{eq:emp_process_covar} are $\bounded(h^{-1})$ both pointwise and uniformly.

The fourth term in~\eqref{eq:emp_process_covar} can be written as a constant times a sum of terms of the form 
\[h \left(\b{D}_{n,h,u,j}^{-1} - \b{D}_{0,h,u,j}^{-1}\right)P_0 \left[ w_{h, u, k} K_{h,u}\left( \Gamma_{n,h,b,v}\xi_n - \gamma_{n,h,b,v}\right) \xi_n \right].  \]
By Lemma~\ref{lm:Dmatrix}, $\| \b{D}_{n,h,u,j}^{-1} - \b{D}_{0,h,u,j}^{-1}\|_\infty = \bounded(\{nh\}^{-1/2})$. Using a change of variables, we can also show that $\| w_{h, u, k} K_{h,u}\xi_n \|_{L_2(P_0)} = \bounded(h^{-1/2})$ and $\| \Gamma_{n,h,b,v}\xi_n - \gamma_{n,h,b,v}\|_{L_2(P_0)} = \bounded(h^{-1/2})$. Hence, under the pointwise conditions this term is $\bounded(\{nh\}^{-1/2}) = \fasterthan(1)$. Under the uniform results, the same rates hold except that Lemma~\ref{lm:Dmatrix_unif} includes an extra $\log h^{-1}$ term. Hence, under the uniform conditions, this term is $\bounded(\{nh / \log h^{-1}\}^{-1/2}) = \bounded(h)$ uniformly over $u,v \in \s{A}_0$. A similar analysis applies to the sixth term in~\eqref{eq:emp_process_covar}.

For the fifth term in~\eqref{eq:emp_process_covar}, we can show that $\| \Gamma_{0,h,b,u}(\xi_n - \xi_\infty)\|_{L_2(P_0)} = \fasterthan(h^{-1/2})$ and $\|\Gamma_{0,h,b,v}\left(\xi_n + \xi_\infty\right) - \gamma_{0,h,b,v} \|_{L_2(P_0)} = \bounded(h^{-1/2})$, so the pointwise result follows by the Cauchy-Schwarz inequality. For the uniform result, by~\ref{cond:unif_nuisance_rate}(e), 
\[ \sup_{u \in\s{A}_0} \| \Gamma_{0,h,b,u}(\xi_n - \xi_\infty)\|_{L_2(P_0)} = \fasterthan\left(h^{-1/2}h^{\frac{V}{2(2-V)}}\{\log n\}^{-\frac{1}{2-V}}\right),\]
so we obtain a uniform rate of $\fasterthan\left(h^{\frac{V}{2(2-V)}}\{\log n\}^{-\frac{1}{2-V}}\right)$ for this term.

For the seventh term in~\eqref{eq:emp_process_covar}, by the boundedness of $\mu_n$ we have for any $u,v$
\begin{align*}
\left|\d{P}_n\left(  h \eta_{n,h,b,u} \eta_{n,h,b,v} \right)\right| &\lesssim h   \int \left| \Gamma_{n,h,b,u}\right| \, dF_n \int \left| \Gamma_{n,h,b,v}\right| \, dF_n.
\end{align*}
As in the proof of Lemmas~\ref{lm:Dmatrix} and~\ref{lm:Dmatrix_unif}, we can show that each of these terms is $\bounded(1)$ pointwise and uniformly under the appropriate conditions. The same holds when $\d{P}_n$ is replaced with $P_0$ above. Similarly, for the eighth term in~\eqref{eq:emp_process_covar}, we have
\begin{align*}
\left|\d{P}_n\left(  h \Gamma_{n,h,b,u}\eta_{n,h,b,v} \xi_n \right)\right| &\lesssim h   \d{P}_n \left| \Gamma_{n,h,b,u}\xi_n\right|  \int \left| \Gamma_{n,h,b,v}\right| \, dF_n.
\end{align*}
We can show that both terms are $\bounded(1)$ pointwise and uniformly. The remainder of the terms in~\eqref{eq:emp_process_covar} can be handled using a similar argument.

We have now shown that 
\begin{align*}
h P_0 \left[ \left(\phi_{n,h,b,u}^* - \phi_{\infty,h,b,u}^*\right)\phi_{n,h,b,v}^*\right] = \fasterthan(1)
\end{align*}
under the pointwise conditions, and 
\begin{align*}
\sup_{u,v \in \s{A}_0} h \left| P_0 \left[ \left(\phi_{n,h,b,u}^* - \phi_{\infty,h,b,u}^*\right)\phi_{n,h,b,v}^*\right]\right|  = \fasterthan\left(n^{-p}\right)
\end{align*}
for some $p > 0$ since $nh^5 = \boundeddet(1)$.

\end{proof}

%% file: supp/unif.tex
\clearpage

\section{Lemmas supporting uniform result}

We now present a series of supporting lemmas towards the construction of uniform bands for $\theta_0(a_0)$.

\begin{lemma}\label{lm:unif_bounded_variance}
Let $\sigma^2_{\infty,h,b}(a_0) := hP_0(\phi^{*}_{\infty, h, b, a_0})^2$. If \ref{cond:bounded_K}--\ref{cond:bandwidth},~\ref{cond:unif_doubly_robust}(a), and~\ref{cond:holder_smooth_theta} hold, then there exist constants $c, C \in (0,\infty)$ such that $c \leq \sigma^2_{\infty,h,b}(a_0) \leq C $ for all $a_0 \in \s{A}_0$ and all h small enough.
\end{lemma}
\begin{proof}[\bfseries{Proof of Lemma~\ref{lm:unif_bounded_variance}}]
We note that~\ref{cond:holder_smooth_theta} implies that~\ref{cond:cont_density} holds for all $a_0 \in \s{A}_0$, and~\ref{cond:unif_doubly_robust}(a) implies that~\ref{cond:doubly_robust}(a) holds for all $a_0 \in \s{A}_0$. Therefore, the conditions of Lemma~\ref{lm:lindeberg_feller_CLT} are met for every $a_0 \in \s{A}_0$, so $\sigma_{\infty,h,b}^2(a_0)$ converges to $V_{K, \tau}\sigma_0^2(a_0) / f_0(a_0)$ for each $a_0 \in \s{A}_0$. Since $\s{A}_0$ is compact, this convergence is also uniform in $a_0$. By Lemma~\ref{lm:lindeberg_feller_CLT}, $V_{K,\tau} \in (0, \infty)$, so by \ref{cond:holder_smooth_theta}, there exist $c, C \in (0,\infty)$ such that $c \leq V_{K,\tau}\sigma_0^2(a_0) / f_0(a_0) \leq C$ for all $a_0 \in \s{A}_0$. Hence, for all $h$ small enough, $c \leq \sigma_{\infty,h,b}^2(a_0) \leq C$ for all $a_0 \in \s{A}_0$.
\end{proof}

% \begin{lemma}\label{lm:unif_bounded_variance}
% Let $\sigma^2_{\infty,h,b}(a_0) := P_0(\phi^{*}_{\infty, h, b, a_0})^2$. If \ref{cond:bounded_K}--\ref{cond:bandwidth},~\ref{cond:unif_doubly_robust}(a), and~\ref{cond:holder_smooth_theta} hold, then there exist constants $c, C \in (0,\infty)$ such that $c \leq h\sigma^2_{\infty,h,b}(a_0) \leq C $ for all $a_0 \in \s{A}_0$ and all h small enough.
% \end{lemma}
% \begin{proof}[\bfseries{Proof of Lemma~\ref{lm:unif_bounded_variance}}]
% We note that~\ref{cond:holder_smooth_theta} implies that~\ref{cond:cont_density} holds for all $a_0 \in \s{A}_0$, and~\ref{cond:unif_doubly_robust}(a) implies that~\ref{cond:doubly_robust}(a) holds for all $a_0 \in \s{A}_0$. Therefore, the conditions of Lemma~\ref{lm:lindeberg_feller_CLT} are met for every $a_0 \in \s{A}_0$, so $h \sigma_{\infty,h,b}^2(a_0)$ converges to $V_{K, \tau}\sigma_0^2(a_0) / f_0(a_0)$ for each $a_0 \in \s{A}_0$. Since $\s{A}_0$ is compact, this convergence is also uniform in $a_0$. By Lemma~\ref{lm:lindeberg_feller_CLT}, $V_{K,\tau} \in (0, \infty)$, so by \ref{cond:holder_smooth_theta}, there exist $c, C \in (0,\infty)$ such that $c \leq V_{K,\tau}\sigma_0^2(a_0) / f_0(a_0) \leq C$ for all $a_0 \in \s{A}_0$. Hence, for all $h$ small enough, $c \leq h \sigma_{\infty,h,b}^2(a_0) \leq C$ for all $a_0 \in \s{A}_0$.
% \end{proof}

\begin{lemma}\label{lm:phi_finite_moments} 
If~\ref{cond:bounded_K},~\ref{cond:unif_doubly_robust}(a), and~\ref{cond:holder_smooth_theta} hold, then $\sup_{a_0 \in \s{A}_0}P_0|h\phi^*_{\infty, h,b,a_0}|^k \lesssim h$ for all $h$ small enough.
\end{lemma}
\begin{proof}[\bfseries{Proof of Lemma~\ref{lm:phi_finite_moments}}]
    First by the triangle inequality, we have that 
    \begin{align}
        \sup_{a_0 \in \s{A}_0}\left[P_0|h\phi^*_{\infty, h,b,a_0}|^k\right]^{1/k} &\leq \sup_{a_0 \in \s{A}_0}\left[P_0 \left|h\Gamma_{0,h, b, a_0}\xi_\infty\right|^k\right]^{1/k} + \sup_{a_0 \in \s{A}_0}\left[P_0 \left|h\gamma_{0, h, b, a_0}\right|^k\right]^{1/k}\nonumber\\
        &\qquad + \sup_{a_0 \in \s{A}_0}\left[P_0\left|h\int \Gamma_{0,h, b, a_0}(\bar{a})\left\{ \mu_\infty(\bar{a},w) - \int \mu_\infty(\bar{a}, \bar{w}) \, dQ_0(\bar{w}) \right\}\, dF_0(\bar{a})\right|^k\right]^{1/k}. \label{eq:expr1}
    \end{align}
    We show that three terms on the right hand side of the above display are bounded up to a constant by $h^{1/k}$, which implies that $\sup_{a_0 \in \s{A}_0}P_0|h\phi^*_{\infty, h,b,a_0}|^k$ is bounded up to a constant by $h$. Noting that $|\xi_\infty|$ is $P_0$-almost surely uniformly bounded by~\ref{cond:unif_doubly_robust}(a) and~\ref{cond:holder_smooth_theta}, we have
    \begin{align*}
        \sup_{a_0 \in \s{A}_0}\left[P_0 \left|h\Gamma_{0,h, b, a_0}\xi_\infty\right|^k\right]^{1/k}&\lesssim \sup_{a_0 \in \s{A}_0}\left[\int\left|h\Gamma_{0,h, b, a_0}(a)\right|^k\, dF_0(a)\right]^{1/k} \\ % \sup_{a_0 \in \s{A}_0}\left[\int\left|h\Gamma_{0,h, b, a_0}(a)\right|^k E_0\{|\xi_\infty(Y,A,W)|^k\mid A=a\} \, dF_0(a)\right]^{1/k}\\
        & \leq \sup_{a_0 \in \s{A}_0}\left\{\int\left|h e_1^T \b{D}_{0,h,a_0,1}^{-1} w_{h,a_0,1}(a) K_{h,a_0}(a)\right|^k \,dF_0(a)\right\}^{1/k} \\%\sup_{a_0 \in \s{A}_0}\left\{\int\left|h e_1^T \b{D}_{0,h,a_0}^{-1} w_{h,a_0,1}(a) K_{h,a_0}(a)\right|^k E_0\{|\xi_\infty(Y,A,W)|^k\mid A=a\}\,dF_0(a)\right\}^{1/k} \\
        &\qquad + \sup_{a_0 \in \s{A}_0}\left\{\int\left|h c_{0,h,a_0,2} \tau_n^2e_3^T \b{D}_{0,b,a_0,2}^{-1} w_{b,a_0,2}(a) K_{b,a_0}(a)\right|^k \,dF_0(a)\right\}^{1/k}\\%\sup_{a_0 \in \s{A}_0}\left\{\int\left|he_3^T c_2 \tau^2\b{D}_{0,b,a_0}^{-1} w_{b,a_0,3}(a) K_{b,a_0}(a)\right|^k E_0\{|\xi_\infty(Y,A,W)|^k \mid A=a\} \,dF_0(a)\right\}^{1/k}\\
        &\lesssim h^{1/k} \sup_{a_0 \in \s{A}_0} \left|e_1^T \b{D}_{0,h,a_0,1}^{-1} \b{1}\right| + h^{1/k} \tau_n^{-1/k} \sup_{a_0 \in \s{A}_0}\left| c_{0,h,a_0,2} \tau_n^2 e_3^T\b{D}_{0,b,a_0,2}^{-1} \b{1}\right|  %h^{1/'k}\|K\|_\infty \|f_0\|_\infty^{1/k} \sup_{a_0 \in \s{A}_0} \left|e_1^T \b{D}_{0,h,a_0}^{-1} \b{1}\right|  \sup_{a_0 \in \s{A}_h} \left\{ E_0[|\xi_\infty(Y,A,W)|^k\mid A=a_0]\right\}^{1/k} \\
     %   &\qquad+  h^{1/k}\|K\|_\infty \|f_0\|_\infty^{1/k} \tau^{-1/k} \sup_{a_0 \in \s{A}_0}\left|e_3^T c_2 \tau^2\b{D}_{0,b,a_0}^{-1} \b{1}\right|\sup_{a_0 \in \s{A}_h} \left\{ E_0[|\xi_\infty(Y,A,W)|^k\mid A=a_0] \right\}^{1/k}
    \end{align*}
    By Lemma~\ref{lm:D0_altform} and~\ref{cond:holder_smooth_theta}, $\sup_{a_0 \in \s{A}_0} |e_1^T \b{D}_{0,h,a_0,1}^{-1}|$ and $\sup_{a_0 \in \s{A}_0} |c_{0,h,a_0,2} \tau_n^2e_3^T \b{D}_{0,b,a_0,2}^{-1}|$ are both $\boundeddet(1)$. %Next, we can bound $E_0[|\xi_\infty|^k \mid A=a_0]$ by $E_0\{(|Y|^k + K_0^k)K_1^{-k} + 2K_0^k \mid A=a_0\}$ using the triangle inequality, where $K_0$, $K_1$ are defined in \ref{cond:uniform_entropy_nuisances}. 
    Therefore, the first term on the right hand side of~\eqref{eq:expr1} is bounded up to a constant by $h^{1/k}$ as $h \longrightarrow 0$. The second term on the right hand side of~\eqref{eq:expr1} can similarly be shown to be bounded up to a constant by $h^{1/k}$ using Lemma~\ref{lm:D0_altform} and the fact that $\mu_0$ is uniformly bounded since $|Y|$ is $P_0$-almost surely uniformly bounded by~\ref{cond:holder_smooth_theta}.
  %  \begin{align*}
  %  &\sup_{a_0 \in \s{A}_0}\left[P_0 \left|h\gamma_{0, h, b, a_0}\right|^k \right]^{1/k}\\
 %   &\qquad\leq\left\{ \sup_{a_0 \in \s{A}_0}|e_1^T \b{D}^{-1}_{0,h, s}|^k \int |h w_{h,s, 1}(a) K_{h, s}(a)w_{h,s, 1}^T(a)|^k \, dF_0(a)\, |\b{D}^{-1}_{0,h, s} P_0 \left( w_{h,s, 1} K_{h, s} \theta_0\right)|^k\right\}^{1/k}\\
   % &\qquad\qquad+\left\{\sup_{a_0 \in \s{A}_0} |e_3^T c_2 \tau^2 \b{D}_{0,b,a_0}^{-1}|^k \int |hw_{b,a_0,3}(a) K_{b,a_0}(a)  w_{b,a_0,3}^T(a)|^k \, dF_0(a)\, |\b{D}_{0,b,a_0}^{-1} P_0\left(  w_{b,a_0,3} K_{b,a_0} \theta_0 \right)|^k\right\}^{1/k}\\
   % &\qquad\leq h^{1/k} \left\{\sup_{a_0 \in \s{A}_0}|e_1^T \b{D}^{-1}_{0,h, s}|^k  \|f_0\|_\infty \int |[1,u]^T K(u)[1,u]|^k \, du\, |\b{D}^{-1}_{0,h, s} P_0 \left( w_{h,s, 1} K_{h, s} \theta_0\right)|^k\right\}^{1/k} \\
   % &\qquad\qquad+h^{1/k} \left\{\sup_{a_0 \in \s{A}_0}|e_3^T c_2 \tau^2 \b{D}_{0,b,a_0}^{-1}|^k |\tau|^{k-1}  \|f_0\|_\infty  \int |[1,v,v^2,v^3]^T K(v) [1,v,v^2,v^3]|^k \, dv\, |\b{D}_{0,b,a_0}^{-1} P_0\left(  w_{b,a_0,3} K_{b,a_0} \theta_0 \right)|^k\right\}^{1/k}.
   % \end{align*}
   
    Finally, the last term  of~\eqref{eq:expr1} can be bounded using Jensen's inequality as follows:
    \begin{align*}
        &\sup_{a_0 \in \s{A}_0}\int\left|h\int \Gamma_{0,h, b, a_0}(\bar{a})\left\{ \mu_\infty(\bar{a},w) - \smallint \mu_\infty(\bar{a}, \bar{w}) \, dQ_0(\bar{w}) \right\}\,dF_0(\bar{a})\right|^k \, dQ_0(w) 
        \\
        &\qquad \leq \sup_{a_0 \in \s{A}_0}\iint \left|h\Gamma_{0,h, b, a_0}(\bar{a})\left\{ \mu_\infty(\bar{a},w) - \int \mu_\infty(\bar{a}, \bar{w}) \, dQ_0(\bar{w}) \right\} \right|^k \, dF_0(\bar{a}) \, dQ_0(w) \\
        &\qquad\lesssim \sup_{a_0 \in \s{A}_0}\int \left|h\Gamma_{0,h, b, a_0}(\bar{a})\right|^k \, dF_0(\bar{a}),
    \end{align*}
    where the final inequality follows from the uniform bound on $\mu_\infty$ assumed by~\ref{cond:unif_doubly_robust}(a). The final integral is bounded up to a constant by $h$ as shown above.
\end{proof}

\begin{lemma}\label{lm:subgaussian}
Let $\{Z_{\infty, h,b}(a_0): a_0 \in \s{A}_0\}$ be the mean-zero Gaussian process  with covariance function
\[\Sigma_{\infty,h,b} : (u,v) \mapsto hP_0(\phi^*_{\infty, h, b, u }\phi^*_{\infty, h, b, v}) / [\sigma_{\infty,h,b}(u)\sigma_{\infty,h,b}(v)].\] 
If~\ref{cond:bounded_K}--\ref{cond:bandwidth},~\ref{cond:unif_doubly_robust}(a), and~\ref{cond:holder_smooth_theta} hold, then $\{Z_{\infty,h,b} : a_0 \in \s{A}_0\}$ is tight in $\ell^\infty(\s{A}_0)$, $E_0\left[ \sup_{a_0\in \s{A}_0} |Z_{\infty, h,b}(a_0)|\right] \leq  C\sqrt{\log(1/h)}$ for a constant $C$ not depending on $h$, and 
\[E\left[ \sup_{|u-v| < \delta} \left| Z_{\infty,h,b}(u) -Z_{\infty,h,b}(v)\right|\right] \leq C'h^{-1/2}([2\delta] \wedge h)^{1/2} \left[ \log\frac{|\s{A}_0|}{2([2\delta] \wedge h)}\right]^{1/2}\]
for all $[2\delta] \wedge h \leq |\s{A}_0| / 8$ and a constant $C'$ not depending on $h$.
\end{lemma}
\begin{proof}[\bfseries{Proof of Lemma~\ref{lm:subgaussian}}]
Let $\rho_{\infty, h,b}$ be the standard deviation semi-metric corresponding to $Z_{\infty,h,b}$; i.e.\ 
\[\rho_{\infty,h,b}^2(u, v) := E\left[ Z_{\infty,h,b}(u) - Z_{\infty,h,b}(v) \right]^2 = |\Sigma_{\infty,h,b}(u,u) - 2\Sigma_{\infty,h,b}(u,v)+\Sigma_{\infty,h,b}(v,v)| = 2 \left| 1 -\Sigma_{\infty,h,b}(u,v)\right|\]
because $\Sigma_{\infty,h,b}(u,u) = 1$ for any $u$ by definition. We first establish that $\rho_{\infty,h,b}$ is Lipschitz in $|u-v|$. By Lemma~\ref{lm:unif_bounded_variance}, there exist constants $c,C \in (0, \infty)$ such that for all $h$ small enough and all $a_0 \in \s{A}_0$, $c \leq \sigma^2_{\infty,h,b}(a_0) \leq C$.
% since
% \begin{align*}
%     \rho^2(u, v) &= |\Sigma_0(u,u) - 2\Sigma_0(u,v)+\Sigma_0(v,v)|\\
%     &= |\{\Sigma_0(u,u) - \Sigma_0(u,v)\}+\{\Sigma_0(u,v)-\Sigma_0(v,v)\}|.
% \end{align*}
We can then bound $|1 - \Sigma_{\infty,h,b}(u,v)|$ as follows:
% \begin{align*}
%     |1 - \Sigma_{\infty,h,b}(u,v)|  &= \left|1-P_0\left(\phi_{\infty, h,b,u}^{*2}\right)^{-1/2}P_0\left(\phi_{\infty, h,b,v}^{*2}\right)^{-1/2}P_0\phi^{*}_{\infty, h, b, u}\phi^{*}_{\infty, h, b, v}\right|\nonumber\\
%     &= \left| 1 - \sigma_{\infty,h,b}(u) \sigma_{\infty,h,b}(v)^{-1} - \sigma_{\infty,h,b}(u)^{-1}\sigma_{\infty,h,b}(v)^{-1}P_0\phi^{*}_{\infty, h, b, u}(\phi^{*}_{\infty, h, b, v} - \phi^{*}_{\infty, h, b, u})\right| \nonumber \\
%     &\leq\frac{\left| \sigma_{\infty,h,b}(u) - \sigma_{\infty,h,b}(v)\right|}{ \sigma_{\infty,h,b}(v)} + \frac{\left| P_0\phi^{*}_{\infty, h, b, u}(\phi^{*}_{\infty, h, b, v} - \phi^{*}_{\infty, h, b, u})\right|}{\sigma_{\infty,h,b}(u)\sigma_{\infty,h,b}(v)} \nonumber\\
%     &\leq \frac{h^{1/2}\left| \sigma_{\infty,h,b}(u) - \sigma_{\infty,h,b}(v)\right| + h^{1/2}\left[ P_0 (\phi^{*}_{\infty, h, b, u} - \phi^{*}_{\infty, h, b, v})^2\right]^{1/2}}{h^{1/2}\sigma_{\infty,h,b}(v)} \nonumber \\
%     &\leq c^{-1/2} h^{1/2}\left| \sigma_{\infty,h,b}(u) - \sigma_{\infty,h,b}(v)\right| + c^{-1/2}h^{1/2}\left[ P_0 (\phi^{*}_{\infty, h, b, u} - \phi^{*}_{\infty, h, b, v})^2\right]^{1/2}.
% \end{align*}
\begin{align*}
    |1 - \Sigma_{\infty,h,b}(u,v)|  &= \left|1-hP_0(\phi^*_{\infty, h, b, u }\phi^*_{\infty, h, b, v}) / [\sigma_{\infty,h,b}(u)\sigma_{\infty,h,b}(v)]\right|\nonumber\\
    &= \left| 1 - \frac{\sigma_{\infty,h,b}(u)}{\sigma_{\infty,h,b}(v)}  - \frac{hP_0\phi^{*}_{\infty, h, b, u}(\phi^{*}_{\infty, h, b, v} - \phi^{*}_{\infty, h, b, u})}{\sigma_{\infty,h,b}(u)\sigma_{\infty,h,b}(v)}\right| \nonumber \\
    &\leq\frac{\left| \sigma_{\infty,h,b}(u) - \sigma_{\infty,h,b}(v) \right|}{  \sigma_{\infty,h,b}(v)} + \frac{h\left| P_0\phi^{*}_{\infty, h, b, u}(\phi^{*}_{\infty, h, b, v} - \phi^{*}_{\infty, h, b, u})\right|}{ \sigma_{\infty,h,b}(u) \sigma_{\infty,h,b}(v)} \nonumber\\
    &\leq \frac{\left| \sigma_{\infty,h,b}(u) - \sigma_{\infty,h,b}(v)\right| + h^{1/2}\left[ P_0 (\phi^{*}_{\infty, h, b, u} - \phi^{*}_{\infty, h, b, v})^2\right]^{1/2}}{\sigma_{\infty,h,b}(v)} \nonumber \\
    &\leq c^{-1/2} \left| \sigma_{\infty,h,b}(u) - \sigma_{\infty,h,b}(v)\right| + c^{-1/2}h^{1/2}\left[ P_0 (\phi^{*}_{\infty, h, b, u} - \phi^{*}_{\infty, h, b, v})^2\right]^{1/2}.
\end{align*}
Since $\sigma_{\infty,h,b}(u)$ and $\sigma_{\infty,h,b}(v)$ are bounded away from zero for all $h$ small enough and $x \mapsto x^{1/2}$ is Lipschitz on $[\varepsilon, \infty)$ for any $\varepsilon > 0$, we also have
\begin{align*}
    \left| \sigma_{\infty,h,b}(u) - \sigma_{\infty,h,b}(v)\right| &= \left| \left[ hP_0 (\phi^{*}_{\infty, h, b, u})^2 \right]^{1/2} - \left[ hP_0 (\phi^{*}_{\infty, h, b, v})^2 \right]^{1/2}\right| \\
    &\lesssim h \left| P_0 (\phi^{*}_{\infty, h, b, u})^2 - P_0 (\phi^{*}_{\infty, h, b, v})^2 \right|  \\
    &= h \left| P_0 \left(\phi^{*}_{\infty, h, b, u} -  \phi^{*}_{\infty, h, b, v}\right) \phi^{*}_{\infty, h, b, u} + P_0 \left(\phi^{*}_{\infty, h, b, u} -  \phi^{*}_{\infty, h, b, v}\right) \phi^{*}_{\infty, h, b, v}\right|  \\
    &\leq  h^{1/2} \left[ P_0 (\phi^{*}_{\infty, h, b, u} - \phi^{*}_{\infty, h, b, v})^2 \{ \sigma_{\infty,h,b}(u)^2 + \sigma_{\infty,h,b}(v)^2\} \right]^{1/2} \\
    &\lesssim h^{1/2} \left[ P_0 (\phi^{*}_{\infty, h, b, u} - \phi^{*}_{\infty, h, b, v})^2 \right]^{1/2}.
\end{align*}
Hence, we can turn our attention to bounding $h^{1/2}[P_0 (\phi^{*}_{\infty, h, b, u} - \phi^{*}_{\infty, h, b, v})^2]^{1/2}$. We have
\begin{align*}
   h^{1/2}\left[P_0 \left(\phi^{*}_{\infty, h, b, u} - \phi^{*}_{\infty, h, b, v}\right)^2\right]^{1/2} &\leq h^{1/2}\left[P_0\{(\Gamma_{0,h,b,u}- \Gamma_{0,h,b,v})^2\xi_\infty^2\}\right]^{1/2}  + h^{1/2}\left[P_0(\gamma_{0,h,b,u}
   -\gamma_{0,h,b,v})^2\right]^{1/2}\\
   &\qquad + h^{1/2}\left[P_0\left(\int\left\{ \Gamma_{0,h,b,u} - \Gamma_{0,h,b,v}\right\}\left\{  \mu_\infty - \int \mu_\infty \, dQ_0\right\} \, dF_0\right)^2\right]^{1/2}.
\end{align*}
For the first term in the display, since $|\xi_\infty|$ is $P_0$-almost surely uniformly bounded, we have
\begin{align}
    &\left[P_0\{(\Gamma_{0,h,b,u}- \Gamma_{0,h,b,v})^2\xi_\infty^2\}\right]^{1/2} \nonumber\\
    &\qquad\lesssim \left[\int  \left\{e_1^T \b{D}_{0,h,u,1}^{-1} w_{h,u,1}(a) K_{h,u}(a)-e_1^T \b{D}_{0,h,v,1}^{-1} w_{h,v,1}(a) K_{h,v}(a)\right\}^2  f_0(a) \, da \right]^{1/2}\nonumber\\
    &\qquad\qquad +\tau_n^2  \left[ \int \left\{ c_{0,h,u,2}e_3^T \b{D}_{0,b,u,2}^{-1} w_{b,u,2}(a) K_{b,u}(a)-c_{0,h,v,2}e_3^T  \b{D}_{0,b,v,2}^{-1} w_{b,v,2}(a) K_{b,v}(a)\right\}^2 f_0(a) \,da\right]^{1/2} \label{eq:Gamma_u-Gamma_v}.
\end{align}
For the first term, we note that since the support of $K$ is contained in $[-1,1]$, if $|u-v| > 2h$, then either $K_{h,u}(a) = 0$ or $K_{h,v}(a) = 0$ for all $a$. Hence, if  $|u-v| > 2h$, then 
\begin{align*}
    &\left[\int  \left\{e_1^T \b{D}_{0,h,u,1}^{-1} w_{h,u,1}(a) K_{h,u}(a)-e_1^T \b{D}_{0,h,v,1}^{-1} w_{h,v,1}(a) K_{h,v}(a)\right\}^2  f_0(a) \, da \right]^{1/2} \\
    &\qquad = \left[\int  \left\{e_1^T \b{D}_{0,h,u,1}^{-1} w_{h,u,1}(a) K_{h,u}(a)\right\}^2f_0(a) \, da + \int \left\{e_1^T \b{D}_{0,h,v}^{-1} w_{h,v,1}(a) K_{h,v}(a)\right\}^2  f_0(a) \, da \right]^{1/2} \\
    &\qquad = h^{-1/2}\left[\int  \left\{e_1^T \b{D}_{0,h,u,1}^{-1} (1,t) K(t)\right\}^2f_0(u + th) \, dt + \int \left\{e_1^T \b{D}_{0,h,v,1}^{-1} (1,t)K(t)\right\}^2  f_0(v + th) \, dt \right]^{1/2},
\end{align*}
which is bounded up to a constant by $h^{-1/2}$ by the boundedness of $K$, $f_0$, and Lemma~\ref{lm:D0_altform}. If $|u-v| \leq 2h$, then we further decompose
\begin{align*}
    &\left[\int  \left\{e_1^T \b{D}_{0,h,u,1}^{-1} w_{h,u,1}(a) K_{h,u}(a)-e_1^T \b{D}_{0,h,v,1}^{-1} w_{h,v,1}(a) K_{h,v}(a)\right\}^2  f_0(a) \, da \right]^{1/2} \\
    &\qquad \leq \left[\int  \left\{e_1^T \b{D}_{0,h,u,1}^{-1} \left[w_{h,u,1}(a) - w_{h,v,1}(a)\right]\right\}^2 K_{h,u}(a)^2 f_0(a) \, da \right]^{1/2} \\
    &\qquad\qquad + \left[\int  \left\{e_1^T \left[\b{D}_{0,h,u,1}^{-1}-\b{D}_{0,h,v,1}^{-1} \right] w_{h,v,1}(a)\right\}^2 K_{h,u}(a)^2 f_0(a) \, da \right]^{1/2} \\
    &\qquad\qquad +  \left[\int  \left\{e_1^T \b{D}_{0,h,v,1}^{-1} w_{h,v,1}(a)\right\}^2 \left\{K_{h,u}(a) - K_{h,v}(a)\right\}^2 f_0(a) \, da \right]^{1/2}.
\end{align*}
For the first term, we have $w_{h,u,1}(a) - w_{h,v,1}(a) = (0, (v-u)/h)$, so $e_1^T \b{D}_{0,h,u,1}^{-1}\left[ w_{h,u,1}(a) - w_{h,v,1}(a)\right] = h^{-1} (v-u)  \b{D}_{0,h,u,1}^{-1}[1,2] =  (v-u)  \boundeddet(1)$ since $\b{D}_{0,h,u,1}^{-1} = f_0(u)^{-1} \b{S}_2^{-1} + \boundeddet(h)$ by Lemma~\ref{lm:D0_altform} and $\b{S}_2$ is a diagonal matrix. Hence, the first term is bounded up to a constant by 
\begin{align*}
   |u-v|\left[ \int \left\{ K_{h,u}(a) \right\}^2  f_0(a) \, da \right]^{1/2}  &= h^{-1/2}|u-v| \left[ \int \left\{ K(t) \right\}^2  f_0(u + th) \, dt \right]^{1/2},
\end{align*}
which is bounded up to a constant by $h^{-1/2}|u-v|$ for all $h$ small enough. 

For the second term, we can write $\b{D}_{0,h,u,1}^{-1}- \b{D}_{0,h,v,1}^{-1} = \b{D}_{0,h,v,1}^{-1}(\b{D}_{0,h,v,1}- \b{D}_{0,h,u,1})\b{D}_{0,h,v,1}^{-1}$, and we have by definition 
\begin{align*}
    \b{D}_{0,h,v,1}- \b{D}_{0,h,u,1}&= P_0 \left(w_{h,v,1}w_{h,v,1}^T  K_{h,v} -  w_{h,u,1} w_{h,u,1}^T K_{h,u}\right) \\
    &= P_0 \left(\left[w_{h,v,1}w_{h,v,1}^T -  w_{h,u,1} w_{h,u,1}^T \right] K_{h,v} +w_{h,u,1} w_{h,u,1}^T\left[K_{h,v} - K_{h,u}\right]\right).
\end{align*}
We then note that 
\[ w_{h,v,1}(a)w_{h,v,1}(a)^T -  w_{h,u,1}(a) w_{h,u,1}(a)^T = h^{-1}(v-u) \begin{pmatrix} 0 & 1 \\ 1 & [2a - u - v] / h \end{pmatrix}.\]
With the change of variables $t = (a-v)/h$, we then have
\begin{align*}
     \left|P_0 \left(\left[w_{h,v,1}w_{h,v,1}^T -  w_{h,u,1} w_{h,u,1}^T \right] K_{h,v} \right)\right| &= h^{-1}|u-v|  \left|\int\begin{pmatrix} 0 & 1 \\ 1 & 
    2t + (v-u) / h\end{pmatrix} K(t) f_0(v + th) \, dt \right|.
\end{align*}
Since we are considering the case $|u-v| \leq 2h$, the absolute value of this expression is bounded up to a constant by $h^{-1}|u-v|$. We also have 
\begin{align*}
    \left|P_0 \left(w_{h,u,1} w_{h,u,1}^T\left[K_{h,v} - K_{h,u}\right]\right)\right| &= \left| \int (1, t) (1,t)^T\left[ K(t) - K(t + [u-v] / h)\right] f_0(u + th) \, dt\right|,
\end{align*}
which is bounded up to a constant by  $h^{-1} |u-v|$ by the Lipschitz assumption on $K$. Hence, 
\begin{align*}
    &\left[\int  \left\{e_1^T \left[\b{D}_{0,h,u,1}^{-1}-\b{D}_{0,h,v,1}^{-1} \right] w_{h,v,1}(a)\right\}^2 K_{h,u}(a)^2 f_0(a) \, da \right]^{1/2} \\
    &\qquad \leq Ch^{-1}|u-v|\left[\int  \left\{e_1^T \b{D}_{0,h,v,1}^{-1} \b{1} \b{D}_{0,h,u,1}^{-1} w_{h,v,1}(a)\right\}^2 K_{h,u}(a)^2 f_0(a) \, da \right]^{1/2} \\
    &\qquad = Ch^{-3/2}|u-v|\left[\int  \left\{e_1^T \b{D}_{0,h,v,1}^{-1} \b{1} \b{D}_{0,h,u,1}^{-1} (1, t)\right\}^2  K(t + (v-u)/h) f_0(v + th) \, da \right]^{1/2},
\end{align*}
which is bounded up to a constant by $h^{-3/2}|u-v|$.

For the final term, since $K$ is assumed to be Lipschitz, we have
\begin{align*}
    &\left[ \int \left\{ e_1^T \b{D}_{0,h,v,1}^{-1} w_{h,v,1}(a) \left[K_{h,u}(a) - K_{h,v}(a)\right] \right\}^2 f_0(a) \, da \right]^{1/2} \\
    &\qquad\lesssim \left[ \int \left\{ e_1^T \b{D}_{0,h,v,1}^{-1} w_{h,v,1}(a)\right\}^2 h^{-2}\left[(a-u)/h -(a-v)/h \right]^2  f_0(a) \, da \right]^{1/2} \\
    &\qquad = h^{-3/2}|u-v|\left[ \int \left\{ e_1^T \b{D}_{0,h,v,1}^{-1} (1, t)\right\}^2 f_0(u + th) \, dt \right]^{1/2},
\end{align*}
which is bounded up to a constant by $h^{-3/2}|u-v|$.

Putting it together, we have that
\[ h^{1/2}\left[\int  \left\{e_1^T \b{D}_{0,h,u,1}^{-1} w_{h,u,1}(a) K_{h,u}(a)-e_1^T \b{D}_{0,h,v,1}^{-1} w_{h,v,1}(a) K_{h,v}(a)\right\}^2  f_0(a) \, da \right]^{1/2}\]
is bounded up to a constant by $h^{-3/2}|u-v|$ when $|u-v| \leq 2h$ and is bounded up to a constant when $|u-v| > 2h$. Analysis of the second term of~\eqref{eq:Gamma_u-Gamma_v} follows the same logic, and yields the same result. We can also show using the above techniques that  $h^{1/2}\left[P_0(\gamma_{0,h,b,u}-\gamma_{0,h,b,v})^2\right]^{1/2}$ satisfies the same bound. Finally, since $\mu_\infty$ is uniformly bounded, we have
\begin{align*}
    \left[P_0\left(\int\left\{ \Gamma_{0,h,b,u} - \Gamma_{0,h,b,v}\right\}\left\{  \mu_\infty - \int \mu_\infty \, dQ_0\right\} \, dF_0\right)^2\right]^{1/2} &\leq C \int \left| \Gamma_{0,h,b,u} -\Gamma_{0,h,b,v}\right| \, dF_0 \\
    &\leq C\left[ \int \left(\Gamma_{0,h,b,u} -\Gamma_{0,h,b,v}\right)^2 \, dF_0  \right]^{1/2},
\end{align*}
which is the same as the expression we bounded above.
 
We have now shown that there exist a constants $C_1$ and $C_2$ not depending on $h$, $u$, or $v$ such that $\rho_{\infty,h,b}(u,v) \leq C_1 h^{-1/2} |u-v|^{1/2}$ for $|u-v| \leq 2h$, and $\rho_{\infty,h,b}(u,v) \leq C_2$ for  $|u-v| > 2h$. Without loss of generality, we can take $C_2 = \sqrt{2} C_1$ so the bound is continuous. Suppose $\varepsilon \in (0,  C_1 / \sqrt{2}]$ and $|u-v| \leq 4C_1^{-2} h \varepsilon^2$. Then $|u-v| \leq 2h$, so 
\[ \rho_{\infty,h,b}(u,v) \leq C_1 h^{-1/2} |u-v|^{1/2} \leq 2\varepsilon.\]
Hence, for $\varepsilon \leq C_1/\sqrt{2}$, $N(\varepsilon, \s{A}_0, \rho_{\infty,h,b}) \leq C_1^2|\s{A}_0| / (4h\varepsilon^2)$. For $\varepsilon > C_1 / \sqrt{2}$ and $|u-v| \leq 2h$, we have $\rho_{\infty,h,b}(u,v) \leq C_1 h^{-1/2} |u-v|^{1/2} \leq \sqrt{2}C_1 < 2\varepsilon$. For $\varepsilon > C_1 / \sqrt{2}$ and $|u-v| > 2h$, we have $\rho_{\infty,h,b}(u,v) \leq \sqrt{2} C_1 < 2 \varepsilon$. Therefore, all $u,v \in \s{A}_0$ fit in a single $\rho_{\infty,h,b}$ ball of radius $\varepsilon$, so $N(\varepsilon, \s{A}_0, \rho_{\infty,h,b}) = 1$ for $\varepsilon > C_1 / \sqrt{2}$. 

This implies that $\s{A}_0$ can be covered by finitely many $\rho_{\infty,h,b}$ balls of radius $\varepsilon$ for every $\varepsilon > 0$. Hence, the semimetric space $(\s{A}_0, \rho_{\infty,h,b})$ is totally bounded, and hence separable (see, e.g.\ page 17 of \citealp{vandervaart1996}). Furthermore, since $Z_{\infty,h,b}$ is a Gaussian process, it is sub-Gaussian with respect to its intrinsic semimetric $\rho_{\infty,h,b}$. We thus conclude $\{Z_{\infty,h,b}(a_0) : a_0 \in \s{A}_0\}$ is a separable sub-Gaussian process with respect to $\rho_{\infty,h,b}$. We then have by Corollary 2.2.8 of \cite{vandervaart1996} that
\begin{align*}
     E\left[ \sup_{a_0 \in \s{A}_0} |Z_{\infty,h,b}(a_0)|\right] &\lesssim E\left[ |Z_{\infty,h,b}(a_1)|\right] + \int_0^\infty \left\{ \log N(\varepsilon, \s{A}_0, \rho_{\infty,h,b})\right\}^{1/2} \, d\varepsilon \\
     &\leq 1 + \int_0^{C_1/\sqrt{2}} \left\{ \log \left( C_1^2|\s{A}_0| / [4h\varepsilon^2]\right) \right\}^{1/2} \, d\varepsilon \\
     &\leq 1 + \frac{C_1 \left[ 1 + \log \left(|\s{A}_0| / [2h]\right) \right]}{2 \left[ \log \left(|\s{A}_0| / [2h]\right)\right]^{1/2}},
\end{align*}
which is bounded up to a constant by $\left[\log h^{-1} \right]^{1/2}$ for all $h$ small enough. For every $\delta \in (0, C_1 / \sqrt{2})$, we also have by Corollary 2.2.8 of \cite{vandervaart1996} that
\begin{align*}
     E\left[ \sup_{\rho_{\infty,h,b}(s,t) < \delta} \left|Z_{\infty, h,b}(a_1) - Z_{\infty, h,b}(t)\right|\right] &\lesssim \int_0^\delta\left\{ \log N(\varepsilon, \s{A}_0, \rho_{\infty,h,b})\right\}^{1/2} \, d\varepsilon.
\end{align*}
Since the integral is finite over $[0,\infty)$ (as shown above), the integral over $[0,\delta]$ goes to zero as $\delta$ goes to zero. Hence, the sample paths of $Z_{\infty,h,b}$ are almost surely uniformly $\rho_{\infty,h,b}$-continuous. Since $(\s{A}_0, \rho_{\infty,h,b})$ is totally bounded, this implies that $Z_{\infty,h,b}$ is tight in $\ell^\infty(\s{A}_0)$.   

Finally, since $\rho_{\infty,h,b}(u,v) \leq C_1 h^{-1/2} (\min\{|u-v|, 2h\})^{1/2}$, we have again by Corollary 2.2.8 of \cite{vandervaart1996} that
\begin{align*}
E\left[ \sup_{|u-v| < \delta} \left| Z_{\infty,h,b}(u) -Z_{\infty,h,b}(v)\right|\right] &\leq E\left[ \sup_{\rho_{\infty,h,b}(u, v) < C_1 h^{-1/2} (\delta \wedge [2h])^{1/2}} \left| Z_{\infty,h,b}(u) -Z_{\infty,h,b}(v)\right|\right] \\
&\lesssim \int_0^{C_1 h^{-1/2} (\delta \wedge [2h])^{1/2}} \left\{ \log N(\varepsilon, \s{A}_0, \rho_{\infty,h,b})\right\}^{1/2} \, d\varepsilon \\
    &\leq \int_0^{C_1 h^{-1/2} (\delta \wedge [h/2])^{1/2}}  \left\{ \log \left( C_1^2|\s{A}_0| / [4h\varepsilon^2]\right) \right\}^{1/2} \, d\varepsilon \\
     &\lesssim h^{-1/2}([2\delta] \wedge h)^{1/2} \left[ \log\frac{|\s{A}_0|}{2([2\delta] \wedge h)}\right]^{1/2}
\end{align*}
as long as $[2\delta] \wedge h \leq |\s{A}_0| / 8$.
\end{proof}

\begin{lemma}\label{lm:gaussian_approx}
If~\ref{cond:bounded_K}--\ref{cond:bandwidth},~\ref{cond:unif_doubly_robust}(a), and~\ref{cond:holder_smooth_theta} hold, %, and $\sup_{a_0 \in \s{A}_{\varepsilon_2}} E_0[|Y|^4 \mid A = s] < \infty$ and $E_0[|Y|^q] < \infty$ for some $q \geq 4$.
then
\begin{align*}
   %\sup_{a_0 \in \s{A}_0}\left|\mathbb{G}_n \frac{\phi^*_{\infty, h, b, a_0}}{\sigma_{\infty,h,b}(a_0)}\right| - \sup_{a_0 \in \s{A}_0}\left|Z_{\infty,h,b}(a_0)\right| =\bounded\left(n^{-1/2+1/q}h^{-1/2}\{\log n\}^{3/2} + \{nh\}^{-1/4}\{\log n\}^{5/4} + \{nh\}^{-1/6}\log n\right). 
   \sup_{a_0 \in \s{A}_0}\left|\mathbb{G}_n \frac{h^{1/2}\phi^*_{\infty, h, b, a_0}}{\sigma_{\infty,h,b}(a_0)}\right| - \sup_{a_0 \in \s{A}_0}\left|Z_{\infty,h,b}(a_0)\right| =\bounded\left(\{ n h\} ^{-1/2}\{\log n\}^{3/2} + \{nh\}^{-1/4}\{\log n\}^{5/4} + \{nh\}^{-1/6}\log n\right). 
\end{align*}
Hence, if in addition $nh^p \longrightarrow \infty$ for some $p > 1$, then %$h \geq n^{-p}$ for some $p \in (0, 1-2/q)$, then
\begin{align*}
    \sup_{t \in \d{R}} \left|P \left(\sup_{a_0 \in \s{A}_0} \left| \d{G}_n \frac{h^{1/2}\phi_{\infty, h, b,a_0}^*}{\sigma_{\infty,h,b}(a_0)} \right|  \leq t \right)-P \left(\sup_{a_0 \in \s{A}_0}\left|Z_{\infty,h,b}(a_0)\right| \leq t \right)\right| = \fasterthandet(1).
\end{align*}
\end{lemma}
\begin{proof}[\bfseries{Proof of Lemma~\ref{lm:gaussian_approx}}]
These results are an application of Corollary~2.2 and Lemma~2.4 of \cite{Chernozhukov_2014}. For each $a_0 \in \s{A}_0$ and $h,b > 0$, we consider the following function:
\[\eta_{h,b,a_0} := (y,a,w) \mapsto h^{1/2}\sigma^{-1}_{\infty,h,b}(a_0)\phi_{\infty, h, b,a_0}^*(y,a,w).\]
We then define the class of functions $\s{H}_{h,b} := \{\eta_{h,b,a_0} : a_0 \in \s{A}_0\}$. In the notation of \cite{Chernozhukov_2014}, we have $\s{F} = \s{H}_{h,b} \cup -\s{H}_{h,b}$, 
\begin{align*}
    Z &= \sup_{f \in \s{F}} \d{G}_n f = \sup_{\eta \in \s{H}_{h,b}} \left|\d{G}_n \eta\right| = \sup_{a_0 \in \s{A}_0} \left|\d{G}_n \frac{h^{1/2}\phi^*_{\infty, h, b, a_0}}{\sigma_{\infty,h,b}(a_0)}\right|,
\end{align*}
and $\tilde{Z} = \sup_{a_0 \in \s{A}_0} \left|Z_{\infty,h,b}(a_0)\right|$, where by definition, $Z_{\infty,h,b}(a_0)$ is a mean-zero Gaussian process on $\s{A}_0$ with covariance function $(u,v) \mapsto P_0(\eta_{h,b,u}\eta_{h,b,v})$.

We now verify the conditions of Corollary~2.2 of \cite{Chernozhukov_2014}. First, $\s{H}_{h,b}$ is pointwise measurable because $K$ is uniformly continuous by~\ref{cond:bounded_K}, and by Corollary~\ref{cor:phi_infty_VC}, $\s{H}_{h,b}$ is VC type. For each $a_0 \in \s{A}_0$ and $h,b > 0$, we have $P_0 \eta_{h,b,a_0}^2 = 1$  by definition of $\sigma_{\infty, h,b}(a_0)$. By Lemmas~\ref{lm:unif_bounded_variance} and~\ref{lm:phi_finite_moments}, we have  
\begin{align*}
    \sup_{\eta \in \s{H}_{h,b}} P_0 |\eta|^3 =\sup_{a_0 \in \s{A}_0} \frac{P_0 \left|h^{1/2}\phi_{\infty, h, b,a_0}^*\right|^3}{\sigma_{\infty,h,b}(a_0)^3} \lesssim h^{-1/2} 
\end{align*}
Hence, in the notation of  \cite{Chernozhukov_2014}, we have $\sigma^2 = 1$ and $b = h^{-1/2}$ up to a constant not depending on $h$. We next establish that $\s{H}_{h,b}$ is uniformly bounded up to a constant by $h^{-1/2}$. We let $\b{1}$  be a vector of 1's of the appropriate dimension. Using the boundedness and bounded support of $K$ as well as the uniform boundedness of $\mu_\infty$,  for each $a_0 \in \s{A}_0$, we have that
\begin{align*}
    \left|h\phi^*_{\infty, h,b,a_0}(y,a,w)\right| %&= h\left|\Gamma_{0,h, b, a_0}(a)\xi_\infty(y,a,w) + \gamma_{0, h, b, a_0}(a)+\int \Gamma_{0,h, b, a_0}(\bar{a})\left\{ \mu_\infty(\bar{a},w) - \int \mu_\infty(\bar{a}, \bar{w}) \, dQ_0(\bar{w}) \right\}\, dF_0(\bar{a})\right|\nonumber\\
    &\leq\left|e_1^T \b{D}_{0, h, a_0,1}^{-1} \b{1} + c_{0,h,a_0,2}(h/b)^3 e_3^T \b{D}_{0, b,a_0,2}^{-1} \b{1} \right|  \left\{ \left|\xi_\infty(y,a,w)\right| + 2K_0\right\}\|K\|_\infty\\
    &\qquad+\left|e_1^T\b{D}^{-1}_{0,h, a_0,1}  \b{1} \b{1}^T \b{D}^{-1}_{0,h, a_0,1} \right|\left| P_0 \left( w_{h,a_0,1} K_{h,a_0} \theta_0 \right)\right|\|K\|_\infty \\
     &\qquad+  \left|c_{0,h,a_0,2} (h/b)^3 e_3^T\b{D}_{0,b,a_0,2}^{-1} \b{1} \b{1}^T \b{D}_{0,b,a_0,2}^{-1} \right| \left| P_0 \left( w_{b,a_0,2} K_{b,a_0} \theta_0 \right) \right| \|K\|_\infty \\
     &\qquad + (h/b)^2 \left|e_1^T \b{D}_{0, h, a_0,1}^{-1} \right|\left| \b{1} + \b{1} \b{1}^T \b{D}_{0, h, a_0,1}^{-1} P_0( \tilde{w}_{h,a_0,1} K_{h,a_0})\right| \|K\|_\infty \\
     &\qquad\qquad \times \left|e_3^T\b{D}_{0,b,a_0,2}^{-1} P_0 \left( w_{b,a_0,2} K_{b,a_0} \theta_0 \right) \right| 
\end{align*}
By Lemma~\ref{lm:D0_altform}, the elements of $\b{D}_{0, h,a_0,1}^{-1}$ and $\b{D}_{0, b,a_0,2}^{-1}$ are uniformly bounded over $a_0 \in \s{A}_0$ and for all $h$ small enough. By the uniform boundedness of $\theta_0$ and $f_0$ in an enlargement of $\s{A}_0$, we can also show using a change of variables that $\left| P_0 \left( w_{h,a_0,1} K_{h,a_0} \theta_0 \right)\right|$ and  $\left| P_0 \left( w_{b,a_0,2} K_{b,a_0} \theta_0 \right) \right|$ are uniformly bounded for all $a_0 \in \s{A}_0$ and $h$ small enough. Hence, there are finite positive constants $C_1$ and $C_2$ not depending on $h$, $a_0$, or $(y,a,w)$ such that $\left|h\phi^*_{\infty, h,b,a_0}(y,a,w)\right| \leq C_1 + C_2 |\xi_\infty(y,a,w)|$ for all $(y,a,w)$, $a_0 \in \s{A}_0$ and $h$ small enough. Therefore, by Lemma~\ref{lm:unif_bounded_variance}, an envelope function for $\s{H}_{h,b}$ is given by $h^{-1/2} (C_1' + C_2' |\xi_\infty|)$ for finite positive constants $C_1'$ and $C_2'$. By~\ref{cond:unif_doubly_robust}(a) and~\ref{cond:holder_smooth_theta}, $|\xi_\infty|$ is uniformly bounded. Thus, $\s{H}_{h,b}$ is uniformly bounded up to a constant by $b = h^{-1/2}$, so the moment and envelope conditions of Corollary~2.2 of \cite{Chernozhukov_2014} hold.
%By the triangle inequality, $P_0|\xi_\infty|^q$ is bounded by $E_0[(|Y|^q + K_0^q)K_1^{-q} + 2K_0^q]$. Thus, the envelope function of $\s{H}_{h,b}$ has absolute $q$th moment bounded by a constant since $E_0[|Y|^q]$ is finite, so the envelope condition of Corollary~2.2 of \cite{Chernozhukov_2014} holds.

We have now checked all the conditions of Corollary~2.2 of \cite{Chernozhukov_2014}, so with $\gamma_n = (\log n)^{-1}$, it follows that for all $h$ small enough and a constant $C$ not depending on $h$,
\begin{align*}
    &P\left(\left|\sup_{a_0 \in \s{A}_0}\left|\mathbb{G}_n h^{1/2}\frac{\phi^*_{\infty, h, b, a_0}}{\sigma_{\infty,h,b}(a_0)} \right|- \sup_{a_0 \in \s{A}_0}\left|{Z}_{\infty, h,b}(s)\right|\right| > \frac{Ch^{-1/2}\{\log n\}^{3/2}}{n^{1/2}} + \frac{Ch^{-1/4}\{\log n\}^{5/4}}{n^{1/4}} + \frac{Ch^{-1/6}\log n}{n^{1/6}} \right) \\
    &\qquad \lesssim \frac{1}{\log n} + \frac{\log n}{n}
\end{align*}
We conclude that
\begin{align*}
    \left|\sup_{a_0 \in \s{A}_0}\left|\mathbb{G}_n \frac{\phi^*_{\infty, h, b, a_0}}{\sigma_{\infty,h,b}(a_0)}\right| - \sup_{a_0 \in \s{A}_0}\left|Z_{\infty,h,b}(a_0)\right|\right| =\bounded\left(r_n\right)
\end{align*}
for
\[r_n = \{nh\}^{-1/2}\{\log n\}^{3/2} + \{nh\}^{-1/4}\{\log n\}^{5/4} + \{nh\}^{-1/6}\log n.\]

For the second statement, we use Lemma~2.4 of \cite{Chernozhukov_2014}. In their notation, we have $\s{F}_n = (h^{-1/2}\s{H}_{h,b}) \cup (-h^{-1/2}\s{H}_{h,b})$. We have already established that this class is pointwise measurable, that its envelope is square integrable, and that its variance function is uniformly bounded above and below for all $n$. Lemma~\ref{lm:subgaussian} implies that $\s{F}_n$ is $P_0$-pre-Gaussian. For the final condition of Lemma~2.4 of \cite{Chernozhukov_2014}, by Lemma~\ref{lm:subgaussian}, we have $E_0 \sup_{a_0\in \s{A}_0} |Z_{\infty, h,b}(a_0)| = \boundeddet\left(\{\log(1/h)\}^{1/2}\right)$. The assumption that $nh^p \longrightarrow \infty$ for some $p > 1$ then implies that $r_n \{ \log (1/h)\}^{1/2}  =\fasterthandet(1)$, which verifies the last condition of Lemma~2.4 of \cite{Chernozhukov_2014}.
%Now since $p < 1$, % $p < 1 - 2/q$, $p/2 + 1/q -1/2 < 0$, and together with the assumption that $h \geq n^{-p}$ for $p \in (0, 1-2/q)$, there exist $p_1, p_2, p_3 > 0$ such that $n^{-1/2 + 1/q}h^{-1/2} \leq n^{-p_1}$, $\{nh\}^{-1/4} \leq n^{-p_2}$, and $\{nh\}^{-1/6} \leq n^{-p_3}$. These polynomial rates in $n$ dominate the poly-log factors in $n$ and $1/h$, so $r_n \{ \log (1/h)\}^{1/2}  =\fasterthandet(1)$, which verifies the last condition of Lemma~2.4 of \cite{Chernozhukov_2014}.

\end{proof}

\begin{lemma}\label{lemma:anticoncentration}
If the conditions of Lemma~\ref{lm:subgaussian} hold, then for any $\varepsilon > 0$,
\[\sup_{t \in \d{R}}P \left(\left| \sup_{a_0\in\s{A}_0}\left| Z_{\infty, h, b}(a_0)\right| - t \right|  \leq \varepsilon \left[\log h^{-1}\right]^{-1/2}\right)  \leq C\varepsilon  + \fasterthandet(1)\]
as $h \longrightarrow 0$ for $C$ not depending on $h$ or $\varepsilon$.
\end{lemma}
\begin{proof}[\bfseries{Proof of Lemma~\ref{lemma:anticoncentration}}]
We use Lemma~A.1 from the supplementary material of \cite{Chernozhukov_2014}. We define the class of functions $\s{H}_{h,b} := \{ \phi_{\infty, h,b,s}^* / \sigma_{\infty, h,b,}(s) : s \in \s{S}\}$, and in the notation of Lemma~A.1 of \cite{Chernozhukov_2014}, we set $\s{F} = \s{H}_{h,b} \cup -\s{H}_{h,b}$. By Lemma~\ref{lm:subgaussian}, a tight Gaussian process in $\ell^\infty(\s{F})$ exists, so $\s{F}$ is $P_0$-pre-Gaussian, and by definition, $\n{Var}_0(f) = 1$ for all $f \in \s{F}$. Hence, the conditions of Lemma~A.1 of \cite{Chernozhukov_2014} are satisfied, and we have
\begin{align*}
    &\sup_{t \in \d{R}}P \left(\left| \sup_{a_0\in\s{A}_0}\left| Z_{\infty, h, b}(a_0)\right| - t \right|  \leq \varepsilon \left[\log h^{-1}\right]^{-1/2}\right) \\
    &\qquad\leq C \varepsilon \left[\log h^{-1}\right]^{-1/2} \left[ E_0\left\{ \sup_{a_0\in\s{A}_0} \left|Z_{\infty, h, b}(a_0)\right| \right\} + \left\{\log\left( \left[\log h^{-1}\right]^{1/2} / \varepsilon \right)\right\}^{1/2}  \right]
\end{align*}
for $C$ not depending on $h$ or $\varepsilon$. By Lemma~\ref{lm:subgaussian}, $E_0 \left\{\sup_{a_0\in\s{A}_0}|Z_{\infty,h,b
o}(s)|\right\} = \boundeddet(\{\log h^{-1}\}^{1/2})$. Furthermore, \[\varepsilon \left[\log h^{-1}\right]^{-1/2} \left\{\log\left( \left[\log h^{-1}\right]^{1/2} / \varepsilon \right)\right\}^{1/2} = \fasterthandet(1)\]
as $h \longrightarrow 0$ for any $\varepsilon > 0$.  The result follows.
\end{proof}

\begin{lemma}\label{lemma:sup_empirical_process}
If~\ref{cond:bounded_K}--\ref{cond:uniform_entropy_nuisances} and~\ref{cond:unif_nuisance_rate}--\ref{cond:holder_smooth_theta} hold, then 
\[\sup_{t \in \d{R}}\left|  P_0 \left(\sup_{a_0 \in \s{A}_0}(nh)^{1/2} \left| \frac{\theta_{n,h,b}(a_0) - \theta_0(a_0)}{\sigma_{n,h,b}(a_0)} \right| \leq t\right) -  P_0 \left(\sup_{a_0 \in \s{A}_0}\left| Z_{\infty, h, b}(a_0)\right| \leq t\right) \right| = \fasterthandet(1).\]
\end{lemma}
\begin{proof}[\bfseries{Proof of Lemma~\ref{lemma:sup_empirical_process}}]
Since~\ref{cond:unif_doubly_robust} implies that~\ref{cond:doubly_robust} holds for all $a_0 \in \s{A}_0$, by Lemma~\ref{lm:first-order-decomposition},
\begin{align}
    P \left(\sup_{a_0 \in \s{A}_0}(nh)^{1/2} \left| \frac{\theta_{n,h,b}(a_0) - \theta_0(a_0)}{\sigma_{n,h,b}(a_0)} \right| \leq t\right) &= P \left(\sup_{a_0 \in \s{A}_0} \left| G_n(a_0) + R_{n}(a_0) \right| \leq t\right), \nonumber
\end{align}
where we define 
\[R_{n}(a_0) := (nh)^{1/2}\sum_{j=1}^6 R_{n,h,b,a_0,j} / \sigma_{n,h,b}(a_0) +  \frac{\sigma_{\infty,h,b}(a_0) -\sigma_{n,h,b}(a_0)}{\sigma_{n,h,b}(a_0)} \d{G}_n \frac{h^{1/2}\phi_{\infty, h, b,a_0}^*}{\sigma_{\infty,h,b}(a_0)}\]
and $G_n(a_0) := \d{G}_n \frac{h^{1/2}\phi_{\infty, h, b,a_0}^*}{\sigma_{\infty,h,b}(a_0)}$. By the triangle inequality, we have
\begin{align*}
    P \left(\sup_{\s{A}_0} \left| G_n\right| \leq t - \sup_{\s{A}_0} |R_{n}| \right) &\leq  P \left(\sup_{\s{A}_0} \left| G_n + R_{n}\right| \leq t\right)\leq P \left(\sup_{\s{A}_0} \left|G_n\right| \leq t + \sup_{\s{A}_0} |R_{n}| \right)
\end{align*}
for each $t \in \d{R}$. We then have
\begin{align*}
    &P \left(\sup_{\s{A}_0} \left| G_n + R_{n} \right| \leq t\right) - P \left(\sup_{\s{A}_0}\left| Z_{\infty, h, b}\right| \leq t\right) \\
    &\qquad\leq P \left(\sup_{\s{A}_0} \left| G_n\right| \leq t + \sup_{\s{A}_0} |R_{n}| \right) - P \left(\sup_{\s{A}_0}\left| Z_{\infty, h, b}\right| \leq t + \sup_{\s{A}_0} |R_{n}|\right) \\
    &\qquad\qquad + P \left(\sup_{\s{A}_0}\left| Z_{\infty, h, b}\right| \leq t + \sup_{\s{A}_0} |R_{n}|\right) - P \left(\sup_{\s{A}_0}\left| Z_{\infty, h, b}\right| \leq t\right) \\
    &\qquad\leq \left| P \left(\sup_{\s{A}_0} \left|G_n\right| \leq t + \sup_{\s{A}_0} |R_{n}| \right) - P \left(\sup_{\s{A}_0}\left| Z_{\infty, h, b}\right| \leq t + \sup_{\s{A}_0} |R_{n}|\right) \right| + P \left(t < \sup_{\s{A}_0}\left| Z_{\infty, h, b}\right| \leq t + \sup_{\s{A}_0} |R_{n}|\right) \\
    &\qquad \leq \sup_{t \in \d{R}} \left| P \left(\sup_{\s{A}_0} \left| G_n\right| \leq t  \right) - P \left(\sup_{\s{A}_0}\left| Z_{\infty, h, b}\right| \leq t \right) \right| + P \left(\left| \sup_{\s{A}_0}\left| Z_{\infty, h, b}\right| - t \right|  \leq \sup_{\s{A}_0} |R_{n}|\right).
\end{align*}
Similarly, 
\begin{align*}
    &P \left(\sup_{\s{A}_0} \left|G_n + R_{n} \right| \leq t\right) - P \left(\sup_{\s{A}_0}\left| Z_{\infty, h, b}\right| \leq t\right) \\
    &\qquad \geq P \left(\sup_{\s{A}_0} \left| G_n\right| \leq t - \sup_{\s{A}_0} |R_{n}| \right) - P \left(\sup_{\s{A}_0}\left| Z_{\infty, h, b}\right| \leq t - \sup_{\s{A}_0} |R_{n}|\right) \\
    &\qquad\qquad + P \left(\sup_{\s{A}_0}\left| Z_{\infty, h, b}\right| \leq t - \sup_{\s{A}_0} |R_{n}|\right) - P \left(\sup_{\s{A}_0}\left| Z_{\infty, h, b}\right| \leq t\right) \\
    &\qquad \geq -\left| P \left(\sup_{\s{A}_0} \left| G_n\right| \leq t - \sup_{\s{A}_0} |R_{n}| \right) - P \left(\sup_{\s{A}_0}\left| Z_{\infty, h, b}\right| \leq t - \sup_{\s{A}_0} |R_{n}|\right) \right|  - P \left(t - \sup_{\s{A}_0} |R_{n}|< \sup_{\s{A}_0}\left| Z_{\infty, h, b}\right| \leq t  \right) \\
    &\qquad \geq -\sup_{t \in \d{R}} \left| P \left(\sup_{\s{A}_0} \left|G_n\right| \leq t  \right) - P \left(\sup_{\s{A}_0}\left| Z_{\infty, h, b}\right| \leq t \right) \right| - P \left(\left| \sup_{\s{A}_0}\left| Z_{\infty, h, b}\right| - t \right|  \leq \sup_{\s{A}_0} |R_{n}|\right).
\end{align*}
Hence,
\begin{align*}
     &\sup_{t \in \d{R}}\left|  P \left(\sup_{a_0\in\s{A}_0}(nh)^{1/2} \left| \frac{\theta_{n,h,b}(s) - \theta_0(s)}{\sigma_{n,h,b}(a_0)} \right| \leq t\right) -  P \left(\sup_{a_0\in\s{A}_0}\left| Z_{\infty, h, b}(a_0)\right| \leq t\right) \right| \\
     &\qquad\leq \sup_{t \in \d{R}} \left| P \left(\sup_{a_0\in\s{A}_0} \left| \d{G}_n \frac{h^{1/2}\phi_{\infty, h, b,a_0}^*}{\sigma_{\infty,h,b}(a_0)}\right| \leq t  \right) - P \left(\sup_{a_0\in\s{A}_0}\left| Z_{\infty, h, b}(a_0)\right| \leq t \right) \right|\\
     &\qquad \qquad + \sup_{t \in \d{R}}P \left(\left| \sup_{a_0\in\s{A}_0}\left| Z_{\infty, h, b}(a_0)\right| - t \right|  \leq \sup_{a_0\in\s{A}_0} |R_{n}(a_0)|\right)
\end{align*}
The first term on the right hand side is $\fasterthandet(1)$ by Lemma~\ref{lm:gaussian_approx}. For the second term, for any $\varepsilon > 0$, we can write
\begin{align}
     \sup_{t \in \d{R}}P \left(\left| \sup_{a_0\in\s{A}_0}\left| Z_{\infty, h, b}(a_0)\right| - t \right|  \leq \sup_{a_0\in\s{A}_0} |R_{n}(a_0)|\right) &\leq  \sup_{t \in \d{R}}P \left(\left| \sup_{a_0\in\s{A}_0}\left| Z_{\infty, h, b}(a_0)\right| - t \right|  \leq \varepsilon \left[\log h^{-1}\right]^{-1/2}\right) \nonumber\\
     &\qquad + P\left( \sup_{a_0\in\s{A}_0} |R_{n}(a_0)| > \varepsilon \left[\log h^{-1}\right]^{-1/2}\right). \label{eq:anticon}
\end{align}
By Lemma~\ref{lemma:anticoncentration}, the first term is bounded by $C\varepsilon + \fasterthandet(1)$ for $C$ not depending on $\varepsilon$ or $h$.  For the second term, since $nh \longrightarrow \infty$, we have
\begin{align*}
\left(\log h^{-1}\right)^{1/2}\sup_{a_0\in\s{A}_0} |R_{n,a_0}| 
 &\lesssim \left[\sup_{a_0 \in \s{A}_0}  \sigma_{n,h,b}(a_0)^{-1}\right] \left[ \sum_{j=1}^6 (nh \log n)^{1/2}\sup_{a_0 \in \s{A}_0} |R_{n,h,b,a_0,j}|\right. \\
 &\qquad +  \left. (\log n)^{1/2} \sup_{a_0 \in \s{A}_0} |\sigma_{n,h,b}(a_0) -\sigma_{\infty,h,b}(a_0)| \sup_{a_0 \in \s{A}_0} \left|\d{G}_n \frac{h^{1/2}\phi_{\infty, h, b,a_0}^*}{\sigma_{\infty,h,b}(a_0)}\right| \right].
\end{align*}
By Lemmas~\ref{lm:unif_bounded_variance} and~\ref{lemma:covar}, $\sup_{a_0\in\s{A}_0}\sigma_{\infty,h,b}(a_0)^{-1} = \bounded(1)$. By Lemma~\ref{lm:R1}, $\sup_{a_0 \in \s{A}_0} |R_{n,h,b,a_0,1}| = \boundeddet(h^{2 + \delta_4})$ for some $\delta_4 > 0$. Since $nh^5 = \boundeddet(1)$, we then have $\sup_{a_0 \in \s{A}_0} |R_{n,h,b,a_0,1}| = \fasterthandet(\{nh\log n\}^{-1/2})$. By Corollary~\ref{cor:supR2R3}, $\sup_{a_0 \in \s{A}_0} |R_{n,h,b,a_0,2}|$ and $\sup_{a_0 \in \s{A}_0} |R_{n,h,b,a_0,2}|$ are $\fasterthan( \{nh \log n\}^{-1/2})$ as long as $nh^p = \fasterthandet(1)$ for some $p > 0$, which holds for $p = 6$ since $nh^5 = \boundeddet(1)$ and $h = \fasterthandet(1)$. By Lemma~\ref{lm:R4}, $\sup_{a_0 \in \s{A}_0} |R_{n,h,b,a_0,4}| = \fasterthan(\{nh \log n\}^{-1/2})$. By Lemma~\ref{lemma:R5}, $\sup_{a_0 \in \s{A}_0} |R_{n,h,b,a_0,5}| = \bounded(\{nh\}^{-1}) = \fasterthan(\{nh\log n\}^{-1/2})$ since $nh / \log h^{-1}$ and $nh / \log n$ both go to $\infty$ by~\ref{cond:unif_nuisance_rate}. Finally, by Lemma~\ref{lm:R6}, $\sup_{a_0 \in \s{A}_0} |R_{n,h,b,a_0,6}| = \bounded(\{nh / \log h^{-1}\}^{-1}) = \fasterthan(\{nh\log n\}^{-1/2})$ by the same logic. Hence, $\sum_{j=1}^6 \sup_{a_0 \in \s{A}_0} |R_{n,h,b,a_0,j}| = \fasterthan(\{nh \log n\}^{1/2})$.

By Lemmas~\ref{lm:subgaussian} and~\ref{lm:gaussian_approx}, 
\[ \sup_{a_0 \in \s{A}_0} \left|\d{G}_n \frac{h^{1/2}\phi_{\infty, h, b,a_0}^*}{\sigma_{\infty,h,b}(a_0)}\right| = \boundeddet( \{\log n\}^{1/2}) + \fasterthan(1),\]
and by Lemma~\ref{lemma:covar}, $\sup_{a_0 \in \s{A}_0} |\sigma_{n,h,b}(a_0) -\sigma_{\infty,h,b}(a_0)| = \fasterthan(\{\log n\}^{-1})$.

We have now shown that $\left(\log h^{-1}\right)^{1/2}\sup_{a_0\in\s{A}_0} |R_{n,a_0}| = \fasterthan(1)$, which implies that the second term on the right hand side of~\eqref{eq:anticon} goes to zero for any $\varepsilon > 0$. Since $\varepsilon$ was arbitrary, this implies that 
\[  \sup_{t \in \d{R}}P \left(\left| \sup_{a_0\in\s{A}_0}\left| Z_{\infty, h, b}(a_0)\right| - t \right|  \leq \sup_{a_0\in\s{A}_0} |R_{n}(a_0)|\right) = \fasterthandet(1),\]
which concludes the proof.
\end{proof}

\begin{lemma}\label{lemma:finite_approx}
If the conditions of Lemma~\ref{lm:subgaussian} hold and $\omega_n = \fasterthandet\left(h^p\right)$ for some $p > 1$, then 
\[\sup_{t \in \d{R}}\left| P_0 \left(\sup_{a_0 \in \s{A}_0}\left| Z_{\infty, h, b}(a_0)\right| \leq t\right) -  P_0 \left(\max_{a_0 \in \s{A}_n}\left| Z_{\infty, h, b}(a_0)\right| \leq t\right) \right| = \fasterthandet(1).\]
\end{lemma}
\begin{proof}[\bfseries{Proof of Lemma~\ref{lemma:finite_approx}}]
We write
\begin{align*}
\left| P_0 \left(\sup_{\s{A}_0}\left| Z_{\infty, h, b}\right| \leq t\right) -  P_0 \left(\max_{\s{A}_n}\left| Z_{\infty, h, b}\right| \leq t\right) \right| &= \left| E_0 \left[ I\left( \sup_{\s{A}_0}\left| Z_{\infty, h, b}\right| \leq t\right) -  I\left(\max_{\s{A}_n}\left| Z_{\infty, h, b}\right| \leq t\right) \right] \right|\\
&\leq E_0 \left[ \left| I\left( \sup_{\s{A}_0}\left| Z_{\infty, h, b}\right| \leq t\right) -  I\left(\max_{\s{A}_n}\left| Z_{\infty, h, b}\right| \leq t\right)\right| \right] \\
&\leq P_0 \left( \sup_{\s{A}_0}\left| Z_{\infty, h, b}\right| \leq t, \max_{\s{A}_n}\left| Z_{\infty, h, b}\right| >t \right) \\
&\qquad + P_0 \left( \sup_{\s{A}_0}\left| Z_{\infty, h, b}\right| > t, \max_{\s{A}_n}\left| Z_{\infty, h, b}\right| \leq t \right).
\end{align*}
We address the two probabilities in the final expression in the same way, so we only provide the derivation for the first term. For any $\varepsilon > 0$, we can write
\begin{align*}
\sup_{t \in \d{R}} P_0 \left( \sup_{\s{A}_0}\left| Z_{\infty, h, b}\right| \leq t, \max_{\s{A}_n}\left| Z_{\infty, h, b}\right| >t \right) &\leq \sup_{t \in \d{R}}P_0 \left( \sup_{\s{A}_0}\left| Z_{\infty, h, b}\right| \leq t - \varepsilon \left[\log h^{-1}\right]^{-1/2}, \max_{\s{A}_n}\left| Z_{\infty, h, b}\right| >t \right) \\
&\qquad + \sup_{t \in \d{R}}P_0 \left( t- \varepsilon\left[\log h^{-1}\right]^{-1/2} < \sup_{\s{A}_0}\left| Z_{\infty, h, b}\right| \leq t \right) \\
&\leq P_0 \left( \left|\sup_{\s{A}_0}\left| Z_{\infty, h, b}\right| - \max_{\s{A}_n}\left| Z_{\infty, h, b}\right| \right| >  \varepsilon\left[\log h^{-1}\right]^{-1/2} \right) \\
&\qquad +\sup_{t \in \d{R}} P_0 \left( \left| \sup_{\s{A}_0}\left| Z_{\infty, h, b}\right| - t\right| < \varepsilon\left[\log h^{-1}\right]^{-1/2} \right).
\end{align*}
By Lemma~\ref{lemma:anticoncentration}, the second term in the last inequality is bounded by $C \varepsilon + \fasterthandet(1)$. For the first term, we note that by the definition of the mesh $\omega_n$ of $\s{A}_n$, 
\[ \left|\sup_{\s{A}_0} | Z_{\infty,h,b}| - \max_{\s{A}_n} | Z_{\infty,h,b}| \right| \leq \sup_{|u-v| \leq \omega_n} \left|Z_{\infty,h,b}(u) - Z_{\infty,h,b}(v)\right| \]
almost surely. By Markov's inequality and Lemma~\ref{lm:subgaussian}, we then have
\begin{align*}
P_0 \left( \left|\sup_{\s{A}_0}\left| Z_{\infty, h, b}\right| - \max_{\s{A}_n}\left| Z_{\infty, h, b}\right| \right| >  \varepsilon\left[\log h^{-1}\right]^{-1/2} \right) &\leq \varepsilon^{-1} \left[\log h^{-1}\right]^{1/2} E_0 \left[\left|\sup_{\s{A}_0}\left| Z_{\infty, h, b}\right| - \max_{\s{A}_n}\left| Z_{\infty, h, b}\right| \right|\right]  \\
&\leq \varepsilon^{-1} \left[\log h^{-1}\right]^{1/2} E_0 \left[\sup_{|u-v| \leq \omega_n} \left|Z_{\infty,h,b}(u) - Z_{\infty,h,b}(v)\right|\right]  \\
&\lesssim \varepsilon^{-1} \left[\log h^{-1}\right]^{1/2}  h^{-1/2} ([2\omega_n] \wedge h)^{1/2} \left[ \log \frac{|\s{A}_0|}{2 ([2\omega_n] \wedge h)}\right]^{1/2} \\
&= \varepsilon^{-1}\left[\log h^{-1}  ([2\omega_n / h] \wedge 1) \log \frac{|\s{A}_0|}{2 ([2\omega_n] \wedge h)}\right]^{1/2}.
\end{align*}
Now $\omega_n = \fasterthandet\left(h^p \right)$ implies that $\omega_n / h \longrightarrow 0$, so for all $n$ large enough the above expression is equal to
\[ \varepsilon^{-1}\left[h^{-1} \log h^{-1}  \omega_n \log \frac{|\s{A}_0|}{4\omega_n}\right]^{1/2}.\]
This goes to zero by the assumed rate for $\omega_n$ for every $\varepsilon > 0$.
\end{proof}

\begin{lemma}\label{lemma:approx_process}
If~\ref{cond:bounded_K}--\ref{cond:uniform_entropy_nuisances} and~\ref{cond:unif_nuisance_rate}--\ref{cond:holder_smooth_theta}, $nh^5 = \boundeddet(1)$, and $m n^d = \boundeddet(1)$ for some $d \in (0, \infty)$, then
\begin{align*}
     \sup_{t \in \d{R}}\left| P_0 \left(\max_{a_0 \in \s{A}_n}\left| Z_{\infty, h, b}(a_0)\right| \leq t\right) -  P_0 \left(\max_{a_0 \in \s{A}_n}\left| Z_{n, h, b}(a_0)\right| \leq t \mid \b{O}_n\right) \right| = \fasterthan(1).
\end{align*}
\end{lemma}
\begin{proof}[\bfseries{Proof of Lemma~\ref{lemma:approx_process}}]
Given $O_1, \dotsc, O_n$, $Z_{n, h, b}$ is a Gaussian process on $\s{A}_n$. Hence, we can use the Gaussian comparison result of Theorem~2 of \cite{Chernozhukov_2014}. Let $\s{A}_n := \{a_1 < a_2 < \cdots < a_m\}$. In their notation, we then have $Y_j = Z_{\infty, h,b}(a_j)$ and $X_j = Z_{n,h,b}(a_j)$ for $1 \leq j \leq m$, and $Y_j = -Z_{\infty, h,b}(a_j)$ and $X_j = -Z_{n,h,b}(a_j)$ for $1+ m \leq j \leq 2m$. Hence, $\max_{1 \leq j \leq 2m} Y_j = \max_{a_0 \in \s{A}_n}\left| Z_{\infty, h, b}(a_0)\right|$ and $\max_{1 \leq j \leq 2m} X_j = \max_{a_0 \in \s{A}_n}\left| Z_{n, h, b}(a_0)\right|$. Hence, we have $p = 2m$. Since $Z_{\infty, h,b}$ and $Z_{n,h,b}$ are both normalized, we have $\sigma_{jj}^X = \sigma_{jj}^Y = 1$ for all $j$. By Lemma~\ref{lm:subgaussian}, we have
\[ a_p \leq E\left[ \sup_{a_0 \in \s{A}_0} \left| Z_{\infty,h,b}(a_0)\right| \right] = \boundeddet\left( \{\log h^{-1}\}^{1/2} \right).\]
We define 
\begin{align*}
    \Delta_n &:= \sup_{u, v \in \s{A}_0} \left| \d{P}_n (h \phi_{n,h,b,u}^* \phi_{n,h,b,v}^*) / [\sigma_{n,h,b}(u)\sigma_{n,h,b}(v)] - P_0 (h \phi_{\infty,h,b,u}^* \phi_{\infty,h,b,v}^*) / [\sigma_{\infty,h,b}(u)\sigma_{\infty,h,b}(v)]\right|. %\\
  %  &\leq \sup_{u, v \in \s{A}_0}\frac{ \left| \d{P}_n (h \phi_{n,h,b,u}^* \phi_{n,h,b,v}^*) - P_0 (h \phi_{\infty,h,b,u}^* \phi_{\infty,h,b,v}^*) \right|}{  \sigma_{\infty,h,b}(u)\sigma_{\infty,h,b}(v)} \\
   % &\qquad + \sup_{u, v \in \s{A}_0} \left| \d{P}_n (h \phi_{n,h,b,u}^* \phi_{n,h,b,v}^*) \right|\frac{ \left|\sigma_{\infty,h,b}(u)\sigma_{\infty,h,b}(v) - \sigma_{n,h,b}(u)\sigma_{n,h,b}(v)\right|}{\sigma_{\infty,h,b}(u)\sigma_{\infty,h,b}(v)\sigma_{n,h,b}(u)\sigma_{n,h,b}(v)} \\
  %  &= \fasterthan(n^{-p})
\end{align*}
Then by Theorem~2 of \cite{Chernozhukov_2014}, 
\begin{align*}
 \sup_{t \in \d{R}}\left| P_0 \left(\max_{a_0 \in \s{A}_n}\left| Z_{\infty, h, b}(a_0)\right| \leq t\right) -  P_0 \left(\max_{a_0 \in \s{A}_n}\left| Z_{n, h, b}(a_0)\right| \leq t \mid \b{O}_n\right) \right| &\leq \left[ \Delta_n \{\log (2m)\} \log \left(h^{-1} \vee \Delta_n^{-1}\right) \right]^{1/3}.
\end{align*}
By Lemmas~\ref{lemma:covar} and~\ref{lm:unif_bounded_variance}, $\Delta_n = \fasterthan(n^{-p})$ for some $p > 0$. Since $\log m \lesssim \log n$ by assumption, the right hand side of the preceding display is $\fasterthan(1)$.

% Finally, we prove that $Z_{\infty,h,b}$ is tight by showing almost all sample paths of $Z_{\infty,h,b}$ is uniformly continuous, which justifies the approximation of supremum via maximum over a finite set. This is implied if it follows that, for any sequence $s_n$ and $t_n$ in $\s{A}_0$ such that $|s_n - t_n| \longrightarrow 0$, we have $\left|Z_{\infty,h}(s_n)-Z_{\infty,h}(t_n)\right| \longrightarrow 0$ in probability (i.e., Page 41 of \cite{vandervaart1996}). We can prove this claim by showing $E\sup_{|s_n-t_n|\leq \delta}\left|Z_{\infty,h}(s_n)-Z_{\infty,h}(t_n)\right| \longrightarrow 0$. Since $Z_{\infty,h}$ is a Gaussian Process, it satisfies the following maxmimal inequality with an intrinsic standard deviation metric $\rho$;
% \begin{align}
%     E\sup_{|s-t|\leq \delta}\left|Z_{\infty,h}(s)-Z_{\infty,h}(t)\right| &\lesssim \int_0^\infty \sqrt{\log N(\varepsilon, \s{S}, \rho)}\, d\varepsilon\nonumber\\
%     &\leq \sup_{|s-t|\leq \delta}\rho\left(Z_{\infty,h}(s), Z_{\infty,h}(t)\right)\int_0^{|\s{S}|} \sqrt{\log N(\varepsilon, \s{S}, |\cdot|)}\, d\varepsilon\label{eq:chaining_change_of_var}
% \end{align}
% As shown in the proof of Lemma~\ref{lm:subgaussian}, $\rho\left(Z_{\infty,h}(s), Z_{\infty,h}(t)\right)=\bounded(|s-t|)$ assuming \ref{cond:uniform_entropy_nuisances}--\ref{cond:cont_cond_var} and~\ref{cond:lipschitz_bounded_f0_sigma0}, and thus $\fasterthan(1)$ when $|s-t| \longrightarrow 0$. Further, the second term of equation \eqref{eq:chaining_change_of_var} is finite since $\s{S}$ is compact. 
\end{proof}

%% file: supp/additional_sim.tex
\clearpage
\section{Additional results from numerical studies}\label{supp:numerical}

This section presents additional results from the numerical study described in Section 4 of the main text. \rev{In addition to data adaptive nuisance estimators, we consider the following parametric estimators. First, we estimated $\mu_0$ using a correctly specified logistic regression model to obtain estimators of $\gamma_1$, $\gamma_2$, $\gamma_3$ and $\gamma_4$. We estimated $g_0$ using maximum likelihood estimation with a correctly specified parametric model to obtain an estimator of $\beta$. To investigate the double-robustness of our estimator, we also estimated $\mu_0$ and $g_0$ using incorrectly specified parametric models. For $\mu_0$, we used the incorrectly specified logistic regression model that assumes $\mu_0(a, w) = \n{expit}\left(\tilde\gamma_1^T \tilde{w} + \tilde\gamma_2a \right)$ for some $\tilde\gamma_1 \in \d{R}^3$, $\tilde\gamma_2 \in \d{R}$, where $\tilde{w} := (1,w_1,w_2)$. For $g_0$, we used the incorrectly specified linear regression model that assumes that given $W = w$, $A$ follows a normal distribution with mean $ \tilde{\gamma_3}^T \tilde{w}$ for some $\tilde\gamma_3 \in \d{R}^3$ and constant variance}.

\rev{Figures~\ref{fig:bias_variance_500_2500} and \ref{fig:bias_variance_500_2500_ml} display} the pointwise empirical bias and variance for sample sizes $n=500, 1000$, and $2500$. \rev{Figure~\ref{fig:bias_variance_500_2500} corresponds to the nuisance parameter estimation based on parametric models and Figure \ref{fig:bias_variance_500_2500_ml} corresponds to the nuisance parameter estimation based on SuperLearner.} The squared bias of the local linear estimators remains large for $n=2500$ \rev{unless undersmoothing is employed}. The variance of the local linear estimator, on the other hand, is smaller than that of the debiased estimator compared with the corresponding bandwidth selection procedure. \rev{The variance of undersmoothed local linear estimator is comparable with the debiased estimator using LOOCV, and it increases faster than the optimal rate $n^{4/5}$.} The variance of all estimators is generally larger when the outcome regression model is misspecified. \rev{The conclusion is similar when using parametric and SuperLearner-based nuisance estimators.} 

Figure~\ref{fig:mse_500_2500_ml} displays the pointwise empirical mean squared error (MSE) of the estimators for $n = 500$, $1000$ and $2500$. \rev{The results include both parametric model-based and SuperLearner-based nuisance estimators.} For all sample sizes, the debiased local linear estimator with the plug-in bandwidth selection procedure attains the smallest MSE for most interior points we considered.  The local linear estimator using plug-in bandwidth selection attains a small MSE at points where the second derivative of $\theta_0$ is close to zero, but for points where the second derivative is far from zero, it has a larger MSE due to its bias. For the debiased estimator, the LOOCV bandwidth selection yields similar MSEs regardless of optimizing over both $b$ and $h$ or just over $h$ with $b = h$ fixed. \rev{When $n=500$, the undersmoothed local linear estimator displays a slightly smaller MSE than the debiased estimator with LOOCV bandwidth selection; however, as $n$ increases, the MSE grows, indicating suboptimal convergence rate of MSE when undersmoothing is employed. The undersmoothing technique consistently yields larger MSEs compared to the debiased local linear estimator with the plug-in bandwidth selection for all interior points and sample sizes considered. All estimators display larger MSE towards the boundary points.}

Figure~\ref{fig:pointwise_coverage_500_2500} displays the empirical coverage of pointwise 95\% confidence intervals for sample sizes $n=500, 1000$ and $n = 2500$. The coverage of confidence intervals based on the local linear estimator does not improve when the sample size is larger. On the other hand, the coverage of the confidence intervals based on the debiased estimator are slightly lower for $n = 500$, and generally very close to 95\% when $n = 2500$. The coverage accuracy is particularly good when at least one of the nuisance estimators is based on a correctly specified parametric model. On the other hand, we observe slightly worse coverage near the boundary when only conditional density is correctly specified and is based on SuperLearner.

\rev{Figure~\ref{fig:pointwise_ci_width_supp} displays the median length of pointwise 95\% confidence intervals for two sample sizes, $n=500, 1000$ and $n=2500$. We observe that the widths of the confidence intervals are generally comparable between local linear estimators and the debiased estimators when the plug-in method is used, indicating that the bias correction has a relative minor impact on confidence interval widths. When undersmoothing is used, the confidence intervals widen and align with the debiased local linear estimators obtained through LOOCV bandwidth selection procedures. Comparatively, the confidence intervals based on the debiased local linear estimator with the plug-in method exhibit narrower widths than those of the local linear estimator with undersmoothing. This gap widens as the sample size $n$ increases, as undersmoothing results in a suboptimal convergence rate.}

Figures~\ref{fig:pairwise_coverage_500_2500} \rev{and ~\ref{fig:pairwise_coverage_500_2500_ml} display} the empirical coverage of confidence intervals for the causal effect $\theta_0(a) - \theta_0(0.5)$ based for sample sizes $n=500, 1000$ and $n = 2500$. The conclusion is similar to the case with $n=1000$ from the main text; the confidence intervals based on the asymptotic independence are conservative when the distance between the evaluation points is small. This is still a problem at $n=2500$. The influence function-based confidence intervals perform particularly well \rev{when nuisance function is estimated based on parametric models as seen in Figure~\ref{fig:pairwise_coverage_500_2500}. The conclusion for SuperLearner-based methods are similar except that the influence function-based confidence intervals based on the plug-in bandwidth selection overcover when the outcome regression is misspecified.}

\rev{Figures~\ref{fig:uniform_coverage-supp} displays the empirical coverage of the uniform confidence bands over $\mathcal{A}_0 = [0,1.0]$ obtained from nuisance estimators based on parametric models and SuperLearner. The coverage accuracy of uniform bands does not show a significant difference between parametric and data-adaptive nuisance estimators. The plug-in bandwidth selection performs well for sample sizes greater than $1000$ while methods based on cross-validation show slight undercoverage for large sample sizes.}

Figure~\ref{fig:bandwidth_distribution} displays distributions of the bandwidth $h$ selected by the \rev{five} procedures. We note that the bandwidth for the plug-in method is the same for the local linear and debiased estimators. The plug-in method generally selects larger bandwidths than the LOOCV methods. This explains the larger bias and smaller variance of the estimators based on the plug-in method. The distribution of $h$ for the LOOCV method that selects both $h$ and $b$ has the largest variance. Figure~\ref{fig:tau_distribution} shows the densities of $\tau = h/b$ for the LOOCV method that selects both $h$ and $b$. Surprisingly, the procedure seems to favor smaller $\tau$ than one. This explains the larger variance of the LOOCV method that selects both $h$ and $b$ relative to the LOOCV method that fixes $b = h$.

\rev{Finally, Figure~\ref{fig:true-dose-response} displays the true covariate-adjusted regression function used across the numerical study and its second derivative. The location that corresponds to a large second derivative in its absolute value coincides with the area where local linear estimators demonstrate larger bias and poor coverage.}
\setcounter{figure}{6}
\begin{figure}[ht]
\centering
\begin{subfigure}{\textwidth}
  \centering
  \includegraphics[width=1\linewidth]{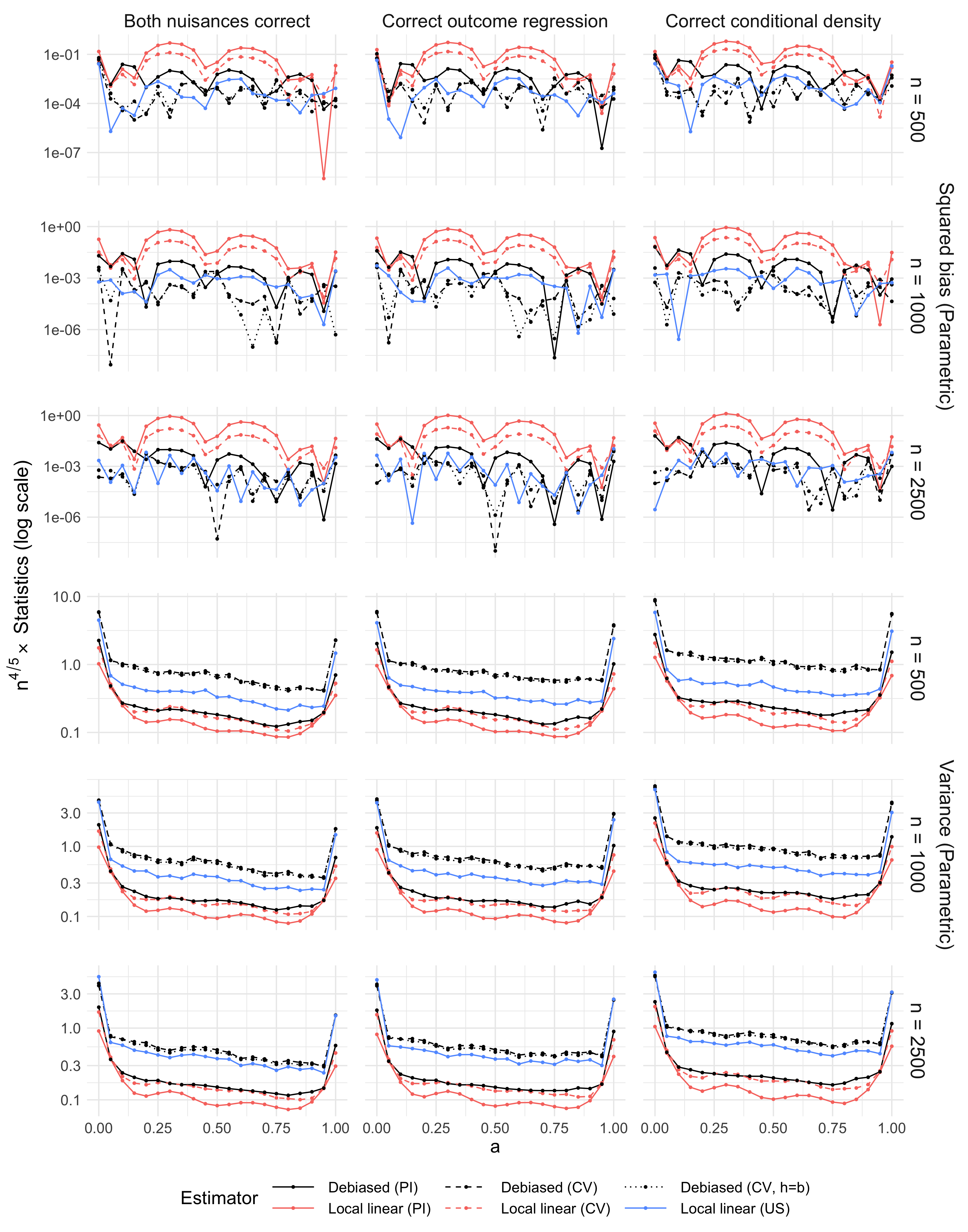}
\end{subfigure}
\caption{Empirical squared bias and variance of the estimators from sample sizes $n=500, 1000$ and $2500$ \rev{when parametric models are used for estimating nuisance functions.} The values are scaled by $n^{4/5}$ and displayed on the log scale.}
\label{fig:bias_variance_500_2500}
\end{figure}

\begin{figure}[ht]
\centering
\begin{subfigure}{\textwidth}
  \centering
  \includegraphics[width=1\linewidth]
  {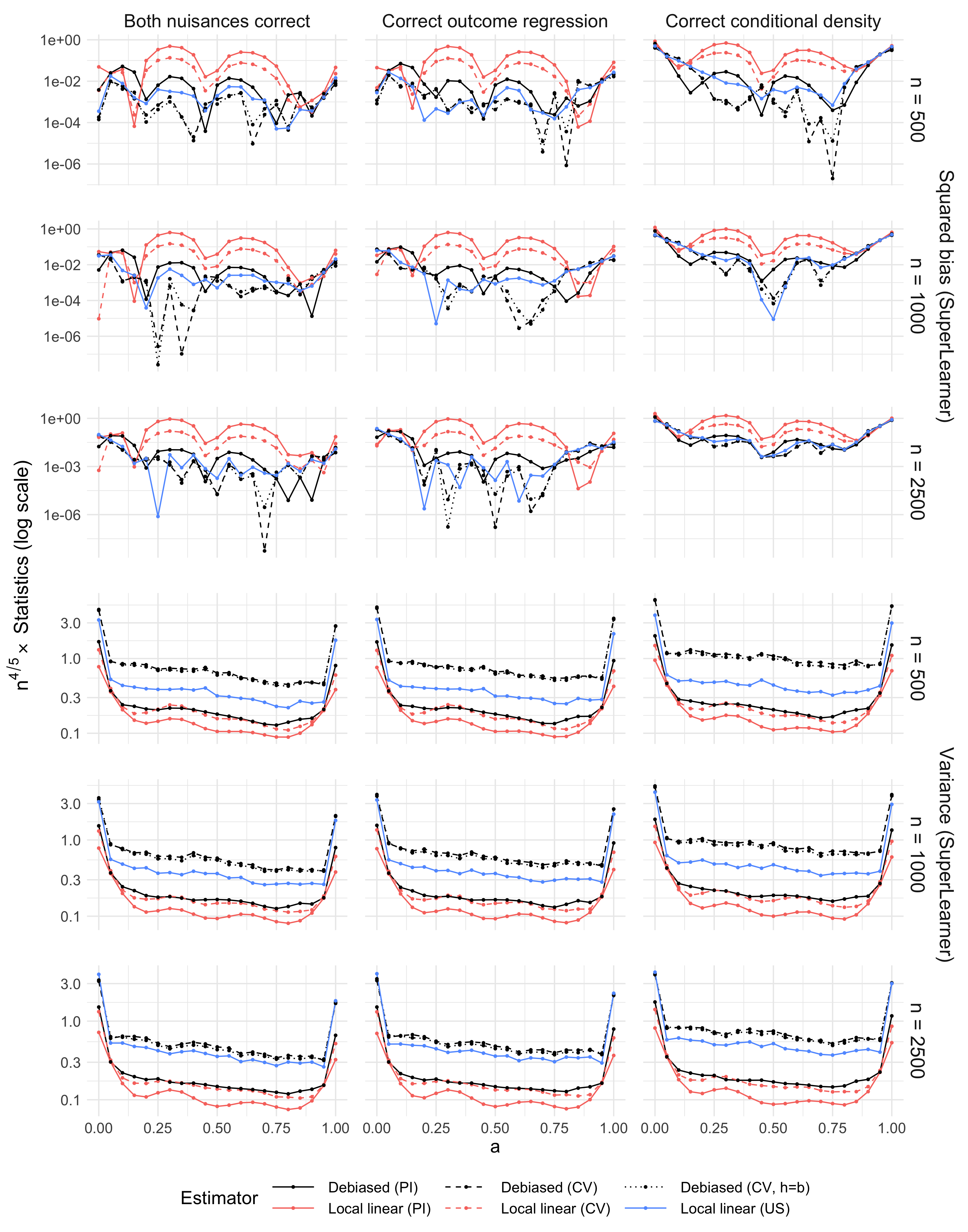}
\end{subfigure}
\caption{\rev{Empirical squared bias and variance of the estimators from sample sizes $n=500, 1000$ and $2500$ when SuperLearner-based methods are used for estimating nuisance functions. The values are scaled by $n^{4/5}$ and displayed on the log scale.}}
\label{fig:bias_variance_500_2500_ml}
\end{figure}

\begin{figure}[ht]
\centering
\begin{subfigure}{\textwidth}
  \centering
  \includegraphics[width=1\linewidth]{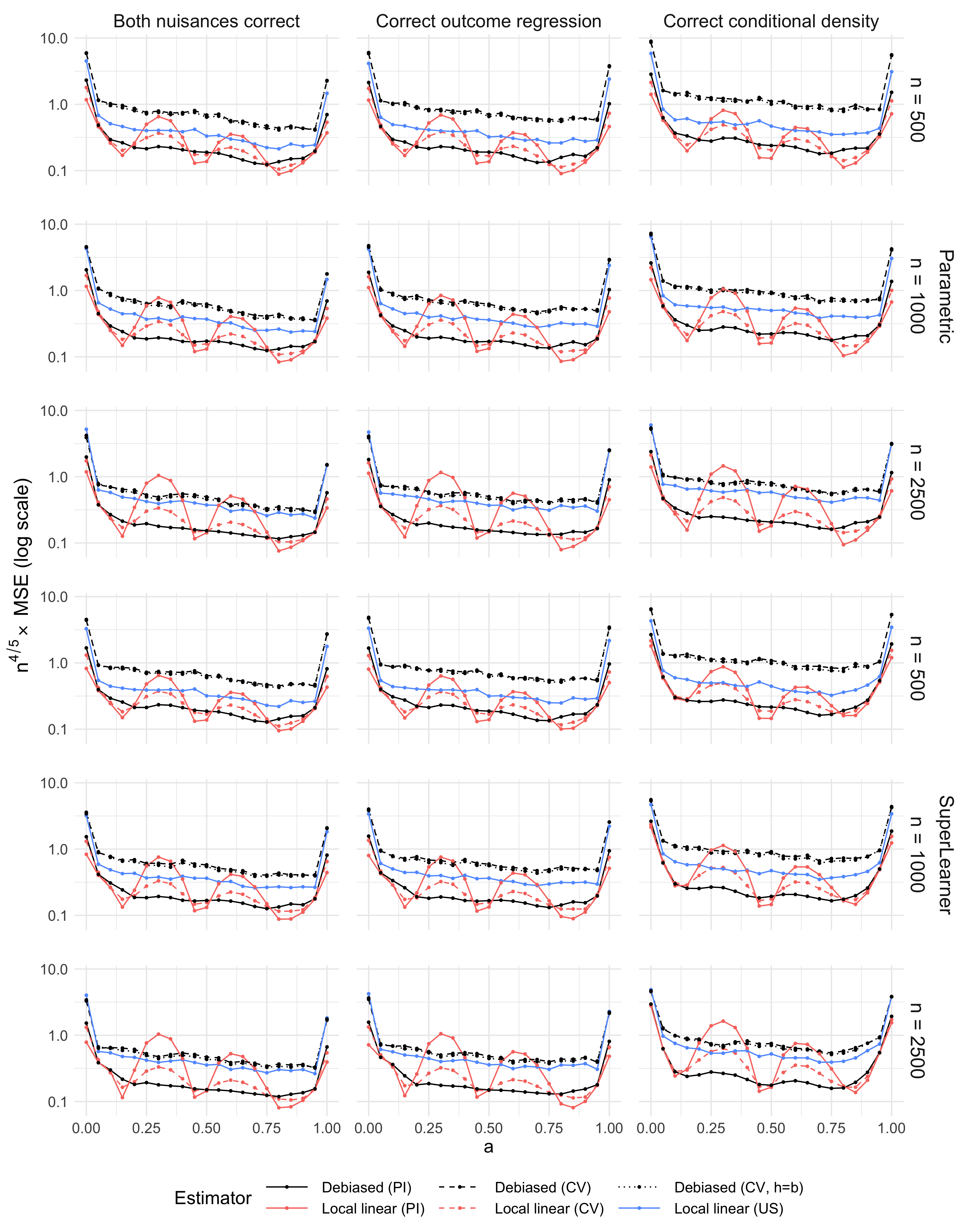}
\end{subfigure}
\caption{Empirical mean squared error of the estimators \rev{based on both parametric and SuperLearner estimators for nuisance functions}. The values are scaled by $n^{4/5}$ and displayed on the log scale.}
\label{fig:mse_500_2500_ml}
\end{figure}

\begin{figure}[ht]
\centering
\begin{subfigure}{\textwidth}
  \centering
  \includegraphics[width=1\linewidth]{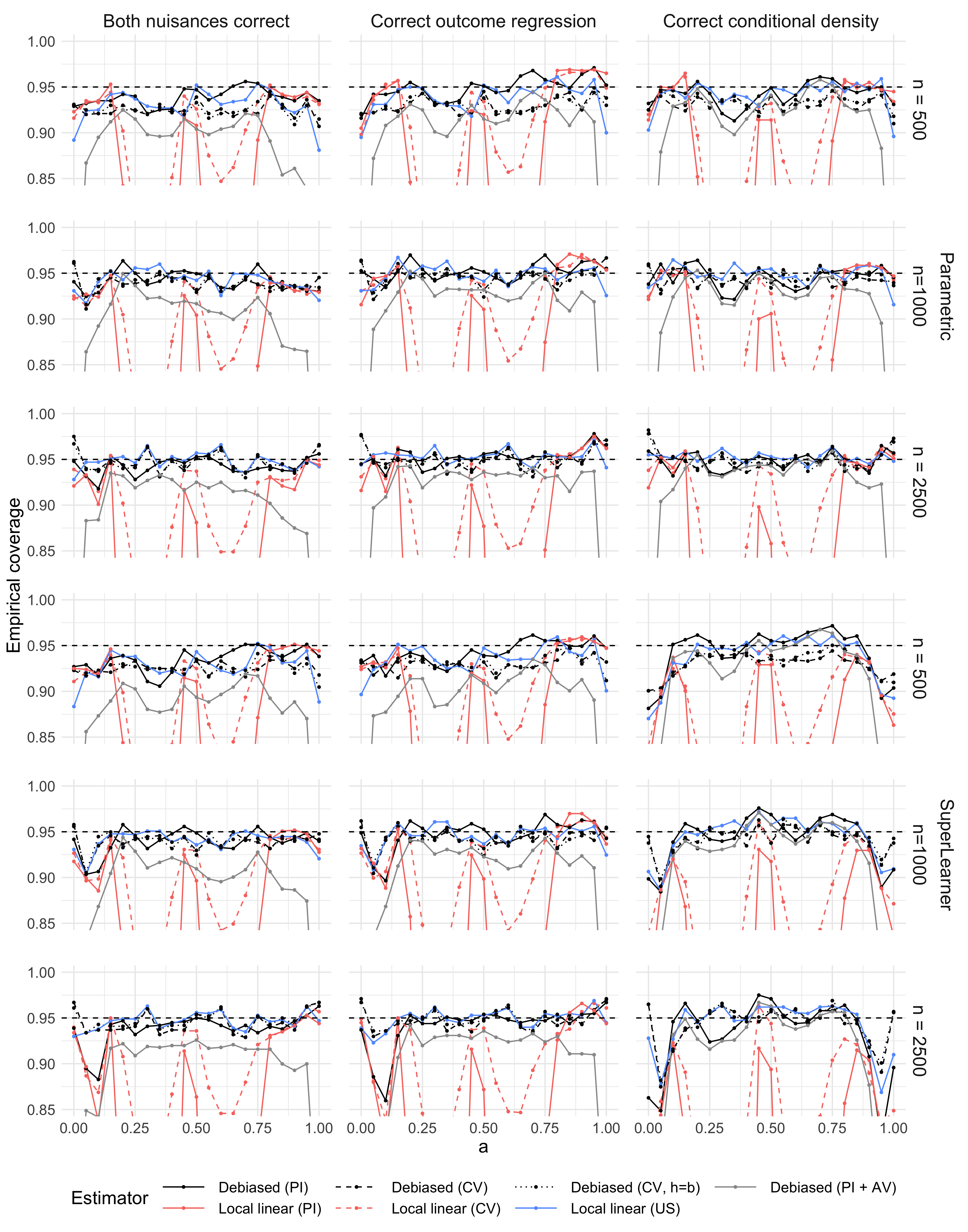}
  %\caption{}
\end{subfigure}
\caption{Empirical coverage of 95\% pointwise confidence intervals based on the debiased local linear estimator and the local linear estimator for sample sizes $n=500$ and $2500$.}
\label{fig:pointwise_coverage_500_2500}
\end{figure}

\begin{figure}[ht]
\centering
\begin{subfigure}{\textwidth}
  \centering
  \includegraphics[width=1\linewidth]{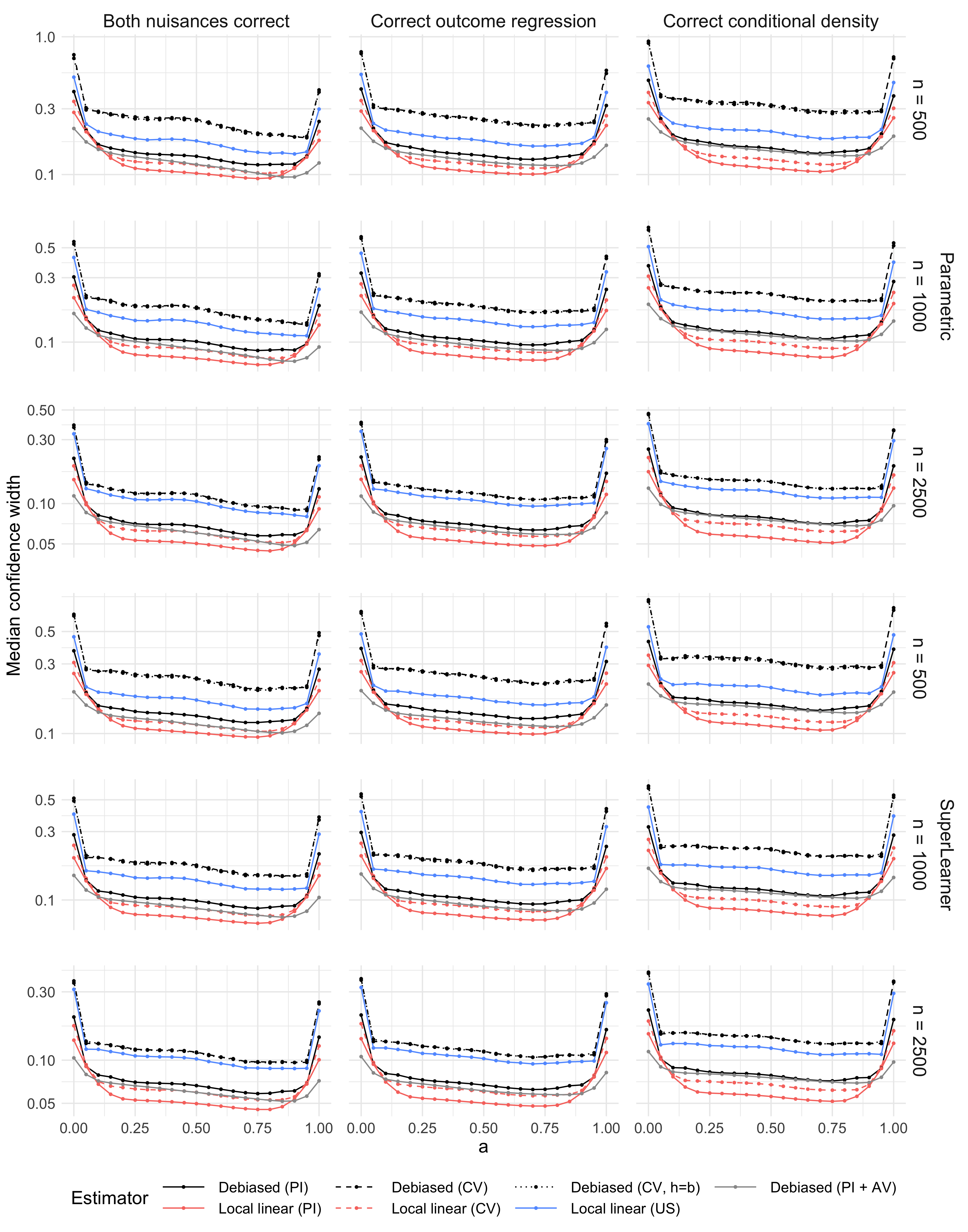}
  %\caption{}
\end{subfigure}
\caption{\rev{Median width of 95\% pointwise confidence intervals based on the debiased local linear estimator and the local linear estimator with sample sizes $n=500, 1000$ and $2500$.}}
\label{fig:pointwise_ci_width_supp}
\end{figure}

% \begin{figure}[ht]
% \centering
% \begin{subfigure}{\textwidth}
%   \centering
%   \includegraphics[width=1\linewidth]{figs/pts_coverage_compare_with_plugin_ml_500-rev.png}
%   %\caption{}
% \end{subfigure}
% \caption{}
% \end{figure}

\begin{figure}[ht]
\centering
\begin{subfigure}{\textwidth}
  \centering
  \includegraphics[width=1\linewidth]{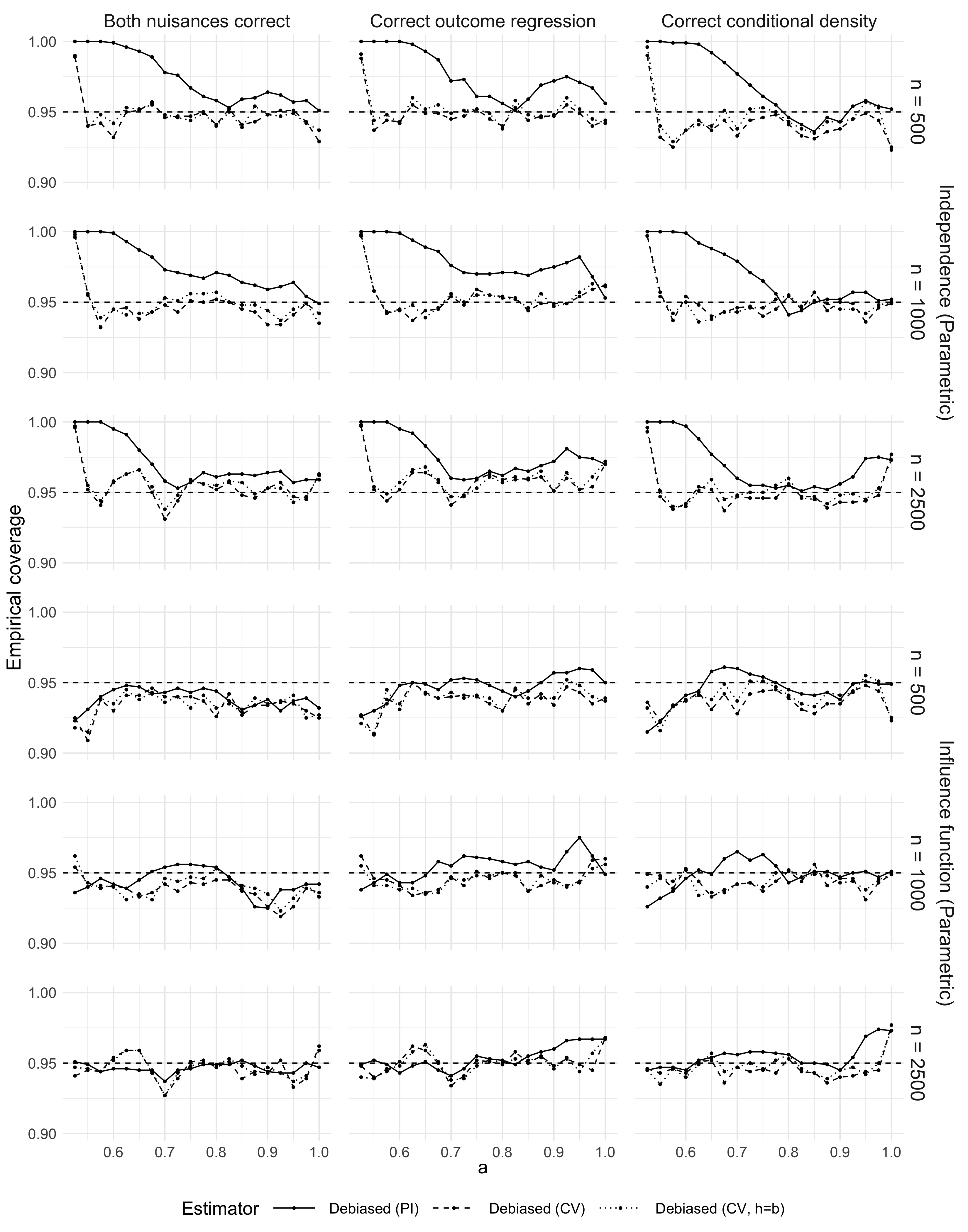}
  % \caption{Empirical coverage of 95\% pointwise confidence intervals when the sample size is $n=2500$.}
\end{subfigure}
\caption{Empirical coverage of 95\% pointwise confidence intervals for $\theta_0(a) - \theta_0(0.5)$ based on the debiased estimator when parametric methods are used for nuisance function estimation. The intervals in the top two rows use the sum of two variance estimators; those in the bottom rows use the influence function-based variance estimator. }
\label{fig:pairwise_coverage_500_2500}
\end{figure}

\begin{figure}[ht]
\centering
\begin{subfigure}{\textwidth}
  \centering
  \includegraphics[width=1\linewidth]{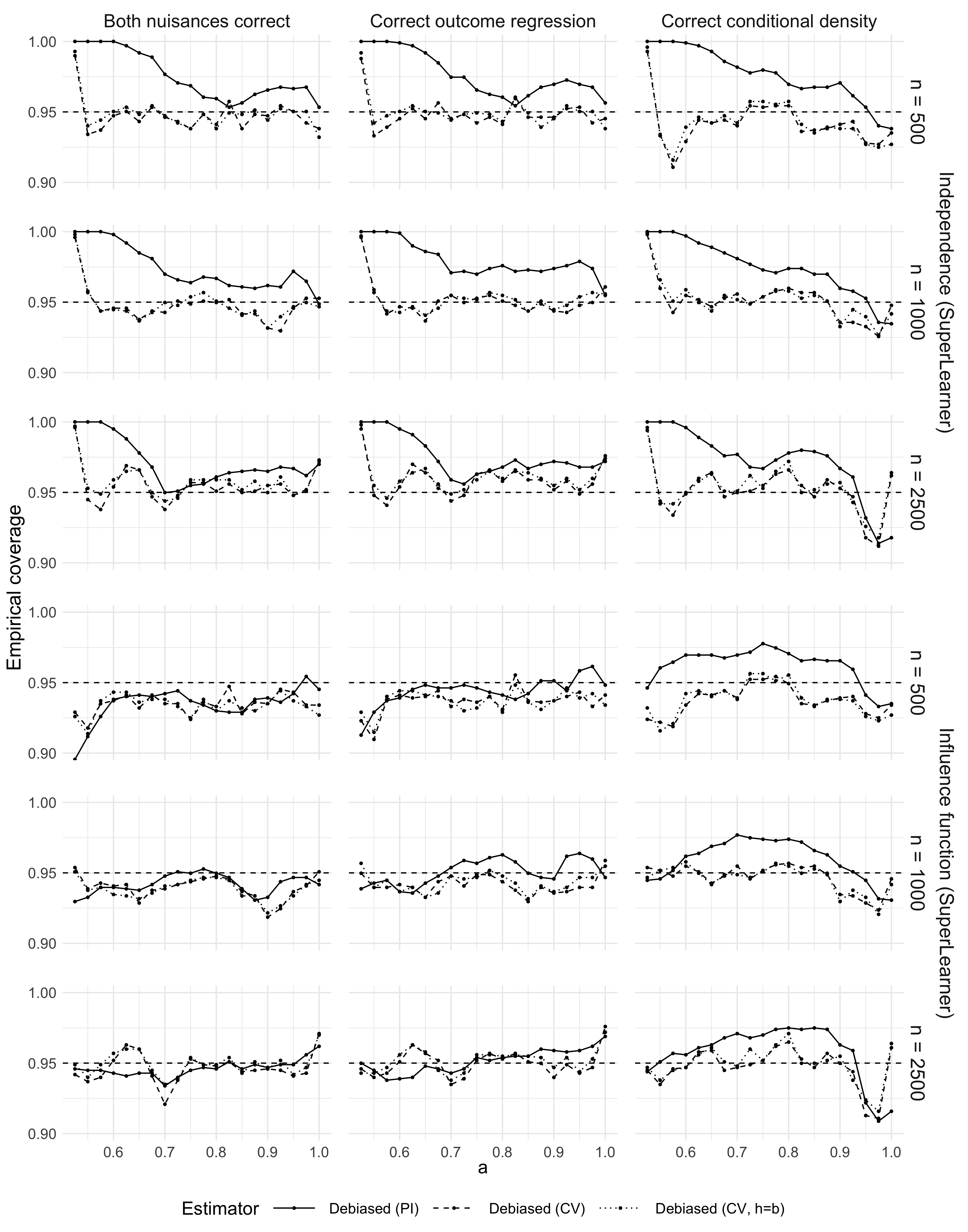}
  % \caption{Empirical coverage of 95\% pointwise confidence intervals when the sample size is $n=2500$.}
\end{subfigure}
\caption{Empirical coverage of 95\% pointwise confidence intervals for $\theta_0(a) - \theta_0(0.5)$ based on the debiased estimator when SuperLearner-based methods are used for nuisance function estimation. The intervals in the top two rows use the sum of two variance estimators; those in the bottom rows use the influence function-based variance estimator. }
\label{fig:pairwise_coverage_500_2500_ml}
\end{figure}

\begin{figure}[ht]
    \centering
    \includegraphics[width=6.5in]{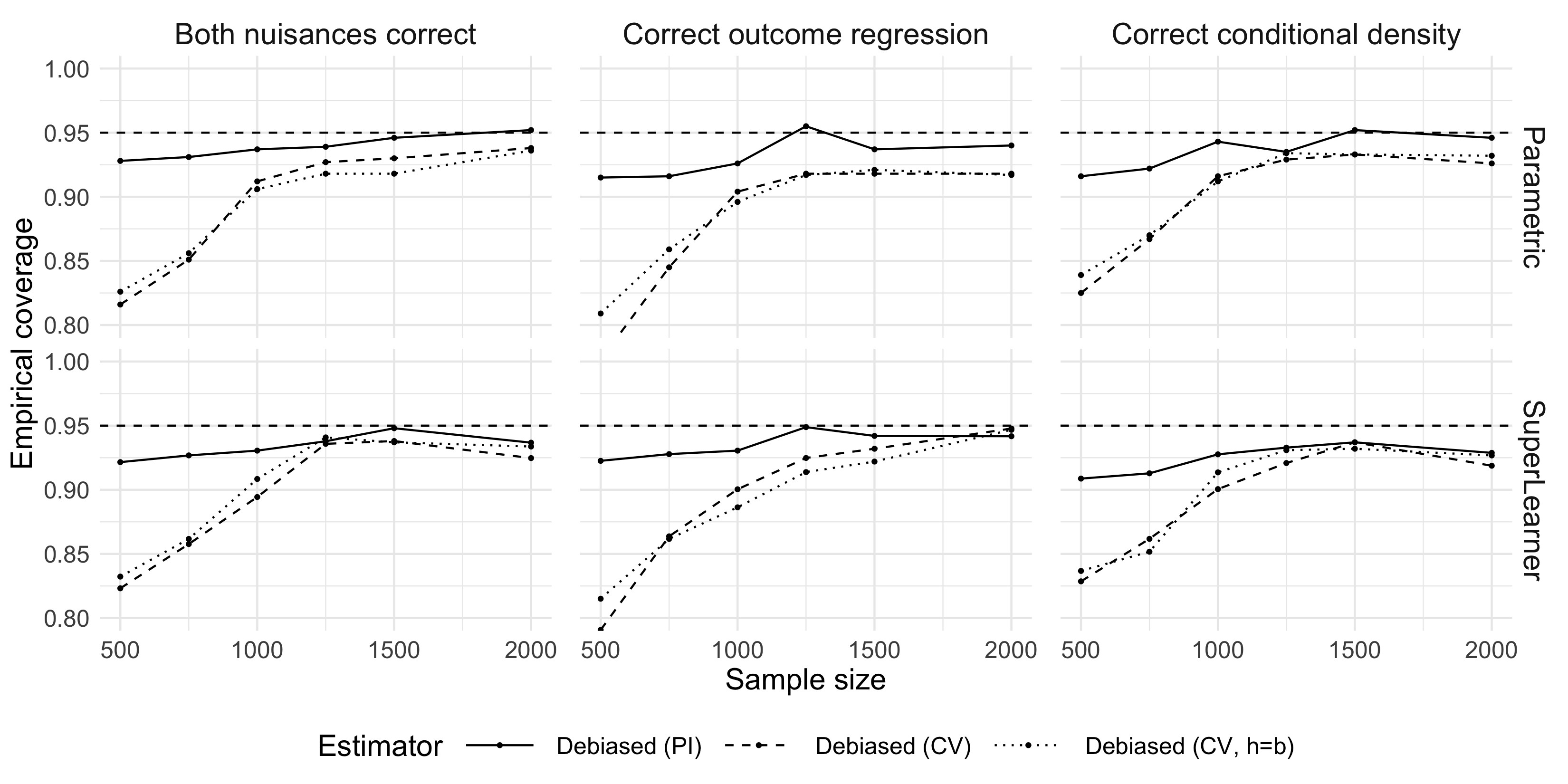}
    \caption{Empirical coverage of $95\%$ uniform confidence bands based on the debiased estimator over $\mathcal{A}_0 = [0,1]$.}
    \label{fig:uniform_coverage-supp}
\end{figure}

% \begin{figure}[ht]
%     \centering
%     \includegraphics[width=6.5in]{figs/unif_coverage_comparison-rev.png}
%     \caption{Empirical coverage of $95\%$ uniform confidence bands based on the debiased estimator over $\mathcal{A}_0 = [0.1, 0.9]$.}
%     \label{fig:uniform_coverage-supp2}
% \end{figure}

\begin{figure}[ht]
  \centering
  \includegraphics[width=1\linewidth]{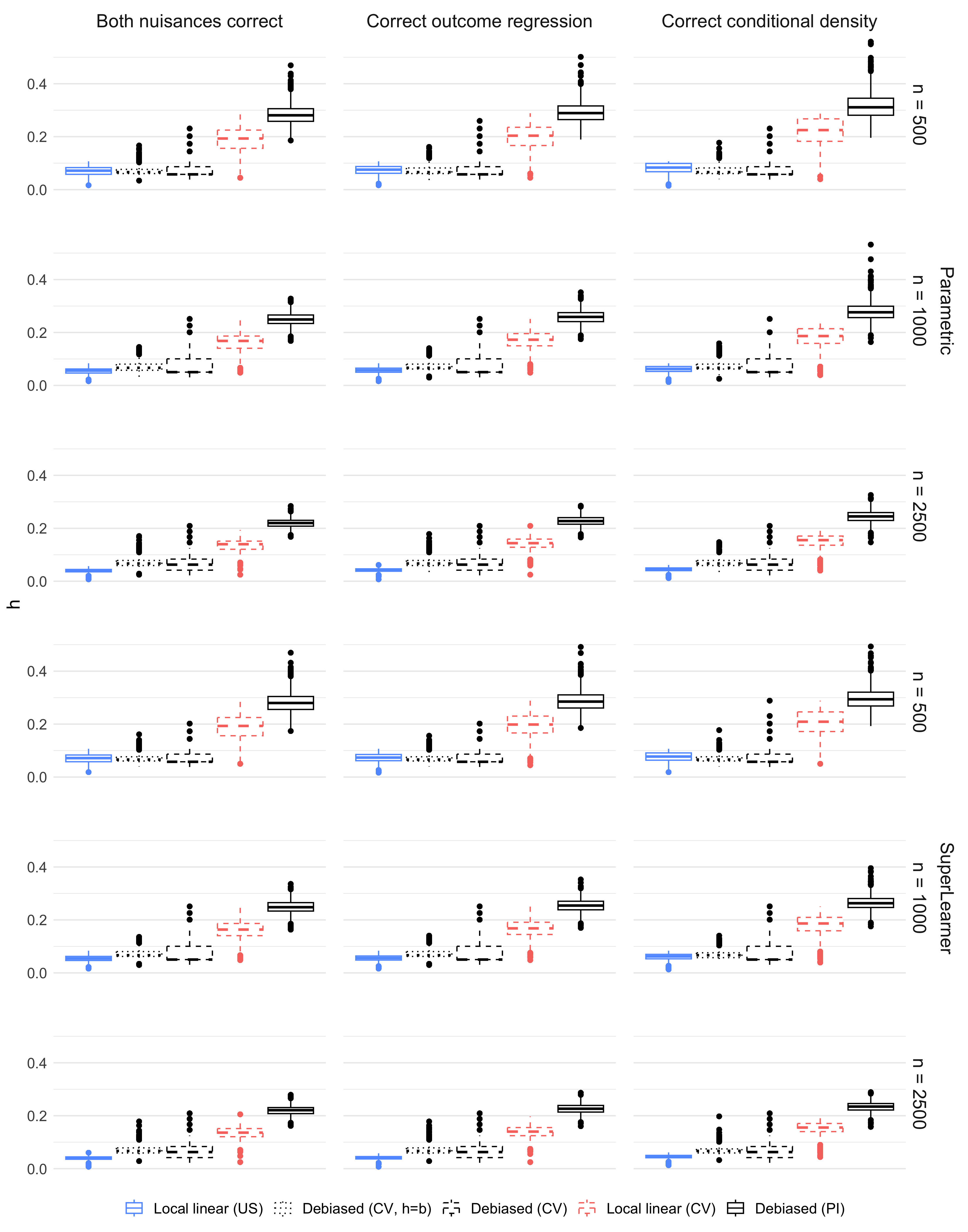}
  \caption{Distribution of the bandwidth $h$ selected by the \rev{five} different procedures over the 1000 simulations in each setting and for each sample size.}
\label{fig:bandwidth_distribution}
\end{figure}

\begin{figure}[ht]
  \centering
  \includegraphics[width=1\linewidth]{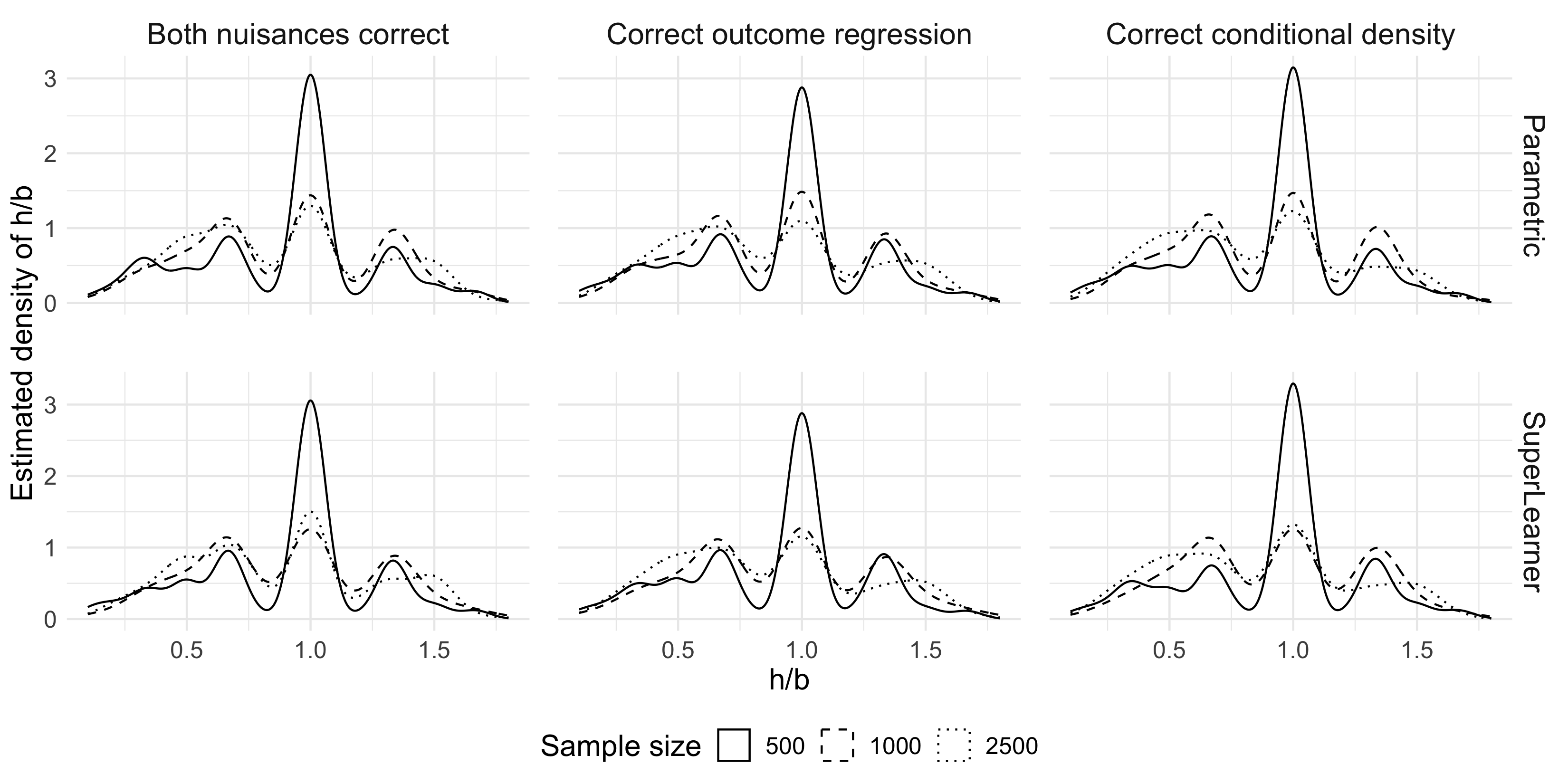}
  \caption{Density of $h/b$ found by the LOOCV method that selects both $h$ and $b$ for the debiased estimator.}
\label{fig:tau_distribution}
\end{figure}

\begin{figure}
\centering
\begin{subfigure}{3.25in}
  \centering
  \includegraphics[width=2.9in]{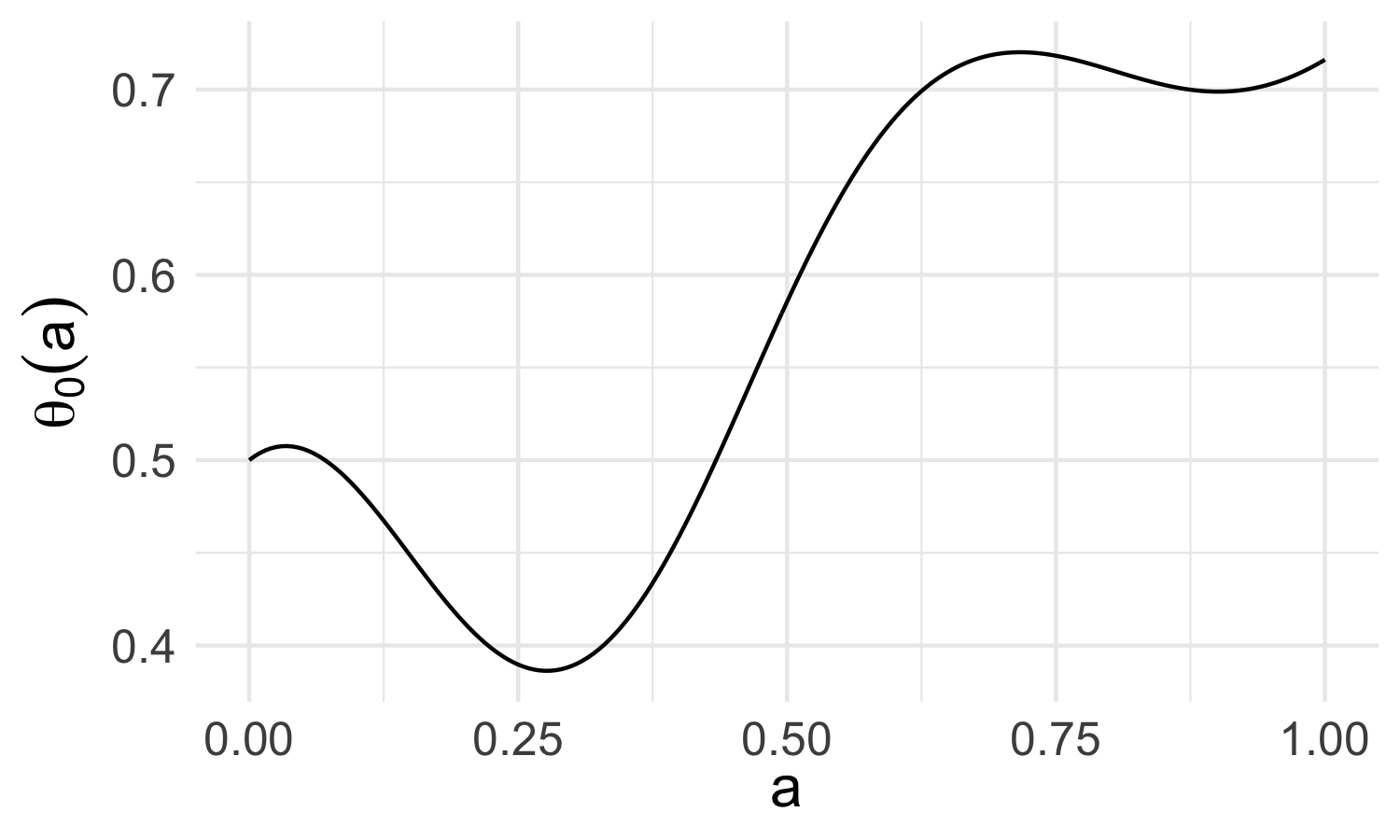}
 % \caption{True parameter $\theta_0$ as a function of $a$}
  \label{fig:dose-response}
\end{subfigure}%
\begin{subfigure}{3.25in}
  \centering
  \includegraphics[width=2.9in]{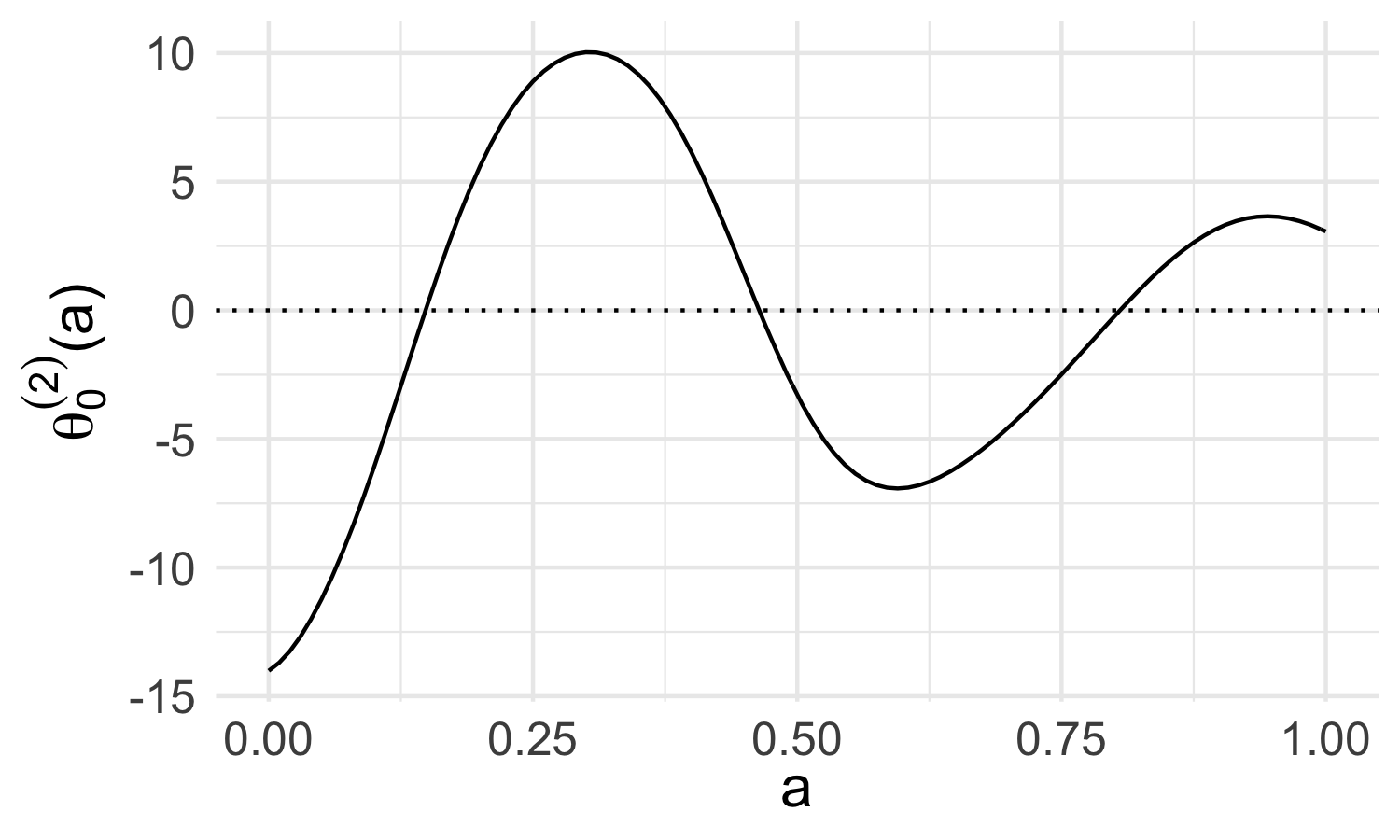}
  %\caption{The second derivative of $\theta_0$ as a function of $a$}
  \label{fig:dose-response-deriv}
\end{subfigure}
\caption{The true function $\theta_0$ and its second derivative $\theta_0''$ used in the numerical experiements.}
\label{fig:true-dose-response}
\end{figure}